\documentclass[DIV=12,10pt,index=totoc,toc=bibliography]{scrartcl}

\usepackage{adjustbox}
\usepackage[section]{algorithm}
\usepackage[noend]{algpseudocode}
\usepackage{amsmath}
\usepackage{amssymb}
\usepackage{amsthm}
\usepackage[english]{babel}
\usepackage{booktabs}
\usepackage[
    style=alphabetic,
    giveninits=true,
    maxnames=42,
    maxalphanames=42,
    natbib,
]{biblatex}
\usepackage{csquotes}
\usepackage[inline]{enumitem}
\usepackage{hyphenat}
\usepackage[utf8]{inputenc}
\usepackage{longtable}
\usepackage{makeidx}
\usepackage{mathtools}
\usepackage{stackengine}
\usepackage{stmaryrd}
\usepackage[
    subrefformat=simple,
    labelformat=simple,
]{subcaption}
\usepackage{tabularx}
\usepackage{thmtools}
\usepackage{thm-restate}
\usepackage[
    backgroundcolor=orange!5,
    bordercolor=orange,
    linecolor=orange,
    textsize=footnotesize,
]{todonotes}
\usepackage{tikz}
\usepackage{isomath}
\usepackage[bb=boondox]{mathalfa}
\DeclareMathSymbol{:}{\mathpunct}{operators}{"3A}

\usepackage{tcolorbox}
\usepackage{tocstyle}
\usepackage{xstring}

\usetikzlibrary{external}
\usepackage[hidelinks]{hyperref}

\usepackage[indexonlyfirst,nopostdot,toc]{glossaries}
\usepackage{glossary-longbooktabs}
\newglossarystyle{tabgroup}{%
\setglossarystyle{long3col-booktabs}%
\renewcommand\glsdescwidth{0.5\hsize}%
{\begin{longtable}{lp{\glsdescwidth}r}}%
{\end{longtable}}%
}
\setglossarystyle{tabgroup}

\newcommand{\alg}[1]{\mathfrak{#1}}
\newcommand{\lalg}[1]{\mathrm{#1}}
\newcommand{\walg}[1]{\mathbb{#1}}
\newcommand{\gclass}[1]{\mathcal G_{\lalg{#1}}}
\newcommand{\wclass}[1]{\mathcal K_{#1}}
\newcommand{\welem}[1]{\mathbb{#1}}
\newcommand{\sem}[2][\alg A]{#2_{#1}}

\newcommand{\cladp}{\mathrm{ADP}}
\newcommand{\cfges}{\lalg{CFG}^\emptyset}
\newcommand{\AST}{\mathrm{AST}}

\newcommand{\findc}{\operatorname{fin-dc}}
\newcommand{\dcomp}{\operatorname{d-comp}}
\newcommand{\finidpo}{fin, id, \preceq}
\newcommand{\bd}{\walg{BD}}
\newcommand{\constrees}{\mathscr H}
\newcommand{\clsd}{\mathrm{closed}}

\theoremstyle{definition}
\declaretheorem[numberwithin=section,qed=$\Box$]{definition}
\declaretheorem[qed=$\Box$,sibling=definition]{example}
\declaretheorem[sibling=definition]{theorem}
\declaretheorem[sibling=definition,name=Theorem]{boxtheorem}
\tcolorboxenvironment{boxtheorem}{colback=white,left=1mm,right=1mm,top=1mm,bottom=1mm}
\declaretheorem[sibling=definition]{observation}
\declaretheorem[sibling=definition]{lemma}
\declaretheorem[numbered=no,name=Lemma]{lemma*}

\declaretheorem[sibling=definition,name=Corollary]{boxcorollary}
\tcolorboxenvironment{boxcorollary}{colback=white,left=1mm,right=1mm,top=1mm,bottom=1mm}
\renewcommand{\qedsymbol}{$\blacksquare$}

\algnewcommand\algorithmicforeach{\textbf{for each}}
\algnewcommand\algorithmicvariables{\textbf{Variables:}}
\algnewcommand\Variables{\item[\algorithmicvariables]}
\algdef{S}[FOR]{ForEach}[1]{\algorithmicforeach\ #1\ \algorithmicdo}
\algrenewcommand\algorithmicrequire{\textbf{Input:}}
\algrenewcommand\algorithmicensure{\textbf{Output:}}
\algnewcommand\Fixedcomment[1]{\hfill\makebox[0.4\textwidth][l]{$\triangleright$ #1}}

\setlist{%
    topsep=0.3em,
    partopsep=0.5em,
    parsep=0.15em,
    itemsep=0.15em,
}
\setlist[enumerate,1]{label=(\roman*)}

\usetikzlibrary{backgrounds}
\usetikzlibrary{calc}
\usetikzlibrary{graphs}
\usetikzlibrary{quotes}
\usetikzlibrary{tikzmark}
\makeatletter
\@ifpackageloaded{biblatex}{\renewrobustcmd*{\bibinitdelim}{\,}}
\makeatother

\DeclareMathOperator{\id}{id}
\DeclareMathOperator*{\inford}{inf_\preceq}
\DeclareMathOperator*{\supord}{sup_\preceq}
\DeclareMathOperator*{\minord}{min_\preceq}
\DeclareMathOperator*{\maxord}{max_\preceq}
\DeclareMathOperator*{\argminord}{arg\,min_\preceq}
\DeclareMathOperator*{\argmaxord}{arg\,max_\preceq}
\DeclareMathOperator*{\argmax}{arg\,max}
\DeclareMathOperator*{\argmin}{arg\,min}
\DeclareMathOperator{\sort}{sort}
\DeclareMathOperator{\rk}{rk}
\makeatletter
\DeclareRobustCommand {\T}{\qopname \newmcodes@ o{T}}
\makeatother
\DeclareMathOperator{\pos}{pos}
\DeclareMathOperator{\yield}{yield}
\newcommand{\infsum}[1][\oplus]{\mathop{\vphantom{\sum}\mathchoice
  {\vcenter{\hbox{\Large $\sum^{#1}$}}}
  {\vcenter{\hbox{\normalsize $\sum^{#1}$}}}{\sum^{#1}}{\sum^{#1}}}\displaylimits}
\newcommand{\infsumop}[1][\oplus]{\mathop{\vphantom{\sum}\mathchoice
  {\vcenter{\hbox{\normalsize $\sum^{#1}$}}}
  {\vcenter{\hbox{\normalsize $\sum^{#1}$}}}{\sum^{#1}}{\sum^{#1}}}\displaylimits}
\DeclareMathOperator{\mul}{mul}
\DeclareMathOperator{\maxrk}{maxrk}
\DeclareMathOperator{\wt}{wt}
\DeclareMathOperator{\fparse}{parse}

\DeclareMathOperator{\decodes}{decodes}

\DeclareMathOperator{\med}{med}
\DeclareMathOperator{\factors}{factors}
\DeclareMathOperator{\adp}{adp}

\DeclareMathOperator{\wds}{wds}
\DeclareMathOperator{\cnc}{cwds}
\DeclareMathOperator{\ffint}{int}
\DeclareMathOperator{\trg}{trg}
\DeclareMathOperator{\lhs}{lhs}
\DeclareMathOperator{\rhs}{rhs}
\DeclareMathOperator{\height}{height}

\DeclareMathOperator{\cotrees}{cutout}
\DeclareMathOperator*{\maxv}{max_{\mathbb{BD}}}
\DeclareMathOperator{\Omegav}{\mathnormal{\Omega}_{\mathbb{BD}}}

\newcommand{\tc}[2]{\operatorname{tc_{#1, \mathnormal{#2}}}}
\DeclareMathOperator*{\maxn}{max_{\mathnormal n}}
\DeclareMathOperator{\nbest}{\mathnormal{n}best}
\DeclareMathOperator{\takenbest}{take\mathnormal{n}best}
\DeclareMathOperator{\mulnkk}{mul\mathnormal{n}_{\welem k}^{\mathnormal{(k)}}}
\newcommand{\mulnkki}[1]{\operatorname{mul\mathnormal{n}_{\welem k_{\mathnormal{#1}}}^{\mathnormal{(k_{#1})}}}}
\DeclareMathOperator{\seq}{seq}
\DeclareMathOperator{\prefof}{\preceq_{\mathrm{pref}}}
\newcommand{\disunion}{\mathbin{\dot{\cup}}}
\DeclareMathOperator{\pref}{pref}

\newcommand\letvdash[1]{\mathrel{\stackengine{1ex}{\vdash}{\;\;\scriptscriptstyle#1}{O}{r}{F}{T}{L}}}
\stackMath%

\newcommand{\pto}{\to\hspace{-2mm}\shortmid \hspace{2mm}}
\newcommand{\wtrans}{\letvdash{w}}
\newcommand{\nwtrans}[1]{\mathrel{(\letvdash{w})^{#1}}}

\newcommand{\ltrue}{\mathsf{t\mkern-1mu t}}
\newcommand{\lfalse}{\mathsf{f\mkern-1mu f}}

\newcommand{\terminal}[1]{\textbf{#1}}
\newcommand{\tJan}{\terminal{Jan}}
\newcommand{\tPiet}{\terminal{Piet}}
\newcommand{\tMarie}{\terminal{Marie}}
\newcommand{\tzag}{\terminal{zag}}
\newcommand{\thelpen}{\terminal{helpen}}
\newcommand{\tlezen}{\terminal{lezen}}

\newcommand{\term}[1]{\textbf{#1}}
\newcommand{\nont}[1]{\mathrm{#1}}

\newcommand{\sans}{{\mathsf{a}}}
\newcommand{\sinp}{{\mathsf{i}}}

\newcommand{\Vnew}{V_{\mathrm{new}}}
\newcommand{\changed}{\mathit{changed}{}}
\newcommand{\select}{\mathit{select}}
\newcommand{\wthom}[1][\walg K]{\def\ArgI{{#1}}\wthomR}
\newcommand{\wthomR}[1]{\wt\mathopen{}\left(#1\right)\mathclose{}_{\ArgI}}
\newcommand{\wtphom}[1][\walg K]{\def\ArgI{{#1}}\wtphomR}
\newcommand{\wtphomR}[1]{\wt'\mathopen{}\left(#1\right)\mathclose{}_{\ArgI}}
\newcommand{\nbzeroes}{(\underbrace{0, \dots, 0}_{\text{$n$ times}})}
\newcommand{\nbweight}{(\welem k, \underbrace{0, \dots, 0}_{\mathclap{\text{$n-1$ times}}})}
\newcommand{\nbweighti}[1]{(\welem k_{#1}, \underbrace{0, \dots, 0}_{\text{$n-1$ times}})}
\newcommand{\rzo}{{\mathbb R_0^1}}

\newcommand{\ruleindex}{r \in R: \\ r = (A \to \sigma(A_1, \dots, A_k))}
\newcommand{\TRc}{(T_R)^{(c)}}
\newcommand{\TRpc}{(T_{R'})^{(c)}}
\newcommand{\reqrule}{r = \big( A \to \sigma(A_1, \dots, A_k) \big)}
\makeatletter
\newcommand{\rgcomment}[1]{\ifmeasuring@\text{(#1)}\else\omit\hfill$\displaystyle\text{(#1)}$\fi\ignorespaces}
\makeatother
\newcommand{\isep}{\mathbin{{.}{.}}\nobreak}

\newcommand{\fruit}{\terminal{fruit}}
\newcommand{\flies}{\terminal{flies}}
\newcommand{\like}{\terminal{like}}
\newcommand{\bananas}{\terminal{bananas}}

\newcommand{\wlmclass}[1]{\mathcal W_{\mathrm{#1}}}

\newcommand*{\wrtglm}{\ensuremath{\big((G, \alg L), (\walg K, \oplus, \mathbb 0, \Omega, \infsumop), \wt\big)}}

\usetocstyle{KOMAlike}

\usepackage{mathrsfs}
\renewcommand{\gclass}[1]{\mathscr G_{\mathrm{#1}}}
\renewcommand{\wclass}[1]{\mathscr K_{\mathrm{#1}}}
\newcommand\leftparen(
\newcommand\rightparen)
\renewcommand{\wlmclass}[2][]{\ensuremath{\mathscr W_{#1}
\expandarg%
\StrLen{#2}[\strlen]
\ifnum\strlen>0%
(#2)%
\else%
#2%
\fi%
}}
\tikzset{basic/.style={text height=2ex,text depth=0.5ex,fill=lightgray!70,rounded corners,font=\small}}
\tikzset{lhs/.style={left,xshift=0.4em,fill=none}}
\tikzset{rhs/.style={right,xshift=-0.4em,fill=none}}

\let\nbest\relax
\DeclareMathOperator{\nbest}{\mathnormal{n}\walg{BST}}
\renewcommand{\term}[1]{\mathsf{#1}}
\renewcommand{\tc}[2]{\operatorname{tc_{\mathnormal{#1}, \mathnormal{#2}}}}
\renewcommand{\terminal}[1]{\mathsf{#1}}
\renewcommand{\TRc}{\T_R^{(c)}}
\renewcommand{\TRpc}{\T_{R'}^{(c)}}
\renewcommand{\alg}[1]{\mathcal{#1}}

\title{Weighted Parsing for  \\ Grammar-Based Language Models\\ over Multioperator Monoids}

\date{\today}

\author{Richard Mörbitz and Heiko Vogler\\Technische Universität Dresden, Germany}

\makeindex
\newglossary*{abbreviations}{List of abbreviations}
\makeglossaries
\newglossaryentry{alg:cfg}{
    name={\ensuremath{\lalg{CFG}^\Delta}},
    description={CFG algebra over $\Delta$},
    sort={alg:cfg},
    type={abbreviations}
}
\newglossaryentry{alg:lcfrs}{
    name={\ensuremath{\lalg{LCFRS}^\Delta}},
    description={LCFRS algebra over $\Delta$},
    sort={alg:lcfrs},
    type={abbreviations}
}
\newglossaryentry{alg:tag}{
    name={\ensuremath{\lalg{TAG}^\Delta}},
    description={TAG algebra over $\Delta$},
    sort={alg:tag},
    type={abbreviations}
}
\newglossaryentry{alg:yield}{
    name={\ensuremath{\lalg{YIELD}^\Delta}},
    description={YIELD algebra over $\Delta$},
    sort={alg:yield},
    type={abbreviations}
}
\newglossaryentry{gclass:all}{
    name={\ensuremath{\gclass{all}}},
    description={class of all RTG-LMs},
    sort={gclass:all},
    type={abbreviations}
}
\newglossaryentry{gclass:cfg}{
    name={\ensuremath{\gclass{CFG}}},
    description={class of RTG-LMs with language algebra $\lalg{CFG}$},
    sort={gclass:cfg},
    type={abbreviations}
}
\newglossaryentry{gclass:lcfrs}{
    name={\ensuremath{\gclass{LCFRS}}},
    description={class of RTG-LMs with language algebra $\lalg{LCFRS}$},
    sort={gclass:lcfrs},
    type={abbreviations}
}
\newglossaryentry{gclass:tag}{
    name={\ensuremath{\gclass{TAG}}},
    description={class of RTG-LMs with language algebra $\lalg{TAG}$},
    sort={gclass:tag},
    type={abbreviations}
}
\newglossaryentry{gclass:yield}{
    name={\ensuremath{\gclass{YIELD}}},
    description={class of RTG-LMs with language algebra $\lalg{YIELD}$},
    sort={gclass:yield},
    type={abbreviations}
}
\newglossaryentry{gclass:acyc}{
    name={\ensuremath{\gclass{acyc}}},
    description={class of acyclic RTG-LMs},
    sort={gclass:acyc},
    type={abbreviations}
}
\newglossaryentry{gclass:mon}{
    name={\ensuremath{\gclass{mon}}},
    description={class of monadic RTG-LMs},
    sort={gclass:mon},
    type={abbreviations}
}
\newglossaryentry{gclass:nl}{
    name={\ensuremath{\gclass{nl}}},
    description={class of nonlooping RTG-LMs},
    sort={gclass:nl},
    type={abbreviations}
}
\newglossaryentry{gclass:findc}{
    name={\ensuremath{\gclass{\findc}}},
    description={class of RTG-LMs with finitely decomposable language algebra},
    sort={gclass:findc},
    type={abbreviations}
}
\newglossaryentry{wclass:all}{
    name={\ensuremath{\wclass{\mathrm{all}}}},
    description={class of complete M-monoids},
    sort={mclass:all},
    type={abbreviations}
}
\newglossaryentry{wclass:sr}{
    name={\ensuremath{\wclass{sr}}},
    description={class of M-monoids associated with semirings},
    sort={mclass:sr},
    type={abbreviations}
}
\newglossaryentry{wclass:sup}{
    name={\ensuremath{\wclass{sup}}},
    description={class of superior M-monoids},
    sort={mclass:sup},
    type={abbreviations}
}
\newglossaryentry{wclass:tropical}{
    name={\ensuremath{\walg{T}}},
    description={tropical M-monoid},
    sort={mclass:tropical},
    type={abbreviations}
}
\newglossaryentry{wclass:viterbi}{
    name={\ensuremath{\walg{V}}},
    description={Viterbi M-monoid},
    sort={mclass:viterbi},
    type={abbreviations}
}
\newglossaryentry{wclass:bd}{
    name={\ensuremath{\walg{BD}}},
    description={best derivation M-monoid},
    sort={mclass:bd},
    type={abbreviations}
}
\newglossaryentry{wclass:nbest}{
    name={\ensuremath{\nbest}},
    description={n-best M-monoid},
    sort={mclass:nbest},
    type={abbreviations}
}
\newglossaryentry{wclass:dcomp}{
    name={\ensuremath{\wclass{\dcomp}}},
    description={class of d-complete M-monoids},
    sort={mclass:dcomp},
    type={abbreviations}
}
\newglossaryentry{wclass:dist}{
    name={\ensuremath{\wclass{dist}}},
    description={class of distributive M-monoids},
    sort={mclass:dist},
    type={abbreviations}
}
\newglossaryentry{wclass:finidpo}{
    name={\ensuremath{\wclass{\finidpo}}},
    description={class of finite and idempotent M-monoids with a certain monotonicity property},
    sort={mclass:finidpo},
    type={abbreviations}
}
\newglossaryentry{wclass:int}{
    name={\ensuremath{\wclass{int}}},
    description={class of intersection M-monoids},
    sort={mclass:int},
    type={abbreviations}
}
\newglossaryentry{wclass:adp}{
    name={\ensuremath{\wclass{\mathrm{ADP}}}},
    description={class of complete M-monoids},
    sort={mclass:adp},
    type={abbreviations}
}
\newglossaryentry{wlmclass:closed}{
    name={\ensuremath{\mathscr W_{\mathrm{closed}}(\gclass{all}, \wclass{\dcomp} \cap \wclass{dist})}},
    description={class of closed wRTG-LMs},
    sort={wlmclass:closed},
    type={abbreviations}
}
\newglossaryentry{wlmclass:bd}{
    name={\ensuremath{\mathscr W_{<1}(\gclass{all}, \wclass{\bd})}},
    description={class of best derivation wRTG-LMs},
    sort={wlmclass:bd},
    type={abbreviations}
}

\begin{document}

\maketitle

\begin{abstract}
    We develop a general framework for weighted parsing which is built on top of grammar-based language models and employs multioperator monoids as weight algebras.
    It generalizes previous work in that area (semiring parsing, weighted deductive parsing) and also covers applications outside the classical scope of parsing, e.g., algebraic dynamic programming.
    We show an algorithm for weighted parsing and,  for a large class of weighted grammar-based  language models, we  prove formally that it terminates and is correct.
\end{abstract}

\clearpage

\tableofcontents

\clearpage

\section{Introduction}%
\label{sec:intro}

In natural language processing (NLP), parsing is the syntactic analysis of sentences.
Given a sentence $a$ from some natural language $\alg L$, e.g.,
\[
    a = \fruit \ \flies \ \like \ \bananas \enspace,
\]
the goal is to produce some syntactic description of $a$.
This syntactic description can reflect three different kinds of relationships between words occurring in $a$: sequence, dependency, and constituency~\cite{hutsom92}.

Parsing is usually performed using some language model.
Here we will focus on that branch of constituency parsing in which the language models are provided by some kind of formal grammar, like context-free grammars (CFG)~\cite{Chomsky1963}, linear context-free rewriting systems (LCFRS)~\cite{vijweijos87}, multiple context-free grammars (MCFG)~\cite{Sekietal91}, or tree-adjoining grammars (TAG)~\cite{jossch97}.
Formally, each sentence $a$ from the natural language $\alg L$ is mapped to an element of the set $\constrees$ of constituent trees of some given grammar $G$:
\[
    \fparse: \alg L \to \constrees \enspace. \tag{constituent parsing}
\]

In many NLP applications, it is also desirable to obtain information about a given sentence which is different from its constituent tree.
We refer to the quantity which is to be computed as \emph{parsing objective}, and we call the mapping from $\alg L$ to the set of parsing objectives a \emph{parsing problem}.
In the following, we give some examples of parsing problems (cf.~\cite{Goodman1999}).
First, it can be the case that some sentence $a$ is not grammatical according to $G$ and therefore no constituent tree exists for it.
We can answer the question whether a sentence is grammatical or not by mapping each sentence to an element of the set $\mathbb B = \{\ltrue, \lfalse\}$ of Boolean truth values:
\[
    \fparse: \alg L \to \mathbb B \enspace. \tag{recognition}
\]

Second, natural languages are ambiguous and hence a sentence can have several constituent trees, each representing a different meaning (cf.\ Figure~\ref{fig:asts}).
We can tell how many constituent trees there are for some sentence by mapping it to a natural number:
\[
    \fparse: \alg L \to \mathbb N \enspace. \tag*{\makebox[2cm][r]{(number of derivations)}}
\]

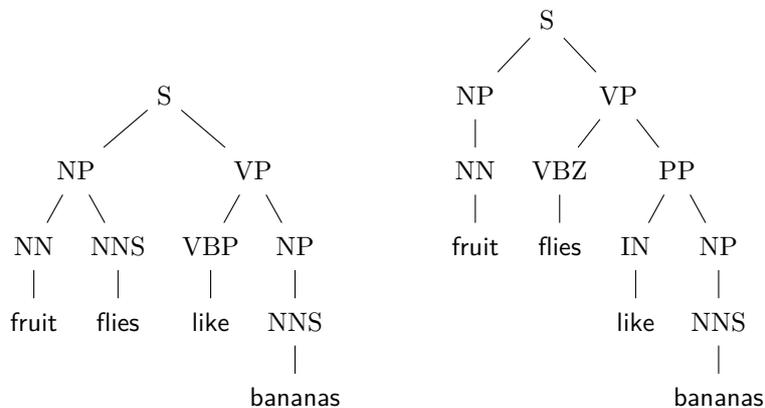
\begin{figure}[h]
    \centering
    \begin{tikzpicture}[baseline=(s),tree layout,every node/.style={text height=2ex,text depth=0.5ex}]
        \node (s) {$\nont{S}$}
        child {
            node {$\nont{NP}$}
            child {
                node {$\nont{NN}$}
                child {
                    node {$\fruit$}
                }
            }
            child {
                node {$\nont{NNS}$}
                child {
                    node {$\flies$}
                }
            }
        }
        child {
            node {$\nont{VP}$}
            child {
                node {$\nont{VBP}$}
                child {
                    node {$\like$}
                }
            }
            child {
                node {$\nont{NP}$}
                child {
                    node {$\nont{NNS}$}
                    child {
                        node {$\bananas$}
                    }
                }
            }
        };
    \end{tikzpicture}
    \hspace{1cm}
    \begin{tikzpicture}[baseline=(s),tree layout,every node/.style={text height=2ex,text depth=0.5ex}]
        \node (s) {$\nont{S}$}
        child {
            node (np) {$\nont{NP}$}
            child {
                node (nn) {$\nont{NN}$}
                child {
                    node {$\fruit$}
                }
            }
        }
        child {
            node (vp) {$\nont{VP}$}
            child {
                node (vbz) {$\nont{VBZ}$}
                child {
                    node {$\flies$}
                }
            }
            child {
                node (pp) {$\nont{PP}$}
                child {
                    node (in) {$\nont{IN}$}
                    child {
                        node {$\like$}
                    }
                }
                child {
                    node (np') {$\nont{NP}$}
                    child {
                        node (nns) {$\nont{NNS}$}
                        child {
                            node {$\bananas$}
                        }
                    }
                }
            }
        };
    \end{tikzpicture}
    \caption{Two constituent trees for the sentence $a = \fruit \ \flies \ \like \ \bananas$.}
    \label{fig:asts}
\end{figure}

Our original idea to parsing is addressed by assigning to each sentence $a$ a \emph{set} of constituent trees, i.e., an element of $\mathcal P(\constrees)$, the powerset of $\constrees$.
This set is empty if $a$ is not grammatical and otherwise contains all constituent trees for $a$:
\[
    \fparse: \alg L \to \mathcal P(\constrees) \enspace. \tag{set of derivations}
\]
We note that, since $\constrees$ is usually infinite, a sentence can be mapped to an infinite set.

Third, it is usually unfeasible to work with all constituent trees of a sentence.
Instead, we want to restrict ourselves to the \enquote{most suitable} ones.
This is usually done by employing a probabilistic language model, e.g., a probabilistic context-free grammar (PCFG).
With a PCFG, we can assign to each constituent tree a \emph{probability}, i.e., a value in $\rzo$ (the interval of real numbers between $0$ and $1$).
Then we can map each sentence $a$ to the highest probability among every constituent tree for $a$:
\[
    \fparse: \alg L \to \mathbb R_0^1 \enspace. \tag{best probability}
\]
The \emph{best parse} of a sentence $a$ is the combination of the constituent tree for $a$ having the highest probability and this probability value.
Since several constituent trees can have the same probability, it is necessary to return a set of constituent trees:
\[
    \fparse: \alg L \to \mathbb R_0^1 \times \mathcal P(\constrees) \enspace. \tag{best derivation}
\]
The elements of $\fparse(a)$ are commonly called \emph{Viterbi parses of $a$}.
In an obvious way, one can also compute a set of best constituent trees (best $n$ derivations).

Instead of best probabilities, one can be interested in mapping a sentence $a$ to the sum of the probabilities of each constituent tree of $a$.
The resulting value is an element of $\mathbb R_+$, the set of non-negative real numbers.
\[
    \fparse: \alg L \to \mathbb R_+ \enspace. \tag{string probability}
\]

The central idea of \citet{Goodman1999} is to abstract from the particular parsing problems from above by considering the parsing problem
\[
    \fparse: \alg L \to \walg K
\]
for any complete semiring $\walg K$, which he called \emph{semiring parsing problem}.
By choosing appropriate semirings, Goodman showed that the particular parsing problems from above are instances of the semiring parsing problem~\cite[Figure 5]{Goodman1999}.
More precisely, a complete semiring is an algebra $(\walg{K},\oplus,\otimes,\walg{0},\walg{1},\infsumop)$, where $\oplus$ (\emph{addition}) and $\otimes$ (\emph{multiplication}) are binary operations on $\walg K$ and $\infsumop$ is an extension of $\oplus$ to infinitely many arguments (where $\oplus$, $\otimes$, and $\infsumop$ satisfy certain algebraic laws).
For the semiring parsing problem, we assume that each rule of the grammar $G$ (modeling $\alg L$) is assigned an element of $\walg K$, called its \emph{weight}.
Then, for each sentence $a$, the value $\fparse(a)$ is the addition of the weights of all abstract syntax trees of $a$ (using $\infsumop$ if $a$ has infinitely many abstract syntax trees); the weight of a single abstract syntax tree is the multiplication of the weights of all occurrences of rules in that tree.

\bigskip

In this paper, we introduce \emph{weighted RTG-based language models} (wRTG-LMs) as a new framework for weighted parsing (where RTG stands for regular tree grammar) and define the \emph{M-monoid parsing problem} (cf.~Section \ref{sec:weighted-RTG-based-grammars}).
In the following, we briefly explain these two concepts.

A wRTG-LM is as a tuple
\[
    \overline G = \Big((G, (\alg L, \phi)), \ (\walg K, \oplus, \mathbb 0, \Omega, \infsumop), \ \wt\Big)
\]
where
\begin{itemize}
    \item $G$ is an RTG~\cite{Brainerd1969} and $(\alg L,\phi)$ is a \emph{language algebra},
    \item $(\walg K, \oplus, \mathbb 0, \Omega, \infsumop)$ is a complete M-monoid~\cite{Kuich1999}, called \emph{weight algebra}, and
    \item $\wt$ maps each rule of $G$ to an operation from $\Omega$.
\end{itemize}
Let us explain these components.

We call the tuple $(G, (\alg L,\phi))$ \emph{RTG-based language model} (RTG-LM) and its meaning is based on the initial algebra semantics approach of \citet{Goguen1977} as follows.
The RTG $G$ generates a set of trees.
Each generated tree is evaluated in the language algebra $(\alg L,\phi)$ to a \emph{syntactic object} $a$, i.e., an element of the modeled language $\alg L$; in this sense, the tree describes the grammatical structure of $a$.
For instance, if $\alg L$ is a natural language, then its sentences (if viewed as sequences of words) are the yields of such trees.
As another example, $\alg L$ could be a set of trees or a set graphs, and then each tree generated by $G$ represents the structures of such a syntactic object.
Moreover, each grammar of the above mentioned classes (e.g., CFG, LCFRS, MCFG, and TAG) can be formalized as an RTG-LM.

Complete M-monoids can be understood as a generalization of complete semirings in the sense that the multiplication $\otimes$ is replaced by a set $\Omega$ of finitary operations.
Then each rule $r$ of $G$ is assigned an operation $\wt(r) \in \Omega$, and the weight of an abstract syntax tree of $G$ is obtained by evaluating the corresponding term over operations in the M-monoid (similar to the evaluation of arithmetic expressions).
For example, each complete semiring can be viewed as a complete M-monoid by embedding the multiplication into $\Omega$.
Moreover, complete M-monoids can be used for parsing objectives beyond the above mentioned ones.
For instance, the intersection of a fixed CFG and a sentence $a$ from a corpus is used in EM training~\cite{demlairub77} of PCFGs~\cites{Bak79}{LarYou90}{NedSat08}, and this intersection may be viewed as a parsing objective.
Thus, the set $\walg K$ of parsing objectives is a set of CFGs and the set $\Omega$ contains operations which combine a number of CFGs into a single CFG (according to the construction of \citet{BarPerSha61}).
As another example, the objectives of an algebraic dynamic programming (ADP)~\cite{GieMeySte04} problem can be viewed as parsing objectives. Then the operations of the corresponding complete M-monoid combine partial solutions to solutions of larger subproblems.
Examples of ADP problems are computing the minimum edit distance or optimal matrix chain multiplication. Hence complete M-monoids form a very flexible class of weight algebras.

Now we turn to the second concept: the \emph{M-monoid parsing problem}. 
It is defined as follows. \\[4mm]
\textbf{Given:}
\begin{enumerate}
    \item a wRTG-LM $\big((G,(\alg L,\phi)),(\walg{K},\oplus,\welem{0},\Omega,\psi,\infsum),\wt\big)$    and
    \item an $a \in \alg L$,
\end{enumerate}
\textbf{Compute:}
\(
    \displaystyle\fparse(a) = \infsum_{d \in \AST(G, a)} \sem[\walg K]{\wt(d)} \enspace.
\)\\[4mm]
where
\begin{itemize}
\item $\AST(G, a)$ is the set of all abstract syntax trees (AST) generated by $G$ that evaluate in the language algebra $(\alg L,\phi)$ to $a$,
\item  $\wt(d)$ is the tree over operations obtained from $d$ by replacing each occurrence of a rule by $\wt(r)$,
  \item $\sem[\walg K]{\wt(d)}$ is the evaluation of $\wt(d)$ in the weight algebra $(\walg{K},\oplus,\welem{0},\Omega,\psi,\infsum)$.
\end{itemize}

By our considerations from above, the semiring parsing problem is an instance of the M-monoid parsing problem (cf.\ Section~\ref{sec:mmonoids-associated-with-semirings}). This holds also true for the computation of the intersection of a gammar and a sentence (cf.\ Section~\ref{sec:intersection}) and each ADP problem (cf.\ Section~\ref{sec:adp}).

We also propose an algorithm to solve the M-monoid parsing problem under certain conditions, and we call our algorithm the \emph{M-monoid parsing algorithm} (cf.~Section \ref{sec:algorithm}). Here we are faced with the difficulty that the sum
\[
  \infsum_{d \in \AST(G, a)} \sem[\walg K]{\wt(d)}
  \]
  can have infinitely many summands (infinite sum). Clearly, this cannot be done by a naive terminating algorithm.
Hence the applicability of our M-monoid parsing problem is restricted to cases in which the infinite sum coincides with some finite sum (cf. Theorem~\ref{thm:tr-trc}).
In the literature on weighted parsing, a few algorithms which are limited to specific weighted parsing problems have been investigated.
\begin{itemize}
    \item The semiring parsing algorithm has been proposed by \citet{Goodman1999} to solve the semiring parsing problem.
        It is a pipeline with two phases.
        The first phase takes as input a context-free grammar, a deduction system~\cite{shischper95}, and a syntactic object $a$ and computes a context-free grammar $G'$ (using a construction idea of~\citealp{BarPerSha61}).
        The second phase takes $G'$ as input and attempts to calculate $\fparse(a)$ (see above).
        Since, in general, $\fparse(a)$ is an infinite sum, this only succeeds if $G'$ is acyclic.
        Goodman states that in applications, this computation needs to be replaced by instructions specific to the used semiring.
    \item The weighted deductive parsing algorithm by \citet{ned03} addresses this problem by restricting itself to weighted parsing problems where the weight algebra is superior (cf.\ Section~\ref{sec:superior-mmonoids}).
        Nederhof's algorithm is a two-phase pipeline, too, where the first phase is the same as in Goodman's approach (but allowing more flexible deduction systems).
        In the second phase, he employs the algorithm of \citet{Knuth1977}, which is a generalization of Dijkstra's shortest path algorithm~\cite{dij59}.
        This also works in cases where $G'$ is cyclic.
    \item The single source shortest distance algorithm by \citet{Mohri2002} is applicable to graphs of which the edges are weighted with elements of some semirings that  is closed for the graph (cf.\ Section~\ref{sec:closed-definition}).
        This is a much weaker restriction than acyclicity or superiority.
        While Mohri's algorithm is not a parsing algorithm, it can be used in the second phase as an alternative to Knuth's algorithm if the CFG $G'$ is non-branching, i.e., a linear grammar~\cite[Def.~1]{kha74}.
\end{itemize}

In the same way as the algorithms of Goodman and Nederhof, our M-monoid parsing algorithm is also a two-phase pipeline (cf.\ Figure~\ref{fig:alg}).         The inputs are a wRTG-LM $\overline G$ and a syntactic object $a$.
In the first phase, we use the \emph{canonical weighted deduction system} to compute a new wRTG-LM $\overline G{}'$ (cf. Section~\ref{sec:weighted-deduction-systems}).
This is similar to the first phases of Goodman and Nederhof, but our deduction system always reflects the CYK algorithm instead of being an additional input.
In the second phase, we employ our \emph{value computation algorithm} (cf. Section~\ref{sec:value-computation-algorithm}).
It is a generalization of Mohri's algorithm, which is in spirit of Knuth's generalization of Dijkstra's algorithm.
Thus the M-monoid parsing algorithm is applicable to every closed wRTG-LM, which includes the cases in which the algorithms of Goodman and Nederhof are applicable.

\begingroup\setlength\textfloatsep{10.0pt plus 2.0pt minus 10.0pt}
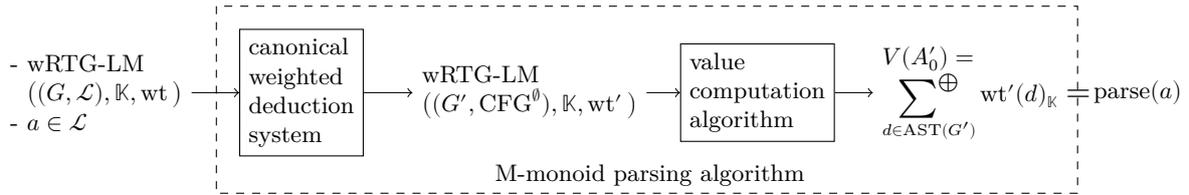
\begin{figure}[t]
    \centering
    \begin{adjustbox}{center}
        \begin{tikzpicture}[font=\small]
            \node[align=left] (input) {- wRTG-LM \\ \hphantom{- }$\big((G, \alg L), \walg K, \wt\big)$ \\ - $a \in \alg L$};
            \node[draw, xshift=2.7cm, align=left] (cnc) at (input) {canonical \\ weighted \\ deduction \\ system};
            \node[xshift=3cm, align=left] (cfges) at (cnc) {wRTG-LM \\ $\big((G', \cfges), \walg K, \wt'\big)$};
            \node[draw, xshift=3cm, align=left] (algorithm) at (cfges) {value \\ computation \\ algorithm};
            \node[anchor=west, align=left] (ensure) at ($(algorithm.east)+(0.5,0)$) {$V(A_0') =$ \\ $\displaystyle\infsum_{d \in \AST(G')} \wt'(d)_{\walg K}$};
            \node[xshift=2.2cm] (output) at (ensure) {$\fparse(a)$};
            \draw[->] (input) -- (cnc);
            \draw[->] (cnc) -- (cfges);
            \draw[->] (cfges) -- (algorithm);
            \draw[->] (algorithm) -- (ensure);
            %\draw[double equal sign distance,line width=0.5pt] (ensure) -- (output);
            \node[yshift=-0.3mm,xshift=0.3mm] at ($(ensure.east) !.5! (output.west)$) {\large $=$};
            \draw[dashed] ($(cnc.north west)+(-0.3,0.3)$) rectangle ($(ensure.south east)+(0.125,-0.6)$);
            \node[yshift=-0.3cm,xshift=-0.05cm] at ($(cnc.south west) !.5! (ensure.south east)$) {M-monoid parsing algorithm};
        \end{tikzpicture}
    \end{adjustbox}
    \caption{Two-phase pipeline for solving the M-monoid parsing problem ($A_0'$ is the initial nonterminal of~$G'$).}
    \label{fig:alg}
\end{figure}
\endgroup

In this paper, we formally prove the following results concerning the M-monoid parsing algorithm.
\begin{itemize}
    \item The value computation algorithm terminates for every closed wRTG-LM as input (Theorem~\ref{thm:vca-terminating}).
    \item The value computation algorithm is correct for every closed wRTG-LM as input (Corollary~\ref{cor:vca-correct}).
    \item The M-monoid parsing algorithm terminates and is correct for every closed wRTG-LM with finitely decomposable language algebra and for every nonlooping wRTG-LM with finitely decomposable language algebra as input (Theorem~\ref{thm:terminating-correct}).
\end{itemize}
These proofs are based on two fundamental results on closed wRTG-LMs (Theorem~\ref{thm:outside-trees-subsumed} and Theorem~\ref{thm:tr-trc}).
Moreover, we prove that several classes of wRTG-LMs are closed (Theorems~\ref{thm:applications} and~\ref{thm:applications2}).
We show that the two advanced parsing problems from above are indeed instances of the M-monoid parsing problem:
\begin{itemize}
    \item Computing the intersection of a grammar and a syntactic object is an M-monoid parsing problem (Theorem~\ref{thm:intersection}).
    \item Every ADP problem is an M-monoid parsing problem (Theorem~\ref{thm:ADP-M-monoid}).
\end{itemize}
Finally, we prove that the M-monoid parsing algorithm is applicable to six particular classes of wRTG-LMs (Corollary~\ref{cor:applicability}).

The key ideas and main results of this article were presented at FSMNLP~2019~\cite{moevog19}.
This paper is self-contained in the sense that we recall all necessary definitions and we provide full proofs of each result. These characteristics are the reason for the length of the paper. Those readers who are familiar with universal algebra and regular tree grammars may skip the preliminaries on first reading and consult them on demand. With a few exceptions, we have placed the proofs of the statements into appendices; so, for those readers who are not so much interested in them can read more smoothly through the text.

\section{Preliminaries}

\subsection{Basic mathematical notions}

\paragraph{Number sets.}
We denote the set of \emph{natural numbers} including $0$ by~$\mathbb N$ and the set of \emph{real numbers} by~$\mathbb R$.
The usual addition and multiplication on~$\mathbb N$ and~$\mathbb R$ are denoted by~$+$ and~$\cdot$,
respectively.
Furthermore, we define the following sets:
\begin{itemize}
    \item $\mathbb N_+ = \mathbb N \setminus \{ 0\}$ (the set of \emph{natural numbers without~$0$}),

    \item $\mathbb R_+ = \{ x \in \mathbb R \mid x \ge 0 \}$ (the set of \emph{non-negative real numbers}),

    \item $\mathbb R_0^1 = \{ x \mid x \in \mathbb R \text{ and } 0 \le x \le 1 \}$ (the set of \emph{real numbers between $0$ and $1$}),
        and
    \item $\mathbb R_0^\infty = \{ x \mid x \in \mathbb R \text{ and } x \ge 0 \} \cup \{ \infty \}$ (the set of \emph{non-negative real numbers with infinity}),
        where we extend the usual addition and multiplication in the following way to operate with $\infty$:
        \begin{align*}
            r + \infty &= \infty \tag{for every $r \in \mathbb R_0^\infty$} \\
            r \cdot \infty &= \infty \tag{for every $r \in \mathbb R_0^\infty \setminus \{ 0 \}$} \\
            \infty \cdot 0 &= 0 \enspace.
        \end{align*}
\end{itemize}
For every $j, k \in \mathbb N$, we denote the set $\{ j, \dots, k \} \subseteq \mathbb N$ by $[j, k]$.
Furthermore, we write $[k]$ rather than~$[1, k]$.

\paragraph{Boolean set.}
Let~$\ltrue$ denote true and~$\lfalse$ denote false.
We define $\mathbb B = \{ \ltrue, \lfalse \}$ (the \emph{Boolean set}).

\paragraph{Sets, binary relations, orders, and families.}
In the following let $A$, $B$, and $C$ be sets.
We denote the cardinality of $A$ by $|A|$.
The power set of $A$ is denoted by $\mathcal P(A)$.
If $A$ contains exactly one element, then we identify $A$ with its element.
\index{u@$\disunion$}
\index{disjoint union}
We say that~$C$ is the \emph{disjoint union of~$A$ and~$B$}, denoted by $C = A \disunion B$, if $C = A \cup B$ and $A \cap B = \emptyset$.

\index{relation}
\index{relation!binary}
A \emph{binary relation on~$A$ and~$B$} is a subset $R$ of $A \times B$.
Let~$a \in A$ and $b\in B$.
We write $a R b$ instead of $(a_1,a_2) \in R$. For each $A' \subseteq A$ we define $R(A')= \{b \in B \mid a \in A', aRb\}$.
\index{relation!reverse}
The \emph{inverse relation of~$R$} is the relation  $R^{-1} \; = \{ (b,a) \mid a R b\}$
on $B$ and $A$.

\index{relation!right-unique}
\index{relation!functional}
We call~$R$ \emph{right-unique} (and \emph{functional}), if $|\{ b \mid a R b \}| \le 1$ (resp., $|\{ b \mid a R b \}| = 1$) for every $a \in A$.
\index{partial function}
\index{mapping}
Usually, a right-unique relation (and a functional relation) on $A$ and $B$ is called \emph{partial function} (resp., \emph{mapping}) and it is denoted by $f: A \pto B$ (resp., $f: A \rightarrow B$).
Since we identify a set with one element with this element, we write $f(a) = b$ rather than $f(\{a\}) = \{b\}$ for a mapping $f$.
A mapping~$f: A \to B$ is
\index{mapping!injective}
\index{mapping!surjective}
\index{mapping!bijective}
\begin{itemize}
    \item \emph{injective}, if $|f^{-1}(b)| \le 1$ for every $b \in B$,
    \item \emph{surjective}, if $|f^{-1}(b)| \ge 1$ for every $b \in B$,
        and
    \item \emph{bijective}, if it is both injective and surjective.
\end{itemize}

%Let $k \in \mathbb N$, $A_0,\dots,A_k$ be sets, $a \in A_0$, and $f_i: A_{i+1} \to A_i$ be a mapping for every $i \in [k-1]$.
%We denote $f_k( f_{k-1}( \dots f_1( f_0(a) ) \dots ) )$\todo{R: Wo brauchen wir das? H: Früher hatte das glaube ich für das Zelegen von Bäumen eine Bedeutung. kann weg}, i.e., the successive application of~$f_0$ to~$a$,~$f_1$ to $f_0(a)$, and so on until~$f_k$ is applied to the result of $f_{k-1}$, by $f_k \dots f_0(a)$.

Let $k \in \mathbb N$, $A_1,\dots,A_k$ be sets and $g: A_1 \times \dots \times A_k \rightarrow A$ be a mapping.
\index{mapping!extension to sets}
The \emph{extension of~$g$ to sets} is the mapping $\widehat{g}: \mathcal P(A_1) \times \dots \times \mathcal P(A_k) \rightarrow \mathcal P(A)$, which is defined for every $F_1 \subseteq A_1,\dots,F_k \subseteq A_k$ as
\[ \widehat{g}(F_1,\dots,F_k) = \{ g(a_1,\dots,a_k) \mid a_1 \in F_1,\dots,a_k \in F_k \} \enspace. \]
In the sequel, we will denote the extension also by $g$.

\index{relation!endorelation}
An \emph{endorelation on~$A$} is a binary relation on~$A$ and~$A$.
\index{relation!identity relation}
The \emph{identity relation on~$A$}, denoted by $\id(A)$, is the endorelation on~$A$ which is defined as $\id(A) = \{ (a,a) \mid a \in A \}$.
Let $a, b \in A$ and $k \in \mathbb N$.
We write $a R^k b$ if there are $a_1, \dots, a_k \in A$ such that $a R a_1, a_1 R a_2, \dots, a_{k-1} R a_k, a_k R b$.
In particular, $a R^0 a$ for every $a \in A$.

\index{relation!reflexive}
\index{relation!antisymmetric}
\index{relation!transitive}
\index{relation!total}
\index{relation!well-founded}
In the following let $R \subseteq A \times A$ be an endorelation on~$A$.
We call~$R$
\begin{itemize}
    \item \emph{reflexive}, if $\id(A) \subseteq R$,
    \item \emph{transitive}, if $a_1 R a_2$ and $a_2 R a_3$ implies $a_1 R a_3$ for every $a_1,a_2,a_3 \in A$,
    \item \emph{antisymmetric}, if $a_1 R a_2$ and $a_2 R a_1$ implies $a_1 = a_2$ for every $a_1,a_2 \in A$,
    \item \emph{total}, if $a_1 R a_2$ or $a_2 R a_1$ for every $a_1,a_2 \in A$,
        and
    \item \emph{well-founded}, if for each non-empty subset $B \subseteq A$ there is an element $b \in B$ such that for each element $b' \in B$ it is not true that $b' R b$ holds.
\end{itemize}

\index{transitive closure}
\index{reflexive and transitive closure}
The \emph{transitive closure of~$R$},
denoted by~$R^+$,
is the smallest transitive endorelation~$R'$ on~$A$ such that $R \; \subseteq \; R'$.
The \emph{reflexive and transitive closure of~$R$},
denoted by~$R^*$,
is the smallest reflexive and transitive endorelation~$R'$ on~$A$ such that $R \; \subseteq \; R'$.

\index{order!partial order}
\index{order!total order}
We call $(A,R)$
\begin{itemize}
    \item a \emph{partial order}, if~$R$ is reflexive, antisymmetric, and transitive.
    \item a \emph{total order}, if $(A,R)$ is a partial order and~$R$ is total.
\end{itemize}

\begin{quote}
    \em In the following, if we deal with a partial order, then we will use the symbol $\preceq$ rather than~$R$.  Moreover, as a convention, we  denote~$\preceq^{-1}$ by~$\succeq$.
\end{quote}

\index{order!strict ordering induced by}
In the following let $(A,\preceq)$ be a partial order.
The \emph{strict ordering relation induced by~$\preceq$} is the binary relation $\prec\; = \; \preceq \setminus \id(A)$.

\index{order!well-partial order}
\index{order!well-order}
We say that $(A,\preceq)$ is
\begin{itemize}
    \item a \emph{well-partial order}, if $(A,\preceq)$ is a partial order and the strict ordering relation induced by~$\preceq$ is well-founded.
    \item a \emph{well-order}, if $(A,\preceq)$ is a total and well-partial order.
\end{itemize}

\begin{example}
    The \emph{natural order on pairs of natural numbers}, $\leq \; \subseteq \mathbb N^2$, is defined as follows:
    for every $(a_1,b_1), (a_2,b_2) \in \mathbb N^2$, $(a_1,b_1) \leq (a_2,b_2)$ if one of the following holds:
    \begin{enumerate}
        \item $a_1 < a_2$, or
        \item $a_1 = a_2$ and $b_1 \leq b_2$.
    \end{enumerate}
    We point out that $(\mathbb N^2, \leq)$ is a well-order.
    This statement is proved in Appendix~\ref{sec:proofs-preliminaries}.
\end{example}

\begin{lemma}[restate={[name={}]lempochains}]\label{lem:po-chains}
    For every partial order $(A, \preceq)$, $n \in \mathbb N$, and $a_1, \dots, a_n \in A$ the following holds:
    if $a_1 \preceq \dots \preceq a_n$ and $a_1 = a_n$, then $a_1 = \dots = a_n$.
\end{lemma}

\begin{proof}
    For the proof of Lemma~\ref{lem:po-chains}, we refer to Appendix~\ref{sec:proofs-preliminaries}.
\end{proof}

\index{lower bound}
\index{infimum}
\index{minimum}
In the following let $(A,\preceq)$ be a partial order and let $X \subseteq A$.
If there is an $a \in A$ such that for every $b \in X$ we have $a \preceq b$,
then~$a$ is a \emph{lower bound of~$X$}.
If, additionally, for every lower bound~$a'$ of~$X$ it holds that $a' \preceq a$,
then~$a$ is the \emph{infimum of~$X$}, denoted by $\inford X$.
If $\inford X \in X$, then~$\inford X$ is the \emph{minimum of~$X$}, denoted by $\minord X$.

\index{upper bound}
\index{supremum}
\index{maximum}
Dually, if there is an $a \in A$ such that for every $b \in X$ we have $b \preceq a$,
then~$a$ is an \emph{upper bound of~$X$}.
If, additionally, for every upper bound~$a'$ of~$X$ it holds that $a \preceq a'$,
then~$a$ is the \emph{supremum of~$X$}, denoted by $\supord X$.
If $\supord X \in X$, then~$\supord X$ is the \emph{maximum of~$X$}, denoted by $\maxord X$.

\index{$\argminord$}
\index{$\argmaxord$}
    Let $B$ be a set, $X \subseteq B$, and $f: B \rightarrow A$ be a mapping.
    We define
    \begin{align*}
        \argminord_{a \in X} f(a) &= \{ a \in X \mid f(a) \preceq f(a') \text{\ for every $a' \in X$} \} \\
  %      \intertext{and}
        \argmaxord_{a \in X} f(a) &= \{ a \in X \mid f(a') \preceq f(a) \text{\ for every $a' \in X$} \} \enspace. \qedhere
    \end{align*}

If the partial order $(A,\preceq)$ is clear from the context,
then we will drop $\preceq$ from $\inford$,
$\supord$, $\minord$, $\maxord$, $\argminord$,
and $\argmaxord$ and we will simply write $\inf$, $\sup$,
$\min$, $\max$, $\argmin$, and $\argmax$, respectively.
If $A = \{ \iota_1, \iota_2 \}$ for two arbitrary elements $\iota_1$ and $\iota_2$, then we write $\min (\iota_1, \iota_2)$ and $\max (\iota_1, \iota_2)$ rather than $\min \{ \iota_1, \iota_2 \}$ and $\max \{ \iota_1, \iota_2 \}$, respectively.

\index{index set}
In the following let $I$ be a countable set (\emph{index set}) and~$A$ be a set.
\index{family}
An \emph{$I$-indexed family over~$A$} (or: \emph{family over $A$}) is a mapping $f: I \rightarrow A$.
As usual, we represent each~$I$-indexed family~$f$ over~$A$ by $(f(i) \mid i \in I)$ and abbreviate $f(i)$ by $f_i$.
The set of all $I$-indexed families over~$A$ is denoted by $A^I$. Let~$J$ be a countable index set.
\index{partition}
A \emph{$J$-partition of~$I$} is a $J$-indexed family $(I_j \mid j \in J)$ over~$\mathcal P(I)$, where
\begin{enumerate*}
    \item $\bigcup_{j \in J} I_j = I$, and
    \item $I_j \cap I_{j'} = \emptyset$ for every $j,j' \in J$ with $j \not= j'$.
\end{enumerate*}

\paragraph{Strings and formal languages.}
\index{string}
In the following let $A$ be a set and $k \in \mathbb N$.
The set of \emph{strings of length~$k$ over~$A$} is the set $A^k = \{ a_1 \dots a_k \mid a_1,\dots,a_k \in A \}$.
\index{empty string}
In particular, $A^0 = \{ \varepsilon \}$,
where~$\varepsilon$ denotes the \emph{empty string}.
\index{string!set of strings}
The set of \emph{strings over~$A$} is the set $A^* = \bigcup_{i \in \mathbb N} A^i$.

Note that in our notation, we make no difference between the set of strings of length~$k$ over~$A$ and the $k$-fold Cartesian product $\underbrace{A \times \dots \times A}_{\text{$k$ times}}$, which is also denoted by~$A^k$. Thus, for $k=0$, we identify $\varepsilon$ and $()$.

\index{string!length}
\index{string!slice}
Let $w \in A^*$ with $w = w_1 \dots w_k$,
for some $k \in \mathbb N$ and $w_i \in A$ for every $i \in [k]$.
The \emph{length of~$w$}, denoted by $|w|$,  is~$k$, i.e., $|w| = k$.
For every $i, j \in [k]$ we denote the $(i,j)$-slice of~$w$ by $w_{i \isep j} = w_i \dots w_j$.

\index{string!concatenation}
Let $w' \in A^*$ be another string such that $w' = w_1' \dots w_{k'}'$ for some $k' \in \mathbb N$ and $w_1',\dots,w_{k'}' \in A$.
The \emph{concatenation of~$w$ and~$w'$},
denoted by $w \circ w'$, is the string $w_1 \dots w_k w_1' \dots w_{k'}'$.
We usually leave out the operation symbol and just write $w w'$ rather than $w \circ w'$.

\index{string!substring}
\index{string!prefix}
\index{string!suffix}
Let $w,w' \in A^*$.
If there are $u,v \in A^*$ such that $w = u w' v$,
then~$w'$ is a \emph{substring of~$w$}, $u$ is a \emph{prefix of~$w$}, which we denote by $u \prefof w$, and $v$ is a \emph{suffix of~$w$}.

\index{alphabet}
\index{formal language}
\index{formal language!concatenation}
If the set $A$ is nonempty and finite, then we call it \emph{alphabet}.
If $A$ is an  alphabet, then each subset of~$A^*$ is called \emph{formal language over~$A$}.
Let $L,L'\subseteq A^*$.
The \emph{concatenation of~$L$ and~$L'$} is the formal language
\[ L \circ L' = \{ w w' \mid w \in L \text{ and } w' \in L' \} \enspace. \qedhere \]

\subsection{Universal algebra}

\paragraph{Sorts and signatures.}
\index{sort}
\index{S@$S$-sorted set}
Let $S$ be a set (\emph{sorts}).
An \emph{$S$-sorted set} is a tuple $(A,\sort)$,
where $A$ is a set and $\sort: A \rightarrow S$ is a mapping.
Let $s \in S$, then we denote the set $\sort^{-1}(s)$ by~$A_s$.

We call $(A,\sort)$ \emph{empty} (respectively, \emph{nonempty}, \emph{finite},
and \emph{infinite}) if~$A$ is so.
\index{S@$S$-sorted alphabet}
An \emph{$S$-sorted alphabet} is a nonempty and finite $S$-sorted set.
An $S$-sorted set $(B,\sort')$ is a \emph{subset} of $(A,\sort)$,
if $B \subseteq A$ and $\sort'(b) = \sort(b)$ for every $b \in B$.
Let $(A,sort)$ and $(B,\sort')$ be two $S$-sorted sets.
\index{sort-preserving}
A mapping $f\colon A \rightarrow B$ is called \emph{sort-preserving} if $f(A_s) \subseteq B_s$ for each $s \in S$.
Moreover, let $(A_1,\sort_1)$ be a subset of $(A,\sort)$.
\index{sort-preserving!restriction}
The \emph{restriction of $f$ to $A_1$} is the mapping $f|_{A_1}\colon A_1 \rightarrow B$ such that $f|_{A_1}(a)=f(a)$ for each $a \in A_1$.

\index{ranked set}
A \emph{ranked set} is an $\mathbb N$-sorted set.
By convention, we call its mapping $\rk$ rather than $\sort$, i.e., $\rk: A \rightarrow \mathbb{N}$.
Each $(S^*\times S)$-sorted set $(\Sigma,\sort)$ \emph{can be viewed as the ranked set} $(\Sigma,\rk)$ where for every $k \in \mathbb N$, $s,s_1,\dots,s_k \in S$, and $\sigma \in \Sigma_{(s_1 \dots s_k,s)}$ we define $\rk(\sigma) = k$.
On the other hand, every ranked set $(A,\rk)$ can be considered as an $(S^*\times S)$-sorted set for some set $S$ of sorts with $|S|=1$.
\index{ranked alphabet}
A \emph{ranked alphabet} is a nonempty and finite ranked set.

\index{trg}
For each $(S^* \times S)$-sorted set~$\Sigma$ we define a mapping $\trg: \Sigma \to S$ such that $\trg(\sigma) = s$ if $\sigma \in \Sigma_{(s_1 \dots s_k,s)}$.
Moreover, for each $s \in S$, we denote the set $\trg^{-1}(s)$ by~$\Sigma_s$.

We often denote an $S$-sorted set $(A,\sort)$ and a ranked set $(A,\rk)$ only by~$A$; then the mappings will be denoted by $\sort_A$ and $\rk_A$, respectively.

\begin{quote}
    \em In the rest of this paper, if $S$ and $\Sigma$ are unspecified, then they stand for an arbitrary set of sorts and an arbitrary $(S^* \times S)$-sorted set, respectively.
    Moreover, for the sake of brevity, whenever we  write ``$\sigma \in \Sigma_{(s_1\ldots s_k,s)}$'', then we mean that $k \in \mathbb{N}$ and $s,s_1,\dots,s_k \in S$.
\end{quote}

\paragraph{$S$-sorted algebras and $S$-sorted $\Sigma$-homomorphisms.}
Sorted algebras have been introduced by~\cite{Higgins1963}.
We use the notation of~\cite{Goguen1985} and also refer to~\cite{Goguen1977}.

\index{algebra}
\index{S@$S$-sorted $\Sigma$-algebra}
\index{algebra!carrier set}
\index{algebra!interpretation mapping}
An \emph{$S$-sorted $\Sigma$-algebra} (or: algebra) is a pair $(\alg{A},\phi)$,
where~$\alg{A}$ is an $S$-sorted set (\emph{carrier set}) and~$\phi$ is a mapping from $\Sigma$ to operations on~$\alg{A}$ (\emph{interpretation mapping}) such that the following condition holds:
for every  $\sigma \in \Sigma_{(s_1 \dots s_k,s)}$ we have $\phi(\sigma)\colon \alg{A}_{s_1} \times \dots \times \alg{A}_{s_k} \rightarrow \alg{A}_s$. If $|S|=1$, then we call $(\alg{A},\phi)$ simply \emph{$\Sigma$-algebra}.
\index{factor@$\factors$}
For every $a \in \alg A$ we let
\[
    \factors(a)= \{b \in \alg{A} \mid b (<_{\mathrm{factor}})^* a\}
\]
where for every $a,b \in \alg{A}$, $b <_{\mathrm{factor}} a$ (\emph{$b$ is a factor of $a$}) if there are a $k \in \mathbb N$ and a $\sigma \in \Sigma_k$ such that $b$ occurs in some tuple $(b_1, \dots, b_k)$ with $\phi(\sigma)(b_1, \dots, b_k) = a$.
\index{finitely decomposable}
We call $(\alg A, \phi)$ \emph{finitely decomposable}~\cite[Def.~1.15]{kla84} if $\factors(a)$ is finite for every $a \in \alg A$.
We note that, in particular, for every finitely decomposable $S$-sorted $\Sigma$-algebra $(\alg A,\phi)$, $\sigma \in \Sigma$, and $a \in \alg L$, the set $\phi(\sigma)^{-1}(a)$ is finite.

\index{homomorphism}
\index{S@$S$-sorted $\Sigma$-homomorphism}
Let $(\alg{A},\phi)$ and $(\alg B,\psi)$ be $S$-sorted $\Sigma$-algebras.
Moreover, let $h\colon \alg{A} \rightarrow \alg B$ be a sort-preserving mapping. We call $h$ an \emph{$S$-sorted $\Sigma$-homomorphism (from $(\alg{A},\phi)$ to $(\alg B,\psi)$)},
if for every $\sigma \in \Sigma_{(s_1 \dots s_k,s)}$
and $a_1 \in \alg{A}_{s_1},\dots,a_k \in \alg{A}_{s_k}$ it holds that
\[ h(\phi(\sigma)(a_1,\dots,a_k)) = \psi(\sigma)(h(a_1),\dots,h(a_k)) \enspace. \]
If we write $h\colon (\alg{A},\phi) \rightarrow (\alg B,\psi)$, then we mean that $h$ is an $S$-sorted $\Sigma$-homomorphism from $(\alg{A},\phi)$ to $(\alg B,\psi)$.

\paragraph{$S$-sorted $\Sigma$-term algebra.}
\index{tree}
In the following let $X$ be a countable $S$-sorted set.
The set of \emph{trees over~$\Sigma$ and~$X$},
denoted by~$\T_\Sigma(X)$, is the smallest $S$-sorted set~$T$ such that
\begin{enumerate}
    \item $X_s \subseteq T_s$ for every $s \in S$, and
    \item for every $\sigma \in \Sigma_{(s_1 \dots s_k, s)}$ and $t_1 \in T_{s_1},\dots,t_k \in T_{s_k}$ we have  $\sigma(t_1,\dots,t_k) \in T_s$.
\end{enumerate}
If~$X = \emptyset$, we write~$\T_\Sigma$ instead of~$\T_\Sigma(X)$. If $\sigma \in \Sigma$ with $\rk(\sigma)=0$, then we abbreviate the tree $\sigma()$ by~$\sigma$.
\begin{quote}
    \em In the rest of this paper, if we write ``$t$ has the form $\sigma(t_1,\ldots,t_k)$'', then we mean that there are $\sigma \in \Sigma_{(s_1 \dots s_k, s)}$ and $t_1 \in \T_\Sigma(X)_{s_1},\dots,t_k \in \T_\Sigma(X)_{s_k}$ such that $t=\sigma(t_1,\ldots,t_k)$.
\end{quote}
We note that for each $t \in \T_\Sigma(X)$ the choices of $k \in \mathbb{N}$, $s,s_1,\ldots,s_k \in S$,  $\sigma \in \Sigma_{(s_1 \dots s_k, s)}$, and $t_1 \in \T_{s_1},\dots,t_k \in \T_{s_k}$ such that  $t=\sigma(t_1,\ldots,t_k)$ are unique.

\index{S@$S$-sorted $\Sigma$-term algebra}
The \emph{$S$-sorted $\Sigma$-term algebra over $X$} is the $S$-sorted $\Sigma$-algebra $(\T_\Sigma(X),\phi_\Sigma)$, where for every  $\sigma \in \Sigma_{(s_1 \dots s_k,s)}$ and $t_i \in (\T_\Sigma(X))_{s_i}$ with $i \in [k]$ we define  $\phi_\Sigma(\sigma)(t_1,\dots,t_k) = \sigma(t_1,\dots,t_k)$.

\begin{theorem}[cf.~{\cite[Prop.~2.6]{Goguen1977}}]
    Let  $(\alg{A},\phi)$ be an $S$-sorted $\Sigma$-algebra. If $h\colon X \rightarrow \alg{A}$ is a sort-preserving mapping, then   there exists a unique $S$-sorted $\Sigma$-homomorphism $\widetilde{h}\colon (\T_\Sigma(X),\phi_\Sigma) \rightarrow (\alg{A},\phi)$ extending $h$, i.e., such that $\widetilde{h}|_X=h$. Thus, in particular,  there exists a unique $S$-sorted $\Sigma$-homomorphism $g\colon (\T_\Sigma,\phi_\Sigma) \rightarrow (\alg{A},\phi)$.
\end{theorem}

\index{derived operation}
For a string $u=s_1\ldots s_n$ with $n \in \mathbb{N}$ and  $s_i \in S$ we let $X_u = \{x_{1,s_1},\ldots,x_{n,s_n}\}$ be a set of variables.
The set $X_u$ can be viewed as an $S$-sorted set with $(X_u)_s = \{x_{i,s_i} \mid s_i=s\}$.
Let $(\alg A,\phi)$ be an $S$-sorted $\Sigma$-algebra and $t \in \T_\Sigma(X_u)_s$ for some $s \in S$.
The \emph{$t$-derived operation on $\alg{A}$}, denoted by $\sem{t}$, is the operation $\sem{t}: \alg{A}_{s_1} \times \ldots \times \alg{A}_{s_n} \rightarrow \alg{A}_s$ defined by $\sem{t}(a_1,\ldots,a_n) = \widetilde{h}(t)$ where $h: X_u \rightarrow \alg{A}$ with $h(x_{i,s_i}) = a_i$.

Obviously, for each $t \in \T_\Sigma(X_\varepsilon)_s$ (i.e., $t \in (\T_\Sigma)_s$), $\sem{t}: \{()\} \to \alg{A}_s$.
We abbreviate $\sem{t}()$ by $\sem{t}$.
Then $\sem{t} = g(t)$ where $g\colon (\T_\Sigma,\phi_\Sigma) \rightarrow (\alg{A},\phi)$ is the unique $\Sigma$-homomorphism.

\begin{observation}[cf.~{\cite[Prop.~2.5]{Goguen1977}}]\label{obs:tree-derived-operations}
    Let $(\alg{A},\phi)$ be a $\Sigma$-algebra.
    Then for every $k \in \mathbb N$, $t \in \T_\Sigma(X_k)$, and $t_1,\dots,t_k \in \T_\Sigma$ it holds that
    \[ (t_{\T_\Sigma}(t_1,\dots,t_k))_{\alg{A}} = \sem{t}\big((t_1)_{\alg{A}},\dots,(t_k)_{\alg{A}}\big) \enspace. \]
\end{observation}
We remark that an extension of this result to $S$-sorted $\Sigma$-algebras is straightforward.

\index{tree homomorphism}
\index{S@$S$-sorted tree homomorphism}
Now let $\Delta$ be another $(S^* \times S)$-sorted set. Moreover, we let $h: \Sigma \rightarrow \T_\Delta(X)$ such that, for each $\sigma \in \Sigma_{(s_1\ldots s_k,s)}$, we have  $h(\sigma) \in \T_\Delta(X_u)_s$ with $u = s_1\ldots s_k$.
Then we define the $\Sigma$-algebra $(\T_\Delta,\phi_h)$ such that for every $\sigma \in \Sigma_{(s_1\ldots s_k,s)}$ we let
$\phi_h(\sigma) = h(\sigma)_{(\T_\Delta,\phi_\Delta)}$, i.e., the $h(\sigma)$-derived operation on the $\Delta$-term algebra $(\T_\Delta,\phi_\Delta)$.
Then we call the unique $\Sigma$-homomorphism from $(\T_\Sigma,\phi_\Sigma)$ to $(\T_\Delta,\phi_h)$ the \emph{$S$-sorted tree homomorphism induced by $h$}.

\index{tree relabeling}
\index{$S$-sorted tree relabeling}
In the particular case that for every $\sigma \in \Sigma_{(s_1\ldots s_k,s)}$ there is a $\delta \in \Delta_{(s_1\ldots s_k,s)}$ such that $h(\sigma)=\delta(x_{1,s_1},\ldots,x_{k,s_k})$,  we call the $S$-sorted tree homomorphism induced by $h$ the \emph{$S$-sorted tree relabeling induced by $h$}.

If $|S|=1$, then we call an $S$-sorted tree homomorphism ($S$-sorted tree relabeling) simply \emph{tree homomorphism} (\emph{tree relabeling}, respectively).

\paragraph{Useful functions on trees.}
\index{tree!position}
\index{tree!leaf}
Let $t \in \T_\Sigma(X)$.
The set $\pos(t) \subseteq \mathbb N^*$ is inductively defined as follows:
\[ \pos(t) = \begin{cases}
    \{ \varepsilon \} &\text{if } t \in X \\
    \{ \varepsilon \} \cup \bigcup_{i=1}^k \{ i \} \circ \pos(t_i) &\text{if $t$ has the form $\sigma(t_1,\dots,t_k)$.}
\end{cases} \]
Let $p \in \pos(t)$.
We call~$p$ a \emph{leaf} if there is no $k \in \mathbb N$ with $pk \in \pos(t)$.

\index{tree!subtree}
\index{tree!replacement}
\index{tree!label}
\index{tree!label sequence}
In order to formalize manipulations (e.g.\ rewriting) of trees, we introduce the following functions.
Let~$t \in \T_\Sigma(X)$ and $p \in \pos(t)$.
Moreover, let $s \in \T_\Sigma(X)$.
We define
\begin{itemize}
    \item the \emph{subtree of~$t$ at position~$p$}, denoted by $t|_p$,
    \item the \emph{tree obtained by replacing the subtree of~$t$ at position~$p$ by~$s$}, denoted by $t[s]_p$,
    \item the \emph{label of~$t$ at position~$p$}, denoted by $t(p)$, and
    \item the \emph{label sequence of~$t$ from root to position~$p$}, denoted by $\seq(t, p)$,
\end{itemize}
by structural induction as follows:
\begin{itemize}
    \item For every $h \in X$, $h|_\varepsilon = h$,  $h[s]_\varepsilon = s$, $h(\varepsilon) = h$, and $\seq(h, \varepsilon) = h$.

    \item If $t = \sigma(t_1,\ldots, t_k)$, then $t|_\varepsilon = t$, $t[s]_\varepsilon = s$, $\sigma(t_1,\ldots,t_k)(\varepsilon) = \sigma$, and $\seq(\sigma(t_1,\ldots,t_k), \varepsilon) = \sigma$.

        Moreover, for every $1 \leq i \leq k$ and $p' \in \pos(t_i)$, we define $t|_{ip'} = t_i|_{p'}$,
        \[ t[s]_{ip'} = \sigma(t_1,\ldots,t_{i-1},t_i[s]_{p'}, t_{i+1},\ldots,t_k)\enspace, \]
        $\sigma(t_1,\ldots,t_k)(ip') = t_i(p')$, and $\seq(\sigma(t_1,\ldots,t_k), ip') = \sigma \seq(t_i, p')$.
\end{itemize}

\index{tree!slice}
Let $t \in \T_\Sigma(X)$ and $p, p' \in \pos(t)$ such that $p \prefof p'$.
We define the \emph{label sequence of~$t$ from position~$p$ to position~$p'$}, denoted by $\seq(t, p, p')$, and the \emph{slice of~$t$ from position~$p$ to position~$p'$}, denoted by $t[p, p']$, as follows:
\begin{align*}
    \seq(t, p, p') &= \seq(t|_p, p'') \\
    t[p, p'] &= (t|_p)[\sigma(x_{1,s_1}, \dots, x_{k,s_k})]_{p''} \enspace,
\end{align*}
where $p'' \in \pos(t)$ such that $p p'' = p'$ and $t(p') = \sigma$ with $\sigma \in \Sigma_{(s_1 \dots s_k,s)}$.
Let $t_1 \in (\T_\Sigma)_{s_1}, \dots, t_k \in (\T_\Sigma)_{s_k}$.
We write $t[p, p'](t_1, \dots, t_k)$ rather than $(t[p, p'])_{\T_\Sigma}(t_1, \dots, t_k)$.

We define the mapping $\height: \T_\Sigma(X) \to \mathbb N$ inductively as follows:
\[ \height(t) = \begin{cases}
    0 &\text{if $t \in X \cup \Sigma$} \\
    1 + \max \{ \height(t_1),\dots,\height(t_k) \} &\text{if $t$ has the form $\sigma(t_1,\dots,t_k)$ and $k > 0$.}
\end{cases} \]

\begin{lemma}[restate={[name={}]lemheightfinite}]\label{lem:fixed-height-finite-trees'}
    Let~$\Sigma$ be a ranked set and $k = \max \{ \rk(\sigma) \mid \sigma \in \Sigma \}$ such that $k > 0$.
    Then for each $h \in \mathbb N$ it holds that $|\{ t \in \T_\Sigma \mid \height(t) \leq h \}| \leq |\Sigma|^{(\sum_{i=0}^h k^i)}$.
    In particular, $\{ t \in \T_\Sigma \mid \height(t) \leq h \}$ is finite.
\end{lemma}

\begin{proof}
    For the proof of Lemma~\ref{lem:fixed-height-finite-trees'}, we refer to Appendix~\ref{sec:proofs-preliminaries}.
\end{proof}

\index{tree!yield}
For each $\Delta \subseteq \Sigma \cup X$, we define the mapping $\yield_\Delta: \T_\Sigma(X) \rightarrow \Delta^*$ for each $t \in \T_\Sigma(X)$ of the form $\sigma(t_1,\dots,t_k)$  as follows:
\[ \yield_\Delta(t) = \begin{cases}
    \sigma &\text{if } k=0 \text{\ and } \sigma \in \Delta\\
    \varepsilon &\text{if } k=0 \text{\ and } \sigma \not\in \Delta\\
    \yield_\Delta(t_1) \dots \yield_\Delta(t_k) &\text{if $k > 0$}.
\end{cases} \]
If $\Delta = \Sigma \cup X$, then we simply write $\yield$ rather than $\yield_\Delta$.

\paragraph{Cycles in trees.}

\index{cyclic}
\index{acyclic}
\index{cycle}
\index{cycle!elementary}
Let $R$ be a ranked set and $\rho \in R^*$.
We call~$\rho$
\begin{itemize}
    \item \emph{cyclic}, if there are $i, j \in [|\rho|]$ such that $i < j$ and $w_i = w_j$,
    \item \emph{acyclic}, if~$\rho$ is not cyclic,
    \item \emph{a cycle}, if $|\rho| > 1$ and $\rho_1 = \rho_{|\rho|}$,
    \item \emph{an elementary cycle}, if~$\rho$ is a cycle and both $\rho_1 \dots \rho_{|\rho| - 1}$ and $\rho_2 \dots \rho_{|\rho|}$ are acyclic.
\end{itemize}
\index{cyclic!c@$(c,w)$-cyclic}
Let $c \in \mathbb N$ and $\rho, w \in R^*$ such that~$w$ is an elementary cycle.
We say that~$\rho$ is \emph{$(c,w)$-cyclic} if there are $v_0, \dots, v_c \in R^*$ such that $\rho = v_0 w v_1 \dots w v_c$ and for every $i \in [0, c]$, $w$ is not a substring of~$v_i$.
Thus, intuitively, $\rho$ is $(c,w)$-cyclic if~$w$ repeats exactly~$c$ times in~$\rho$.
\index{cyclic!c@$c$-cyclic}
We say that~$\rho$ is \emph{$c$-cyclic} if there is a $w \in R^*$ such that~$\rho$ is $(c,w)$-cyclic and for every $w' \in R^*$ and $c' \in \mathbb N$ with $c' > c$, $\rho$ is not $(c',w')$-cyclic.

Let $c \in \mathbb N$, $d \in \T_R$, and $p_1, p_2, p \in \pos(d)$ such that $p_1 \prefof p_2 \prefof p$ and~$p$ is a leaf.
We let $\seq(d, p_1, p_2) = w$.
We say that~$p$ is \emph{cyclic} (\emph{acyclic}, \emph{$(c,w)$-cyclic}, \emph{$c$-cyclic}), if $\seq(d, p)$ is cyclic (resp.\ acyclic, $(c,w)$-cyclic, $c$-cyclic).
We say that~$d$ is \emph{acyclic}, if every leaf $p \in \pos(d)$ is acyclic.
Furthermore, we say that~$d$ is \emph{$c$-cyclic}, if there is a leaf $p \in \pos(d)$ such that~$p$ is $c$-cyclic and for every leaf $p \in \pos(d)$ and $c' \in \mathbb N$ with $c' > c$, $p$ is not $c'$-cyclic.
For every $c \in \mathbb N$, we denote the set $\{ d \in \T_R \mid c' \in \mathbb N, c' \leq c, \text{\ and $d$ is $c'$-cyclic} \}$ by~$\TRc$.

\begin{observation}
    For every $c \in \mathbb N$ it holds that $\TRc \subseteq \T_R^{(c+1)}$.
    Furthermore, $\T_R = \bigcup_{i \in \mathbb N} \T_R^{(i)}$.
\end{observation}

\subsection{Monoids, semirings, and M-monoids}

\index{monoid}
A \emph{monoid} is a tuple $(\walg{K},\oplus,\welem{0})$, where
\begin{itemize}
    \item $\walg{K}$ is a set (\emph{carrier set}),
    \item $\oplus: \walg{K} \times \walg{K} \rightarrow \walg{K}$ is a mapping such that for every $\welem{k}_1,\welem{k}_2,\welem{k}_3 \in \walg{K}$ it holds that $(\welem{k}_1 \oplus \welem{k}_2) \oplus \welem{k}_3 = \welem{k}_1 \oplus (\welem{k}_2 \oplus \welem{k}_3)$ (\emph{associativity}), and
    \item $\welem{0} \in \walg{K}$ such that for every $\welem{k} \in \walg{K}$ it holds that $\welem{0} \oplus \welem{k} = \welem{k} = \welem{k} \oplus \welem{0}$ (\emph{identity element}).
\end{itemize}
\index{monoid!commutative}
We call $(\walg{K},\oplus,\mathbb 0)$  \emph{commutative} if for every $\welem{k}_1,\welem{k}_2 \in \walg{K}$ it holds that $\welem{k}_1 \oplus \welem{k}_2 = \welem{k}_2 \oplus \welem{k}_1$.
\index{idempotent}
We call $(\walg{K},\oplus,\mathbb 0)$  \emph{idempotent} if for every $\welem{k} \in \walg{K}$ it holds that $\welem{k} \oplus \welem{k} = \welem{k}$.

We define the binary relation $\preceq \, \subseteq \walg K \times \walg K$ for every $\welem k_1, \welem k_2 \in \walg K$ as follows:
$\welem k_1 \preceq \welem k_2$ if there is a $\welem k \in \walg K$ such that $\welem k_1 \oplus \welem k = \welem k_2$.

\begin{lemma}[restate={[name={}]lemnatordrt}]\label{lem:natord-refl-trans}
    For every monoid $(\walg K, \oplus, \welem 0)$, the binary relation~$\preceq$ on~$\walg K$ is reflexive and transitive.
\end{lemma}

\begin{proof}
    For the proof of Lemma~\ref{lem:natord-refl-trans}, we refer to Appendix~\ref{sec:proofs-preliminaries}.
\end{proof}

\index{monoid!naturally ordered}
We call $(\walg K, \oplus, \mathbb 0)$ \emph{naturally ordered} if~$\preceq$ is a partial order.

\begin{lemma}[restate={[name={}]lemnatord},cf.~{\cites[Proposition~3.2]{Karner1992}[Theorem~1.8]{HebWei1998}[p.\,6]{DroKui2009}}]\label{lem:natord-subsume}
    Let $(\walg K, \oplus, \mathbb 0)$ be a monoid.
    Then~$\walg K$ is naturally ordered if and only if for every $\welem k_1, \welem k_2, \welem k_3 \in \walg K$ with $\welem k_1 = \welem k_1 \oplus \welem k_2 \oplus \welem k_3$ it holds that $\welem k_1 = \welem k_1 \oplus \welem k_2$.
\end{lemma}

\begin{proof}
    For the proof of Lemma~\ref{lem:natord-subsume}, we refer to Appendix~\ref{sec:proofs-preliminaries}.
\end{proof}

\index{infinitary sum operation}
\index{monoid!infinitary sum operation}
An \emph{infinitary sum operation on~$\walg{K}$} is a family $(\infsum_{\hspace{-1mm}I} \mid$ $I$ is a countable index set$)$ of mappings $\infsum_{\hspace{-1mm}I}: \walg{K}^I \rightarrow \walg{K}$.
Instead of  $\infsum_I (\welem{k}_i \mid i \in I)$  we write $\infsum_{i \in I} \welem{k}_i$. If~$I$ is finite, then we denote $\infsum_{i \in I} \welem{k}_i$ by $\bigoplus_{i \in I} \welem{k}_i$.

\index{monoid!complete}
A \emph{complete monoid} (cf.~\cite[p.\,124--125]{Eilenberg1974};~\cite[p.\,247--248]{Golan1999};~\cite[p.\,5]{Fulop2018}) is a tuple $(\walg{K},\oplus,\mathbb 0,\infsum)$,
where $(\walg{K},\oplus,\mathbb 0)$ is a commutative monoid and $\infsum$ is an infinitary sum operation on $\walg{K}$
such that for each $I$-indexed family $(\welem{k}_i \mid i \in I)$ over $\walg{K}$ the following axioms hold:
\begin{itemize}
    \item If $I=\emptyset$, then $\infsum_{i \in \emptyset} \welem{k}_i = \welem{0}$,
    \item If $I= \{j\}$, then $\infsum_{i \in \{j\}} \welem{k}_i = \welem{k}_j$,
    \item If $I= \{j,j'\}$, then  $\infsum_{i \in \{j,j'\}} \welem{k}_i = \welem{k}_j + \welem{k}_{j'}$,
    \item $\infsum_{i \in I} \welem{k}_i = \infsum_{j \in J} \big( \infsum_{i \in I_j} \welem{k}_i \big)$ for every $I$-indexed family $(\welem{k}_i \mid i \in I)$ over $\walg{K}$ and $J$-partition of~$I$.
\end{itemize}

Intuitively, $\infsum$ extends~$\oplus$ while preserving the properties of associativity, commutativity, and the identity element~$\welem{0}$ of~$\oplus$.
However, using the above definition of complete, certain convergence properties of finite sums need not necessarily apply to infinite sums as well.
We solve this problem by requiring an additional property of $\infsum$ as follows.

Let $(\walg K, \oplus, \mathbb 0, \infsum)$ be a complete commutative monoid.
\index{monoid!d-complete}
We call~$\walg K$ \emph{d-complete} (cf.~\cite{Karner1992}), if for every $\welem k \in \walg K$ and family $(\welem k_i \mid i \in \mathbb N)$ of elements of~$\walg K$ the following holds:
if there is an $n_0 \in \mathbb N$ such that for every $n \in \mathbb N$ with $n \ge n_0$, $\infsum_{\substack{i \in \mathbb N: \\ i \le n}} \welem k_i = \welem k$, then $\infsum_{i \in \mathbb N} \welem k_i = \welem k$.

\begin{lemma}[cf.~{\cite[Proposition~3.1]{Karner1992}}]\label{lem:d-complete}
    Let $(\walg K, \oplus, \mathbb 0, \infsum)$ be a complete commutative monoid.
    Then the following statements are equivalent:
    \begin{enumerate}
        \item $\walg K$ is d-complete,
        \item for every $\welem k \in \walg K$ and family $(\welem k_i \mid i \in \mathbb N)$, if $\welem k \oplus \welem k_i = \welem k$ for every $i \in \mathbb N$, then
            \[ \welem k \oplus \infsum_{i \in \mathbb N} \welem k_i = \welem k \enspace, \]
            and
        \item for every countable set~$I$, family $(\welem k_i \mid i \in I)$ of elements of~$\walg K$, and finite subset~$E$ of~$I$ the following holds:
            if for every finite set~$F$ with $E \subseteq F \subseteq I$ it holds that
            \[ \infsum_{i \in E} \welem k_i = \infsum_{i \in F} \welem k_i \enspace, \]
            then
            \[ \infsum_{i \in E} \welem k_i = \infsum_{i \in I} \welem k_i \enspace. \]
    \end{enumerate}
\end{lemma}

Instead of giving a proof here, we refer the reader to~\cite{Karner1992}.
Although he stated this lemma for complete semirings rather than monoids, only the properties of the semiring's underlying monoid were used.
Thus the same proof is applicable to our lemma.

\begin{lemma}[restate={[name={}]lemdcompnatord}]\label{lem:d-complete-natord}
    Every d-complete monoid is naturally ordered.
\end{lemma}

\begin{proof}
    For the proof of Lemma~\ref{lem:d-complete-natord}, we refer to Appendix~\ref{sec:proofs-preliminaries}.
\end{proof}

\index{monoid!completely idempotent}
A complete monoid $(\walg{K},\oplus,\welem{0},\infsum)$ is \emph{completely idempotent}~\cite{drovog14} if for every $\welem{k} \in \walg{K}$ and index set $I$ we have $\infsum_{i \in I} \welem{k} = \welem{k}$.

\begin{lemma}[restate={[name={}]lemiidc}]\label{lem:inf-idp-d-complete}
    Every completely idempotent monoid is d-complete.
\end{lemma}

\begin{proof}
    For the proof of Lemma~\ref{lem:inf-idp-d-complete}, we refer to Appendix~\ref{sec:proofs-preliminaries}.
\end{proof}

\index{semiring}
A \emph{semiring} is tuple $(\walg{K},\oplus,\otimes,\welem{0},\welem{1})$ such that
\begin{itemize}
    \item $(\walg{K},\oplus,\welem{0})$ is a commutative monoid and $(\walg{K},\otimes,\welem{1})$ is a monoid,
    \item for every $\welem{k}_1,\welem{k}_2,\welem{k}_3 \in \walg{K}$ it holds that $\welem{k}_1 \otimes (\welem{k}_2 \oplus \welem{k}_3) = (\welem{k}_1 \otimes \welem{k}_2) \oplus (\welem{k}_1 \otimes \welem{k}_3)$ and $(\welem{k}_1 \oplus \welem{k}_2) \otimes \welem{k}_3 = (\welem{k}_1 \otimes \welem{k}_3) \oplus (\welem{k}_2 \otimes \welem{k}_3)$ (\emph{distributivity of~$\otimes$} over~$\oplus$), and
    \item for every $\welem{k} \in \walg{K}$ it holds that $\welem{k} \otimes \welem{0} = \welem{0} = \welem{0} \otimes \welem{k}$ (\emph{absorbing element}).
\end{itemize}
\index{semiring!commutative}
We call $(\walg{K},\oplus,\otimes,\welem{0},\welem{1})$ \emph{commutative}, if $\otimes$ is commutative.
\index{semiring!naturally ordered}
\index{semiring!idempotent}
We call $(\walg{K},\oplus,\otimes,\welem{0},\welem{1})$ \emph{naturally ordered}, if $(\walg K, \oplus, \welem 0)$ is naturally ordered, and \emph{idempotent}, if $(\walg K, \oplus, \welem 0)$ is idempotent.
We call~$\oplus$ \emph{addition} and~$\otimes$ \emph{multiplication}.

\begin{example}\label{ex:boolean-semiring}
    We consider the \emph{Boolean semiring} $(\mathbb B, \lor, \land, \lfalse, \ltrue)$, where $\lor$ is logical disjunction and $\land$ is logical conjunction.
    It is easy to see that $\mathbb B$ is commutative and idempotent.

    Let $\infsum[\lor]$ be the infinitary sum operation on~$\walg K$ defined as follows:
    for every countable set~$I$ and family $(\welem k_i \mid i \in I)$ of elements of~$\walg K$, if there is an $i \in I$ such that $a_i = \ltrue$, then $\infsum[\lor]_{i \in I} a_i = \ltrue$ and otherwise $\infsum[\lor]_{i \in I} a_i = \lfalse$.
    It is easy to see that $(\mathbb B, \lor, \lfalse, \infsum[\lor])$ is completely idempotent.
    Thus, by Lemma~\ref{lem:inf-idp-d-complete}, $\walg K$ is d-complete.

    Following~\cite[Example~3.1]{Karner1992}, we extend the Boolean semiring by $\infty$, i.e., we consider the semiring $(\mathbb B^{(\infty)}, \lor, \land, \lfalse, \ltrue)$, where $\mathbb B^{(\infty)} = \mathbb B \cup \{ \infty \}$ and $\lor$ and $\land$ are extended as follows to operate with $\infty$:
    \begin{align*}
        \welem K \lor \infty &= \infty &&\text{for every $\welem k \in \mathbb B \cup \{ \infty \}$}
        \\
        \welem K \land \infty &= \infty &&\text{for every $\welem k \in \{ \ltrue, \infty \}$}
        \\
        \welem k \land \lfalse &= \lfalse \enspace.
    \end{align*}
    We define this semiring to be commutative as well, thus its operations are fully specified.
    We define the infinitary sum operation $\infsum[\lor']$ such that for every family $(\welem k_i \mid i \in I)$ over $\mathbb B^{(\infty)}$
    \[
        \infsum[\lor']_{i \in I} \welem k_i = \begin{cases}
            \bigoplus_{i \in I: \welem k_i \not= \lfalse} \welem k_i & \text{if $\{ i \in I \mid \welem k_i \not= \lfalse \}$ is finite} \\
            \infty & \text{otherwise} \enspace.
        \end{cases}
    \]
    Then $(\mathbb B^{(\infty)}, \lor, \lfalse, \infsum[\lor'])$ is complete, but not d-complete.
    The latter can be seen using the family $(\ltrue \mid i \in I)$ over $\mathbb B^{(\infty)}$ for some infinite and countable set $I$.
    While $\infsum[\lor']_{i \in E} \ltrue = \ltrue$ for every finite and nonempty subset $E$ of $I$, we have that $\infsum[\lor']_{i \in I} \ltrue = \infty$.
    (In particular, $\mathbb B^{(\infty)}$ is not completely idempotent.)
\end{example}

\begin{sloppypar}
\index{semiring!complete}
A \emph{complete semiring} (cf.~{\cite[p.\,124--125]{Eilenberg1974};~\cite[p.\,247--248]{Golan1999};~\cite[p.\,5]{Fulop2018}}) is a tuple $(\walg{K},\oplus,\otimes,\welem{0},\welem{1},\infsum)$,
where $(\walg{K},\oplus,\otimes,\welem{0},\welem{1})$ is a semiring,
$(\walg{K},\oplus,\welem{0},\infsum)$ is a complete monoid,
and
\[\welem{k} \otimes \big( \infsum_{i \in I} \welem{k}_i \big) = \infsum_{i \in I} \big( \welem{k} \otimes \welem{k}_i \big) \ \text{ and } \ \big( \infsum_{i \in I} \welem{k}_i \big) \otimes \welem{k} = \infsum_{i \in I} \big( \welem{k}_i \otimes \welem{k} \big)
\]
hold for every $\welem{k} \in \walg{K}$,
countable index set~$I$, and $I$-indexed family $(\welem{k}_i \mid  i \in I)$ over  $\walg{K}$.
\end{sloppypar}

\index{M-monoid}
\index{null operation}
A \emph{multioperator monoid} (for short: M-monoid, cf.~\cite{Kuich1999}) is a tuple $(\walg{K},\oplus,\welem{0},\Omega,\psi)$ such that
\begin{itemize}
    \item $(\walg{K},\oplus,\welem{0})$ is a commutative monoid,
    \item $(\walg{K},\psi)$ is an $\Omega$-algebra for some ranked set $\Omega$,
        and
    \item $\welem{0}^k \in \Omega$ for every $k \in \mathbb N$, where $\psi(\welem{0}^k): \walg{K}^k \rightarrow \walg{K}$ such that $\psi(\welem{0}^k)(\welem{k}_1,\dots,\welem{k}_k) = \welem{0}$ for every $\welem{k}_1,\dots,\welem{k}_k \in \walg{K}$. We call the operation $\welem{0}^k$ a \emph{null operation}.
\end{itemize}
\index{M-monoid!distributive}
The M-monoid  $(\walg{K},\oplus,\welem{0},\Omega,\psi)$ is \emph{distributive} if for each $\omega \in \Omega$, $k \in \mathbb N$, $i \in [k]$,
and $\welem{k},\welem{k}_1,\dots,\welem{k}_k \in \walg{K}$, the operation $\psi(\omega)$ \emph{distributes over $\oplus$}, i.e.,
\begin{align*}
    &\psi(\omega)(\welem{k}_1,\dots,\welem{k}_{i-1},\welem{k}_i \oplus \welem{k},\welem{k}_{i+1},\dots,\welem{k}_k) \\
    &= \psi(\omega)(\welem{k}_1,\dots,\welem{k}_{i-1},\welem{k}_i,\welem{k}_{i+1},\dots,\welem{k}_k) \oplus \psi(\omega)(\welem{k}_1,\dots,\welem{k}_{i-1},\welem{k},\welem{k}_{i+1},\dots,\welem{k}_k)
\end{align*}
\index{M-monoid!absorbing}%
and $\welem{0}$ is \emph{absorbing}, i.e., $\psi(\omega)(\welem{k}_1,\dots,\welem{k}_k) = \welem{0}$ if $\welem{0} \in \{ \welem{k}_1,\dots,\welem{k}_k \}$.
\index{M-monoid!naturally ordered}
\index{M-monoid!absorbing}
We call $(\walg K, \oplus, \welem 0, \Omega, \psi)$ \emph{naturally ordered}, if $(\walg K, \oplus, \welem 0)$ is naturally ordered, and \emph{idempotent}, if $(\walg K, \oplus, \welem 0)$ is idempotent.

In the following, we will often identify $\omega \in \Omega$ with $\psi(\omega)$. Then we will omit the mapping~$\psi$ from the tuple $(\walg{K},\oplus,\welem{0},\Omega,\psi)$ and simply write $(\walg{K},\oplus,\welem{0},\Omega)$.
Also, for the sake of convenience, we will omit in examples and constructions the explicit specification of the null operations $\welem{0}^k$ in the definition of~$\Omega$.

\index{M-monoid!complete}
\begin{sloppypar}
A \emph{complete M-monoid} is a tuple $(\walg{K},\oplus,\welem{0},\Omega,\infsum)$,
where $(\walg{K},\oplus,\welem{0},\Omega)$ is an M-monoid and $(\walg{K},\oplus,\welem{0},\infsum)$ is a complete monoid.
\index{M-monoid!completely idempotent}%
\index{M-monoid!d-complete}%
A complete M-monoid $(\walg{K},\oplus,\welem{0},\Omega)$ is \emph{d-complete} (\emph{completely idempotent}) if  $(\walg{K},\oplus,\welem{0},\infsum)$ is d-complete (completely idempotent).
\end{sloppypar}

\begin{quote}
    \em As usual, we will identify any algebra defined in this section with its carrier set~$\walg{K}$, whenever the type of the algebra is clear from the context.
\end{quote}

\subsection{Regular tree grammars}

\index{regular tree grammar}
\index{RTG}
\index{S@$S$-sorted regular tree grammar}
An \emph{$S$-sorted regular tree grammar} (for short: $S$-sorted RTG) is a tuple $G = (N,\Sigma,A_0,R)$,
where
\begin{itemize}
    \item $N$ is an $S$-sorted alphabet (\emph{nonterminals})
    \item $\Sigma$ is an $(S^* \times S)$-sorted alphabet such that $N \cap \Sigma =\emptyset$ (\emph{terminals}),
    \item $A_0 \in N$ (\emph{initial nonterminal}), and
    \item $R$ is a finite $(N^* \times N)$-sorted set  (\emph{set of rules}) such that $R \subseteq (N \times \T_\Sigma(N))$  and for every $k \in \mathbb N$, $A,A_1,\ldots, A_k \in N$ the following holds:
        if $(B,t) \in R_{(A_1\ldots A_k,A)}$,
        then $B=A$, $\yield_N(t) = A_1 \ldots A_k$, and $\sort_S(A)=\sort_S(t)$.
\end{itemize}

For each rule $(A,t)$, we call~$A$ the \emph{left-hand side} and~$t$ the \emph{right-hand side} of that rule and denote them by $\lhs(r)$ and $\rhs(r)$, respectively.
\index{RTG!maximal rank}
\index{RTG!normal form}
The \emph{maximal rank of~$G$} is defined as $\maxrk(G) = \max \{ \rk_R(\rho) \mid \rho \in R \}$ where $R$ is viewed as a finite ranked set. If the right-hand side of each rule contains exactly one terminal, then $G$ is called \emph{in normal form}.  If $|S|=1$, then an $S$-sorted RTG is a classical regular tree grammar (cf.~\cite{Brainerd1969}).
We usually denote an element $(A,t)$ of~$R$ as $A \rightarrow t$.
For better readability, we show a list $A_1 \rightarrow t_1\ \dots\ A \rightarrow t_k$ of rules with the same left-hand side in the form $A \rightarrow t_1 + \dots + t_k$.

\index{abstract syntax tree}
\index{$\AST(G)$}
The set of \emph{abstract syntax trees (over $R$)} is the set $\AST(G) = (\T_R)_{A_0}$.
We can retrieve from each abstract syntax tree $d$ the tree in~$\T_\Sigma$ which is represented by~$d$. For this we view $R$ as $(S^* \times S)$-sorted set by defining the mapping $\sort: R \to S$ as follows: for every $r = (A \to t)$ in $R$ with  $\yield_N(t) = A_1 \ldots A_n$, we let $\sort(r)= (\sort(A_1)\ldots \sort(A_n),\sort(A))$.
Moreover, we define the mapping $h: R \rightarrow \T_\Sigma(X)$ for each $r \in R$ as follows.
If $r =(A \rightarrow t)$ with $\yield_N(t) = A_1 \ldots A_n$, then $h(r)$ is obtained from $t$ by replacing the $i$-th occurrence of a nonterminal in $t$ (counted from left-to-right) by the variable $x_{i,\sort(A_i)}$ for every $i \in [n]$.
Clearly,  $h(r) \in \T_\Sigma(X_u)$ with $u = s_1 \ldots s_n$.
Then we denote the $S$-sorted tree homomorphism induced by $h$ by $\pi_\Sigma$.
We note that $\pi_\Sigma: \T_R \rightarrow \T_\Sigma$ and we can say that $\pi_\Sigma(d)$ retrieves from each $d \in \T_R$ the tree in~$\T_\Sigma$ which is represented by~$d$.
\index{P@$\pi_\Sigma$}

\index{derivation}
It is obvious that each abstract syntax tree corresponds to a left derivation of the RTG and vice versa.

\index{regular tree grammar!generated language}
\index{L@$L(G,A)$}
\index{L@$L(G)$}
For every $A \in N$, the \emph{(formal) tree language generated by~$G$ from~$A$} is the set
\[
    L(G,A) = \{ \pi_\Sigma(d) \mid d \in (\T_R)_A\}\enspace.
\]
We note that,  if $A \in N_s$ for some $s \in S$, then $L(G,A) \subseteq (\T_\Sigma)_s$.
The \emph{(formal) tree language generated by~$G$} is the set $L(G) = L(G,A_0)$.
\index{RTG!unambiguous}
We call $G$ \emph{unambiguous} if for each $t \in L(G)$ there is a unique $d \in (\T_R)_{A_0}$ such that $\pi_\Sigma(d)=t$.

It was proved in~\cite[Theorem~3.16]{Brainerd1969} (also cf.~\cite[Theorem~3.22]{Engelfriet2015}) that for each regular tree grammar $G$ there is a regular tree grammar $G'$ such that $G'$ is in normal form and $L(G)=L(G')$. In a straightforward way, this result can be lifted to $S$-sorted RTG.

\section{Weighted RTG-based language models and the M-monoid parsing problem}
\label{sec:weighted-RTG-based-grammars}

In this section we introduce our framework of weighted RTG-based language models and use it do define the M-monoid parsing problem.
We compare our approach to interpreted regular tree grammars~\cite{KolKuh2011}, a similar framework which makes use of the initial algebra semantics~\cite{Goguen1977}, too.

\subsection{Weighted RTG-based language models}

We approach an algebraic definition of weighted grammars in two steps.
First we define RTG-based language models as an expressive grammar formalism and then we extend this definition by a weight component.

The idea behind RTG-based language models is to specify both the syntax of a language and the language itself within one formalism.
This is based on the \emph{initial algebra semantics}~\cite[Sect.~3.1]{Goguen1977}.
Here we use RTGs to describe the syntax and we use language algebras to generate the modeled language from these syntactic descriptions.

\index{RTG-based language model}
\index{RTG-LM}
\index{RTG-LM!language algebra}
Formally, an \emph{RTG-based language model} (RTG-LM) is a tuple $(G,(\alg L,\phi))$ where
\begin{itemize}
    \item $G=(N,\Sigma,A_0,R)$ is an $S$-sorted RTG for some set $S$ of sorts and
    \item $(\alg L,\phi)$ is an $S$-sorted $\Gamma$-algebra (\emph{language algebra}) such that $\Sigma \subseteq \Gamma$ (as $(S^* \times S)$-sorted set).
\end{itemize}

\index{L@$L(G)_{\alg L}$}
\index{language generated by~$(G,(\alg L,\phi))$}
\index{RTG-LM!generated language}
The \emph{language generated by~$(G,(\alg L,\phi))$}, denoted by $L(G)_{\alg L}$, is defined as
\[
    L(G)_{\alg L} = \{ \sem[\alg L]{\pi_\Sigma(d)} \mid d \in \AST(G) \}\enspace.
\]

We note that $L(G)_{\alg L} \subseteq \alg L_{\sort(A_0)}$.
\index{syntactic object}
We call the elements of~$\alg L$ \emph{syntactic objects}.
\index{$\AST(G, a)$}
For each $a \in \alg L$, we define the set of \emph{trees which evaluate to $a$} as $\AST(G, a) = \{ d \in \AST(G) \mid \pi_\Sigma(d)_{\alg L} = a \}$.
\index{RTG-LM!ambiguous}
We call $(G, (\alg L, \phi))$ \emph{ambiguous} if there are $d_1, d_2 \in \AST(G)$ such that $\sem[\alg L]{\pi_\Sigma(d_1)} = \sem[\alg L]{\pi_\Sigma(d_2)}$ and $d_2 \not= d_2$.
We note that there are two characteristics of ambiguity.
\begin{enumerate}
    \item There are a syntactic object $a \in \alg L$ and two trees $t_1, t_2 \in \T_\Sigma$ such that $\sem[\alg L]{(t_1)} = \sem[\alg L]{(t_2)} = a$ and $t_1 \not= t_2$.
        This mirrors semantic ambiguity in the modeled language.
        For instance, if $\alg L$ is a string language and $a$ a sentence, then $t_1$ and $t_2$ represent different groupings of the words in $a$ into constituents (cf.\ Figure~\ref{fig:asts}).
    \item There are $d_1, d_2 \in \AST(G)$ and a $t \in \T_\Sigma$ such that $\pi_\Sigma(d_1) = \pi_\Sigma(d_2) = t$ and $d_1 \not= d_2$.
        Then $d_1$ and $d_2$ represent the same syntactic description of the syntactic object $\sem[\alg L]{t}$, but that description may be obtained using different rules of the RTG $G$.
        This kind of ambiguity is called \emph{spurious ambiguity} and it is often not wanted.
\end{enumerate}
In the rest of this section, we will not differentiate between different kinds of ambiguity.
Methods for deciding or removing spurious ambiguity are beyond the scope of this paper.

Now we enrich RTG-LMs by a weight component.
This consists of an M-monoid (the weight algebra) for computing the weights of ASTs and a mapping that assigns to each rule of the RTG $G$ an M-monoid operation.

\index{weighted RTG-based language model}
\index{wRTG-LM}
\index{wt}
A \emph{weighted RTG-based language model} (wRTG-LM) is a tuple
\[
    \overline{G} = \big((G,(\alg L,\phi)), \ (\walg{K},\oplus,\welem{0},\Omega,\psi,\infsumop), \ \wt\big)\enspace,
\]
where
\begin{itemize}
    \item $(G,(\alg L,\phi))$ is an RTG-LM; we assume that $R$ is the set of rules of $G$,
    \item $(\walg{K},\oplus,\welem{0},\Omega,\psi,\infsum)$ is a complete M-monoid,
        and
    \item $\wt: R \rightarrow \Omega$ is mapping such that for each $r \in R$ the operation $\wt(r)$ has arity $\rk_R(r)$ (viewing $R$ as a ranked set).
        The tree relabeling induced by $\wt$ by is the mapping $\widetilde{\wt}: \T_R \rightarrow \T_\Omega$ .
        We denote $\widetilde{\wt}$ by $\wt$, too.
\end{itemize}
\index{language algebra}
\index{weight algebra}
We call
\begin{itemize}
    \item $(\alg L,\phi)$ the \emph{language algebra of $\overline{G}$},
    \item $(G,(\alg L,\phi))$ the \emph{RTG-LM of $\overline{G}$}, and \item $(\walg{K},\oplus,\welem{0},\Omega,\psi,\infsum)$ the \emph{weight algebra of~$\overline{G}$}.
\end{itemize}
If we abbreviate the two involved algebras by their respective carrier sets, then a wRTG-LM is specified by a tuple $((G,\alg L),\walg{K},\wt)$.

Intuitively, each wRTG-LM consists of two components:
\index{syntax component}
\index{weight component}
a \emph{syntax component} and a \emph{weight component}.
The syntax component is defined by the $\Sigma$-algebra $(\alg L, \phi)$ and the mapping $\pi_\Sigma: \T_R \to \T_\Sigma$.
The weight component is defined by the $\Omega$-algebra $(\walg K, \psi)$ and the mapping $\wt: \T_R \to \T_\Omega$.

\subsection{M-monoid parsing problem}

In the previous subsection we have introduced wRTG-LMs as the formal foundation of our approach to weighted parsing.
Now we will develop the weighted parsing problem that naturally emerges from wRTG-LMs.
We call this problem \emph{M-monoid parsing problem}.

Given a wRTG-LM \wrtglm\ and a syntactic object $a$, the relevant syntactic descriptions for parsing $a$ are the elements of the set $\AST(G,a)$, i.e., the set of ASTs of $G$ which evaluate to $a$.
We can map each tree from this set to a weight by first applying the tree relabeling $\wt$ to it and then evaluating the resulting tree over $\T_\Omega$ using the unique homomorphism $(.)_{\walg K}$.
Thus we obtain an $\AST(G,a)$-indexed family of elements of $\walg K$.
We note that since several ASTs can be evaluated to the same weight, it is not appropriate to use a set rather than a family here.
We accumulate this family of weights to a single element of $\walg K$ using the infinitary sum operation $\infsumop$.

\index{M-monoid parsing problem}
Formally, the \emph{M-monoid parsing problem} is the following problem. \\[3mm]
\textbf{Given:}
\begin{enumerate}
    \item a wRTG-LM $\big((G,(\alg L,\phi)),(\walg{K},\oplus,\welem{0},\Omega,\psi,\infsum),\wt\big)$    and
    \item an $a \in \alg L_{\sort(A_0)}$,
\end{enumerate}
\textbf{Compute:} $\displaystyle\fparse_{(G,\alg L)}(a) = \infsum_{d \in \AST(G, a)} \sem[\walg K]{\wt(d)}$ \enspace. \hfill \refstepcounter{equation}(\theequation)\label{eq:parsing-problem}
\\[3ex]
We note that for finite index sets, $\infsum$ can be replaced by~$\bigoplus$.
Whenever $(G,\alg L)$ is clear from the context, we will just write $\fparse$ rather than $\fparse_{(G,\alg L)}$.

In Figure~\ref{fig:overview} we illustrate how the syntax component and the weight component of the input wRTG-LM of the M-monoid parsing problem play together.

\begin{figure}
    \centering
    \begin{tikzpicture}[every node/.style={align=center}, remember picture]
        \node (asts) {trees generating $a$ \\[1ex] $\AST(G, a) \subseteq (\T_R)_{A_0}$};
        \node[below left=1cm and 0cm of asts] (pisigma) {subset of $\T_\Sigma$ \\[1ex] $\pi_\Sigma\big(\AST(G, a)\big) \subseteq L(G, A_0) \subseteq \T_\Sigma$};
        \node[below right=1cm and -0.2cm of asts] (wt) {family over $\T_\Omega$ \\[1ex] $\big(\wt(d) \mid d \in \AST(G, a)\big)$};
        \node[below=1.5cm of pisigma] (a) {syntactic object \\[1ex] $\{ a \} = \pi_\Sigma\big(\AST(G, a)\big)_{\alg L} \in L(G, A_0)_{\alg L} \subseteq \alg L_{\sort A_0}$};
        \node[below=1.5cm of wt] (k) {family over $\walg K$ \\[1ex] $\big(\wt(d)_{\walg K} \mid d \in \AST(G, a)\big)$};
        \node[below=1.5cm of k] (parse) {$\displaystyle\fparse(a) = \infsum_{d \in \AST(G, a)} \wt(d)_{\walg K}$};
        \draw[->] ($(asts.south)-(0.5,0)$) -- node[above left] {$\pi_\Sigma$} (pisigma);
        \draw[->] ($(asts.south)+(0.5,0)$) -- node[above right] {$\wt$} (wt);
        \draw[->] (pisigma) -- node[left] {$(.)_{\alg L}$} (a);
        \draw[->] (wt) -- node[right] {$(.)_{\walg K}$} (k);
        \draw[->] (a) -- node[below left] {$\fparse$} (parse);
        \draw[->] (k) -- node[right] {$\infsum$} (parse);
        \coordinate (syntax-c) at (asts.south -| a.east);
        \coordinate (weight-c) at (asts.south -| parse.west);
        \coordinate (syntax-o) at (asts.south -| a.west);
        \coordinate (weight-o) at (asts.south -| parse.east);
        \draw[dashed,color=gray,rounded corners] (syntax-c) rectangle ($(a.south west)-(0,0.1)$);
        \draw[dashed,color=gray,rounded corners] (weight-c) rectangle (parse.south east);
        \node[above right=0mm and 1mm of syntax-o,text=gray,align=left] {computations in the \\ syntax component};
        \node[above left=0mm and 1mm of weight-o,text=gray,align=right] {computations in the \\ weight component};
    \end{tikzpicture}
    \caption{Overview of the M-monoid parsing problem for a wRTG-LM $\big((G, \alg L), (\walg K, \oplus, \Omega, \infsum), \wt\big)$ and a syntactic object $a$.}
    \label{fig:overview}
\end{figure}

\begin{figure}
    \tikzset{lhs/.style={left,xshift=0.4em,fill=none}}
    \tikzset{rhs/.style={right,xshift=-0.4em,fill=none}}
    \tikzset{basic/.style={text height=2ex,text depth=0.5ex,fill=lightgray!70,rounded corners,font=\small}}
    \begin{adjustbox}{max width=\textwidth}
        \begin{tikzpicture}
            \node[label=above:$d \in \AST(G)$] (ast) {\scalebox{0.8}{\begin{tikzpicture}[level 1/.style={sibling distance=4.1cm},level 2/.style={sibling distance=2.1cm},every node/.style=basic]
                \node (s) {$\langle x_1 x_2 \rangle$}
                child {
                    node (np) {$\langle x_1 x_2 \rangle$}
                    child {
                        node (nn) {$\langle \fruit \rangle$}
                    }
                    child {
                        node (nns){$\langle \flies \rangle$}
                    }
                }
                child {
                    node (vp) {$\langle x_1 x_2 \rangle$}
                    child {
                        node (vbp) {$\langle \like \rangle$}
                    }
                    child {
                        node (np') {$\langle x_1 \rangle$}
                        child {
                            node (nns') {$\langle \bananas \rangle$}
                        }
                    }
                };
                \node[lhs] at (s.west) {$\nont{S} \to$};
                \node[lhs] at (np.west) {$\nont{NP} \to$};
                \node[lhs] at (nn.west) {$\nont{NN} \to$};
                \node[lhs] at (nns.west) {$\nont{NNS} \to$};
                \node[lhs] at (vp.west) {$\nont{VP} \to$};
                \node[lhs] at (vbp.west) {$\nont{VBP} \to$};
                \node[lhs] at (np'.west) {$\nont{NP} \to$};
                \node[lhs] at (nns'.west) {$\nont{NNS} \to$};
                \node[rhs] at (s.east) {$(\nont{NP},\nont{VP})$};
                \node[rhs] at (np.east) {$(\nont{NN},\nont{NNS})$};
                \node[rhs] at (vp.east) {$(\nont{VBP},\nont{NP})$};
                \node[rhs] at (np'.east) {$(\nont{NNS})$};
            \end{tikzpicture}}};
            \node[label=above:$t \in \T_\Sigma$] (pisigma) at ($(ast)+(-5.6,0)$) {\scalebox{0.8}{\begin{tikzpicture}[tree layout,every node/.style=basic]
                \node (s) {$\langle x_1 x_2 \rangle$}
                child {
                    node (np) {$\langle x_1 x_2 \rangle$}
                    child {
                        node (nn) {$\langle \fruit \rangle$}
                    }
                    child {
                        node (nns){$\langle \flies \rangle$}
                    }
                }
                child {
                    node (vp) {$\langle x_1 x_2 \rangle$}
                    child {
                        node (vbp) {$\langle \like \rangle$}
                    }
                    child {
                        node (np') {$\langle x_1 \rangle$}
                        child {
                            node (nns') {$\langle \bananas \rangle$}
                        }
                    }
                };
            \end{tikzpicture}}};
            \node[label=above:in $\T_{\Omega}$] (wt) at ($(ast)+(6,0)$) {\scalebox{0.8}{\begin{tikzpicture}[tree layout]
                \node (s) {$\tc{1.0}{r_1}$}
                child {
                    node (np) {$\tc{0.5}{r_3}$}
                    child {
                        node (nn) {$\tc{1.0}{r_8}$}
                    }
                    child {
                        node (nns){$\tc{0.4}{r_9}$}
                    }
                }
                child {
                    node (vp) {$\tc{0.6}{r_6}$}
                    child {
                        node (vbp) {$\tc{1.0}{r_{12}}$}
                    }
                    child {
                        node (np') {$\tc{0.3}{r_4}$}
                        child {
                            node (nns') {$\tc{0.6}{r_{10}}$}
                        }
                    }
                };
            \end{tikzpicture}}};
            \node[anchor=east] (bestderv) at ($(wt.east)+(0,-2.0)$) {$\big(0.0216, \{ r_1(r_3(r_8,r_9),r_6(r_{12},r_4(r_{10}))) \}\big)$};
            \node[anchor=east] (bestderv') at ($(bestderv.east)-(0,1.3)$) {$\big(0.0144, \{ r_1(r_2(r_8),r_5(r_{11},r_7(r_{13},r_4(r_{10})))) \}\big)$};
            \coordinate (bdmid) at ($(bestderv)!.5!(bestderv')$);
            \node (max) at (wt |- bdmid) {$\maxv\vphantom{f}$};
            \node[anchor=west] (fflb) at (pisigma.west |- max) {$a = \fruit\ \flies\ \like\ \bananas$};
            \node (wt') at ($(max)-(0,2.0)$) {$\wt(d') \in \T_\Omega$};
            \node (ast') at (ast |- wt') {$d' \in \AST(G)$};
            \node (pisigma') at (pisigma |- wt') {$\pi_\Sigma(d') \in \T_\Sigma$};
            \draw[rounded corners,->] ($(ast.west)+(0.5,1.0)$) -- node[above] {$\pi_\Sigma$} ($(pisigma.east)+(-0.5,1.0)$);
            \draw[rounded corners,->] ($(ast.east)+(-0.5,1.0)$) -- node[above] {$\wt$} ($(wt.west)+(0.5,1.0)$);
            \draw[->] ($(pisigma.south)+(0,0.7)$) -- node[left] {$(.)_{\lalg{CFG}^\Delta}$} (pisigma |- fflb.north);
            \draw[->] ($(wt.south)+(0,0.7)$) -- node[left] {$(.)_{\mathbb{BD}}$} (wt |- bestderv.north);
            \draw[->] (wt'.north) -- node[left] {$(.)_{\mathbb{BD}}$} (wt |- bestderv'.south);
            \draw[->] (ast'.east) -- node[above] {$\wt$} (wt'.west);
            \draw[->] (ast'.west) -- node[above] {$\pi_\Sigma$} (pisigma'.east);
            \draw[->] (pisigma'.north) -- node[left] {$(.)_{\lalg{CFG}^\Delta}$} (pisigma |- fflb.south);
            \draw[->] ($(fflb.east)+(1,0)$) -- node[above] {$\fparse$} ($(max.west)+(-1,0)$);
        \end{tikzpicture}
    \end{adjustbox}
    \caption{Illustration of the M-monoid parsing problem for the wRTG-LM $\big((G,\lalg{CFG}^\Delta),\walg{BD},\wt\big)$ and the syntactic element $a =  \fruit \ \flies \ \like \ \bananas$ of $\Delta^*$.}
    \label{fig:overview-ex}
\end{figure}

\begin{example}\label{ex:best-derivation-mmonoid}
    In the introduction we have mentioned the \emph{best parsing} problem.
    Given a grammar $G$ and a sentence $a$, the goal was to compute the highest probability $p$ among all constituent trees of $a$ in $G$ and the set of all constituent trees with probability $p$.
    Here we show that the best parsing problem is an instance of the M-monoid parsing problem.
    For this, we slightly modify the best parsing problem: instead of constituent trees, we compute ASTs.
    Due to our choice of the underlying RTG -- the nonterminals correspond to syntactic categories -- we can obtain from each AST one of the desired constituent trees.
    We note that this approach is common in practical applications of parsing and furthermore, we could directly compute constituent trees by employing a different weight algebra.
    Since ASTs correspond to derivations, our problem is called \emph{best derivation problem} (cf.~\cite[Figure~5]{Goodman1999}).

    In this example we define a wRTG-LM $\overline G$ for computing the best derivations of a grammar whose language contains, among others, the sentence $\fruit\ \flies\ \like\ \bananas$.
    We start by giving the syntax component which represents this particular grammar.
    Later we introduce the general \emph{best derivation M-monoid} and use it in the weight component.
    In the end, we compute the best derivation of $\fruit\ \flies\ \like\ \bananas$.

    We consider the $S$-sorted RTG $G = (N, \Sigma, \nont{S}, R)$ with a singleton set of sorts (e.g., $S = \{ \iota \}$).
    It is defined as follows.
    \begin{itemize}
        \item $N = N_\iota = \{\nont{S},\nont{NP},\nont{VP},\nont{PP},\nont{NN},\nont{NNS},\nont{VBZ},\nont{VBP},\nont{IN}\}$,
        \item $\Sigma = \Sigma_{(\iota\iota,\iota)} \cup  \Sigma_{(\iota,\iota)} \cup  \Sigma_{(\varepsilon,\iota)}$ and
            $\Sigma_{(\iota\iota,\iota)} =  \{ \sigma \}$,  $\Sigma_{(\iota,\iota)} = \{ \gamma \}$, and
            $\Sigma_{(\varepsilon,\iota)} = \{ \alpha_{\fruit}, \alpha_{\flies}, \alpha_{\like}, \alpha_{\bananas} \}$, and
        \item \(R\) contains the rules (ignoring the numbers above the arrows for the time being):
            \begin{figure}[h]
                \setlength{\abovedisplayskip}{0pt}\small
                \begin{alignat*}{6}
                    r_1:    \ &  &  \nont{S} &\stackrel{1.0}{\longrightarrow} \sigma(\nont{NP},\nont{VP}) & \quad
                    r_6:    \ &  &  \nont{VP} &\stackrel{0.6}{\longrightarrow} \sigma(\nont{VBP},\nont{NP}) & \quad
                    r_{11}: \ &  &  \nont{VBZ} &\stackrel{1.0}{\longrightarrow} \alpha_{\flies} \\
                    r_2:    \ &  &  \nont{NP} &\stackrel{0.2}{\longrightarrow} \gamma(\nont{NN}) & \quad
                    r_7:    \ &  &  \nont{PP} &\stackrel{1.0}{\longrightarrow} \sigma(\nont{IN},\nont{NP}) & \quad
                    r_{12}: \ &  &  \nont{VBP} &\stackrel{1.0}{\longrightarrow} \alpha_{\like} \\
                    r_3:    \ &  &  \nont{NP} &\stackrel{0.5}{\longrightarrow} \sigma(\nont{NN},\nont{NNS}) & \quad
                    r_8:    \ &  &  \nont{NN} &\stackrel{1.0}{\longrightarrow} \alpha_{\fruit} & \quad
                    r_{13}: \ &  &  \nont{IN} & \stackrel{1.0}{\longrightarrow} \alpha_{\like} \\
                    r_4:    \ &  &  \nont{NP} &\stackrel{0.3}{\longrightarrow} \gamma(\nont{NNS}) & \quad
                    r_9:    \ &  &  \nont{NNS} &\stackrel{0.4}{\longrightarrow} \alpha_{\flies} \\
                    r_5:    \ &  &  \nont{VP} &\stackrel{0.4}{\longrightarrow} \sigma(\nont{VBZ},\nont{PP}) & \quad
                    r_{10}: \ &  &  \nont{NNS} &\stackrel{0.6}{\longrightarrow} \alpha_{\bananas}
                    \enspace.
                \end{alignat*}
                %\caption{\label{fig:rules-ex} Rules of the RTG of Example~\ref{ex:cfg}.}
            \end{figure}
    \end{itemize}

    We define the language algebra $(\alg L, \phi)$ as a $\Sigma$-algebra with $\alg L = \{ \fruit, \flies, \like, \bananas \}^*$ and
    \begin{align*}
        \phi(\sigma)(a_1, a_2) &= a_1 a_2 && \text{for every $a_1, a_2 \in \alg L$} \\
        \phi(\gamma)(a) &= a && \text{for every $a \in \alg L$} \\
        \phi(\alpha_a) &= a && \text{for every $a \in \{ \fruit, \flies, \like, \bananas \}$.}
    \end{align*}

    \begin{sloppypar}
    Intuitively, $\alg L$ is a string algebra with the following capabilities.
    It can produce each of the syntactic objects $\fruit$, $\flies$, $\like$, and $\bananas$ using a constant operation (i.e., $\alpha_a$ for every $a \in \{ \fruit, \flies, \like, \bananas \}$).
    Furthermore, it can concatenate two syntactic objects (using $\sigma$) and contains an identity mapping (cf.\ $\gamma$).
    \end{sloppypar}

    We proceed to the definition of the best derivation M-monoid.
    We want to use this single M-monoid to describe the computation of the best derivation of every RTG-LM.
    For this we choose as carrier set an artificially large set and assume that it contains every rule of every RTG.

    Let $R_\infty$ be a ranked set such that $(R_\infty)_k$ is infinite for each $k \in \mathbb{N}$.
    \index{best derivation M-monoid}
    We define the \emph{best derivation M-monoid} to be the complete M-monoid
    \[
         \big(\gls{wclass:bd},\ \maxv,\ (0, \emptyset),\ \Omegav, \ \infsumop[\maxv]\big) \enspace,
    \]
    where
    \begin{itemize}
        \item $\mathbb{BD} = \mathbb R_0^1 \times \mathcal P(\T_{R_\infty})$,
        \item for every $(p_1, D_1), (p_2, D_2) \in \mathbb{BD}$,
            \[
                \maxv\big((p_1, D_1), (p_2, D_2)\big) = \begin{cases}
                    (p_1, D_1) &\text{if $p_1 > p_2$} \\
                    (p_2, D_2) &\text{if $p_1 < p_2$} \\
                    (p_1, D_1 \cup D_2) &\text{otherwise,}
                \end{cases}
            \]
        \item $\Omegav = \{ \tc{p}{r} \mid p \in \rzo \ \text{and} \ r \in R_\infty \}$, where for each $p \in \rzo$, $k \in \mathbb N$, and $r \in (R_\infty)_k$, we define \( \tc{p}{r}: \mathbb{BD}^{k} \to \mathbb{BD} \) (tc abbreviates top concatenation)
            such that for every $(p_1, D_1), \dots, (p_{k}, D_{k}) \in \mathbb{BD}$,
            \[
                \tc{p}{r}\big((p_1, D_1), \dots, (p_{k}, D_{k})\big) = (p',D')
            \]
            where $p'= p \cdot p_1 \cdot \ldots \cdot p_{k}$ and $D'=\{ r(d_1, \dots, d_{k}) \mid d_i \in D_i, 1 \le i \le k\}$,        and
        \item for every family $((p_i, D_i) \mid i \in I)$ over $\mathbb{BD}$, we define $\infsum[\maxv]_{i \in I} \welem (p_i, D_i) = (p, D)$, where $p = \sup \{ p_i \mid i \in I \}$ and $D = \bigcup_{i \in I: p_i = p} D_i$.
            (We note that this supremum exists because~$1$ is an upper bound of every subset of~$\rzo$ and every bounded subset of~$\mathbb R$ has a supremum.)
    \end{itemize}

    We finish the definition of the weight component of $\overline G$ by defining the mapping $\wt: R \to \Omegav$.
    Since $R_\infty$ is infinite, we can assume that $R_k \subseteq (R_\infty)_k$ for every $k \in \mathbb N$.
    We let $\wt(r) = \tc{p}{r}$ for every $r \in R$, where $p$ is shown above the arrow of~$r$.
    Intuitively, $\wt$ associates with each rule a pair where the first component is a number in $\rzo$ and the second component is a singleton set which contains the rule itself.

    We have shown an AST $d \in \AST(G, \fruit\ \flies\ \like\ \bananas)$ in the center of the upper row of Figure~\ref{fig:overview-ex}.
    To its left we have illustrated its evaluation to the syntactic object $a = \fruit\ \flies\ \like\ \bananas$ in the syntactic component.
    We obtain $\pi_\Sigma(d)$ by dropping the non-highlighted parts of $d$.
    The application of the homomorphism $(.)_{\alg L}: \T_\Sigma \to \alg L^*$ to $\pi_\Sigma(d)$ yields $a$.
    To the right of $d$ it can be seen how it is evaluated to $\big(0.0216, \{ r_1(r_3(r_8,r_9),r_6(r_{12},r_4(r_{10}))) \}\big)$ in the weight component.
    The probability of $d$ (i.e., the real number $0.0216$) is obtained as the product of the numbers which are associated to the rules occurring in $d$.
    The set of ASTs of $a$ with this probability consists only of $d$.
    This holds for every $d \in \AST(G,a)$.

    In the lower row of Figure~\ref{fig:overview-ex} we have indicated that there is a second AST $d'$ which is evaluated to~$a$, too.
    We obtain
    \[
        \wt(d')_{\walg{BD}} = (0.0144, \{ r_1(r_2(r_8),r_5(r_{11},r_7(r_{13},r_4(r_{10})))) \}) \enspace.
    \]
    Thus
    \(
        \maxv \big(\wt(d)_{\walg{BD}}, \wt(d')_{\walg{BD}}\big) = \wt(d)_{\walg{BD}}
    \).
    As one might expect, it is more likely that~$a$ refers to the preferences (to $\like\ \bananas$) of certain insects ($\fruit\ \flies$).
\end{example}

\subsection{Comparison with interpreted regular tree grammars (IRTG)}

\begin{table}
    \centering\small
    \begin{tabular}{lll}
        \toprule
        \multicolumn{2}{l}{wRTG-LM} & IRTG \\
        \midrule
        \multicolumn{2}{l}{$\overline G = \big((G, \alg L), (\walg K, \oplus, \mathbb 0, \Omega, \infsum), \wt\big)$} & $\overline G = (G, \mathcal I_1, \mathcal I_2)$ \\[1ex]
        \multicolumn{2}{l}{RTG $G = (N, \Sigma, A_0, R)$}  & RTG $G = (N, \Sigma, A_0, R)$ \\
        \multicolumn{2}{l}{\textbullet abstract syntax trees $\AST(G)$}  & \textbullet tree language $L(G)$ \\[1ex]
        syntax component & weight component  & interpretation $\mathcal I_i = (h_i, \alg A_i)$ ($i \in [2]$) \\
        \textbullet tree relabeling $\pi_\Sigma: \T_R \to \T_\Sigma$ & \textbullet tree relabeling $\wt: \T_R \to \T_\Omega$  & \textbullet tree homomorphism $h_i: \T_\Sigma \to \T_{\Delta_i}$ \\
        \textbullet $\Gamma$-algebra $\alg L$ ($\Sigma \subseteq \Gamma$\,) & \textbullet complete M-monoid $\walg K$  & \textbullet $\Delta_i$-algebra $\alg A_i$ \\
        \textbullet evaluation $(.)_{\alg L}: \T_\Sigma \to \alg L$ & \textbullet evaluation $(.)_{\walg K}: \T_R \to \walg K$ & \textbullet evaluation $(.)_{\alg A_i}: \T_{\Delta_i} \to \alg A_i$ \\[1ex]
        \multicolumn{2}{l}{$L(G)_{\alg L} = \{ \pi_\Sigma(d)_{\alg L} \mid d \in \AST(G) \}$}  & $L(\overline G) = \{ (h_1(t)_{\alg A_1}, h_2(t)_{\alg A_2}) \mid t \in L(G) \}$ \\
        \bottomrule
    \end{tabular}
    \caption{Comparison of a wRTG-LM to an IRTG with two interpretations.}
    \label{tab:comp-irtg-wrtglm}
\end{table}

We compare our framework of wRTG-LMs with \emph{interpreted regular tree grammars} (IRTGs, \autocite{KolKuh2011}).
For this, we briefly recall the basic notions of IRTGs.
An IRTG $\overline G$ consists of an RTG $G = (N, \Sigma, A_0, R)$ and several \emph{interpretations}.
Each interpretation is a pair $(h, \alg A)$, where $h: \T_\Sigma \to \T_\Delta$ is a tree homomorphism and $\alg A$ is a $\Delta$-algebra.
The language generated by $\overline G$ is the set of all tuples which are obtained by interpreting trees of $L(G)$ in the several algebras.
Formally, if $\overline G$ consists of the interpretations $(h_1, \alg A_1), \dots, (h_n, \alg A_n)$ with $n \in \mathbb N$, then the language generated by $\overline G$ is the set
\[
    L(\overline G) = \{ (h_1(t)_{\alg A_1}, \dots, h_n(t)_{\alg A_n}) \mid t \in L(G) \} \enspace.
\]
In the right column of Table~\ref{tab:comp-irtg-wrtglm}, we illustrate the concept of IRTGs for the special case of two interpretations (i.e., $n = 2$).

In our comparison of wRTG-LMs and IRTGs, we consider wRTG-LMs as IRTGs with two $\Sigma$\hyp{}interpretations.
We view each wRTG-LM $\big((G, \alg L), (\walg K, \oplus, \mathbb 0, \Omega, \infsum), \wt\big)$ as the IRTG $\big(G, (\pi_\Sigma, \alg L), (\wt, \walg K)\big)$.
This is done as shown in Table~\ref{tab:comp-irtg-wrtglm}:
\begin{itemize}
    \item the wRTG-LM and the IRTG consist of the same RTG $G = (N, \Sigma, A_0, R)$,
    \item the syntax component corresponds to the first interpretation $\mathcal I_1$, and
    \item the weight component corresponds to the second interpretation $\mathcal I_2$.
\end{itemize}

We point out that this view of wRTG-LMs as IRTGs does not conform to the definition of IRTGs.
While the core component of a wRTG-LM is the set $\AST(G)$ of abstract syntax trees, the core component of an IRTG is the tree language $L(G)$.
A second, minor difference is that the language of an IRTG consists of tuples of interpreted trees, while the language of a wRTG-LM consists of syntactic objects (i.e., trees evaluated in the language algebra).

Finally, we compare the M-monoid parsing problem to the \emph{decoding} problem of IRTGs.
Decoding is motivated by modeling translation between natural languages using synchronous grammars.
It is defined as follows.
Given an IRTG $\overline G = \big(G, (h_1, \alg A_1), (h_2, \alg A_2)\big)$ and a syntactic object $a$, compute the set
\[
    \decodes(a) = \{ h_2(t)_{\alg A_2} \mid t \in L(G) \land h_1(t)_{\alg A_1} = a \} \enspace.
\]
Compared to the M-monoid parsing problem, we consider the language algebra $\alg L$ as the input language and the weight algebra $\walg K$ as the output language of our translation.
We can derive the M-monoid parsing problem from the IRTG decoding problem by applying two changes.
First, we need to compute a family of elements of $\alg A_2$ rather than a set.
This is because in the M-monoid parsing problem, if several abstract syntax trees have the same weight, then this weight contributes to the value of $\fparse(a)$ multiple times.
Second, we map this family to a single element of $\alg A_2$ using the infinitary sum operation.
The application of these transformations yields Equation~\eqref{eq:parsing-problem}.

\section{Classes of weighted RTG-based language models}

In this section we define several subclasses of wRTG-LMs. 
For this, we use two parameters:
\begin{enumerate}
    \item a subclass $\gclass{}$ of the class of all RTG-LMs $\gls{gclass:all}$ and
    \item a subclass $\wclass{}$ of the class of all complete M-monoids $\gls{wclass:all}$.
\end{enumerate}
\index{G@$(\gclass{}, \wclass{})$-LM}
Now let $\gclass{} \subseteq \gclass{all}$ and $\wclass{} \subseteq \wclass{all}$.
Then a \emph{$(\gclass{},\wclass{})$-LM} is a wRTG-LM $\big((G,(\alg L,\phi)), \ (\walg{K},\oplus,\welem{0},\Omega,\infsumop), \ \wt\big)$ such that
\begin{enumerate}
    \item its RTG-LM $(G,(\alg L,\phi))$ is in $\gclass{}$ and
    \item its weight algebra $(\walg{K},\oplus,\welem{0},\Omega,\infsumop)$ is in $\wclass{}$.
\end{enumerate}
\index{W@$\wlmclass{\gclass{},\wclass{}}$}
We denote the class of all $(\gclass{},\wclass{})$-LMs by $\wlmclass{\gclass{},\wclass{}}$.

Moreover, we will introduce the subclass $\wlmclass[\clsd]{\gclass{}, \wclass{}}$ which imposes an additional restriction on wRTG-LMs.
This class is central to the termination and correctness of the M-monoid parsing algorithm.

\subsection{Classes of RTG-based language models}
\label{sec:classes-rtglms}

In this subsection we recall four particular classes of RTG-LMs:
context-free grammars, linear context-free rewriting systems, tree-adjoining grammars, and yield-grammars.
We mention that also context-free hypergraph grammars~\cite{baucou87,habkre87} can be viewed as RTG-LMs \cite{cou91} (also cf. \cite{dregebvog16}). Each of these classes is determined by a particular class of language algebras. Additionally, in Subsection \ref{subsect:general-classes}, we define three more classes of RTG-LMs which are determined by (a) particular subclasses of regular tree grammars and (b) by an interplay between the involved RTG and the language algebra.

\subsubsection{The CFG-algebras and context-free grammars}

It was suggested in \cite[Sect.~3.1]{Goguen1977} to consider context-free languages as initial many-sorted algebra semantics of context-free grammars.
The context-free grammars are here replaced by RTG.

Let $\Delta$ be a finite set and $S=\{\iota\}$ be a set of sorts (for some arbitrary but fixed $\iota$).
We let $X = \{x_1,x_2,\ldots\}$ be a set of variables.
These variables will be used to denote strings over $\Delta$. For each $k \in \mathbb{N}$, we let $X_k=\{x_1,\ldots,x_k\}$.

We define the $(\{\iota\}^*\times \{\iota\})$-sorted set $\Gamma^{\lalg{CFG},\Delta}$  such that for each $k \ge 0$:
\begin{align*}
    (\Gamma^{\lalg{CFG},\Delta})_{(\iota^k,\iota)} = \{ \langle w \rangle \mid & \;  w = v_0 x_1 v_1 \ldots x_k v_k \text{ for some } v_0, \ldots, v_k \in \Delta^*\}\enspace.
\end{align*}

We define the \emph{CFG-algebra over $\Delta$}%
\index{CFG-algebra over $\Delta$}
to be the $\{\iota\}$-sorted $\Gamma^{\lalg{CFG},\Delta}$-algebra
\((\gls{alg:cfg},\phi)\)
with
\begin{itemize}
    \item $\lalg{CFG}^\Delta = (\lalg{CFG}^\Delta)_\iota = \Delta^*$.
    \item For every $k \in \mathbb N$, $\langle w\rangle \in (\Gamma^{\lalg{CFG},\Delta})_{(\iota^k,\iota)}$, and $u_1,\ldots,u_k \in \Delta^*$ we define
        \[\phi(\langle w\rangle)(u_1,\ldots,u_k) = w'
        \]
        where $w'$ is obtained from $w$ by replacing each $x_i$ by $u_i$ for each $i \in [k]$.
\end{itemize}

A \emph{context-free grammar over $\Delta$}
\index{context-free grammar}%
\index{CFG}%
is an RTG-based LM
\[(G,(\lalg{CFG}^\Delta,\phi))\]
where the $S$-sorted RTG $G$ is in normal form.
\index{context-free language}
A \emph{context-free language} is the formal language generated by some context-free grammar.

We note that the language $L(G)_{\lalg{CFG}^\Delta}$ generated by this context-free grammar is a formal language over $\Delta$.
We also note that, by definition of RTG-LMs, the terminal set $\Sigma$ of $G$ is a $(\{\iota\}^*\times \{\iota\})$-sorted subset of $\Gamma^{\lalg{CFG},\Delta}$.
Thus, for the specification of a particular context-free grammar, we only have to specify the $\Delta$ and an RTG.
We denote the class of all context-free grammars by $\gls{gclass:cfg}$.

Indeed, classical context-free grammars and those which are defined here are in the following, easy one-to-one correspondence. Let $G= (N,\Delta,A_0,R)$ be a usual context-free grammar and let $(G',(\lalg{CFG}^\Delta,\phi)$  be a context-free grammar (as defined here) where  $G'=(N,\Sigma,A_0,R')$.   We say that $G$ and $G'$ \emph{correspond to each other} if the following two statements are equivalent for every $k \in \mathbb{N}$, $A_1,\ldots,A_k \in N$, and $v_0, \ldots, v_k \in \Delta^*$:
\begin{enumerate}
    \item $A \rightarrow v_0 A_1 v_1 \ldots A_k v_k$ is in $R$.
    \item $A \rightarrow \sigma(A_1,\ldots,A_k)$ is in $R'$ with   $\sigma = \langle v_0 x_1 v_1 \ldots x_k v_k\rangle$.
\end{enumerate}
Then, clearly, the languages generated by $G$ and $(G',(\lalg{CFG}^\Delta,\phi))$ are the same.

\begin{example}
    \label{ex:cfg}
    We let~$\Delta = \{ \term{Fruit},\term{flies},\term{like},\term{bananas} \}$.
    We consider the $\{\iota\}$-sorted RTG $G = (N, \Sigma, \nont{S}, R)$ and the language algebra $(\alg L, \phi)$ from Example~\ref{ex:best-derivation-mmonoid}.
    We observe that $\big(G, (\alg L, \phi)\big)$ is a context-free grammar.
    This can be seen by letting $\sigma = \langle x_1 x_2 \rangle$, $\gamma = \langle x_1 \rangle$, and $\alpha_a = \langle a \rangle$ for every $a \in \Delta$.
    Then $(\alg L, \phi) = \lalg{CFG}^\Delta$.

    \begin{figure}
        \begin{adjustbox}{center}
            \begin{tikzpicture}
                \node (ast) {\scalebox{0.8}{\begin{tikzpicture}[level 1/.style={sibling distance=5.0cm},level 2/.style={sibling distance=2.5cm},every node/.style=basic]
                    \node (s) {$\langle x_1 x_2 \rangle$}
                    child {
                        node (np) {$\langle x_1 x_2 \rangle$}
                        child {
                            node (nn) {$\langle \fruit \rangle$}
                        }
                        child {
                            node (nns){$\langle \flies \rangle$}
                        }
                    }
                    child {
                        node (vp) {$\langle x_1 x_2 \rangle$}
                        child {
                            node (vbp) {$\langle \like \rangle$}
                        }
                        child {
                            node (np') {$\langle x_1 \rangle$}
                            child {
                                node (nns') {$\langle \bananas \rangle$}
                            }
                        }
                    };
                    \node[lhs] at (s.west) {$\nont{S} \to$};
                    \node[lhs] at (np.west) {$\nont{NP} \to$};
                    \node[lhs] at (nn.west) {$\nont{NN} \to$};
                    \node[lhs] at (nns.west) {$\nont{NNS} \to$};
                    \node[lhs] at (vp.west) {$\nont{VP} \to$};
                    \node[lhs] at (vbp.west) {$\nont{VBP} \to$};
                    \node[lhs] at (np'.west) {$\nont{NP} \to$};
                    \node[lhs] at (nns'.west) {$\nont{NNS} \to$};
                    \node[rhs] at (s.east) {$(\nont{NP},\nont{VP})$};
                    \node[rhs] at (np.east) {$(\nont{NN},\nont{NNS})$};
                    \node[rhs] at (vp.east) {$(\nont{VBP},\nont{NP})$};
                    \node[rhs] at (np'.east) {$(\nont{NNS})$};
                \end{tikzpicture}}};
                \node[right=2cm of ast] (pisigma) {\scalebox{0.8}{\begin{tikzpicture}[tree layout,every node/.style=basic]
                    \node (s) {$\langle x_1 x_2 \rangle$}
                    child {
                        node (np) {$\langle x_1 x_2 \rangle$}
                        child {
                            node (nn) {$\langle \fruit \rangle$}
                        }
                        child {
                            node (nns){$\langle \flies \rangle$}
                        }
                    }
                    child {
                        node (vp) {$\langle x_1 x_2 \rangle$}
                        child {
                            node (vbp) {$\langle \like \rangle$}
                        }
                        child {
                            node (np') {$\langle x_1 \rangle$}
                            child {
                                node (nns') {$\langle \bananas \rangle$}
                            }
                        }
                    };
                \end{tikzpicture}}};
                \node[below=1.5cm of ast] (ast') {\scalebox{0.8}{\begin{tikzpicture}[level 1/.style={sibling distance=4.2cm},level 2/.style={sibling distance=2.5cm},every node/.style=basic]
                    \node (s) {$\langle x_1 x_2 \rangle$}
                    child {
                        node (np) {$\langle x_1 \rangle$}
                        child {
                            node (nn) {$\langle \fruit \rangle$}
                        }
                    }
                    child {
                        node (vp) {$\langle x_1 x_2 \rangle$}
                        child {
                            node (vbz) {$\langle \flies \rangle$}
                        }
                        child {
                            node (pp) {$\langle x_1 x_2 \rangle$}
                            child {
                                node (in) {$\langle \like \rangle$}
                            }
                            child {
                                node (np') {$\langle x_1 \rangle$}
                                child {
                                    node (nns) {$\langle \bananas \rangle$}
                                }
                            }
                        }
                    };
                    \node[lhs] at (s.west) {$\nont{S} \to$};
                    \node[lhs] at (np.west) {$\nont{NP} \to$};
                    \node[lhs] at (nn.west) {$\nont{NN} \to$};
                    \node[lhs] at (vp.west) {$\nont{VP} \to$};
                    \node[lhs] at (vbz.west) {$\nont{VBZ} \to$};
                    \node[lhs] at (pp.west) {$\nont{PP} \to$};
                    \node[lhs] at (in.west) {$\nont{IN} \to$};
                    \node[lhs] at (np'.west) {$\nont{NP} \to$};
                    \node[lhs] at (nns.west) {$\nont{NNS} \to$};
                    \node[rhs] at (s.east) {$(\nont{NP},\nont{VP})$};
                    \node[rhs] at (np.east) {$(\nont{NN})$};
                    \node[rhs] at (vp.east) {$(\nont{VBZ},\nont{PP})$};
                    \node[rhs] at (pp.east) {$(\nont{IN},\nont{NP})$};
                    \node[rhs] at (np'.east) {$(\nont{NNS})$};
                \end{tikzpicture}}};
                \node (pisigma') at (ast' -| pisigma) {\scalebox{0.8}{\begin{tikzpicture}[tree layout,every node/.style=basic]
                    \node (s) {$\langle x_1 x_2 \rangle$}
                    child {
                        node (np) {$\langle x_1 \rangle$}
                        child {
                            node (nn) {$\langle \fruit \rangle$}
                        }
                    }
                    child {
                        node (vp) {$\langle x_1 x_2 \rangle$}
                        child {
                            node (vbz) {$\langle \flies \rangle$}
                        }
                        child {
                            node (pp) {$\langle x_1 x_2 \rangle$}
                            child {
                                node (in) {$\langle \like \rangle$}
                            }
                            child {
                                node (np') {$\langle x_1 \rangle$}
                                child {
                                    node (nns) {$\langle \bananas \rangle$}
                                }
                            }
                        }
                    };
                \end{tikzpicture}}};
                \node[font=\small,anchor=north west] at (ast.north west) {$d \in \AST(G)$};
                \node[font=\small,anchor=north west] at (ast'.north west) {$d' \in \AST(G)$};
                \node[font=\small,anchor=north east] at (pisigma.north east) {$t \in \T_\Sigma$};
                \node[font=\small,anchor=north east] at (pisigma'.north east) {$t' \in \T_\Sigma$};
                \coordinate (pisigma-south) at ($(pisigma.south)+(0,0.5)$);
                \coordinate (fflbpos) at ($(pisigma-south)!.5!(pisigma'.north)$);
                \node (fflb) at (fflbpos) {$a = \fruit\ \flies\ \like\ \bananas$};
                \draw[->] (ast) -- node[above] {$\pi_\Sigma$} (pisigma);
                \draw[->] (pisigma-south) -- node[left] {$(.)_{\lalg{CFG}^\Delta}$} (fflb);
                \draw[->] (ast') -- node[above] {$\pi_\Sigma$} (pisigma');
                \draw[->] (pisigma') -- node[left] {$(.)_{\lalg{CFG}^\Delta}$} (fflb);
            \end{tikzpicture}
        \end{adjustbox}
        \caption{Two abstract syntax trees for the syntactic object $a =  \fruit \ \flies \ \like \ \bananas$ in the RTG-LM $(G,\lalg{CFG}^\Delta)$ and their evaluation in the $\lalg{CFG}^\Delta$-algebra, see Example~\ref{ex:cfg}.}
        \label{fig:cfg-lm}
    \end{figure}
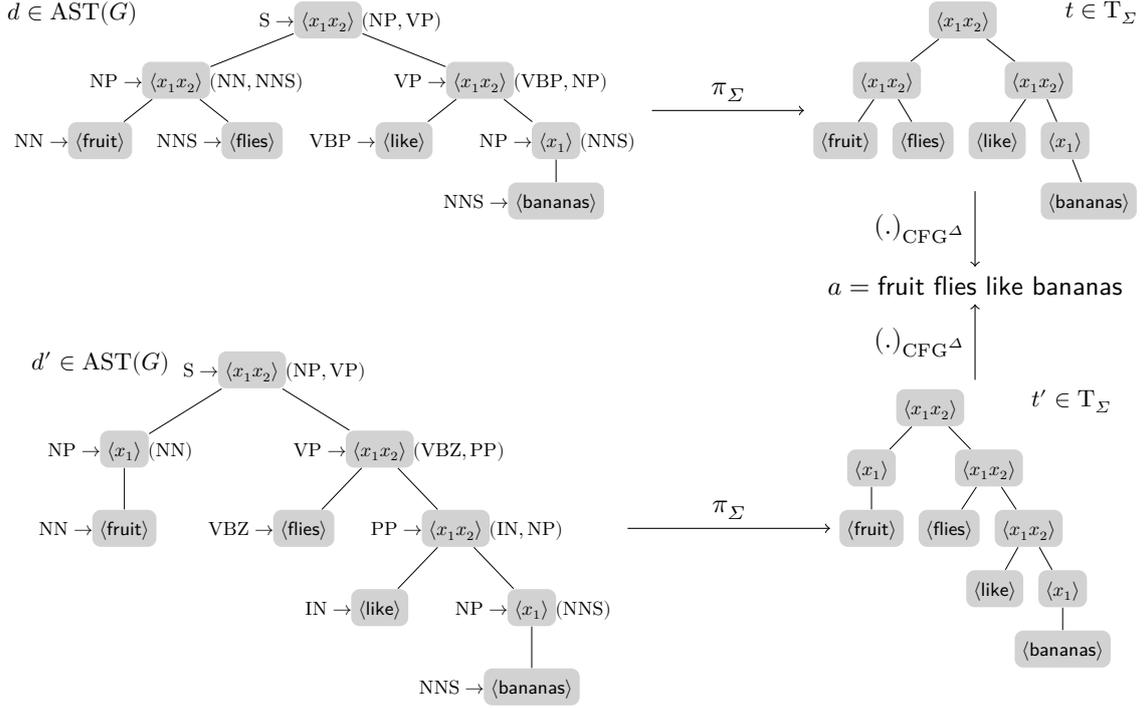

    In Figure~\ref{fig:cfg-lm} we have again illustrated the ASTs $d$ and $d'$ from Figure~\ref{fig:overview-ex} and their evaluation in the syntactic component.
    This time we have used the notions of $\lalg{CFG}^\Delta$ and also shown $d'$ and $\pi_\Sigma(d')$ entirely.
    The AST $d$ in the top row expresses that certain insects ($\fruit \ \flies$) like something ($\bananas$).
    The AST $d'$ in the bottom row expresses how $\fruit$ performs a certain activity (to fly like bananas).
    Hence this RTG-LM is ambiguous.
\end{example}

\subsubsection{The LCFRS-algebras and linear context-free rewriting systems}

The formalization of context-free grammars using the initial algebra semantics can be generalized to (string) linear context-free rewriting systems in a straightforward way.
A formal definition was given by \textcite[Def.~6.2+6.3]{kal10}.
Here we will embed it into our framework of wRTG-LM.

Let $\Delta$ be an alphabet and $S = \mathbb{N}$ be a set of sorts.
In this section it is convenient to use a doubly indexed set of variables.
Let \(k \in \mathbb{N}\) and $l_1,\ldots, l_k \in \mathbb{N}$.
We denote by $X_{l_1,\ldots, l_k}$ the set
\[
    X_{l_1,\ldots, l_k} = \{x^{(i)}_j \mid i \in [k], j \in [l_i]\} \enspace.
\]
Intuitively, each $x^{(i)}_j$ denotes a string and each $x^{(i)}$ represents an $l_i$-tuple of strings.

We define the $(\mathbb{N}^*\times \mathbb{N})$-sorted set $\Gamma^{\lalg{LCFRS},\Delta}$ such that for each $k,n,l_1,\ldots, l_k \in \mathbb{N}$:
\begin{align*}
    (\Gamma^{\lalg{LCFRS},\Delta})_{(l_1\ldots l_k,n)} = \{ \langle w_1,\ldots, w_n\rangle \mid{}  &w_i \in (\Delta \cup X_{l_1,\ldots, l_k})^* \text{ and each variable}\\
    & \text{$x^{(i)}_j\in X_{l_1,\ldots,l_k}$ occurs exactly once in $w_1\ldots w_n$}\} \enspace.
\end{align*}

\index{LCFRS-algebra over $\Delta$}
We define the \emph{LCFRS-algebra over $\Delta$} to be the $\mathbb{N}$-sorted $\Gamma^{\lalg{LCFRS},\Delta}$-algebra $(\gls{alg:lcfrs},\phi)$ with
\begin{itemize}
    \item $\lalg{LCFRS}^\Delta = \bigcup_{n \in \mathbb{N}} (\lalg{LCFRS}^\Delta)_n$ where $(\lalg{LCFRS}^\Delta)_n=(\Delta^*)^n$.
    \item For every $\langle w_1,\ldots, w_n\rangle \in (\Gamma^{\lalg{LCFRS},\Delta})_{(l_1\ldots l_k,n)}$ and  $u^{(1)}_1,\ldots,u^{(1)}_{l_1}, \ldots, u^{(k)}_1,\ldots,u^{(k)}_{l_k} \in \Delta^*$ we define
        \[\phi(\langle w_1,\ldots,w_n\rangle)((u^{(1)}_1,\ldots,u^{(1)}_{l_1}), \ldots, (u^{(k)}_1,\ldots,u^{(k)}_{l_k})) = (w_1',\ldots, w_n')
        \]
        where $w_\kappa'$ ($\kappa \in [n]$) is obtained from $w_\kappa$ by replacing each $x^{(i)}_j$ by $u^{(i)}_j$ ($i \in [k]$, $j \in [l_i]$).
\end{itemize}

\index{linear context-free rewriting system}
\index{LCFRS}
A \emph{linear context-free rewriting system over $\Delta$} is an RTG-LM
\[
    (G,(\lalg{LCFRS}^\Delta,\phi))
\]
where the $\mathbb N$-sorted RTG $G=(N,\Sigma,A_0,R)$ is in normal form and $A_0 \in N_1$.
We note that the language $L(G)_{\lalg{LCFRS}^\Delta}$ generated by this linear context-free rewriting system is a formal language over $\Delta$.

\index{LCFRS!fan-out}
For each $l \in \mathbb N$ and $A \in N_l$ we call $l$ the \emph{fan-out of $A$};
the \emph{fan-out of $G$} is the maximal fan-out of all nonterminals in $N$.

We denote the class of all linear context-free rewriting systems by $\gls{gclass:lcfrs}$.

Intuitively it is clear that, for each context-free grammar $(G,(\lalg{CFG}^\Delta,\phi))$ over $\Delta$ there is an linear context-free rewriting system $(G',(\lalg{LCFRS}^\Delta,\phi))$ over $\Delta$ in which each nonterminal has fan-out 1, which generates the same language as $(G,(\lalg{CFG}^\Delta,\phi))$.
In fact, if we identify the sort $\iota$ of $G$ with the sort $1$ of $G'$, then $G=G'$.
This also holds the other way around if the variables in the $\Sigma$-symbol of each rule occur in the order $x^{(1)}_1,x^{(2)}_1,\ldots,x^{(k)}_1$.

\begin{example}\label{ex:lcfrs}
    We consider the set $\Delta =\{ \tzag, \thelpen, \tlezen, \tJan, \tPiet, \tMarie\}$ and the following $\mathbb{N}$-sorted RTG $G=(N,\Sigma,\mathrm{root},R)$
    with
    \begin{itemize}
        \item $N = N_1 \cup N_2$ and $N_1 = \{\mathrm{root}, \mathrm{nsub}\}$ and $N_2 = \{\mathrm{dobj}\}$,
        \item $\Sigma = \Sigma_{(12,1)} \cup \Sigma_{(12,2)} \cup \Sigma_{(1,2)} \cup \Sigma_{(\varepsilon,1)}$   where
            \begin{align*}
                \Sigma_{(12,1)} &= \{\langle x^{(1)}_1 x^{(2)}_1  \tzag\ x^{(2)}_2\rangle\}\\
                \Sigma_{(12,2)} &= \{\langle x^{(1)}_1 x^{(2)}_1, \thelpen\ x^{(2)}_2\rangle\}\\
                \Sigma_{(1,2)} &= \{\langle x^{(1)}_1,  \tlezen\rangle\}\\
                \Sigma_{(\varepsilon,1)} &= \{\langle \tJan\rangle, \langle \tPiet\rangle, \langle \tMarie\rangle\}.
            \end{align*}
            (We note that $\Sigma$ is an $(\mathbb{N}^*\times \mathbb{N})$-sorted subset of $\Gamma^{\lalg{LCFRS},\Delta}$.)
        \item $R$ is the following set of rules:
            \begin{alignat*}{4}
                r_1: \ & & \mathrm{root} & \to \langle x^{(1)}_1 x^{(2)}_1  \tzag\ x^{(2)}_2\rangle (\textrm{nsub}, \textrm{dobj}) \qquad &
                r_2: \ & & \mathrm{nsub} & \to \langle \tJan\rangle \\
                r_3: \ & & \mathrm{dobj} & \to \langle x^{(1)}_1 x^{(2)}_1, \thelpen\ x^{(2)}_2\rangle (\textrm{nsub}, \textrm{dobj}) \qquad &
                r_4: \ & & \mathrm{nsub} & \to \langle \tPiet\rangle \\
                r_5: \ & & \mathrm{dobj} & \to \langle x^{(1)}_1, \tlezen\rangle (\textrm{nsub}) \qquad &
                r_6: \ & & \mathrm{nsub} & \to \langle \tMarie\rangle
            \end{alignat*}
    \end{itemize}
    Then \(d = r_1(r_4,r_2(r_5,r_3(r_6)))\) is an example of an abstract syntax tree in $\T_R$.
    We have illustrated $d$ and $\pi_\Sigma(d)$ in Figure~\ref{fig:ex-lcfrs}.
    Clearly, $\sem[\lalg{LCFRS}^\Delta]{\pi_\Sigma(d)} = \tJan \ \tPiet \ \tMarie \ \tzag \ \thelpen \ \tlezen$. \qedhere

    \begin{figure}
        \centering
        \begin{tikzpicture}
            \node (ast) {
                \begin{tikzpicture}[every node/.style=basic,level 1/.style={sibling distance=3.5cm},level 2/.style={sibling distance=2.8cm}]
                    \node (zag) {$\langle x^{(1)}_1 x^{(2)}_1 \tzag\ x^{(2)}_2\rangle$}
                    child {
                        node (Jan) {$\langle \tJan \rangle$}
                    }
                    child {
                        node (helpen) {$\langle x^{(1)}_1 x^{(2)}_1, \thelpen\ x^{(2)}_2\rangle$}
                        child {
                            node (Piet) {$\tPiet$}
                        }
                        child {
                            node (lezen) {$\langle x^{(1)}_1, \tlezen\rangle$}
                            child {
                                node (Marie) {$\tMarie$}
                            }
                        }
                    };
                    \node[lhs] at (zag.west) {$\nont{root} \to$};
                    \node[lhs] at (helpen.west) {$\nont{dobj} \to$};
                    \node[lhs] at (lezen.west) {$\nont{dobj} \to$};
                    \node[lhs] at (Jan.west) {$\nont{nsub} \to$};
                    \node[lhs] at (Piet.west) {$\nont{nsub} \to$};
                    \node[lhs] at (Marie.west) {$\nont{nsub} \to$};
                    \node[rhs] at (zag.east) {$(\nont{nsub}, \nont{dobj})$};
                    \node[rhs] at (helpen.east) {$(\nont{nsub}, \nont{dobj})$};
                    \node[rhs] at (lezen.east) {$(\nont{nsub})$};
                \end{tikzpicture}};
            \node at (8,0) (pisigma) {
                \begin{tikzpicture}[every node/.style=basic,tree layout]
                    \graph{
                        "$\langle x^{(1)}_1 x^{(2)}_1 \tzag\ x^{(2)}_2\rangle$" -- { "$\langle \tJan \rangle$" , "$\langle x^{(1)}_1 x^{(2)}_1, \thelpen\ x^{(2)}_2\rangle$" -- { "$\tPiet$" , "$\langle x^{(1)}_1, \tlezen\rangle$" -- "$\tMarie$" } }
                    };
                \end{tikzpicture}};
            \node[below=1cm of pisigma] (sentence) {$a = \tJan\ \tPiet\ \tMarie\ \tzag\ \thelpen\ \tlezen$};
            \draw[->] (ast) to node[above] {$\pi_\Sigma$} (pisigma);
            \draw[->] (pisigma) to node[left] {$\sem[\lalg{LCFRS}^\Delta]{(.)}$} (sentence);
            \node[font=\small,anchor=north west] at ($(ast.north west)+(0,0.5)$) {$d \in \AST(G)$};
            \node[font=\small,anchor=north east] at ($(pisigma.north east)+(0,0.5)$) {$t \in \T_\Sigma$};
        \end{tikzpicture}
        \caption{An abstract syntax tree $d$ for the syntactic object $a = \tJan\ \tPiet\ \tMarie\ \tzag\ \thelpen\ \tlezen$ in the RTG-LM $(G, \lalg{LCFRS}^\Delta)$ and its evaluation in the $\lalg{LCFRS}^\Delta$-algebra, see Example~\ref{ex:lcfrs}.}
        \label{fig:ex-lcfrs}
      \end{figure}
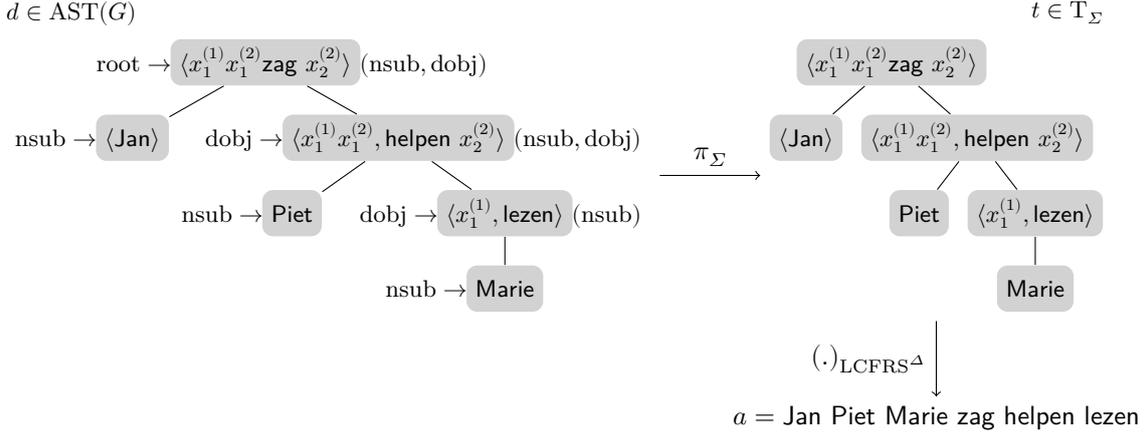
\end{example}

\subsubsection{TAG-algebras and tree-adjoining grammars}

We consider a slight extension of tree-adjoining grammars \cite{jossch97} in which we allow nonterminals (or: states) as in \cite{buenedvog11,buenedvog12}.
Here we restrict ourselves to tree-adjoining grammars generating ranked trees.
Our presentation is essentially the one of \cite{buenedvog11,buenedvog12}.
We also refer to~\cite{tag-irtg} for a formalization of tree-adjoining grammars as interpreted regular tree grammars (IRTG).

Let $S = \{\iota\}$ be a set with one sort $\iota$.
Let $\Delta$ be a finite $(S^* \times S)$-sorted set (of \emph{terminal symbols}).
Let $X$ be an $S$-sorted set of variables and $Z$ be an $(S\times S)$-sorted set of variables (i.e., for each $z \in Z$, we have $\sort(z) = (\iota,\iota)$).
We assume that $\Delta$, $X$, and $Z$ are pairwise disjoint.
Then $\Delta \cup Z$ is an $(S^*\times S)$-sorted set.
Moreover, we let $*$ be a symbol (\emph{foot node}) not in $\Delta \cup X \cup Z$; we let $\sort(*)=\iota$.

For every $m,l \in \mathbb N$, we define the $S$-sorted set $X_m=\{x_{1,\iota},\ldots,x_{m,\iota}\}$ of variables where $X_m \subseteq X$ and $\sort(x_{i,\iota}) = \iota$ (\emph{substitution sites}) and the set $Z_l=\{z_{1,(\iota,\iota)},\ldots,z_{l,(\iota,\iota)}\}$ of variables where $Z_l \subseteq Z$ and $\sort(z_{j,(\iota,\iota)}) = (\iota,\iota)$ (\emph{adjoining sites});
with $i \in [m]$ and $j \in [l]$.

We say that $\zeta \in \T_{\Delta \cup Z_l}(X_m)$ is \emph{linear, nondeleting in $X_m \cup Z_l$} if each element in $X_m \cup Z_l$ occurs exactly once in $\zeta$.
We say that $\zeta \in \T_{\Delta \cup Z_l}(X_m \cup \{*\})$ is \emph{linear, nondeleting in $X_m \cup Z_l \cup \{*\}$} if each element in $X_m \cup Z_l \cup \{*\}$ occurs exactly once in $\zeta$.

We define the $(\{ \iota, (\iota, \iota) \}^* \times \{ \iota, (\iota, \iota ) \})$-sorted set $\Gamma^{\lalg{TAG},\Delta}$ such that $(\Gamma^{\lalg{TAG},\Delta})_{(w,y)} = \emptyset$ for every  $w \not\in \{\iota\}^* \circ \{(\iota, \iota)\}^*$ and $y \in \{\iota, (\iota,\iota)\}$.
For every $m,l \in \mathbb N$ we define
\begin{align*}
    (\Gamma^{\lalg{TAG},\Delta})_{(\iota^m (\iota,\iota)^l,\iota)} &= \{\langle \zeta \rangle \mid \zeta \in \T_{\Delta \cup Z_l}(X_m) \text{ linear, nondeleting in $ X_m \cup Z_l$}\}\\
    (\Gamma^{\lalg{TAG},\Delta})_{(\iota^m (\iota,\iota)^l,(\iota,\iota))} &= \{\langle \zeta \rangle \mid \zeta \in \T_{\Delta \cup Z_l}(X_m \cup \{*\})  \text{ linear, nondeleting in $ X_m \cup Z_l \cup \{*\}$}\}
\end{align*}

\index{TAG-algebra over $\Delta$}
We define the \emph{TAG-algebra over $\Delta$} to be the $\{\iota,(\iota,\iota)\}$-sorted $\Gamma^{\lalg{TAG},\Delta}$-algebra $(\gls{alg:tag},\phi)$ with
\begin{itemize}
    \item $\lalg{TAG}^\Delta = (\lalg{TAG}^\Delta)_\iota \cup (\lalg{TAG}^\Delta)_{(\iota,\iota)}$ with\\ $(\lalg{TAG}^\Delta)_\iota = \T_\Delta$ and $(\lalg{TAG}^\Delta)_{(\iota,\iota)} = \{t \in \T_{\Delta}(\{*\}) \mid * \text{ occurs exactly once in } t \}$ and

    \item for every $\langle \zeta \rangle \in (\Gamma^{\lalg{TAG},\Delta})_{(\iota^m(\iota,\iota)^l,\iota)}$, $t_1,\ldots,t_m \in (\lalg{TAG}^\Delta)_\iota$, and $t_1',\ldots,t_l' \in (\lalg{TAG}^\Delta)_{(\iota,\iota)}$ we define
        \[
            \phi(\langle \zeta\rangle)(t_1,\ldots,t_m,t_1',\ldots,t_l') = \tilde{h}(\zeta)
        \]
        were $\tilde{h}$ is the $S$-sorted tree homomorphism induced by the mapping $h\colon \Delta \cup X_m \cup Z_l \rightarrow \T_\Delta(X)$ defined by
        \begin{align*}
            h(\delta) &= \delta(x_{1,s_1},\ldots,x_{k,s_k}) \hspace*{-1mm} && \text{ for each $\delta \in \Delta$ with $\rk_\Delta(\delta)=k$}\\
            h(x_{i,\iota}) &= t_i && \text{ for each $i \in [m]$}\\
            h(z_{j,(\iota,\iota)}) &= t_j'' && \text{ for each $j \in [l]$ and $t_j''$ is obtained from $t_j'$ by replacing $*$ by $x_{1,\iota}$};
        \end{align*}

    \item for every $\langle \zeta \rangle \in (\Gamma^{\lalg{TAG},\Delta})_{(\iota^l(\iota,\iota)^m,(\iota,\iota))}$, $t_1,\ldots,t_m \in (\lalg{TAG}^\Delta)_\iota$, and $t_1',\ldots,t_l' \in (\lalg{TAG}^\Delta)_{(\iota,\iota)}$ we define
        $\phi(\langle \zeta\rangle)(t_1,\ldots,t_m,t_1',\ldots,t_l') = \tilde{h'}(\zeta)$
        were $h'\colon \Delta \cup X_m \cup Z_l \cup \{*\} \rightarrow \T_{\Delta}(X)$ is  defined in the same way as $h$ but additionally  we let $h'(*) = *$.
\end{itemize}

\index{tree-adjoining grammar}
\index{TAG}
A \emph{tree-adjoining grammar over $\Delta$} is an RTG-LM
\[
    (G,(\lalg{TAG}^{\Delta},\phi))
\]
where the $\{\iota, (\iota,\iota)\}$-sorted RTG $G=(N,\Sigma,A_0,R)$ is in normal form and $A_0 \in N_\iota$.
We note that the language $L(G)_{\lalg{TAG}^\Delta}$ generated by this tree-adjoining grammar is a formal tree language over $\Delta$.
We denote the class of all tree-adjoining grammars by $\gls{gclass:tag}$.

\begin{example}[cf.~\cite{jossch97} and {\cite[Fig.~1]{buenedvog12}}]\label{ex:tag}
    We consider the set
    \[
        \Delta =\{\term{S}, \term{V}, \term{VP}, \term{N}, \term{NP}, \term{D}, \term{Adv}, \term{saw}, \term{Mary}, \term{a}, \term{man}, \term{yesterday}\}
    \]
    of terminal symbols.
    In order to keep notations short, we assume that $\Delta$ is turned into a finite, non-empty ranked set by adding (implicitly) to each symbol of $\Delta$ a finite number of ranks (like  $(\term{NP},2)$ and $(\term{NP},1)$);
    however, in trees over $\Delta$ we drop again this rank information as in Figure~\ref{fig:ex-tag}.

    In Figure~\ref{fig:ex-tag} we show a tree-adjoining grammar $(G,(\lalg{TAG}^{\Delta},\phi))$ where the RTG $G$ has the nonterminals $\nont{A_0}, \nont{A_1}, \nont{A_2}, \nont{F}$ and five rules.
    One possible abstract syntax tree of the RTG $G$ is
    \[
        d = r_1(r_2,r_3(r_4),r_5)\enspace.
    \]
    In Figure~\ref{fig:ex-tag-semantics} we have shown $d$, $\pi_\Sigma(d)$, and its evaluation in the $\lalg{TAG}^\Delta$-algebra, i.e., the syntactic object $a$.
    Unlike in the previous examples, $a$ is not a string, but a tree.
    In particular, it results from evaluating
    \[
        \phi\biggl(\biggl\langle\tikz[inlinetree]{\graph{"$z_1$" -- "$\term{S}$" -- {"$x_1$", "$\term{VP}$" -- {"$\term{V}$" -- "$\term{saw}$", "$x_2$"}}};}\biggr\rangle\biggr)\Bigg(
        \tikz[inlinetree]{\graph{"$\term{NP}$" -- "$\term{N}$" -- "$\term{Mary}$"};},
        \tikz[inlinetree]{\graph{"$\term{NP}$" -- {"$\term{D}$" -- "$\term{a}$", "$\term{N}$" -- "$\term{man}$"}};},
        \tikz[inlinetree]{\graph{"$\term{S}$" -- {"$\term{Adv}$" -- "$\term{yesterday}$", *}};}
        \Bigg)
    \]
    in the $\lalg{TAG}^\Delta$-algebra.
\end{example}

\begin{figure}
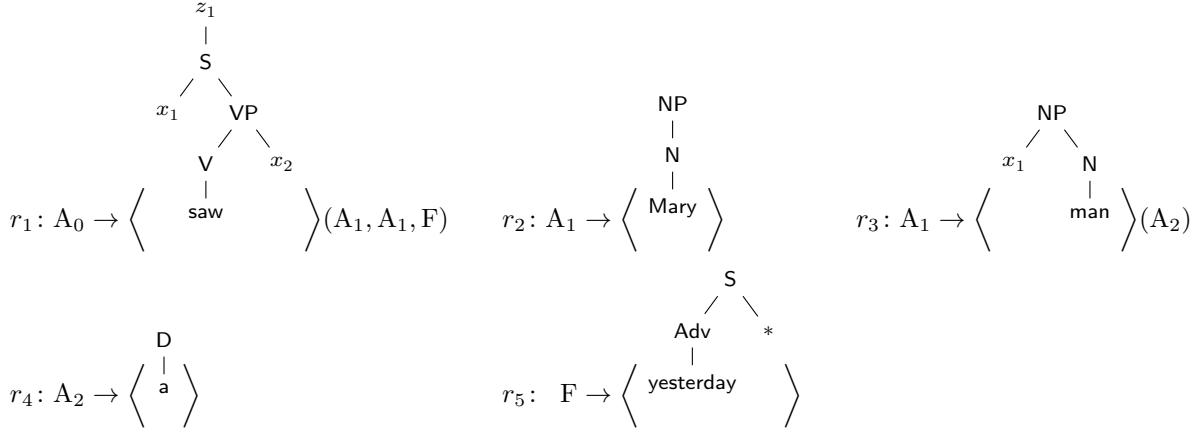

    \begin{alignat*}{6}
        r_1\colon & & \nont{A_0} &\rightarrow \biggl\langle
        \tikz[inlinetree]{\graph{"$z_1$" -- "$\term{S}$" -- { "$x_1$" , "$\term{VP}$" -- { "$\term{V}$" -- "$\term{saw}$" , "$x_2$" } }};}
        \biggr\rangle(\nont{A_1}, \nont{A_1}, \nont{F}) &\qquad
        r_2\colon & & \nont{A_1} &\rightarrow \biggl\langle
        \tikz[inlinetree]{\graph{"$\term{NP}$" -- "$\term{N}$" -- "$\term{Mary}$"};}
        \biggr\rangle &\qquad
        r_3\colon & & \nont{A_1} &\rightarrow \biggl\langle
        \tikz[inlinetree]{\graph{"$\term{NP}$" -- { "$x_1$" , "$\term{N}$" -- "$\term{man}$" }};}
        \biggr\rangle (\nont{A_2}) \\
        r_4\colon & & \nont{A_2} &\rightarrow \biggl\langle
        \tikz[inlinetree]{\graph{"$\term{D}$" -- "$\term{a}$"};}
        \biggr\rangle &
        r_5\colon & & \nont{F} &\rightarrow \biggl\langle
        \tikz[inlinetree]{\graph{"$\term{S}$" -- { "$\term{Adv}$" -- "$\term{yesterday}$" , "$*$" }};}
        \biggr\rangle
    \end{alignat*}
    \caption{Example of a TAG (following \cite{jossch97} and \cite[Fig.~1]{buenedvog12}) where~$\nont{A_0}$ is the initial nonterminal and $z_{1,(\iota,\iota)}$, $x_{1,\iota}$, and $x_{2,\iota}$ are abbreviated by $z_1$, $x_1$, and $x_2$, respectively.}
    \label{fig:ex-tag}
\end{figure}

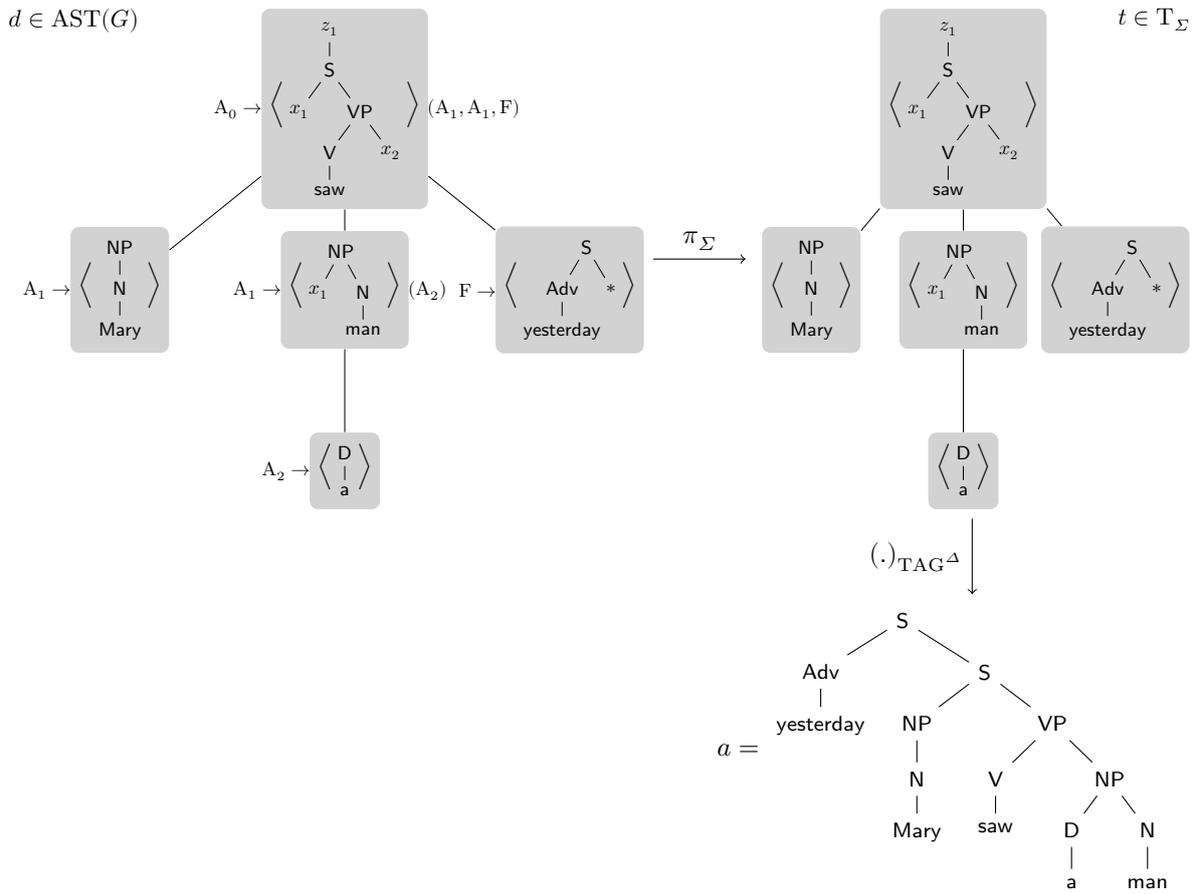
\begin{figure}
    \tikzset{basic/.style={fill=lightgray!70,rounded corners,font=\small}}
    \begin{adjustbox}{center}
    \begin{tikzpicture}
        \node (ast) {\scalebox{0.8}{
            \begin{tikzpicture}[every node/.style=basic,level distance=3cm,level 1/.style={sibling distance=3.7cm}]
                \node (saw) {$\biggl\langle \tikz[inlinetree]{\graph{"$z_1$" -- "$\term{S}$" -- { "$x_1$" , "$\term{VP}$" -- { "$\term{V}$" -- "$\term{saw}$" , "$x_2$" } }};} \biggr\rangle$}
                child {
                    node (Mary) {$\biggl\langle \tikz[inlinetree]{\graph{"$\term{NP}$" -- "$\term{N}$" -- "$\term{Mary}$"};} \biggr\rangle$}
                }
                child {
                    node[sibling distance=0.5cm] (man) {$\biggl\langle \tikz[inlinetree]{\graph{"$\term{NP}$" -- { "$x_1$" , "$\term{N}$" -- "$\term{man}$" }};} \biggr\rangle$}
                    child {
                        node (a) {$\biggl\langle \tikz[inlinetree]{\graph{"$\term{D}$" -- "$\term{a}$"};} \biggr\rangle$}
                    }
                }
                child {
                    node[sibling distance=0.5cm] (yesterday) {$\biggl\langle \tikz[inlinetree]{\graph{"$\term{S}$" -- { "$\term{Adv}$" -- "$\term{yesterday}$" , "$*$" }};} \biggr\rangle$}
                };
                \node[lhs] at (saw.west) {$\nont{A_0} \to$};
                \node[lhs] at (Mary.west) {$\nont{A_1} \to$};
                \node[lhs] at (man.west) {$\nont{A_1} \to$};
                \node[lhs] at (a.west) {$\nont{A_2} \to$};
                \node[lhs] at (yesterday.west) {$\nont{F} \to$};
                \node[rhs] at (saw.east) {$(\nont{A_1}, \nont{A_1}, \nont{F})$};
                \node[rhs] at (man.east) {$(\nont{A_2})$};
            \end{tikzpicture}}};
        \node (pisigma) at (8.5,0) {\scalebox{0.8}{
            \begin{tikzpicture}[every node/.style=basic,level distance=3cm,level 1/.style={sibling distance=2.5cm}]
                \node {$\biggl\langle \tikz[inlinetree]{\graph{"$z_1$" -- "$\term{S}$" -- { "$x_1$" , "$\term{VP}$" -- { "$\term{V}$" -- "$\term{saw}$" , "$x_2$" } }};} \biggr\rangle$}
                child {
                    node {$\biggl\langle \tikz[inlinetree]{\graph{"$\term{NP}$" -- "$\term{N}$" -- "$\term{Mary}$"};} \biggr\rangle$}
                }
                child {
                    node[sibling distance=0.5cm] {$\biggl\langle \tikz[inlinetree]{\graph{"$\term{NP}$" -- { "$x_1$" , "$\term{N}$" -- "$\term{man}$" }};} \biggr\rangle$}
                    child {
                        node {$\biggl\langle \tikz[inlinetree]{\graph{"$\term{D}$" -- "$\term{a}$"};} \biggr\rangle$}
                    }
                }
                child {
                    node[sibling distance=0.5cm] {$\biggl\langle \tikz[inlinetree]{\graph{"$\term{S}$" -- { "$\term{Adv}$" -- "$\term{yesterday}$" , "$*$" }};} \biggr\rangle$}
                };
            \end{tikzpicture}}};
        \node[below=1cm of pisigma] (tree) {
            \tikz[inlinetree,baseline=(S)]{\graph{S/"$\term{S}$" -- { "$\term{Adv}$" -- "$\term{yesterday}$" , S2/"$\term{S}$" -- { "$\term{NP}$" -- "$\term{N}$" -- "$\term{Mary}$" , "$\term{VP}$" -- { "$\term{V}$" -- "$\term{saw}$" , NP2/"$\term{NP}$" -- { "$\term{D}$" -- "$\term{a}$" , N2/"$\term{N}$" -- "$\term{man}$" } } } }};}};
        \node[lhs] at (tree.west) {$a =$};
        \draw[->] (ast) to node[above] {$\pi_\Sigma$} (pisigma);
        \draw[->] (pisigma) to node[left] {$\sem[\lalg{TAG}^\Delta]{(.)}$} (tree);
        \node[font=\small,anchor=north west] at (ast.north west) {$d \in \AST(G)$};
        \node[font=\small,anchor=north east] at (pisigma.north east) {$t \in \T_\Sigma$};
    \end{tikzpicture}
    \end{adjustbox}
    \caption{An abstract syntax tree $d$ in the RTG-LM $(G, \lalg{TAG}^\Delta)$ and its evaluation in the $\lalg{TAG}^\Delta$-algebra, see Example~\ref{ex:tag}.}
    \label{fig:ex-tag-semantics}
\end{figure}

By dropping from TAG-algebras everything which refers to adjoining (the sort $(\iota, \iota)$, the variables in $Z_l$), we might define RTG-algebras and define a  regular tree grammar as an $S$-sorted RTG-LM $(G,(\lalg{RTG}^{\Delta}))$.
This initial algebra presentation of RTG might seem technically exaggerated, because $G$ already is a regular tree grammar.
However, if we stay with the property that the RTG $G$ in the RTG-LM is in normal form, then we have to deal with right-hand sides that do not contain exactly one terminal symbol.
The initial algebra approach above offers one possibility for handling this.

\subsubsection{Yield-algebras and yield grammars}

\citet{GieMeySte04} use yield grammars as a main component of algebraic dynamic programming.
Here we show that yield grammars are a particular subclass of $\gclass{all}$.
Let~$S$ be a set (of sorts) and~$\Delta$ be an $(S^* \times S)$-sorted alphabet.
We denote by $\Delta_0$ the set $\bigcup_{s \in S} \Delta_{(\varepsilon, s)}$.

\index{yield-algebra over $\Delta$}
We define the \emph{$S$-yield-algebra over~$\Delta$} to be the $S$-sorted $\Delta$-algebra $(\gls{alg:yield},\phi)$ with
\begin{itemize}
    \item $\lalg{YIELD}^\Delta = \bigcup_{s \in S} (\lalg{YIELD}^\Delta)_s$ where $(\lalg{YIELD}^\Delta)_s = \{ \langle w,s \rangle \mid w \in (\Delta_0)^* \}$, and
    \item for every $k \in \mathbb N$, $\delta \in \Delta_{(s_1 \dots s_k,s)}$, and $\langle w_1,s_1 \rangle \in (\lalg{YIELD}^\Delta)_{s_1},\dots,\langle w_k,s_k \rangle \in (\lalg{YIELD}^\Delta)_{s_k}$ we define
        \[ \phi(\delta)(\langle w_1,s_1 \rangle,\dots,\langle w_k,s_k \rangle) = \begin{cases}
            \langle \delta,s \rangle &\text{if $k = 0$} \\
            \langle w_1 \dots w_k,s \rangle &\text{otherwise.}
        \end{cases} \]
\end{itemize}
We note that the carrier set of $\lalg{YIELD}^\Delta$ consists of tuples and the second component of each such tuple is a sort.
This has technical reasons, because otherwise there would be no mapping $\sort: \lalg{YIELD}^\Delta \rightarrow S$.
We can easily see this by letting $S = \{ a,b \}$ and $\Delta = \Delta_{(\varepsilon,a)} \cup \Delta_{(aa,a)} \cup \Delta_{(aaa,b)}$ with $\Delta_{(\varepsilon,a)} = \{ \text{a} \}$, $\Delta_{(aa,a)} = \{ \alpha \}$, and $\Delta_{(aaa,b)} = \{ \beta \}$.
Then $t_1 = \alpha(\text{a},\alpha(\text{a},\text{a}))$ is in $(\T_\Delta)_a$ and $t_2 = \beta(\text{a},\text{a},\text{a})$ is in $(\T_\Delta)_b$, but $\yield_{\Delta_0}(t_1) = \yield_{\Delta_0}(t_2) = \text{aaa}$.
We also observe that the unique $\Delta$-homomorphism from $(\T_\Delta,\phi_\Delta)$ to $(\lalg{YIELD}^\Delta,\phi)$ is the mapping $f: \T_\Delta \rightarrow (\Delta^* \times S)$, where $f(t) = (\yield_{\Delta_0}(t),\sort(t))$ for every $t \in \T_\Delta$.

\index{yield-grammar}
An \emph{$S$-sorted yield grammar over~$\Delta$ (yield-grammar)} is an RTG-LM
\[(G,(\lalg{YIELD}^\Delta,\phi))\]
where $G$ is an $S$-sorted RTG.
We denote the class of all yield-grammars by $\gls{gclass:yield}$.

\subsubsection{Further classes of RTG-based language models}
\label{subsect:general-classes}

We introduce three classes of RTG-LMs which are not defined by a particular language algebra (such as $\gclass{CFG}$), but rather by imposing restrictions on the involved RTGs. 
An RTG-LM $\big((N, \Sigma, A_0, R), \alg L)\big)$ is
\begin{itemize}
    \item \emph{acyclic}, if every $d \in \T_R$ is acyclic,
    \item \emph{monadic}, if for every $r \in R$ it holds that $|\rk(r)| \le 1$, and
    \item \emph{nonlooping},
        if for every $d \in \T_R$ and $p \in \pos(d) \setminus \{ \varepsilon \}$ the following holds:
        if $\sem[\alg L]{\pi_\Sigma(d)} = \sem[\alg L]{\pi_\Sigma(d|_p)}$, then $d(\varepsilon) \not= d(p)$.
\end{itemize}

We denote the class of all acyclic RTG-LMs by $\gls{gclass:acyc}$,
the class of all monadic RTG-LMs by $\gls{gclass:mon}$,
the class of all nonlooping RTG-LMs by $\gls{gclass:nl}$,
and the class of all RTG-LMs with finitely decomposable language algebra by $\gls{gclass:findc}$.
The CFG-algebras, LCFRS-algebras, and TAG-algebras are finitely decomposable, i.e., $\gclass{CFG} \subseteq \gclass{\findc}$, $\gclass{LCFRS} \subseteq \gclass{\findc}$, and $\gclass{TAG} \subseteq \gclass{\findc}$.

\subsubsection{Summary of considered classes of RTG-LMs}

We summarize all classes of RTG-LMs introduced in this subsection in Table~\ref{tab:classes-rtglms}.

\begin{table}[h]
    \centering
    \begin{tabular}{ll}
        \toprule
        Notation & Description: the class of all \dots \\
        \midrule
        $\gclass{CFG}$ & context-free grammars \\
        $\gclass{LCFRS}$ & linear context-free rewriting systems \\
        $\gclass{TAG}$ & tree-adjoining grammars \\
        $\gclass{YIELD}$ & yield-grammars \\
        $\gclass{acyc}$ & acyclic RTG-LMs \\
        $\gclass{mon}$ & monadic RTG-LMs \\
        $\gclass{nl}$ & nonlooping RTG-LMs \\
        $\gclass{\findc}$ & RTG-LMs with finitely decomposable language algebra \\
        \bottomrule
    \end{tabular}
    \caption{Classes of RTG-LMs introduced in Section~\ref{sec:classes-rtglms}.}
    \label{tab:classes-rtglms}
\end{table}

\subsection{Classes of weight algebras}
\label{sec:classes-mmonoids}

In this subsection we first show that the weight algebras used by \citet{Goodman1999} and \citet{ned03} are subclasses of $\wclass{all}$, i.e., the class of all complete M-monoids.
We then define additional subclasses of $\wclass{all}$ which will allow us to investigate particular M-monoid parsing problems of this paper with respect to their algorithmic solvability.

\subsubsection{M-monoids that are associated with semirings}
\label{sec:mmonoids-associated-with-semirings}

\index{M-monoid!associated with semiring}
Let $(\walg{K},\oplus,\otimes,\welem{0},\welem{1}, \infsum)$ be a complete semiring.
The \emph{M-monoid associated with~$\walg{K}$} (cf.~\cite[Definition~8.5]{Fulop2009}) is defined as the M-monoid $M(\walg{K}) = (\walg{K},\oplus,\welem{0},\Omega_\otimes,\infsum)$ where $\Omega_\otimes = \bigcup_{k \ge 0} (\Omega_\otimes)_k$ and $(\Omega_\otimes)_k = \{ \mul_{\welem{k}}^{(k)} \mid \welem{k} \in \walg{K} \}$ for every $k \in \mathbb N$.
For every $k \in \mathbb N$ and $\welem{k}, \welem{k}_1, \dots, \welem{k}_k \in \walg{K}$ we define
\[ \mul_{\welem{k}}^{(k)}(\welem{k}_1,\dots,\welem{k}_k) = \welem{k} \otimes \welem{k}_1 \otimes \dots \otimes \welem{k}_k \enspace. \]
In particular, $\mul_{\welem{k}}^{(0)}() = \welem{k}$ for every $\welem{k} \in \walg{K}$.
Note that $\mathbb 1 = \mul_{\mathbb 1}^{(0)}()$.
Clearly, the M-monoid $M(\walg K)$ is complete and distributive.

\citet{Goodman1999} modeled several classic parsing problems by specifying for each of these problems a complete semiring which encapsulates the computation of the problem's solution.
Using the approach from above, we can for each of these semirings define a weight algebra of a wRTG-LM which is associated with that semiring.
In the following, we will do this for some  semirings which we find particularly interesting.
We denote the class of all M-monoids that are associated with complete semirings by $\gls{wclass:sr}$.

\index{tropical M-monoid}
\index{tropical semiring}
\begin{example}\label{ex:tropical-M-monoid}
    The \emph{tropical semiring} is the complete semiring $(\mathbb{R}_0^\infty,\min,+,\infty,0,\inf)$ with the usual binary minimum operation on the reals.
    The \emph{tropical M-monoid} is the M-monoid associated with the tropical semiring, i.e., the M-monoid $\gls{wclass:tropical} = (\mathbb{R}_0^\infty,\min,\infty,\Omega_+,\inf)$.
\end{example}

\index{Viterbi M-monoid}
\index{Viterbi semiring}
\begin{example}\label{ex:Viterbi-M-monoid}
    The \emph{Viterbi semiring} is the complete semiring $(\mathbb{R}_0^1,\max,\cdot,0,1,\sup)$ with the usual binary maximum operation on reals.
    The \emph{Viterbi M-monoid} is the M-monoid $\gls{wclass:viterbi} = (\mathbb{R}_0^1,\max,0,\Omega_\cdot,\sup)$ (with $\cdot$ as index of $\Omega$).
\end{example}

\begin{example}\label{ex:best-derivation-mmonoid-sr}
    Recall the definition of the best derivation M-monoid $\walg{BD}$ from Example~\ref{ex:best-derivation-mmonoid}.
    It is a d-complete and distributive M-monoid and furthermore, $(0, \emptyset)$ is absorptive.
    The proof of this statement is given in Appendix~\ref{sec:proof-best-derivation-mmonoid}.
\end{example}

\begin{example}\label{ex:nbest-mmonoid}
    We define an M-monoid which describes the computation of the probabilities of the $n$ best derivations of a syntactic object (where $n \in \mathbb N_+$).
    This M-monoid is (up to notation) associated with the \emph{Viterbi-$n$-best semiring} \cite[Figure 5]{Goodman1999}.

    Let $n \in \mathbb N_+$.
    In the following, we will denote a family $f: [n] \to \mathbb R_0^1$ as $(f_1, \dots, f_n)$ and let
    \[ \pref(\mathbb N, n) = \bigcup_{\substack{n' \in \mathbb N: \\ n' \ge n}} \{ [n'] \} \cup \{ \mathbb N \} \enspace. \]
    We define the set $\gls{wclass:nbest} = \{ f: [n] \to \rzo \mid \text{for every $i \in [n-1]$: $f(i) \ge f(i+1)$} \}$.
    Furthermore, we define the mapping $\takenbest: (\rzo)^{\pref(\mathbb N, n)} \to \nbest$ such that for every $I \in \pref(\mathbb N, n)$ and family $(f_i \mid i \in I)$ of elements of $\rzo$
    \[ \takenbest((f_i \mid i \in I)) = (g_1, \dots, g_n) \enspace, \]
    where $g \in \nbest$ as follows:
    \begin{enumerate}
        \item If~$I$ is finite, then for every $i \in [n]$, $g_i = f(v(i))$, where $v: [n] \to I$ is recursively defined such that for every $i \in [n]$
            \[ v(i) = \text{an arbitrary} \ j \in \argmax\limits_{j' \in I \setminus v|_i} f_{j'} \enspace, \]
            where for every $i \in [n]$, we let $v|_i = \{ v(j) \mid j \in [n] \ \text{and} \ j < i \}$.
        \item Otherwise, we define the mapping $v: [n] \to I \cup \{ \bot \}$, where $\bot$ is a new element, recursively for every $i \in [n]$ such that
            \[ v(i) = \begin{cases}
                \bot &\text{if $i > 1$ and $v(i-1) = \bot$} \\
                \text{an arbitrary} \ j \in \argmax\limits_{j' \in I \setminus v|_i} f_{j'} &\text{if there is such a $j$} \\
                \bot &\text{otherwise,}
            \end{cases} \]
            where for every $i \in [n]$, we let $v|_i = \{ v(j) \mid j \in [n] \ \text{and} \ j < i \}$.
            Then for every $i \in [n]$
            \[ g_i = \begin{cases}
                f(v(i)) &\text{if $v(i) \not= \bot$} \\
                \sup \{ f_i \mid i \in I \setminus v|_i \} &\text{otherwise.}
            \end{cases} \]
            (We note, again, that this supremum exists because~$1$ is an upper bound of every subset of~$\rzo$ and every bounded subset of~$\mathbb R$ has a supremum.)
    \end{enumerate}

    \begin{sloppypar}
    Moreover, we define the mapping $\cdot_{\nbest}: (\mathbb R_0^1)^{[n]} \times (\mathbb R_0^1)^{[n]} \to (\mathbb R_0^1)^{[n]}$ such that for every $(a_1, \dots, a_n), (b_1, \dots, b_n) \in \nbest$
    \end{sloppypar}
    \[ (a_1, \dots, a_n) \cdot_n (b_1, \dots, b_n) = \takenbest \left( \left( a(\lfloor i / n \rfloor + 1) \cdot b(i \bmod n + 1) \, \middle| \, i \in [0, n^2 - 1] \right) \right) \enspace. \]

    \index{n@$n$-best M-monoid}
    The \emph{$n$-best M-monoid} is the complete M-monoid $(\nbest, \maxn, (\underbrace{0, \dots, 0}_{\text{$n$ times}}), \Omega_n, \infsum[\maxn])$, where
    \begin{itemize}
        \item for every $(a_1, \dots, a_n), (b_1, \dots, b_n) \in \nbest$,
            \[ \maxn \big( (a_1, \dots, a_n), (b_1, \dots, b_n) \big) = \takenbest \big( (a_1, \dots, a_n, b_1, \dots, b_n) \big) \enspace, \]
        \item $\Omega_n = \{ \mulnkk \mid k \in \mathbb N \ \text{and} \ \welem k \in \nbest \}$, where for each $k \in \mathbb N$ and $\welem k \in \nbest$, $\mulnkk: \nbest^k \to \nbest$ such that for each $\welem k_1, \dots, \welem k_k \in \nbest$
            \[ \mulnkk(\welem k_1, \dots, \welem k_k) = (\welem k, \underbrace{0, \dots, 0}_{n - 1 \ \text{times}}) \cdot_n \welem k_1 \cdot_n \ldots \cdot_n \welem k_k \enspace, \]
        \item for every nonempty $I$-indexed family $(\welem k_i \mid i \in I)$ over $\nbest$,
            \[ \infsum[\maxn]_{i \in I} \welem k_i = \takenbest((f_i \mid i \in \mathbb N)) \enspace, \]
            where for every $i \in \mathbb N$, $f_i = \welem k_{\lfloor i / n \rfloor + 1}(i \bmod n + 1)$.
    \end{itemize}

    The n-best M-monoid is a d-complete and distributive M-monoid.
    Furthermore, $\nbzeroes$ is absorptive.
    The proof of this statement is given in Appendix~\ref{sec:proof-nbest-mmonoid}.
\end{example}

\subsubsection{Superior M-monoids}\label{sec:superior-mmonoids}

The weight algebras of \citet{ned03} are \emph{superior}, a notion defined by \citet{Knuth1977}.
They are essentially complete and distributive M-monoids of the form $(\walg K, \min, \mathbb 0, \Omega, \inf)$, where $(\walg K, \preceq)$ is a total order, $\inf(\walg K) \in \walg K$, and $\Omega$ is a set of \emph{superior} functions, i.e., for each $k \in \mathbb N$, $\omega \in \Omega_k$, $i \in [k]$, and $\welem k_1, \dots, \welem k_k \in \walg K$ it holds that
\begin{enumerate}
    \item if $\welem k \preceq \welem k_i$, then
        \(\omega(\welem{k}_1,\dots,\welem{k}_{i-1},\welem{k},\welem{k}_{i+1},\dots,_k) \preceq \omega(\welem{k}_1,\dots,\welem{k}_{i-1},\welem{k}_i,\welem{k}_{i+1},\dots,\welem{k}_k)\),
        and
    \item $\max \{ \welem k_1, \dots, \welem k_k \} \preceq \omega(\welem k_1, \dots, \welem k_k)$.
\end{enumerate}
\index{M-monoid!superior}
We will call such M-monoids \emph{superior M-monoids} and denote the class of all superior M-monoids by $\gls{wclass:sup}$.

It is easy to see that every superior M-monoid is completely idempotent and thus, by Lemma~\ref{lem:inf-idp-d-complete}, also d-complete.

The tropical M-monoid and the Viterbi M-monoid from Example~\ref{ex:tropical-M-monoid} and~\ref{ex:Viterbi-M-monoid}, respectively, are superior.
For the proofs of this statement, we refer to Appendix~\ref{sec:proof-superior-mmonoids}.

\subsubsection{Further classes of complete M-monoids}

We denote the class of all d-complete M-monoids by $\gls{wclass:dcomp}$ and the class of all complete and distributive M-monoids by~$\gls{wclass:dist}$.

We define $\gls{wclass:finidpo}$ to be the class of all M-monoids $(\walg K, \oplus, \mathbb 0, \Omega, \infsum)$ in $\wclass{dist}$ for which
\begin{enumerate*}
    \item $\walg K$ is finite,
    \item $(\walg K, \oplus, \mathbb 0, \infsum)$ is completely idempotent, and
    \item there is a partial order $(\walg K, \preceq)$ such that for every $k \in \mathbb N$, $\omega \in \Omega_k$, and $\welem k_1, \dots, \welem k_k \in \walg K$, $\maxord \{ \welem k_1, \dots, \welem k_k \} \preceq \omega(\welem k_1, \dots, \welem k_k)$.
\end{enumerate*}
Since this third condition looks similar to the second condition on superior M-monoids, we point out the following subtle differences.
First, the carrier set of an M-monoid in $\wclass{\finidpo}$ is finite, which is not necessarily the case for superior M-monoids.
Second, for M-monoids in $\wclass{\finidpo}$, $(\walg K, \preceq)$ is an arbitrary \emph{partial} order which is not related to $\oplus$ at all, while for a superior M-monoid, $(\walg K, \preceq)$ is the same \emph{total} order with respect to which $\min$ is defined.
By Lemma~\ref{lem:inf-idp-d-complete}, $\wclass{\finidpo} \subseteq \wclass{\dcomp}$.

\subsubsection{Summary of considered classes of M-monoids}

We summarize all classes of M-monoids introduced in this subsection in Table~\ref{tab:classes-mmonoids}.
Note that we have written singleton classes of M-monoids without curly braces, e.g., $\walg{BD}$ rather than $\{ \walg{BD} \}$.

\begin{table}[h]
    \centering
    \begin{tabular}{ll}
        \toprule
        Notation & Description: the class of all \dots \\
        \midrule
        $\wclass{sr}$ & M-monoids associated with semirings \\
        $\wclass{sup}$ & superior M-monoids \\
        $\wclass{\dcomp}$ & d-complete M-monoids \\
        $\wclass{dist}$ & distributive M-monoids \\
        $\wclass{\finidpo}$ & finite and idempotent M-monoids with a certain monotonicity property \\
        \midrule
        & Specific M-monoids \\
        \midrule
        $\walg{T}$ & the tropical M-monoid \\
        $\walg{V}$ & the Viterbi M-monoid \\
        $\walg{BD}$ & the best derivation M-monoid \\
        $\nbest$ & the n-best M-monoid \\
        \bottomrule
    \end{tabular}
    \caption{Classes of M-monoids introduced in Section~\ref{sec:classes-mmonoids}.}
    \label{tab:classes-mmonoids}
\end{table}

\subsection{Closed weighted RTG-based language models}\label{sec:closed}

Although superior M-monoids are common weight structures in weighted parsing, they are not general enough to cover all parsing problems (cf., e.g., computing the intersection and ADP in Section~\ref{sec:problems}).
Hence we would like to determine a class $\wlmclass{}$ of wRTG-LMs that properly includes $\wlmclass{\gclass{all}, \wclass{sup}}$ and can describe both
\begin{enumerate*}[label=(\alph*)]
    \item computing the intersection of a grammar and a syntactic object and 
    \item the problems of ADP.
\end{enumerate*}
Furthermore, we recall that Goodman's semiring parsing algorithm only terminates if the input grammar does not contain cyclic derivations.
We would like $\wlmclass{}$ to properly include all wRTG-LMs with that property, too.

If we think of an algorithmic computation of the mapping $\fparse$ in the M-monoid parsing problem, then we encounter the following problem:
the index set of $\infsum$ can be infinite.
However, an algorithm cannot compute an infinite sum and terminate at the same time.
Clearly, if the index set of $\infsum$ is finite for some input, then an algorithm may compute $\fparse$ on this input.
Hence, if for every input with set of abstract syntax trees $D$ over which $\infsum$ is computed there exists a finite subset~$E$ of~$D$ such that
\[
  \infsum_{d \in D} \wthom{d} = \bigoplus_{d \in E} \wthom{d}\enspace,
\]
then an algorithm that correctly computes $\fparse$ on every input may exist.
\textcite{Mohri2002} implemented this idea for graphs weighted with semirings.
He gave an algorithm which solves a problem similar to the M-monoid parsing problem if the input semiring is \emph{closed} for the input graph.
Here we extend this notion to M-monoids that are closed for the input hypergraph.
In order to stay within the domain of parsing, we base our definitions on RTGs rather than on hypergraphs.

\begin{sloppypar}
The rest of this subsection is structured as follows.
In Section~\ref{sec:closed-definition} we will define the class $\wlmclass[\clsd]{\gclass{all}, \wclass{all}}$ of closed wRTG-LMs.
In Section~\ref{sec:closed-properties} we will show that for every wRTG-LM in that class, the infinite sum of the M-monoid parsing problem can be computed by a finite sum.
\end{sloppypar}

\subsubsection{Definition of closed weighted RTG-based language models}
\label{sec:closed-definition}

We note that our motivation for closed wRTG-LMs implies that the weight algebra $\walg K$ of each closed wRTG-LM is d-complete (cf.\ Example~\ref{ex:boolean-semiring}).
Moreover, our definition of closed will only involve a single tree over the set of rules.
In order to entail the desired statement about sets of all abstract syntax trees from this definition, distributivity of $\walg K$ is needed.
Thus, the weight algebra of any closed wRTG-LM is in $\wclass{\dcomp} \cap \wclass{dist}$.

In order to define closed wRTG-LMs, we first need a notion for cutting chunks out of trees.
\begin{quote}
    \em In this section, we let~$R$ denote a ranked set.
\end{quote}

Let $w \in R^*$ be an elementary cycle.
We define the binary relation $\wtrans \, \subseteq \T_R \times \T_R$ such that for each $d, d' \in \T_R$, $d \wtrans d'$ if there are $p, p' \in \pos(d)$ with $\seq(d, p, p') = w$ and $d' = d[d|_{p'}]_p$.
Furthermore, we define the binary relation $\vdash \subseteq \T_R \times \T_R$ such that for each $d, d' \in \T_R$, $d \vdash d'$ if there is an elementary cycle $w \in R^*$ and $d \wtrans d'$.

\begin{lemma}[restate={[name={}]lemtranswf}]\label{lem:transition-well-founded}
    The endorelation~$(\vdash^+)^{-1}$ on~$\T_R$ is well-founded.
\end{lemma}

\begin{proof}
    For the proof of Lemma~\ref{lem:transition-well-founded}, we refer to Appendix~\ref{app:closed-definition}.
\end{proof}

\index{cutout trees}
For each elementary cycle $w \in R^*$, we define the set of \emph{$w$-cutout trees of~$d$} as
\[ \cotrees(d, w) = \{ d' \in \T_R \mid d \wtrans^+ d' \} \enspace. \]
We note that $d \not\in \cotrees(d, w)$ and $\cotrees(d, w)$ is finite.
Moreover, we define the set of \emph{cutout trees of~$d$} as
\[ \cotrees(d) = \{ d' \in \T_R \mid d \vdash^+ d \} \enspace. \]
We note that $d \not\in \cotrees(d)$ either and $\cotrees(d)$ is finite, too.

\begin{lemma}[restate={[name={}]lemcutoutsubset}]\label{lem:subtree-cotrees-subset}
    For every $d, d' \in \T_R$ the following holds:
    if $d \vdash^+ d'$, then $\cotrees(d') \subset \cotrees(d)$.
\end{lemma}

\begin{proof}
    For the proof of Lemma~\ref{lem:subtree-cotrees-subset}, we refer to Appendix~\ref{app:closed-definition}.
\end{proof}

\index{closed}
\index{wRTG-LM!closed}
Let $c \in \mathbb N$ and $\overline G = \big((G, \alg L), (\walg K, \oplus, \welem 0, \Omega), \wt\big)$ be a $(\gclass{all}, \wclass{\dcomp} \cap \wclass{dist})$-LM.
We say that~$\overline G$ is a \emph{$c$-closed wRTG-LM}, if for every $d \in \T_R$ and elementary cycle $w \in R^*$ such that there is a leaf $p \in \pos(d)$ which is $(c+1,w)$-cyclic the following holds:
\begin{equation}\label{eq:c-closed}
    \wthom{d} \oplus \bigoplus_{d' \in \cotrees(d, w)} \wthom{d'} = \bigoplus_{d' \in \cotrees(d, w)} \wthom{d'} \enspace.
\end{equation}
We say that~\emph{$\overline G$ is a closed wRTG-LM} if there is a $c \in \mathbb N$ such that~$\overline G$ is a $c$-closed wRTG-LM.

We denote the class of all closed wRTG-LMs by $\gls{wlmclass:closed}$.

\subsubsection{Properties of closed weighted RTG-based language models}
\label{sec:closed-properties}

\begin{quote}
    \em For the rest of this section, we let $c \in \mathbb N$ and $\overline G \in \wlmclass[\clsd]{\gclass{all}, \wclass{\dcomp} \cap \wclass{dist}}$ with $\overline G = \big((G, \alg L), (\walg K, \oplus, \welem 0, \Omega), \wt\big)$ and $G = (N, \Sigma, A_0, R)$.
\end{quote}

First we generalize the applicability of Equation~\eqref{eq:c-closed} to trees that are more than $(c+1)$-cyclic.

\begin{lemma}[restate={[name={}]lemgenclosed}]\label{lem:closed-bigger-trees'}
    For every $d \in (\T_R)$, $c' \in \mathbb N$ with $c' \geq c + 1$, and elementary cycle $w \in R^*$ such that there is a leaf $p \in \pos(d)$ which is $(c',w)$-cyclic the following holds:
    \[ \wthom{d} \oplus \bigoplus_{d' \in \cotrees(d, w)} \wthom{d'} = \bigoplus_{d' \in \cotrees(d, w)} \wthom{d'} \enspace. \]
\end{lemma}

\begin{proof}
    For the proof of Lemma~\ref{lem:closed-bigger-trees'}, we refer to Appendix~\ref{app:closed-properties}.
\end{proof}

Recall that for every $c \in \mathbb N$, the set $\TRc$ contains those trees over rules that are at most $c$-cyclic.
Formally, $\TRc = \{ d \in \T_R \mid c' \in \mathbb N, c' \leq c, \text{\ and $d$ is $c'$-cyclic} \}$.
The next theorem intuitively states the following.
For every summation over the weights of ASTs, we may remove an arbitrary finite set of ASTs from the summation as long as their cutout trees which are at most $c$-cyclic remain in the index set of the sum.

\begin{boxtheorem}[restate={[name={}]thmclosed}]\label{thm:outside-trees-subsumed}
    For every $l \in \mathbb N$, $D \subseteq \TRc$, and $D' \subseteq \T_R \hspace{-0.2mm} \setminus \hspace{-0.2mm} \TRc$ the following holds:
    if $\bigcup_{d \in D'} (\cotrees(d) \cap \TRc) \subseteq D$, then for every $B \subseteq D'$ with $|B| = l$,
    \[ \bigoplus_{d \in D} \wthom{d} \oplus \infsum_{d \in D'} \wthom{d} = \bigoplus_{d \in D} \wthom{d} \oplus \infsum_{d \in D' \setminus B} \wthom{d} \enspace. \]
\end{boxtheorem}

\begin{proof}
    For the proof of Theorem~\ref{thm:outside-trees-subsumed}, we refer to Appendix~\ref{app:closed-properties}.
\end{proof}

Next we show that for every summation over the weights of a certain set of trees (namely those that are at most $c$-cyclic), we may add an arbitrary finite set of trees to that summation.

\begin{lemma}[restate={[name={}]lemclosed}]\label{lem:outside-trees-spawned}
    For every $l \in \mathbb N$, $A \in N$, and $B \subseteq (\T_R)_A \setminus \TRc$ with $|B| = l$ the following holds:
    \[
        \bigoplus_{d \in (\TRc)_A} \wthom{d} = \bigoplus_{d \in (\TRc)_A \cup B} \wthom{d} \enspace.
    \]
\end{lemma}

\begin{proof}
    For the proof of Lemma~\ref{lem:outside-trees-spawned}, we refer to Appendix~\ref{app:closed-properties}.
\end{proof}

Finally, we extend this result to adding arbitrary (in particular, infinite) set of trees to the summation.

\begin{boxtheorem}\label{thm:tr-trc}
    For every $A \in N$ it holds that
    \[
        \infsum_{d \in (\T_R)_A} \wthom{d} = \bigoplus_{d \in (\TRc)_A} \wthom{d} \enspace.
    \]
\end{boxtheorem}

\begin{proof}
    Let $A \in N$.
    By Lemma~\ref{lem:outside-trees-spawned}, for every finite $D \subseteq (\T_R)_A$ with $(\TRc)_A \subseteq D$
    \[
        \bigoplus_{d \in D} \wthom{d} = \bigoplus_{d \in (\TRc)_A} \wthom{d} \enspace.
    \]
    Thus, by Lemma~\ref{lem:d-complete}~(iii),
    \[
        \infsum_{d \in (\T_R)_A} \wthom{d} = \bigoplus_{d \in (\TRc)_A} \wthom{d} \enspace. \qedhere
    \]
\end{proof}

Theorem~\ref{thm:tr-trc} shows that the sum over the (possibly infinite) set of ASTs of a closed wRTG-LM can indeed be computed by the sum over a finite subset of that set.
It will be essential in the proof of correctness of the value computation algorithm (cf.\ Section~\ref{sec:vca-correct}).

\section{Two particular M-monoid parsing problems}
\label{sec:problems}

In this section, we consider two computational problems which are related more or less closely to parsing, and we present them as instances of the M-monoid parsing problem.
For this, we formalize each of these problems using a particular class of wRTG-LMs.
We start with computing the intersection of a grammar and a syntactic object and then proceed with algebraic dynamic programming.

\subsection{Intersection of a grammar and a syntactic object}
\label{sec:intersection}

\citet{BarPerSha61} have proven that context-free languages are closed under intersection with regular languages.
They gave a constructive proof which, given a CFG $G$ and a finite-state automaton $M$ (modeling the regular language), creates a new CFG, denoted by $G \rhd M$, whose language is the intersection of the languages of $G$ and $M$.
By choosing $M$ such that its language is a single sentence $a$, the derivations of $G \rhd M$ are exactly the derivations of $a$ in $G$.
In the following, we restrict ourselves to this special case and write $G \rhd a$ rather than $G \rhd M$.
We will briefly describe two applications of $G \rhd a$ in NLP and then formalize the construction of $G \rhd a$ (where $G$ is not restricted to CFG) as an instance of the M-monoid parsing problem.

A \emph{parse forest} is a compact (in particular, finite) representation of the set of abstract syntax trees for some syntactic object.
\textcites{BilLan89} (also cf.~\cite{Lan74}) have shown that the intersection of a CFG $G$ and a sentence $a$ is precisely the parse forest of $a$ in $G$.
Their approach has later been referred to as \emph{parsing as intersection} and generalized to language models beyond CFG, e.g., TAG~\cite{Lan94}.

In EM training~\cite{demlairub77} of PCFGs the probabilistic weights of a CFG $G$ are estimated with respect to sentences from a training corpus~\cites{Bak79}{LarYou90}.
This is done by computing $G \rhd a$ for each training sentence $a$.
We note that the cited publications did not explicitly mention the intersection.
This was first done by \textcite{NedSat08} (also cf.~\cite{NedSat03}).
\textcite{dregebvog16} generalized the use of the intersection in EM training to language models beyond CFG.

We will now formally define the intersection of an RTG-LM $(G, \alg L)$ and a syntactic object $a$.
We will then show that computing $G \rhd a$ is an instance of the M-monoid parsing problem, given that $\alg L$ fulfils a certain condition.

\index{intersection}
\index{G@$G \rhd_\psi a$}
Let $(G,(\alg L,\phi))$ and $(G',(\alg L,\phi))$ be RTG-LMs where  $G=(N,\Sigma,A_0,R)$ and $G'=(N',\Sigma,A_0',R')$.
\begin{enumerate}
    \item Let $\psi: N' \rightarrow N$ be a mapping and let $a \in \alg L_{\sort(A_0)}$.

        Then $(G',(\alg L,\phi))$ is the \emph{$\psi$-intersection of~$G$ and~$a$}, denoted by  $G \rhd_\psi a$, if the following conditions hold:
        \begin{itemize}
            \item $L(G')_{\alg L} = L(G)_{\alg L}  \cap \{ a \}$, and
            \item the mapping $\widehat{\psi}\colon \AST(G') \rightarrow \AST(G, a)$ is bijective, where $\widehat{\psi} = \widehat{\psi}'|_{\AST(G')}$ and $\widehat{\psi}': \T_{R'} \to \T_R$ is defined inductively by
                \begin{align*}
                    r'(d_1,\dots,d_k) & \mapsto \psi(r')(\widehat{\psi}'(d_1),\dots,\widehat{\psi}'(d_k))
                \end{align*}
                and $\psi$ is extended in a natural way to rules.
        \end{itemize}

    \item We call $(G',(\alg L,\phi))$ \emph{intersection of~$G$ and~$a$} if there is a mapping $\psi$ as in (i) such that  $(G',(\alg L,\phi))$ is the $\psi$-intersection of~$G$ and~$a$.   \qedhere
\end{enumerate}

If the algebra $(\alg L,\phi)$ is finitely decomposable, then the intersection can be constructed easily from the result of a particular M-monoid parsing problem.
We recall that the CFG-algebras, LCFRS-algebras, and TAG-algebras are finitely decomposable.

Let $(G,(\alg L,\phi))$ be an RTG-LM such that $G =(N,\Sigma,A_0,R)$ is in normal form and $(\alg L,\phi)$ is finitely decomposable.
Moreover, let $a \in \alg L_{\sort(A_0)}$.
\index{intersection M-monoid}
The \emph{intersection M-monoid of $(G, \alg L)$} and $a$ is the finite and complete M-monoid
\[
    \walg{K}((G,(\alg L,\phi)),a) = (\mathcal P(P_{R,a}),\cup,\emptyset,\Omega,\infsumop[\cup]) \enspace,
\]
which we construct as follows.
\begin{itemize}
    \item $P_{R,a} = \{ [A,b] \rightarrow \ \sigma([A_1,a_1],\ldots,[A_k,a_k]) \mid
        \begin{aligned}[t]
            &(A \rightarrow \sigma(A_1,\ldots,A_k)) \in R, b \in \factors(a)_{\sort(A)}, \text{ and} \\
            & (a_1,\ldots,a_k) \in \phi(\sigma)^{-1}(b) \} \enspace;
        \end{aligned}$\\
        we note that $P_{R,a}$ is finite,

    \item $\Omega = \{\omega_r \mid r \in R\}$ where for each $r = (A \rightarrow \sigma(A_1,\ldots,A_k))$ with $\sort(\sigma) =(s_1\ldots s_k,s)$, the operation $\omega_r$ is defined for every $V_1,\ldots,V_k \in \mathcal P(P_{R,a})$ by
        \[
            \omega_r(V_1,\ldots,V_k) = V_1 \cup \ldots \cup V_k \cup V
        \]
        where
        \begin{align*}
            V = \{[A,b] \rightarrow \sigma([A_1,a_1],\ldots,[A_k,a_k]) \mid{} &(\forall i \in [k]): [A_i,a_i] \in \mathrm{lhs}(V_i),\\
            & b = \phi(\sigma)(a_1,\ldots,a_k) \}
        \end{align*}
        and $\mathrm{lhs}(V_i)$ is the set of left-hand sides of all rules in $V_i$,

    \item for each family $(V_i \mid i \in I)$ of elements of $\mathcal P(P_{R,a})$ we define
        \[
            \infsum[\cup]_{i \in I} V_i = \bigcup_{i \in I} V_i\enspace.
        \]
        We note that $\sum\nolimits^{\cup}_{i \in I} V_i$ is a finite set.
\end{itemize}

We remark that the restriction of $G$ to normal form is for simplicity.
An extension of the definition of the intersection M-monoid which allows arbitrary RTGs is possible in a straightforward way.
 
\begin{boxtheorem}[restate={[name={}]thmintersection}]\label{thm:intersection}
    For each RTG-LM with a finitely decomposable algebra and each syntactic object, the construction of their intersection is an M-monoid parsing problem.

    More precisely, let $(G,(\alg L,\phi))$ be an RTG-LM such that $G=(N,\Sigma,A_0,R)$ and $(\alg L,\phi)$ is a finitely decomposable language algebra.
    Moreover, let $a \in \alg L_{\sort(A_0)}$.
    We consider the M-monoid parsing problem with the following input:
    \begin{itemize}
        \item the wRTG-LM $((G,(\alg L,\phi)), \walg{K}((G,(\alg L,\phi)),a), \wt)$ where $\wt(r) = \omega_r$ for each $r \in R$ and
        \item $a$.
    \end{itemize}
    Then $(G',(\alg L,\phi))$ is the $\psi$-intersection of $(G,(\alg L,\phi))$ and $a$, where
    \begin{itemize}
        \item $G' = (N',\Sigma,[A_0,a],\fparse(a))$ with $N'=\mathrm{lhs}(\fparse(a)) \cup \{[A_0,a]\}$ (we note that $\fparse(a)$ is a finite set)  and
        \item $\psi\colon N' \rightarrow N$ is defined by $\psi([A,b]) = A$ for each $[A,b] \in N'$.
    \end{itemize}
\end{boxtheorem}

\begin{proof}
    For the proof of Theorem~\ref{thm:intersection}, we refer to Appendix~\ref{app:intersection}.
\end{proof}

We denote the class of all intersection M-monoids $\walg K((G', \alg L'), a)$, where $(G', \alg L')$ is some RTG-LM and $a$ is some syntactic object, by $\gls{wclass:int}$.

%%%%%%%%%%%%%%%%%%%%%%%%%%%%%%%%%%%%%%%%%%%%%%%%%%%%XXXX

\subsection{Algebraic dynamic programming}
\label{sec:adp}

\emph{Algebraic dynamic programming} (ADP) is a framework for modeling dynamic programming problems which was originally developed by \textcite{GieMeySte04}.
They represented dynamic programming problems using a yield grammar and a so-called evaluation algebra for each problem.
In this section, we will introduce a different formalization of ADP which uses only a single formalism: wRTG-LMs.
Moreover, we will show that each ADP problem is an instance of the M-monoid parsing problem.

In this section we fix the following objects and sets.
We let~$\sans$ be a sort representing \enquote{answers} and~$\sinp$ be a sort representing \enquote{input}.
Moreover, we let $S = \{ \sans,\sinp \}$ be a set of sorts and $\Sigma$ be an $(S^*\times S)$-sorted set such that
\begin{enumerate}
    \item $\Sigma_{(\varepsilon,\sans)} = \emptyset$
        and
    \item $\Sigma_{(s_1 \dots s_k,\sinp)} = \emptyset$ for every $k \in \mathbb N_+$.
\end{enumerate}
Intuitively, in every tree over~$\Sigma$, the leaves are symbols with sort~$\sinp$ and the inner nodes are symbols with sort~$\sans$.

In the following, we will formalize the concepts \emph{objective function} and \emph{evaluation algebra} of~\cite{GieMeySte04} using our own methodology.
We note that we have used sets rather than lists in order to represent multiple answers.
This is motivated by the fact that sets are more commonly understood than lists.
Moreover, we do not want duplicate answers and information about order can be added to answers if they are elements of a set, too.
Thus the advantages lists provide over sets are not needed in our case.
We believe that~\cite{GieMeySte04} chose lists over sets because of their choice to implement ADP in Haskell, where lists are a widespread datastructure.

\index{objective function}
Let $\walg{K}$ be an $S$-sorted set.
An \emph{objective function (for $\walg{K}$)} is a family $(h_s \mid s \in S)$ of mappings $h_s: \mathcal P(\walg{K}_s) \rightarrow \mathcal P(\walg{K}_s)$ which fulfils the following requirements:
\begin{enumerate}
    \item $h_\sinp = \id$,
    \item $h_\sans$ maps each non-empty subset~$F$ of~$\walg{K}_\sans$ to a non-empty subset of~$F$,
    \item $h_\sans(\emptyset) = \emptyset$, and
    \item $h_\sans$ is commutative and associative in the following sense:
        for every non-empty subset~$F$ of~$\walg{K}_\sans$ and every $I$-indexed family $(F_i \mid i \in I)$ of elements $F_i \subseteq F$ such that $F = \bigcup_{i \in I} F_i$
        \begin{equation}\label{eq:obj-function'}
            h_\sans(F) = h_\sans(\bigcup_{i \in I} h_\sans(F_i)) \enspace.
        \end{equation}
        In particular, by choosing $I = \{i\}$ and $F_i=F$, we obtain that $h_\sans(h_\sans(F))=h_\sans(F)$, i.e., $h_\sans$ is idempotent.
\end{enumerate}
We note that since $h_\sinp = \id$, Equation~\ref{eq:obj-function'} also holds if we replace~$\sans$ by~$\sinp$.
Thus, in the following, we will use Equation~\ref{eq:obj-function'} for arbitrary $s \in S$ and say that \emph{$(h_s \mid s \in S)$ is idempotent}.
Moreover, we will simply write~$h$ rather than $(h_s \mid s \in S)$.
\index{objective function!single-valued}
We say that~$h$ is \emph{single-valued} if $|h_\sans(F)| \le 1$ for every $F \subseteq \walg{K}_\sans$.

\index{Bellman's principle of optimality}
Let $(\walg{K},\psi)$ be an $S$-sorted $\Sigma$-algebra and $h$ be an objective function for~$\walg{K}$.
We say that \emph{$h$ satisfies Bellman's principle of optimality} if for every $k \in \mathbb{N}_+$, $s_1,\dots,s_k \in S$, $\sigma \in \Sigma_{(s_1 \dots s_k,\sans)}$, and for every $F_i \subseteq \walg{K}_{s_i}$ with $i \in [k]$ the following holds:
\begin{equation}\label{eq:bellman'}
    h_\sans\big(\psi(\sigma)(F_1,\ldots,F_k)\big) = h_\sans\Big(\psi(\sigma)\big(h_{s_1}(F_1),\ldots,h_{s_k}(F_k)\big)\Big) \enspace.
\end{equation}

Let $(G,(\lalg{YIELD}^\Sigma,\phi))$ be an $S$-sorted yield grammar over $\Sigma$, $(\walg{K},\psi)$ an $S$-sorted $\Sigma$-algebra, and $h$ an objective function for~$\walg{K}$ such that
\begin{enumerate}
    \item $G = (N,\Sigma,A_0,R)$ is unambiguous with $A_0 \in N_\sans$,
    \item $(G, \lalg{YIELD}^\Sigma) \subseteq \gclass{nl}$, and
    \item $h$ satisfies Bellman's principle of optimality.
\end{enumerate}

\index{ADP problem}
The \emph{ADP problem for $(G,(\lalg{YIELD}^\Sigma,\phi))$, $\walg{K}$, and $h$} is the problem of computing, for each $w \in (\Sigma_{(\varepsilon,\sinp)})^*$, the value
\[
    \adp(w) = h_\sans\big( \{ t_{\walg{K}} \mid t \in L(G) \cap \yield_{\Sigma_{(\varepsilon,\sinp)}}^{-1}(w) \} \big)\enspace. \qedhere
\]

We remark that \textcite[p.\,235]{GieMeySte04} do not explicitly require the yield grammar to be unambiguous.
However, they argue against using ambiguous grammars as follows~\cite[p.\,235]{GieMeySte04}:
\enquote{The same candidate has two derivations in the tree grammar: This is bad, as the algorithm will yield redundant answers when asking for more than one, and all counting and probabilistic scoring will be meaningless.}
Here a \emph{candidate} is a $t \in \T_\Sigma$ and a \emph{derivation} of~$t$ is a $d \in (\T_R)_{A_0}$ with $\pi_\Sigma(d) = t$.

Moreover, (ii) is a restriction we impose on $G$ in order to disallow abstract syntax trees that are evaluated in the same way as one of their proper subtrees.
Since the syntactic objects of ADP represent (sub\nobreakdash-)problems which have to be solved, if (ii) did not hold, then the solution of a subproblem would depend on itself, which contradicts dynamic programming.

\begin{example}\label{ex:adp-problem}
    Given two strings $u,v \in (\Sigma_{(\varepsilon,\sinp)})^*$, we can try to edit $u$ into $v$ by traversing $u$ position by position and, at each position $p$, applying one of the following three operations:
    \begin{itemize}
        \item delete the symbol of $u$ at position $p$ and advance $p$ to $p+1$,
        \item insert a symbol into $u$ in front of position $p$ and remain at $p$, and
        \item replace the symbol of $u$ at position $p$ by some other symbol and advance $p$ to $p+1$.
    \end{itemize}
    If the resulting string is $v$, then this edit was successful.
    For the given strings $u,v$ there can be many successful edits.
    In order to find out the ``cheapest'' successful edit, we associate a cost with each of the three operations, e.g., delete  and insert have the cost 1, replace has cost 0 if the replaced symbol and the replacing symbol are equal, otherwise replace has cost 1.
    Then the cost of a  successful edit is the sum of the costs of each occurrence of an operation.
    The \emph{minimum edit distance problem} is the task to calculate, for two given strings $u$ and $v$, the \emph{minimum edit distance between~$u$ and~$v$},
    i.e., the minimum of the costs of all successful edits. We denote this value by $\med(u,v)$.

    Next we formulate the minimum edit distance problem as an ADP problem.
    \begin{enumerate}
        \item We let $\Sigma = \Sigma_{(\varepsilon,\sinp)} \cup \Sigma_{(\sinp,\sans)} \cup \Sigma_{(\sinp \sans,\sans)} \cup \Sigma_{(\sans \sinp,\sans)} \cup \Sigma_{(\sinp \sans \sinp,\sans)}$ with
            \begin{align*}
                \Sigma_{(\varepsilon,\sinp)} &= \{ \text{a},\dots,\text{z} \}\cup \{ \$ \}, \
                \Sigma_{(\varepsilon,\sans)} = \emptyset, \\
                \Sigma_{(\sinp,\sans)} &= \{ \text{nil} \}, \
                \Sigma_{(\sinp \sans,\sans)} = \{ \text{delete} \}, \
                \Sigma_{(\sans \sinp,\sans)} = \{ \text{insert} \}, \text{ and }
                \Sigma_{(\sinp \sans \sinp,\sans)} = \{ \text{replace} \}.
            \end{align*}

        \item We define the $S$-sorted yield grammar $(G,(\lalg{YIELD}^\Sigma,\phi))$ with $G = (N, \Sigma, \nont{A}, R)$ and
            \begin{itemize}
                \item $N = N_\sans = \{ \nont{A} \}$ (where $\nont{A}$ stands for \enquote{alignment}), and
                \item $R$ consists of the following rules:
                    \begin{align*}
                        \nont{A} &\rightarrow \text{nil}(\$) \\
                        \nont{A} &\rightarrow \text{delete}(\delta, \nont{A}) \tag{for every $\delta \in \Sigma_{(\varepsilon,\sinp)}\setminus\{\$\}$} \\
                        \nont{A} &\rightarrow \text{insert}(\nont{A}, \delta) \tag{for every $\delta \in \Sigma_{(\varepsilon,\sinp)}\setminus\{\$\}$} \\
                        \nont{A} &\rightarrow \text{replace}(\delta, \nont{A}, \delta') \enspace. \tag{for every $\delta,\delta' \in \Sigma_{(\varepsilon,\sinp)}\setminus\{\$\}$}
                    \end{align*}
            \end{itemize}

        \item We define the $S$-sorted $\Sigma$-algebra $(\mathbb N \cup \Sigma_{(\varepsilon,\sinp)},\psi)$ such that $(\mathbb N \cup \Sigma_{(\varepsilon,\sinp)})_\sans = \mathbb N$ and $(\mathbb N \cup \Sigma_{(\varepsilon,\sinp)})_\sinp = \Sigma_{(\varepsilon,\sinp)}$,
            where for every $n \in \mathbb N$ and $\delta,\delta_1,\delta_2 \in \Sigma_{(\varepsilon,\sinp)}$
            \begin{align*}
                \psi(\delta) &= \delta \\
                \psi(\text{nil})(\delta) &= 0 \\
                \psi(\text{delete})(\delta,n) &= n + 1 \\
                \psi(\text{insert})(n,\delta) &=  n + 1 \\
                \psi(\text{replace})(\delta_1,n,\delta_2) &= \begin{cases}
                    n &\text{if $\delta_1 = \delta_2$} \\
                    n + 1 &\text{otherwise.}
                \end{cases}
            \end{align*}

        \item We let~$h$ be the objective function such that $h_\sans(F) = \{ \min F \}$ for every non-empty subset~$F$ of~$\mathbb N$. Thus, $h$ is single-valued.
    \end{enumerate}
    Then $\med(u,v) = h_\sans\big( \{ t_{\walg{K}} \mid t \in L(G) \cap \yield_{\Sigma_{(\varepsilon,\sinp)}}^{-1}(u\$v) \} \big)$. Hence, the calculation of $\med(u,v)$ is an ADP problem.
\end{example}

\index{ADP M-monoid over $\walg{K}$ and $h$}
For each ADP problem, we will construct an associated instance of the M-monoid parsing problem as follows.
Let  $(\walg{K},\psi)$ be an $S$-sorted $\Sigma$-algebra and $h$ be an objective function for~$\walg{K}$ that satisfies Bellman's principle of optimality.
We define the \emph{algebra associated with $\walg{K}$ and $h$} as the tuple $(\walg{K}', \oplus, \emptyset, \Sigma',\psi',\infsum)$ such that
\begin{itemize}
    \item $\walg{K}' = \{ h_s(F) \mid s \in S \text{ and } F \subseteq \walg{K}_s \} \cup \{ \bot \}$ where~$\bot$ is a new element,\footnote{$\bot$ helps to guarantee that $\oplus$ is associative, see the proof of Lemma~\ref{lem:adp-mmonoid-complete-distributive}.}
   \item for every $F_1,F_2 \in \walg{K}'$
        \[ F_1 \oplus F_2 = \begin{cases}
            h_s(F_1 \cup F_2) &\text{if there is an $s \in S$ such that $F_1,F_2 \subseteq \walg{K}_s$} \\
            \bot &\text{otherwise,}
        \end{cases} \]
    \item $\Sigma' = \T_\Sigma(X) \cup \{\welem{0}^k \mid k \in \mathbb N\}$ where $\T_\Sigma(X) = \bigcup_{s \in S,u \in S^*} (\T_\Sigma(X_u))_s$ is viewed as a ranked set and each $\welem{0}^k$ has rank~$k$ (we note that for each $\sigma \in \Sigma_{(s_1 \dots s_k,s)}$, $\sigma(x_{1,s_1},\dots,x_{k,s_k}) \in (\T_\Sigma(X_{s_1 \dots s_k}))_s$),
    \item for every $k \in \mathbb{N}$ and $\sigma \in \Sigma'_k$ we define the operation $\psi'(\sigma): (\walg K')^k \to \walg K'$ for every $F_1,\ldots,F_k \in \walg{K}'$ as follows:
        \begin{itemize}
            \item if $\sigma = t$ with $t \in (\T_\Sigma(X_{s_1 \dots s_k}))_s$, then
                \[ \psi'(\sigma)(F_1,\dots,F_k) = \begin{cases}
                    h_s(t_{\walg{K}}(F_1,\dots,F_k)) &\text{if $F_i \subseteq \walg{K}_{s_i}$ for every $i \in [k]$} \\
                    \bot &\text{otherwise,}
                \end{cases} \]
            \item if $\sigma= \welem{0}^k$, then $\psi'(\sigma)(F_1,\ldots,F_k) = \emptyset$, and
        \end{itemize}
    \item $\infsum$ is defined for each $I$-indexed family $(F_i \mid i \in I)$ of elements $F_i \in \walg{K}'$ as
        \[ \infsum_{i \in I} F_i = \begin{cases}
            h_s \Big( \bigcup_{i \in I} F_i \Big) &\text{if there is an $s \in S$ such that $F_i \subseteq \walg{K}_s$ for every $i \in I$} \\
            \bot &\text{otherwise.}
        \end{cases} \]
\end{itemize}
Note that $\emptyset \subseteq \walg{K}_s$ for every $s \in S$.
Thus $\emptyset \in \walg{K}'$, but we cannot assign a sort from~$S$ to~$\emptyset$.
Hence~$\walg{K}'$ is not an $S$-sorted set.

\begin{observation}
    If~$h$ is single-valued or $h_\sans = \id$, then
    \[ \psi'(\sigma)(F_1,\dots,F_k) = \psi(\sigma)(F_1,\dots,F_k) \]
    for every $k \in \mathbb N$, $s,s_1,\dots,s_k \in S$, $\sigma \in \Sigma_{(s_1 \dots s_k,s)}$, and $F_i \subseteq \walg{K}_{s_i}$ for every $i \in [k]$.
\end{observation}

\begin{lemma}[restate={[name={}]lemadpmmonoid}]\label{lem:adp-mmonoid-complete-distributive}
    The algebra associated with $\walg{K}$ and $h$ is a d-complete and distributive M-monoid.
\end{lemma}

\begin{proof}
    For the proof of Lemma \ref{lem:adp-mmonoid-complete-distributive}, we refer to Appendix~\ref{sec:proof-adp-mmonoid}.
\end{proof}

As a consequence of this lemma, we will refer to the algebra associated with~$\walg{K}$ and~$h$ as the \emph{ADP M-monoid over~$\walg{K}$ and~$h$}.

\begin{boxtheorem}[restate={[name={}]thmadpmmonoid}]\label{thm:ADP-M-monoid}
    Each ADP problem is an instance of the M-monoid parsing problem.
    More precisely, let $(G,(\lalg{YIELD}^\Sigma,\phi))$ with $G = (N,\Sigma,A_0,R)$ be a nonlooping RTG-LM.
    Moreover, let $(\walg{K},\psi)$ be  an $S$-sorted $\Sigma$-algebra and~$h$ be an objective function for~$\walg{K}$ that satisfies Bellman's principle of optimality.
    We consider the M-monoid parsing problem with the following input:
    \begin{itemize}
        \item the wRTG-LM
            \[ ((G,(\lalg{YIELD}^\Sigma,\phi)), (\walg{K}', \oplus, \emptyset, \Sigma',\psi',\infsumop), \wt) \]
            where $(\walg{K}', \oplus, \emptyset, \Sigma',\psi',\infsum)$ is the ADP M-monoid over $\walg{K}$ and $h$.
            Moreover, for every $k \in \mathbb N$ and $r = (A \rightarrow t)$ in $R_k$ (viewing~$R$ as a ranked set) we define $\wt(r) = \psi'(t')$, where $t'$ is obtained from $t$ by replacing the $i$th occurrence of a nonterminal by $x_i$ for every $i \in [k]$.
        \item $a \in (\Sigma_{(\varepsilon,\sinp)})^*$.
    \end{itemize}
    Then $\fparse(a) = \adp(a)$.
\end{boxtheorem}

\begin{proof}
    For the proof, we refer to Appendix~\ref{app:adp-mmonoid-parsing}.
\end{proof}

We denote by $\gls{wclass:adp}$ the class of all ADP M-monoids over all algebras $(\walg K, \psi)$ and objective functions $h$ for $\psi$ that satisfy Bellman's principle of optimality.
By Lemma~\ref{lem:adp-mmonoid-complete-distributive}, we have that
\(
    \wlmclass{{\gclass{YIELD} \cap \gclass{nl}},\) \(\wclass{ADP}} \subseteq \wlmclass{\gclass{all}, \wclass{\dcomp} \cap \wclass{dist}}% \enspace.
\).
However, in general
\(
    \wlmclass{\gclass{YIELD} \cap \gclass{nl}, \wclass{\cladp}} \not\subseteq \wlmclass[\clsd]{\gclass{all},\) \(\wclass{\dcomp} \cap \wclass{dist}}% \enspace.
\).
We will address this problem by the additional concepts which we develop in the following sections.
Thus we will be able to show that our M-monoid parsing algorithm can solve every ADP problem (cf.\ Corollary~\ref{cor:applicability}).

\begin{example}[Continuation of Example~\ref{ex:adp-problem}]
    We show how to compute the weight of each rule of the wRTG-LM $\big((G, \lalg{YIELD}^\Delta), (\walg K', \oplus, \emptyset, \psi', \infsum), \wt\big)$, where $(\walg K', \oplus, \emptyset, \psi', \infsum)$ is the ADP M-monoid over $\mathbb N \cup \Sigma_{(\varepsilon,\sinp)}$ and $h$ and $\wt$ is defined as in Theorem~\ref{thm:ADP-M-monoid}.
\begin{align*}
    \wt(\nont{A} \rightarrow \text{nil}(\$))() &= h((\text{nil}')_{\walg{K}}()) = h(\text{nil}_{\walg{K}}()) = h(\{0\}) = \{0\}\\[2mm]
    \wt(\nont{A} \rightarrow \text{delete}(\delta,\nont{A}))(F)
    &= h((\text{delete}(\delta,\nont{A})')_{\walg{K}}(F))
    = h(\text{delete}(\delta,x_1)_{\walg{K}}(F))\\
    &=  h(\{n+1 \mid n \in F\}) = \{1 + \min(F)\} \\[2mm]
    \wt(\nont{A} \rightarrow \text{insert}(\nont{A},\delta))(F)
    &= \{1 + \min(F)\}\\[2mm]
    \wt(\nont{A} \rightarrow \text{replace}(\delta,\nont{A},\delta'))(F) &=
    \begin{cases}
    \{1 + \min(F)\} & \text{if $\delta \not= \delta$}\\
    \{\min(F)\} & \text{otherwise}
    \end{cases}
    \end{align*}
    By Theorem \ref{thm:ADP-M-monoid}, for every $u,v \in (\Sigma_{(\varepsilon,\sinp)} \setminus \{ \$ \})^*$, we have that $\fparse(u\$v) = \med(u,v)$.
\end{example}

\section{M-monoid parsing algorithm}\label{sec:algorithm}

\index{M-monoid parsing algorithm}
The M-monoid parsing algorithm is supposed to solve the M-monoid parsing problem.
As input, it takes a wRTG-LM~$\overline G$ and a syntactic object~$a$.
Its output is intended to be $\fparse(a)$.
The algorithm is a pipeline with two phases (cf.\ Figure~\ref{fig:alg}) and follows the modular approach of \textcites{Goodman1999}{ned03}.
First, a \emph{canonical weighted deduction system} computes from~$\overline G$ and $a$ a new wRTG-LM~$\overline G{}'$ with the same weight structure as $\overline G$, but a different RTG and the language algebra $\cfges$.
Second, $\overline G{}'$ is the input to the \emph{value computation algorithm} (Algorithm~\ref{alg:mmonoid}), which computes the value $V(A_0')$;
this is supposed to be $\infsum_{d \in \AST(G')} \wt(d) = \fparse(a)$.

\subsection{Weighted deduction systems}
\label{sec:weighted-deduction-systems}

The concept of deduction systems is very useful to specify parsing algorithms for strings according to some formal grammar \cite{perwar83,shischper95}.
This concept was extended in~\cite{Goodman1999} and~\cite{ned03} to \emph{weighted deduction systems} in which each inference rule is associated with an operation on some totally ordered set.

A weighted deduction system consists of a \emph{goal item} and a finite set of \emph{weighted inference rules}.
Each inference rule has  the form:
\begin{equation}
    \frac{x_1: I_1, \ ...,\ x_m:I_m}{\omega(x_1,\ldots,x_m): I_0}
    \left\{
        \substack{c_1,\ldots, c_q}
        \right. \label{equ:inf-rule}
\end{equation}
where $m \in \mathbb N$, $\omega$ is an $m$-ary operation (\emph{weight function}), $I_0,\ldots,I_m$ are \emph{items}, and $c_1,\ldots,c_p$ are \emph{side conditions}.
Each item represents a Boolean-valued property (of some combination of nonterminals of the formal grammar $G$ and/or constituents of the string $a=w$).
The meaning of an inference rule is: given that $I_1, \dots, I_m$  and $c_1, \dots, c_p$ are true, $I$ is true as well.
\citet{ned03} pointed out that ``a deduction system having a grammar $G$ [...] and input string $w$ in the side conditions can be seen as a construction $c$ of a context-free grammar $c(G,w)$ [...]''.

Thus, conceptually, a weighted deduction system is a mapping $c$ of which the argument-value relationship is determined by the goal item and the weighted inference rules.
The mapping $c$ takes a grammar $G$ and a string $a$ as arguments and delivers a system $c(G,a)$ of (unconditional) inference rules, called \emph{instantiation} in \cite{ned03}.
Then a parsing algorithm tries to generate the goal item by generating items on demand using the inference rules of $c(G,a)$; in particular, $c(G,a)$ is not fully constructed before applying the parsing algorithm.

Here we generalize the approach of~\cites{perwar83}{shischper95} in two ways:
\begin{enumerate*}[label=(\arabic*)]
    \item instead of string-generating grammars, we consider RTG-LMs over any finitely decomposable language algebra and
    \item instead of unweighted grammars as input, we consider wRTG-LMs (as in~\cite{ned03}).
\end{enumerate*}
For this,
\begin{quote}
    \em in the sequel, we let $(\alg L, \phi)$ be an arbitrary, but fixed finitely decomposable $S$-sorted $\Gamma$-algebra.
\end{quote}

We denote the class of all RTG-LMs with language algebra $\alg L$ by $\gclass{\alg L}$.
Let $\walg{K}$ and $\walg{L}$ be two complete M-monoids.
\index{weighted deduction system}
\index{K@$(\walg K,\walg L)$-weighted deduction system}
A \emph{$(\walg{K},\walg{L})$-weighted deduction system} (or simply: \emph{weighted deduction system}) is a mapping
\[
    \wds_{\walg{K},\walg{L}}: \wlmclass{\gclass{\alg L}, \walg{K}} \times \alg L \to \wlmclass{\gclass{\cfges}, \walg{L}} \enspace,
\]
where the argument-value relationship of $\wds_{\walg{K},\walg{L}}$
% for every weighted RTG-based LM $\overline{G} \in (\gclass{\alg L},\walg{K})$ and each $a \in \alg L$, the weighted RTG-based LM $\wds_{\walg{K},\walg{L}}(\overline{G},a)$
is determined by some goal item and some  finite set of weighted inference rules which may contain references to the arguments.%
\footnote{This definition can be compared to the definition of a function $f: \mathbb{N} \times \mathbb{N} \to \mathbb{N}$ by $f(x,y)=x^2+3y$, in which the argument-value relationship is expressed by an arithmetic expression with references to the arguments $x$ and $y$.}

We allow that the weight algebras $\walg{K}$ and $\walg{L}$ of the argument grammar and the resulting grammar are different in order to enhance flexibility (cf., e.g., \cite[Fig.~3]{ned03}).

In the literature, sound and complete are two important properties that deduction systems must fulfill.
In our context, they could be defined as follows.

We say that $\wds_{\walg{K},\walg{L}}$ is
\begin{itemize}
    \item \emph{sound} if for each $\overline{G}=((G,\alg L),\walg{K},\wt)$ in $\wlmclass{\gclass{\alg L}, \walg{K}}$ and each $a \in \alg L_{\sort(A_0)}$ where $A_0$ is the initial nonterminal of $G$ the following holds:
        if $(G',\cfges)$ is the first component of $\wds_{\walg{K},\walg{L}}(\overline{G}, a)$ and $\varepsilon \in L(G')_{\cfges}$, then $a \in L(G)_{\alg L}$.
    \item \emph{complete} if for each $\overline{G}=((G,\alg L),\walg{K},\wt)$ in $\wlmclass{\gclass{\alg L}, \walg{K}}$ and each $a \in \alg L_{\sort(A_0)}$ where $A_0$ is the initial nonterminal of $G$ the following holds:
        if $a \in L(G)_{\alg L}$, then  $\varepsilon \in L(G')_{\cfges}$, where $(G',\cfges)$ is the first component of $\wds_{\walg{K},\walg{L}}(\overline{G}, a)$.
    \item \emph{unweighted} if $\walg{K}=\walg{L}$ and this M-monoid is the M-monoid associated with the Boolean semiring.
\end{itemize}

In our context, we need a stronger condition on weighted deduction systems.
We call a \index{weighted deduction system!weight-preserving}
weighted deduction system $\wds_{\walg K,\walg K}: \wlmclass{\gclass{\alg L}, \walg K} \times \alg L \to \wlmclass{\gclass{\cfges}, \walg K}$  \emph{weight-preserving}, if for each $\overline G = ((G, \alg L), \walg K, \wt)$ in $\wlmclass{\gclass{\alg L}, \walg K}$ and $a \in \alg L_{\sort(A_0)}$ with $G = (N,\Sigma,A_0,R)$, $\wds_{\walg K,\walg K}(\overline G, a) = ((G', \cfges), \walg K, \wt')$, and $G' = (N',\Sigma',A_0',R')$ there is a bijective mapping
\[
    \psi: \AST(G, a) \to \AST(G')
\]
such that for every $d \in \AST(G, a)$ we have $\wthom{d} = \wthom{\psi(d)}$.

\begin{observation}\label{obs:weight-preserving-parse}
    Let $\overline G = ((G, \alg L), \walg K, \wt)$ be a wRTG-LM with $G = (N,\Sigma,A_0,R)$, $a \in \alg L_{\sort(A_0)}$, and $\wds_{\walg K,\walg K}: \wlmclass{\gclass{\alg L}, \walg K} \times \alg L \to \wlmclass{\gclass{\cfges}, \walg K}$ be a weight-preserving weighted deduction system.
    If $\wds_{\walg K,\walg K}(\overline G, a) = \big((G', \cfges), \walg K, \wt'\big)$, then $\fparse_{(G,\alg L)}(a) = \fparse_{(G',\cfges)}(\varepsilon)$.
\end{observation}

\begin{lemma}[restate={[name={}]lemwdspreserving}]\label{lem:weight-eq-implies-sound-and-complete}
    Each weight-preserving weighted deduction system is sound and complete.
\end{lemma}

\begin{proof}
    For the proof of Lemma~\ref{lem:weight-eq-implies-sound-and-complete}, we refer to Appendix~\ref{app:weight-preserving-wds}.
\end{proof}

Next we define a particular weighted deduction system.
It covers, e.g., the (unweighted) CYK deduction system~\cite{shischper95} and the deduction system for LCFRS of \textcite{kal10}.
We will use this particular weighted deduction system in our M-monoid parsing algorithm.

Let $(\walg K,\oplus,\mathbb 0,\Omega,\infsum)$ be a complete M-monoid such that $\id(\walg K) \in \Omega$.
\index{canonical weighted deduction system}
The \emph{canonical $\walg K$-weighted deduction system}
is the weighted deduction system
\[
    \cnc: \wlmclass{\gclass{\alg L}, \walg K} \times \alg L \to \wlmclass{\gclass{\cfges}, \walg K}
\]
such that for every $\overline G = ((G,\alg L),\walg{K},\wt)$ in $\wlmclass{\gclass{\alg L}, \walg K}$ and $a_0 \in \alg L_{\sort(A_0)}$, where $A_0$ is the initial nonterminal of~$G$, the wRTG-LM $\cnc(\overline G, a_0)$ is defined by
\[
    \cnc(\overline G, a_0) = ((G', \cfges), \walg K, \wt')
\]
where $G'$ and $\wt'$ are obtained from $\overline G$ and $\wt$ as follows.
We let $G=(N,\Sigma,A_0,R)$ and define $\mathrm{rhs}(R) = \{ t \in \T_\Sigma(N) \mid \text{$t$ is the right-hand side of some $r \in R$} \}$.
Then $G'=(N',\Sigma',A_0',R')$ with
\begin{itemize}
    \item $N' = N \times \mathrm{rhs}(R) \times \factors(a_0) \cup \{ [A_0,a_0] \}$ \ (set of \emph{items})
    \item $A_0' = [A_0,a_0]$ \ (\emph{goal item})
    \item For every rule $r=(A \to t)$ in $R$ and $a,a_1,\ldots,a_k \in \factors(a_0)$, let $\yield_N(t) = A_1 \dots A_k$ with $k \in \mathbb{N}$ (i.e., including $k=0$) and $A_1,\dots,A_k \in N$;
        now, if $t'_{\alg L}(a_1,\ldots,a_k)=a$, where~$t'$ is obtained from~$t$ by replacing the $i$th occurrence of a nonterminal by~$x_i$ for every $i \in [k]$, then each rule in the set
        \begin{align*}
            \mathrm{instances}(r) = \{ [A,t,a] &\to \langle x_1\ldots x_k\rangle([A_1,t_1,a_1],\ldots,[A_k,t_k,a_k]) \mid k_1,\dots,k_k \in \mathbb N \text{ and } \\
            &t_1,\dots,t_k \in \mathrm{rhs}(R) \text{ with } \sort(t_i) = \sort(A_i) \text{ for each $i \in [k]$} \}
        \end{align*}
        is in $R'$.
        We define $\wt'(r') = \wt(r)$ for each $r' \in \mathrm{instances}(r)$.
        Moreover, for each rule $r = (A_0 \to t)$ in~$R$ the rule
        \[ r' = ([A_0,a_0] \to \langle x_1 \rangle ([A_0,t,a_0])) \]
        is in~$R'$ and we let $\wt'(r') = \id(\walg K)$.
    \item $\Sigma' = \{ \langle x_1 \dots x_k \rangle \in \Gamma^{\lalg{CFG},\emptyset} \mid 0 \le k \le \maxrk(G) \}$. \qedhere
\end{itemize}

We note that the requirement $\id(\walg K) \in \Omega$ is not a restriction, as the identity relation is defined on every set and can therefore be added to $\walg K$, if necessary.
We also note that the nonterminals of $\cnc(\overline G, a_0)$ contain syntactic objects and right-hand sides of rules.
This is in contrast to the literature, where items of deduction systems contain positions of a string~\cites{shischper95}{Goodman1999}{ned03}{kal10}.
This deviation is due to two reasons.
First, since $\cnc$ is defined for arbitrary finitely decomposable language algebras, string positions are not general enough to represent the language algebra in the nonterminals of $\cnc(\overline G, a)$, but syntactic objects are.
Second, if the nonterminals contained syntactic objects, but not right-hand sides of rules, then we do not know how to compute $\cnc$.

\begin{lemma}[restate={[name={}]lemcncwp}]\label{lem:cnc-weight-preserving}
    The canonical $\walg{K}$-weighted deduction system $\cnc$ is weight-preserving. Hence, $\cnc$ is sound and complete.
\end{lemma}

\begin{proof}
    For the proof of Lemma \ref{lem:cnc-weight-preserving}, we refer to Appendix~\ref{app:cnc-weight-preserving}.
\end{proof}

\begin{example}\label{ex:wds-lcfrs}
    We consider the tropical M-monoid $\walg T = (\mathbb R_0^\infty,\min,\infty,\Omega_+,\inf)$ (cf. Example~\ref{ex:tropical-M-monoid}) and the alphabet~$\Delta$ from Example~\ref{ex:lcfrs}.
    We illustrate the canonical $\walg T$-weighted deduction system
    \[
        \cnc: \wlmclass{\gclass{\lalg{LCFRS}^\Delta}, \walg T} \times \Delta^* \to \wlmclass{\gclass{\cfges}, \walg T}
    \]
    of which the argument-value relationship is determined by the inference rules discussed in \cite[Chapter~7]{kal10}.

    We apply $\cnc$ to the linear context-free rewriting system with $G = (N,\Sigma,A_0,R)$ from Example~\ref{ex:lcfrs} and the string $a = \tJan\ \tPiet\ \tMarie\ \tzag\ \thelpen\ \tlezen$.
    The weights of the rules of~$G$ in the tropical M-monoid are shown in Table~\ref{tab:lcfrs-weighted}.
    Then $\cnc(G, a)$ is the wRTG-LM
    \[
        \cnc(G, a) = \big((G', \cfges), \ (\mathbb R_0^\infty,\min,\infty,\Omega_+,\inf), \ \wt'\big) \enspace,
    \]
    where
    \begin{itemize}
        \item $G'=(N',\Sigma',A_0',R')$ is a $\{\iota\}$-sorted RTG given by
            \begin{itemize} \fussy
                \item $N' = N'_\iota = \{ [A,t,v] \mid A \in \{ \nont{root},\nont{nsub},\nont{dobj} \}, t \in \mathrm{rhs}(R), v \in \factors(a) \} \cup \{ [\nont{root},a] \}$, where
                    \[ \mathrm{rhs}(R) = \{ \begin{aligned}[t]
                        &\langle \tJan \rangle,\langle \tPiet \rangle,\langle \tMarie \rangle, \langle x^{(1)}_1 x^{(2)}_1 \tzag \, x^{(2)}_2 \rangle(\nont{nsub},\nont{dobj}), \\
                        &\langle x^{(1)}_1 x^{(2)}_1,\thelpen \, x^{(2)}_2 \rangle(\nont{nsub},\nont{dobj}),\langle x^{(1)}_1, \tlezen \rangle(\nont{nsub}) \};
                    \end{aligned} \]
                    first, we give an intuition for the computation of $\factors(a)$ by showing the factors of two particular elements of $\lalg{LCFRS}^\Delta$:\\
                    for $\tJan\ \tPiet\ \tMarie\ \tzag\ \thelpen\ \tlezen \in (\lalg{LCFRS}^\Delta)_1$ we have
                    \begin{align*}
                        &\phi\big(\langle x^{(1)}_1 x^{(2)}_1 \tzag\ x^{(2)}_2 \rangle\big)^{-1} (\tJan\ \tPiet\ \tMarie\ \tzag\ \thelpen\ \tlezen) = \\
                        &\qquad \{ \begin{aligned}[t]
                            &\varepsilon, (\tJan\ \tPiet\ \tMarie, \thelpen\ \tlezen), \\
                            &\tJan, (\tPiet\ \tMarie, \thelpen\ \tlezen), \\
                            &\tJan\ \tPiet, (\tMarie, \thelpen\ \tlezen), \\
                            &\tJan\ \tPiet\ \tMarie, (\varepsilon, \thelpen\ \tlezen) \}
                        \end{aligned}
                        \intertext{and for $(\tPiet\ \tMarie, \thelpen\ \tlezen) \in (\lalg{LCFRS}^\Delta)_2$ we have}
                        &\phi\big(\langle x^{(1)}_1 x^{(2)}_1, \thelpen\ x^{(2)}_2 \rangle\big)^{-1} (\tPiet\ \tMarie\ , \thelpen\ \tlezen) = \\
                        &\qquad \{ \varepsilon, (\tPiet\ \tMarie, \tlezen), \tPiet, (\tMarie, \tlezen), \tPiet\ \tMarie, (\varepsilon, \tlezen) \} \enspace.
                    \end{align*}
                    In total, the set $\factors(a)$ is the set
                    \[
                        \factors(a) = \{ \begin{aligned}[t]
                            &\tJan\ \tPiet\ \tMarie\ \tzag\ \thelpen\ \tlezen, \\
                            &(\tJan\ \tPiet\ \tMarie, \thelpen\ \tlezen), (\tPiet\ \tMarie,\thelpen\ \tlezen), \\
                            &(\tMarie, \thelpen\ \tlezen), (\varepsilon, \thelpen\ \tlezen), \\
                            &\varepsilon, \tJan, \tJan\ \tPiet, \tJan\ \tPiet\ \tMarie, \\
                            &(\tJan\ \tPiet\ \tMarie, \tlezen), (\tPiet\ \tMarie, \tlezen), \\
                            &(\tMarie, \tlezen), (\varepsilon, \tlezen), \\
                            &\tPiet, \tPiet\ \tMarie, \tMarie \} \enspace.
                        \end{aligned}
                    \]
                \item $\Sigma' = \Sigma'_{(\varepsilon,\iota)} \cup \Sigma'_{(\iota,\iota)} \cup \Sigma'_{(\iota\iota,\iota)}$
                    where
                    \begin{align*}
                        \Sigma'_{(\varepsilon,\iota)} = \{ \langle \varepsilon \rangle \}\enspace, \
                        \Sigma'_{(\iota,\iota)} = \{ \langle x_1 \rangle \} \enspace, \ \text{and} \
                        \Sigma'_{(\iota\iota,\iota)} = \{ \langle x_1 x_2 \rangle \} \enspace,
                    \end{align*}
                \item $A_0' = [\nont{root},a]$, and
                \item the set of rules $R'$ is given in Figure~\ref{fig:wds:cyk},
            \end{itemize}
            and
        \item For every $r \in R' \setminus R'_{A_0'}$ with $r = ([A,t,u] \to \langle x_1 \dots x_k \rangle ([A_1,t_1,u_1],\dots,[A_k,t_k,u_k]))$ and $k \in \mathbb N$ we define $\wt'(r) = \wt(A \to t)$ and for every $r \in R'_{A_0'}$ we let $\wt'(r) = \mul^{(1)}_0$. \qedhere
    \end{itemize}

    \begin{table}[t]
        \centering
        \begin{tabular}{cc}
            \toprule
            Rule $r \in R$ & $\wt(r)$ \\
            \midrule
            \multicolumn{1}{c}{
                \begin{tabular}{r@{\;{$\to$}\;}l}
                    $\mathrm{root}$ & $\langle x^{(1)}_1 x^{(2)}_1  \tzag\ x^{(2)}_2\rangle (\textrm{nsub}, \textrm{dobj})$ \\
                    $\mathrm{dobj}$ & $\langle x^{(1)}_1 x^{(2)}_1, \thelpen\ x^{(2)}_2\rangle (\textrm{nsub}, \textrm{dobj})$ \\
                    $\mathrm{dobj}$ & $\langle x^{(1)}_1,  \tlezen\rangle (\textrm{nsub})$ \\
                    $\mathrm{nsub}$ & $\langle \tJan\rangle$ \\
                    $\mathrm{nsub}$ & $\langle \tPiet\rangle$ \\
                    $\mathrm{nsub}$ & $\langle \tMarie\rangle$
                \end{tabular}
            }
            &
            \multicolumn{1}{c}{
                \begin{tabular}{r@{\;{$\mapsto$}\;}l}
                    $(\welem{k}_1,\welem{k}_2)$ & $0 + \welem{k}_1 + \welem{k}_2$ \\
                    $(\welem{k}_1,\welem{k}_2)$ & $4 + \welem{k}_1 + \welem{k}_2$ \\
                    $\welem{k}$ & $7 + \welem{k}$ \\
                    $()$ & $3$ \\
                    $()$ & $5$ \\
                    $()$ & $12$
                \end{tabular}
            }
            \\
            \bottomrule
        \end{tabular}
        \caption{The linear context-free rewriting system from Example~\ref{ex:lcfrs} weighted in the tropical M-monoid. The numbers occurring in the definitions of $\wt(r)$ are chosen arbitrarily.}\label{tab:lcfrs-weighted}
    \end{table}

    \begin{figure}[t]
      {\small
        \begin{align*}
            [\nont{root},a] &\to \langle x_1 \rangle ([\nont{root},\langle x^{(1)}_1 x^{(2)}_1 \,\tzag\, x^{(2)}_2 \rangle,a]) \\[3ex]
            \intertext{for every $v_1,v_2 \in \factors(a)$ with $\phi(\langle x^{(1)}_1 x^{(2)}_1 \,\tzag\, x^{(2)}_2 \rangle)(v_1,v_2) \in \factors(a)$, \newline
                $t_1 \in \{ \langle \tJan \rangle,\langle \tPiet \rangle,\langle \tMarie \rangle \}$, and
                $t_2 \in \{ \langle x^{(1)}_1 x^{(2)}_1, \,\thelpen\, x^{(2)}_2 \rangle(\nont{nsub},\nont{dobj}),\langle x^{(1)}_1, \,\tlezen\, \rangle(\nont{nsub}) \}$:}
            [\nont{root},\langle x^{(1)}_1 x^{(2)}_1 \,\tzag\, x^{(2)}_2 \rangle(\nont{nsub},\nont{dobj}),\phi(\langle x^{(1)}_1 x^{(2)}_1 \, \tzag \, x^{(2)}_2 \rangle)(v_1,v_2)] &\to \langle x_1 x_2 \rangle ([\nont{nsub},t_1,v_1],[\nont{dobj},t_2,v_2])
            \\[2ex]
            \intertext{for every $v_1,v_2 \in \factors(a)$ with $\phi(\langle x^{(1)}_1 x^{(2)}_1, \,\thelpen\, x^{(2)}_2 \rangle)(v_1,v_2) \in \factors(a)$, \newline
                $t_1 \in \{ \langle \tJan \rangle,\langle \tPiet \rangle,\langle \tMarie \rangle \}$, and
                $t_2 \in \{ \langle x^{(1)}_1 x^{(2)}_1, \,\thelpen\, x^{(2)}_2 \rangle(\nont{nsub},\nont{dobj}),\langle x^{(1)}_1, \,\tlezen\, \rangle(\nont{nsub}) \}$:}
            [\nont{dobj},\langle x^{(1)}_1 x^{(2)}_1, \,\thelpen\, x^{(2)}_2 \rangle(\nont{nsub},\nont{dobj}),\phi(\langle x^{(1)}_1 x^{(2)}_1, \, \thelpen \, x^{(2)}_2 \rangle)(v_1,v_2)] &\to \langle x_1 x_2 \rangle ([\nont{nsub},t_1,v_1],[\nont{dobj},t_2,v_2])
            \\[2ex]
            \intertext{for every $v_1 \in \factors(a)$ with $\phi(\langle x^{(1)}_1, \,\tlezen\, \rangle)(v_1) \in \factors(a)$ and
                $t_1 \in \{ \langle \tJan \rangle,\langle \tPiet \rangle,\langle \tMarie \rangle \}$:}
            [\nont{dobj},\langle x^{(1)}_1, \,\tlezen\, \rangle(\nont{nsub}),\phi(\langle x^{(1)}_1, \, \tlezen \, \rangle)(v_1)] &\to \langle x_1 \rangle ([\nont{nsub},t_1,v_1])
            \\[3ex]
            [\nont{nsub},\langle \tJan \rangle,\tJan]   &\to \langle \varepsilon \rangle \\
            [\nont{nsub},\langle \tPiet \rangle,\tPiet]  &\to \langle \varepsilon \rangle \\
            [\nont{nsub},\langle \tMarie \rangle,\tMarie] &\to \langle \varepsilon \rangle
            \enspace.
        \end{align*}
        \caption{Application of the canonical $\walg T$-weighted deduction system to the grammar of Example~\ref{ex:lcfrs} and the string $a = \tJan\ \tPiet\ \tMarie\ \tzag\ \thelpen\ \tlezen$ (only the rules are given).}\label{fig:wds:cyk}
      } \end{figure}
\end{example}

We finish this section with a result which shows how the canonical weighted deduction system connects two classes of RTG-LMs.

\begin{lemma}[restate={[name={}]lemnlcncacyc}]\label{lem:no-loops-cnc-acyclic}
    For every $\overline G \in \wlmclass{\gclass{nl} \cap \gclass{\findc}, \wclass{all}}$ and syntactic object $a$ it holds that $\cnc(\overline G, a) \in \wlmclass{\gclass{acyc}, \wclass{all}}$.
\end{lemma}

\begin{proof}
    For the proof of Lemma \ref{lem:no-loops-cnc-acyclic}, we refer to Appendix~\ref{app:nl-cnc-acyc}.
\end{proof}

\clearpage

\subsection{Value computation algorithm}
\label{sec:value-computation-algorithm}

\newlength\variableslength
\newlength\tabularlength

\begin{algorithm}[h]
    \begin{algorithmic}[1]
        \Require a $(\gclass{\cfges}, \wclass{all})$-LM $\big((G', \cfges), (\walg K, \oplus, \welem 0, \Omega, \infsum), \wt'\big)$ with $G' = (N', \Sigma', A_0', R')$
        \setlength\tabularlength{\textwidth}
        \settowidth\variableslength{\algorithmicvariables}
        \addtolength\tabularlength{-\variableslength}
        \addtolength\tabularlength{-1.05em}
        \Variables \begin{tabular}[t]{p{\tabularlength}}
            $V: N' \to \walg K$, $\Vnew \in \walg K$ \Fixedcomment{$\mathcal V: N' \to \mathcal P(\T_{R'})$, $\mathcal \Vnew \subseteq \T_{R'}$} \\
            $\changed \in \mathbb B$ \\
            \Fixedcomment{$\select \in N$'} \\
        \end{tabular}
        \Ensure $V(A_0)$
        % Initialization
        \ForEach{$A \in N'$}\label{l:initialization}
            \State $V(A) \gets \welem 0$\label{l:init-v} \Fixedcomment{$\mathcal V(A) \gets \emptyset$}
        \EndFor
        \Statex \Fixedcomment{$n \gets 0$}
        % Loop
        \Repeat\label{l:loop}
            \State $\changed \gets \lfalse$\label{l:reset-changed}
            \ForEach{$A \in N'$}\label{l:for-loop} \Fixedcomment{$\select \gets A$}
                \State $\Vnew \gets \welem 0$\label{l:init-vnew} \Fixedcomment{$\mathcal \Vnew \gets \emptyset$}
                \ForEach{$r = \big( A \to \sigma(A_1, \dots, A_k) \big)$ in $R'$}\label{l:updates}
                    \State $\Vnew \gets \Vnew \oplus \wt'(r)\big(V(A_1), \dots, V(A_k)\big)$\label{l:update-vnew}
                    \Statex \Comment $\mathcal \Vnew \gets \mathcal \Vnew \cup \{ r(d_1, \dots, d_k) \mid d_1 \in \mathcal V(A_1), \dots, d_k \in \mathcal V(A_k) \}$
                \EndFor
                \If{$V(A) \not= \Vnew$}\label{l:check-difference}
                    \State $\changed \gets \ltrue$\label{l:update-changed}
                \EndIf
                \State $V(A) \gets \Vnew$\label{l:update-v} \Fixedcomment{$\mathcal V(A) \gets \mathcal \Vnew$; $n \gets n + 1$}
            \EndFor
        \Until{$\changed = \lfalse$}\label{l:loop-condition}
    \end{algorithmic}
    \caption{Value computation algorithm}\label{alg:mmonoid}
\end{algorithm}

\begingroup\setlength\emergencystretch{20pt}
\index{value computation algorithm}
The value computation algorithm is given as Algorithm~\ref{alg:mmonoid}.
It takes as input a $(\gclass{\cfges}, \wclass{all})$-LM $\big((G', \cfges), \walg K,\wt'\big)$ with $G' = (N', \Sigma', A_0', R')$ and outputs the value $V(A_0') \in \walg K$, where $V: N' \to \walg K$ is a mapping it maintains.
Furthermore, it maintains the Boolean variable \emph{changed}.
The algorithm consists of two phases:
In the first phase (lines~\ref{l:initialization}--\ref{l:init-v}), for every nonterminal~$A$ it lets $V(A)$ be~$\welem 0$.
In the second phase (lines~\ref{l:loop}--\ref{l:loop-condition}), a \emph{repeat-until loop} is iterated until the variable \emph{changed} has the value~$\lfalse$.
This variable is set to~$\lfalse$ at the start of each iteration (line~\ref{l:reset-changed}), but may be assigned the value~$\ltrue$ in line~\ref{l:update-changed}.
In each iteration of the repeat-until loop, an \emph{inner for loop} iterates over every nonterminal (lines~\ref{l:for-loop}--\ref{l:update-v}).
For each nonterminal~$A$, a value~$\Vnew$ is computed (lines~\ref{l:init-vnew}--\ref{l:update-vnew}), where
\[
    \Vnew(A) = \bigoplus_{\substack{r \in R': \\ r = (A \to \sigma(A_1, \dots, A_k))}} \wt'(r)\big(V(A_1), \dots, V(A_k)\big) \enspace.
\]
If this value differs from $V(A)$, then the variable \emph{changed} is set to~$\ltrue$.
Finally, $V(A)$ is set to~$\Vnew$ (lines~\ref{l:check-difference}--\ref{l:update-v}).
\endgroup

Note that we have placed additional variables and statements behind the comment symbol~$\triangleright$.
These are not part of the algorithm and can be ignored for the time being.
They describe formal properties of the algorithm which we will refer to in the next section.

\section{Termination and correctness of the M-monoid parsing algorithm}

We are interested in two formal properties of Algorithm~\ref{alg:mmonoid} and of the M-monoid parsing algorithm (Figure~\ref{fig:alg}): termination and correctness.

Algorithm~\ref{alg:mmonoid} computes the weights of the ASTs bottom-up and reuses the results of common subtrees (as in dynamic programming); this requires distributivity of the weight algebra.
Moreover, solving the M-monoid parsing problem involves the computation of an infinite sum, which can only be done by a terminating algorithm in special cases (cf.\ the start of Section~\ref{sec:closed}).
We have already shown (cf.\ Theorem~\ref{thm:tr-trc}) that the class of closed wRTG-LMs is a good candidate for such a special case.
Hence, in the following, we will be concerned with inputs of that class.

\subsection{Properties of the value computation algorithm}

In the following, we will study formal properties of Algorithm~\ref{alg:mmonoid}.
We are mainly interested in two questions:
\begin{enumerate}[label=(\alph*)]
    \item Does the algorithm terminate for every input?
    \item Does $V(A_0') = \infsum_{d \in \AST(G')} \wt'(d)_{\walg K}$ hold after termination?
\end{enumerate}

For the analysis of Algorithm~\ref{alg:mmonoid} we have introduced the additional variables $\mathcal V: N' \to \mathcal P(\T_{R'})$, $\mathcal \Vnew \subseteq \T_{R'}$, $\select \in N'$, and $n \in \mathbb N$, where
\begin{itemize}
    \item $\mathcal V: N' \to \mathcal P(\T_{R'})$ captures for each nonterminal~$A$ and each iteration of the inner for loop the subset of~$(\T_{R'})_A$ which contributes to the computation of $V(A)$ in that iteration,
    \item $\mathcal \Vnew \subseteq \T_{R'}$ is used to accumulate the set $\mathcal V(A)$ of abstract syntax trees,
    \item in each iteration of the inner for loop, $\select \in N'$ is the nonterminal of that iteration,
        and
    \item $n \in \mathbb N$ is used to count the iterations of the inner for loop.
\end{itemize}
We have placed these new variables and the statements which modify them behind the comment symbol~$\triangleright$.
In the analysis we will treat this comments as if they were actual statements (\emph{auxiliary statements}).

For technical purposes, we formalize the sequences of values which are taken by~$V$,~$\mathcal V$, $\select$, and $\changed$ during the iterations of the inner for loop as families.

We define the families $(V_n \mid n \in \mathbb N)$,
$(\mathcal V_n \mid n \in \mathbb N)$,
$(\select_n \mid n \in \mathbb N)$, and
$(\changed_n \mid n \in \mathbb N)$
such that for each $n \in \mathbb N$, we have $V_n: N' \to \walg K$,
$\mathcal V_n: N' \to \mathcal P(\T_{R'})$,
$\select_n \in N' \cup \{ \bot \}$,
$\changed_n \in \mathbb B$, and the following holds for every $n \in \mathbb N$:
\begin{itemize}
    \item if lines~\ref{l:init-vnew}--\ref{l:update-v} have been executed~$n$ times, then after executing line~\ref{l:update-v}, including the auxiliary statements, the values of~$V$,~$\mathcal V$, $\select$, and $\changed$ are~$V_n$,~$\mathcal V_n$, $\select_n$, and $\changed_n$ respectively,
    \item intuitively, we define the values of $V_n$, $\mathcal V_n$, $\select_n$, and $\changed_n$ for those~$n$ which are beyond termination of the algorithm by copying the corresponding values from the final iteration.
        Formally, if there is a $k \in \mathbb N$ such that $k < n$ and $\changed_{|N'| \cdot \lfloor k / |N'| \rfloor} = \lfalse$, then we define $V_n = V_k$, $\mathcal V_n = \mathcal V_k$, $\select_n = \bot$ and $\changed_n = \lfalse$.
\end{itemize}
Thus,~$V_0$ and~$\mathcal V_0$ denote the respective values after the execution of lines~\ref{l:initialization}--\ref{l:init-v} and $\select_0$ is the nonterminal chosen by the inner for loop when Algorithm~\ref{alg:mmonoid} executes line~\ref{l:for-loop} for the first time.
We define $\changed_0 = \ltrue$.

Let $n \in \mathbb N$.
We say that the algorithm \emph{terminates after~$n$ iterations of the inner for loop} if~$n$ is the smallest number such that $\changed_{|N'| \cdot \lfloor n / |N'| \rfloor} = \lfalse$.
We say that the algorithm \emph{still runs in the $n$th iteration of the inner for loop} if for every $n' \in \mathbb N$ with $n' \le n$ it holds that $\changed_{|N'| \cdot \lfloor n' / |N'| \rfloor} = \ltrue$.
Let $A \in N'$ and $d \in (\T_{R'})_A$.
We say that \emph{$d$ is first added to $\mathcal V(A)$ in the $n$th iteration of the inner for loop} if $d \in \mathcal V_{n+1}(A)$ and for every $n' \in \mathbb N$ with $n' \le n$, $d \not\in \mathcal V_{n'}(A)$.

\begin{observation}\label{obs:v-nplus1}
    For every $n \in \mathbb N$ and $A \in N'$ such that the algorithm still runs in the $n$th iteration of the inner for loop the following holds:
    If $\select_n = A$, then
    \begin{align*}
        V_{n+1}(A) &= \bigoplus_{\substack{r \in R': \\ r = (A \to \sigma(A_1, \dots, A_k))}} \wt'(r)\big(V_n(A_1), \dots, V_n(A_k)\big)
        \intertext{and}
        \mathcal V_{n+1}(A) &= \bigcup_{\substack{r \in R': \\ r = (A \to \sigma(A_1, \dots, A_k))}} \{ r(d_1, \dots, d_k) \mid d_1 \in \mathcal V_n(A_1), \dots, d_k \in \mathcal V_n(A_k) \} \enspace.
    \end{align*}
    If $\select_n \not= A$, then $V_{n+1}(A) = V_n(A)$ and $\mathcal V_{n+1}(A) = \mathcal V_n(A)$.
\end{observation}

\index{value computation algorithm!correct}
Let $\wlmclass{} \subseteq \wlmclass{\gclass{all}, \wclass{all}}$.
Algorithm~\ref{alg:mmonoid} is \emph{correct for $\wlmclass{}$} if for every $\overline G{}' = \big((G', \cfges), \walg K, \wt'\big)$ in $\wlmclass{}$ with $G' = (N', \Sigma', A_0', R')$ the following holds:
if Algorithm~\ref{alg:mmonoid} is executed with~$\overline G{}'$ as input, then $V(A_0') = \infsum_{d \in \AST(G')} \wtphom{d}$ after termination.

In the following subsections we show that Algorithm~\ref{alg:mmonoid} terminates for every wRTG-LM in the class $\wlmclass[\clsd]{\gclass{\cfges}, \wclass{dist} \cap \wclass{\dcomp}}$ and that is correct for this class.

\begin{quote}
    \em For the rest of this section, we let $c \in \mathbb N$ and $\overline G{}' = \big((G', \cfges), (\walg K, \oplus, \welem 0, \Omega), \wt'\big)$ with $G' = (N', \Sigma', A_0', R')$ be a $c$-closed $(\gclass{\cfges}, \wclass{dist} \cap \wclass{\dcomp})$-LM.
\end{quote}

We start with two general lemmas that are needed for both termination and correctness.

\begin{lemma}[restate={[name={}]lemvisbigsum}]\label{lem:v-is-bigsum}
    For every $n \in \mathbb N$ and $A \in N'$ it holds that $V_n(A) = \bigoplus_{d \in \mathcal V_n(A)} \wtphom{d}$.
\end{lemma}

\begin{proof}
    For the proof of Lemma \ref{lem:v-is-bigsum}, we refer to Appendix~\ref{app:vca-general}.
\end{proof}

\begin{lemma}[restate={[name={}]lemmcvmonotone}]\label{lem:mcv-monotone}
    For every $n \in \mathbb N$ and $A \in N'$ the following holds: for each $n \in \mathbb N$ with $n' > n$, $\mathcal V_n(A) \subseteq \mathcal V_{n'}(A)$.
\end{lemma}

\begin{proof}
    For the proof of Lemma \ref{lem:mcv-monotone}, we refer to Appendix~\ref{app:vca-general}.
\end{proof}

\subsubsection{Termination of the value computation algorithm}

An important step in showing that Algorithm~\ref{alg:mmonoid} terminates on every closed wRTG-LM is the following Lemma.

\begin{lemma}[restate={[name={}]lemcutcycles}]\label{lem:cut-cycles}
    For every $d \in \T_{R'}, n \in \mathbb N$, and $A \in N'$ the following holds:
    if $d \in \mathcal V_n(A)$, then $\cotrees(d) \subseteq \mathcal V_n(A)$.
\end{lemma}

\begin{proof}
    For the proof of Lemma~\ref{lem:cut-cycles}, we refer to Appendix~\ref{app:vca-termination}.
\end{proof}

\begin{quote}
    \em For the rest of this subsection, for every $n \in \mathbb N$ and $A \in N'$ we let $\Delta_n(A) = \mathcal V_{n+1}(A) \setminus \mathcal V_n(A)$.
\end{quote}

From Lemma~\ref{lem:cut-cycles}, we are able to conclude the following.

\begin{lemma}[restate={[name={}]lemmcvgrowsonchange}]\label{lem:mcv-grows-on-change}
    For every $n \in \mathbb N$ and $A \in N'$ the following holds: if $V_{n+1}(A) \not= V_n(A)$, then $\mathcal V_{n+1}(A) \cap \TRpc \supset \mathcal V_n(A) \cap \TRpc$.
\end{lemma}

\begin{proof}
    For the proof of Lemma~\ref{lem:mcv-grows-on-change}, we refer to Appendix~\ref{app:vca-termination}.
\end{proof}

\begin{boxtheorem}\label{thm:vca-terminating}
    For every wRTG-LM $\overline G$ in $\wlmclass[\clsd]{\gclass{\cfges}, \wclass{\dcomp} \cap \wclass{dist}}$ the following holds:
    if the value computation algorithm (Algorithm~\ref{alg:mmonoid}) is executed with~$\overline G$ as input, then it terminates.
\end{boxtheorem}

\begin{proof}
    We define the mapping $\ffint: \mathcal P(\T_{R'}^{(c)}) \times \walg B \to \mathbb N^2$ such that $\ffint(D, b) = (|\T_{R'}^{(c)}| - |D|, \delta(b))$ for each $D \subseteq \T_{R'}$ and $b \in \walg B$,
    where $\delta(\ltrue) = 1$ and $\delta(\lfalse) = 0$.
    Then for each $n \in \mathbb N$ with
    \[
      \bigvee_{n' \in \mathbb N: \, n \cdot |N'| < n' \leq (n+1) \cdot |N'|} \changed_{n'} = \ltrue
    \]
    we have
    \begin{align*}
        &\ffint\left(\TRpc \cap \bigcup_{A \in N'} \mathcal V_{(n+1) \cdot |N'|}(A), \bigvee_{\substack{n' \in \mathbb N: \\ n \cdot |N'| < n' \leq (n+1) \cdot |N'|}} \changed_{n'}\right)
        \\
        &= \left(|\T_{R'}^{(c)}| - \left|\TRpc \cap \bigcup_{A \in N'} \mathcal V_{(n+1) \cdot |N'|}(A)\right|, \delta\left(\bigvee_{\substack{n' \in \mathbb N: \\ n \cdot |N'| < n' \leq (n+1) \cdot |N'|}} \changed_{n'}\right)\right)
        \\
        &> \left(|\T_{R'}^{(c)}| - \left|\TRpc \cap \bigcup_{A \in N'} \mathcal V_{(n+2) \cdot |N'|}(A)\right|, \delta\left(\bigvee_{\substack{n' \in \mathbb N: \\ (n+1) \cdot |N'| < n' \leq (n+2) \cdot |N'|}} \changed_{n'}\right)\right) \tag{*}
        \\
        &= \ffint\left(\TRpc \cap \bigcup_{A \in N'} \mathcal V_{(n+2) \cdot |N'|}(A), \bigvee_{\substack{n' \in \mathbb N: \\ (n+1) \cdot |N'| < n' \leq (n+2) \cdot |N'|}} \changed_{n'}\right)
    \end{align*}
    where \enquote{<} is the strict ordering relation induced by~$\leq$, the natural order on~$\mathbb N^2$, and that~(*) holds can be seen as follows.
    First, for every $A \in N'$, $\mathcal V_{(n+2) \cdot |N'|}(A) \supseteq \mathcal V_{(n+1) \cdot |N'|}(A)$ by Lemma~\ref{lem:mcv-monotone} and thus $\TRpc \cap \mathcal V_{(n+2) \cdot |N'|}(A) \supseteq \TRpc \cap \mathcal V_{(n+1) \cdot |N'|}(A)$.
    Then we distinguish two cases:
    \begin{enumerate}
        \item If there is an $A \in N'$ such that $\TRpc \cap \mathcal V_{(n+2) \cdot |N'|}(A) \supset \TRpc \cap \mathcal V_{(n+1) \cdot |N'|}(A)$, then
            \begin{align*}
                \TRpc \cap \bigcup_{A' \in N'} \mathcal V_{(n+2) \cdot |N'|} &= \Big( \TRpc \cap \mathcal V_{(n+2) \cdot |N'|}(A) \Big) \cup \Big( \TRpc \cap \bigcup_{A' \in N' \setminus \{ A \}} \mathcal V_{(n+2) \cdot |N'|} \Big)
                \\
                &\supset \Big( \TRpc \cap \mathcal V_{(n+1) \cdot |N'|}(A) \Big) \cup \Big( \TRpc \cap \bigcup_{A' \in N' \setminus \{ A \}} \mathcal V_{(n+2) \cdot |N'|} \Big)
                \\
                &\supseteq \Big( \TRpc \cap \mathcal V_{(n+1) \cdot |N'|}(A) \Big) \cup \Big( \TRpc \cap \bigcup_{A' \in N' \setminus \{ A \}} \mathcal V_{(n+1) \cdot |N'|} \Big)
                \\
                &= \TRpc \cap \bigcup_{A' \in N'} \mathcal V_{(n+1) \cdot |N'|}
            \end{align*}
            and thus
            \[ |\T_{R'}^{(c)}| - \left|\TRpc \cap \bigcup_{A' \in N'} \mathcal V_{(n+1) \cdot |N'|}\right| > |\T_{R'}^{(c)}| - \left|\TRpc \cap \bigcup_{A' \in N'} \mathcal V_{(n+2) \cdot |N'|}\right| \enspace.  \]
        \item Otherwise for every $A \in N'$ we have that $\TRpc \cap \mathcal V_{(n+2) \cdot |N'|}(A) = \TRpc \cap \mathcal V_{(n+1) \cdot |N'|}(A)$.
            Then for every $n' \in \mathbb N$ with $(n+1) \cdot |N'| \leq n' < (n+2) \cdot |N'|$, by Lemma~\ref{lem:mcv-monotone}
            \[ \TRpc \cap \mathcal V_{n' + 1}(A) \subseteq \TRpc \cap \mathcal V_{n'}(A) \enspace, \]
            thus by Lemma~\ref{lem:po-chains}
            \[ \TRpc \cap \mathcal V_{n' + 1}(A) = \TRpc \cap \mathcal V_{n'}(A) \enspace, \]
            and by Lemma~\ref{lem:mcv-grows-on-change} $V_{n'+1}(A) = V_{n'}(A)$.
            Thus $\bigvee_{n' \in \mathbb N: \, (n+1) \cdot |N'| < n' \leq (n+2) \cdot |N'|} \changed_{n'} = \lfalse$ and $\delta(\ltrue) = 1 > 0 = \delta(\lfalse)$.
    \end{enumerate}
    This proves (*).

    Now we prove termination by contradiction.
    For this, we assume that
    \[
      \bigvee_{n' \in \mathbb N: \, n \cdot |N'| < n' \leq (n+1) \cdot |N'|} \changed_{n'} = \ltrue
    \]
    for every $n \in \mathbb N$.
    We define the set
    \[ I = \left\{ \ffint(\TRpc \cap \bigcup_{A \in N'} \mathcal V_{(n+1) \cdot |N'|}(A), \ltrue) \,\middle|\, n \in \mathbb N \right\}  \subseteq \mathbb N^2 \enspace. \]
    Clearly~$I$ is nonempty.
    Since $(\mathbb N^2,<)$ is well-founded, the set~$I$ has a minimal element.
    Thus there is an $m \in \mathbb N$ such that for each $n \in \mathbb N$,
    \[ \ffint \left( \TRpc \cap \bigcup_{A \in N'} \mathcal V_{(m+1) \cdot |N'|}(A), \ltrue \right) \leq \ffint \left( \TRpc \cap \bigcup_{A \in N'} \mathcal V_{(n+1) \cdot |N'|}(A), \ltrue \right) \]
    Thus, by choosing $n = m+1$, we obtain
    \begin{align*}
        \ffint \left( \TRpc \cap \bigcup_{A \in N'} \mathcal V_{(m+1) \cdot |N'|}(A), \ltrue \right) &\leq \ffint \left( \TRpc \cap \bigcup_{A \in N'} \mathcal V_{(m+2) \cdot |N'|}(A), \ltrue \right) \\
        &< \ffint \left( \TRpc \cap \bigcup_{A \in N'} \mathcal V_{(m+1) \cdot |N'|}(A), \ltrue \right) \tag{by *}
    \end{align*}
    which is a contradiction.
    Thus, there is an $n \in \mathbb N$ such that $\bigvee_{n' \in \mathbb N: \, n \cdot |N'| < n' \leq (n+1) \cdot |N'|} \changed_{n'} = \lfalse$.
    We let $n_0$ be the smallest $n \in \mathbb N$ such that $\bigvee_{n' \in \mathbb N: \, n_0 \cdot |N'| < n' \leq (n_0+1) \cdot |N'|} \changed_{n'} = \lfalse$.
    Then the algorithm terminates after~$n_0$ executions of lines~\ref{l:init-vnew}--\ref{l:update-v}.
\end{proof}

\subsubsection{Correctness of the value computation algorithm}\label{sec:vca-correct}

\begin{lemma}[restate={[name={}]lempassthrough}]\label{lem:value-passed-through}
    For every $n \in \mathbb N$, $d \in \TRpc$ of the form $d = r(d_1, \dots, d_k)$ with $r = \big(A \to \sigma(A_1, \dots, A_k)\big)$, $\welem k_1, \dots, \welem k_k \in \walg K$, and $I \subseteq [k]$ such that
    \begin{enumerate}
        \item for every $i \in [k] \setminus I$, $d_i \in \mathcal V_n(A_i)$ and
        \item for every $i \in I$, $V_n(A_i) = V_n(A_i) \oplus \welem k_i$
    \end{enumerate}
    the following holds:
    if $\select_n = A$, then $V_{n+1}(A) = V_{n+1}(A) \oplus \wt'(r)(\welem l_1, \dots, \welem l_i)$, where
    \[ \welem l_i = \begin{cases}
        \welem k_i &\text{if $i \in I$} \\
        \wtphom{d_i} &\text{otherwise.}
    \end{cases} \]
\end{lemma}

\begin{proof}
    For the proof of Lemma~\ref{lem:value-passed-through}, we refer to Appendix~\ref{app:vca-correctness}.
\end{proof}

\begin{theorem}\label{thm:vca-correct}
    \setlength\emergencystretch{8pt}%
    For every wRTG-LM $\overline G = \big((G', \cfges), \walg K, \wt'\big)$ in $\wlmclass[\clsd]{\gclass{\cfges}, \wclass{\dcomp} \cap \wclass{dist}}$ where $G' = (N', \Sigma', A_0', R')$ the following holds:
    if the value computation algorithm (Algorithm~\ref{alg:mmonoid}) is executed with~$\overline G$ as input, then after termination for every $A \in N'$ it holds that $V(A) = \infsum_{d \in (\T_{R'})_A} \wt'(d)_{\walg K}$.
\end{theorem}

\begin{proof}
    Let $A \in N'$ and $n \in \mathbb N$ such that Algorithm~\ref{alg:mmonoid} terminates after~$n$ executions of lines~\ref{l:init-vnew}--\ref{l:update-v}.
    By Theorem~\ref{thm:tr-trc} we have that $\infsum_{d \in (\T_{R'})_A} \wtphom{d} = \bigoplus_{d \in (\T_{R'}^{(c)})_A} \wtphom{d}$.
    Furthermore
    \begin{align*}
        V_{n+1}(A) &= \bigoplus_{d \in \mathcal V_{n+1}(A)} \wtphom{d}
        \tag{Lemma~\ref{lem:v-is-bigsum}} \\
        &= \bigoplus_{d \in \mathcal V'_{n+1}(A)} \wtphom{d} \enspace,
        \tag{Theorem~\ref{thm:outside-trees-subsumed}}
    \end{align*}
    where $\mathcal V_{n+1}'(A) = \mathcal V_{n+1}(A) \cap (\T_{R'}^{(c)})_A$.
    Hence it remains to show that $V_{n+1}(A) = \bigoplus_{d \in (\T_{R'}^{(c)})_A} \wtphom{d}$, which we do by an indirect proof.

    Assume that $V_{n+1}(A) \not= \bigoplus_{d \in (\T_{R'}^{(c)})_A} \wtphom{d}$.
    Then there is a $d \in (\T_{R'}^{(c)})_A \setminus \mathcal V_{n+1}(A)$ such that
    \begin{equation}
        V_{n+1}(A) \not= V_{n+1}(A) \oplus \wtphom{d} \enspace. \tag{P1}
    \end{equation}
    We choose $A \in N'$ and $d \in (\T_{R'}^{(c)})_A \setminus \mathcal V_{n+1}'(A)$ such that~$d$ is the smallest tree (in height) in $\TRpc$ with this property.
    By Observation~\ref{obs:v-nplus1}, $\{ r \in R' \mid \lhs(r) = A \ \text{and} \ \rk(r) = 0 \} \subseteq \mathcal V_{n+1}(A)$, hence $\height(d) \geq 1$.
    We let $d = r(d_1, \dots, d_k)$ and $\reqrule$ with $k > 0$.
    Since~$d$ is the smallest tree with property~(P1), it cannot be the case for any $i \in [k]$ that $d_i \in (\T_{R'})_{A_i}^{(c)} \setminus \mathcal V_{n+1}'(A_i)$ and $V_{n+1}(A_i) \not= V_{n+1}(A_i) \oplus \wtphom{d_i}$.
    Hence for every $i \in [k]$, either $d_i \in \mathcal V_{n+1}'(A_i)$ or $V_{n+1}(A_i) = V_{n+1}(A_i) \oplus \wtphom{d_i}$.
    Now we distinguish two cases.
    \begin{enumerate}
        \item If $d_i \in \mathcal V_{n_A}(A_i)$ for every $i \in [k]$, where $n_A \in \mathbb N$ is the greatest number such that $\select_{n_A} = A$, then by Observation~\ref{obs:v-nplus1}, $d \in \mathcal V_{n_A + 1}(A)$.
            We note that $n_A \le n$.
            Thus by Lemma~\ref{lem:mcv-monotone}, $d \in \mathcal V_{n+1}(A)$, which contradicts the definition of~$d$.
        \item Otherwise, let $n_A \in \mathbb N$ the greatest number such that $\select_{n_A} = A$.
            Then for every $i \in [k]$, $V_{n+1}(A_i) = V_{n+1}(A_i) \oplus \wtphom{d_i}$, $d_i \in \mathcal V_{n_A}(A_i)$, or there is an $n' \in \mathbb N$ with $n_A < n' \le n + 1$ such that $d_i \in \mathcal V_{n'}(A_i)$.
            For every $i \in [k]$ such that only the latter holds, since the algorithm terminates after~$n$ executions of lines~\ref{l:init-vnew}--\ref{l:update-v}, we have that $V_{n_i}(A_i) = V_{n_i + 1}(A_i)$, where $n_i \in \mathbb N$ is the greatest number such that $\select_{n_i} = A_i$.
            Then by Lemma~\ref{lem:v-is-bigsum}
            \[ V_{n_i}(A_i) = V_{n_i}(A_i) \oplus \bigoplus_{d' \in \Delta_{n_i}(A_i)} \wtphom{d'} \enspace. \]
            Thus, by Lemma~\ref{lem:natord-subsume} (which is applicable due to Lemma~\ref{lem:d-complete-natord}), $V_{n_i}(A_i) = V_{n_i}(A_i) \oplus \wtphom{d_i}$, and by Observation~\ref{obs:v-nplus1}, $V_{n_A}(A_i) = V_{n_A}(A_i) \oplus \wtphom{d_i}$.
            By termination after~$n$ iterations of the inner for loop and Observation~\ref{obs:v-nplus1}, $V_{n+1}(A) = V_{n_A}(A)$ for every $A \in N'$.
            We let $I = \{ i \in [k] \mid V_{n_A}(A_i) = V_{n_A}(A_i) \oplus \wtphom{d_i} \}$.
            Then by Lemma~\ref{lem:value-passed-through},
            \begin{align*}
                V_{n_A + 1}(A) &= V_{n_A + 1}(A) \oplus \wt'(r)\left(\wtphom{d_1}, \dots, \wtphom{d_k}\right) \\
                &= V_{n_A + 1}(A) \oplus \wtphom{d} \enspace.
            \end{align*}
            Thus $V_{n+1}(A) = V_{n+1}(A) \oplus \wtphom{d}$, which contradicts the definition of~$d$. \qedhere
    \end{enumerate}
\end{proof}

\begin{boxcorollary}\label{cor:vca-correct}
    The value computation algorithm (Algorithm~\ref{alg:mmonoid}) is correct for the class $\wlmclass[\clsd]{\gclass{\cfges}, \wclass{\dcomp} \cap \wclass{dist}}$.
\end{boxcorollary}

\subsection{Properties of the M-monoid parsing algorithm}

\index{M-monoid parsing algorithm!correct}
We say that the M-monoid parsing algorithm is \emph{correct} for some class $\wlmclass{}$ of wRTG-LMs if it computes $\fparse(a)$ for every wRTG-LM in $\wlmclass{}$ and syntactic object $a$.
We want to show that the M-monoid parsing algorithm is correct for every wRTG-LM with finitely decomposable language algebra which is closed or nonlooping.

\begin{lemma}[restate={[name={}]lemclosedpreserved}]\label{lem:1and2-closed}
    For every wRTG-LM $\overline G$ with finitely decomposable language algebra and syntactic object~$a$, the wRTG-LM $\cnc(\overline G, a)$ is closed if
    \begin{itemize}
        \item $\overline G$ is closed or
        \item $\overline G$ is nonlooping and the weight algebra of $\overline G$ is in $\wclass{\dcomp} \cap \wclass{dist}$.
    \end{itemize}
\end{lemma}

\begin{proof}
    For the proof of Lemma \ref{lem:1and2-closed}, we refer to Appendix~\ref{app:mpa-properties}.
\end{proof}

\begin{boxtheorem}\label{thm:terminating-correct}
    The M-monoid parsing algorithm is terminating and correct for every closed wRTG-LM with finitely decomposable language algebra and for every nonlooping wRTG-LM with finitely decomposable language algebra and weight algebra in $\wclass{\dcomp} \cap \wclass{dist}$.
\end{boxtheorem}

\begin{proof}
    The M-monoid parsing algorithm terminates because (a) the computation of $\cnc$ is terminating for every wRTG-LM with finitely decomposable language algebra and (b) Algorithm~\ref{alg:mmonoid} is terminating by Theorem~\ref{thm:vca-terminating}, which we can be applied due to Lemma~\ref{lem:1and2-closed}.
    The M-monoid parsing algorithm is correct because (a) $\cnc$ is weight-preserving (Observation~\ref{obs:weight-preserving-parse} and Lemma~\ref{lem:cnc-weight-preserving}) and (b) Algorithm~\ref{alg:mmonoid} is correct by Theorem~\ref{thm:vca-correct} (which is applicable again due to Lemma~\ref{lem:1and2-closed}), hence
    \[
        \displaystyle\fparse(a) \overset{\text{(a)}}{=} \infsum_{d \in \AST(G')} \wt'(d)_{\walg K} \overset{\text{(b)}}{=} V(A_0') \enspace. \qedhere
    \]
\end{proof}

\section{Application scenarios}

In this section we investigate the applicability of the value computation algorithm (Algorithm~\ref{alg:mmonoid}) and of the M-monoid parsing algorithm.
\index{value computation algorithm!applicable}%
\index{M-monoid parsing algorithm!applicable}%
We say that an algorithm is \emph{applicable} to a class of wRTG-LMs if it is terminating and correct for every wRTG-LM in that class.
We compare the variety of classes of wRTG-LMs to which our algorithms are applicable to that of similar algorithms from the literature.
In the end we informally discuss complexity results of our algorithms.

\subsection{Value computation algorithm}

By Theorem~\ref{thm:vca-terminating} and Corollary~\ref{cor:vca-correct}, the value computation algorithm (Algorithm~\ref{alg:mmonoid}) is applicable to every closed wRTG-LM with language algebra $\cfges$, i.e., to every wRTG-LM in the class $\wlmclass[\clsd]{\gclass{\cfges},$ $\wclass{\dcomp} \cap \wclass{dist}}$.
We start by identifying some classes of closed wRTG-LMs.

\begin{boxtheorem}[restate={[name={}]thmapplications}]\label{thm:applications}
    Each wRTG-LM in each of the following three classes is closed:
    $\wlmclass{\gclass{all}, \wclass{\finidpo}}$,
    $\wlmclass{\gclass{all}, \wclass{sup}}$, and
    $\wlmclass{\gclass{acyc}, \wclass{\dcomp} \cap \wclass{dist}}$.
\end{boxtheorem}

\begin{proof}
    For the proof, we refer to Appendix~\ref{app:applications}.
\end{proof}

By Theorem~\ref{thm:vca-terminating}, Corollary~\ref{cor:vca-correct} and Theorem~\ref{thm:applications}, the value computation algorithm (Algorithm~\ref{alg:mmonoid}) is applicable to each class of wRTG-LMs that is mentioned in Theorem~\ref{thm:applications} if we restrict their language algebra to $\cfges$.
We note that $\wclass{\finidpo} \subseteq \wclass{\dcomp} \cap \wclass{dist}$ and $\wclass{sup} \subseteq \wclass{\dcomp} \cap \wclass{dist}$ by their definition and Lemma~\ref{lem:inf-idp-d-complete}.

Now we compare the applicability of the value computation algorithm (Algorithm~\ref{alg:mmonoid}) to the applicabilities of
\begin{enumerate*}[label=(\alph*)]
    \item the second phase of the semiring parsing algorithm \cite[Figure 10]{Goodman1999},
    \item Knuth's algorithm \cite[Section 3]{Knuth1977}, and
    \item Mohri's algorithm \cite[Figure 2]{Mohri2002}.
\end{enumerate*}
In order to have a basis for a fair comparison, we understand the inputs of these algorithms as particular wRTG-LMs of the form $\big((G', \cfges), (\walg K,\oplus,\welem{0},\Omega,\infsum), \wt'\big)$ with $G' = (N', \Sigma', A_0', R')$.
An algorithm is correct for such a wRTG-LM if it computes $\infsum_{d \in \AST(G')} \wt'(d)_{\walg{K}}$.
Thus the notion applicable is the same for the value computation algorithm and the other ones.

\begin{table*}[h]
    \centering
    {\small
    \begin{tabular}{lll}
        \toprule
        Algorithm & Class of inputs & Comment \\
        \midrule
        (a) Goodman & $\wlmclass{\gclass{\cfges} \cap \gclass{acyc}, \wclass{sr}}$ & acyclic RTGs and complete semirings \\[1ex]
        (b) Knuth & $\wlmclass{\gclass{\cfges}, \wclass{sup}}$ & superior M-monoids \\[1ex]
        (c) Mohri & $\wlmclass[\clsd]{\gclass{\cfges} \cap \gclass{mon}, \wclass{com.\;sr}}$ & monadic RTGs and commutative semirings \\[1ex]
        (d) Algorithm~\ref{alg:mmonoid} & $\wlmclass[\clsd]{\gclass{\cfges}, \wclass{\dcomp} \cap \wclass{dist}}$ & d-complete and distributive M-monoids \\
        \bottomrule
    \end{tabular}
    }
    \caption{
        Comparison of the value computation algorithm (d) to three similar algorithms. The second column represents the class of wRTG-LMs to which the corresponding algorithm is applicable.}
    \label{tab:comparison2}
\end{table*}

Table~\ref{tab:comparison2} shows for each algorithm the class of inputs to which it is applicable.
The algorithms (c) and (d) are applicable to closed wRTG-LMs; moreover, (c) is applicable to a proper subset of the inputs of (d).
Each input to which (b) is applicable is in the class $\wlmclass{\gclass{\cfges}, \wclass{sup}}$ and thus, due to Theorem~\ref{thm:applications}, closed.
The inputs to which (a) is applicable constitute the class $\wlmclass{\gclass{\cfges} \cap \gclass{acyc}, \wclass{sr}}$.
However, according to Theorem~\ref{thm:applications}, only the class $\wlmclass{\gclass{acyc}, \wclass{\dcomp} \cap \wclass{dist}}$ is closed and $\wclass{sr} \not\subseteq \wclass{\dcomp}$ in general.
We remark that for every wRTG-LM in $\wlmclass{\gclass{\cfges} \cap \gclass{acyc}, \wclass{sr}}$ the set of ASTs is finite.
Hence only finite sums have to be computed and the restriction to d-complete M-monoids is not needed; thus the value computation algorithm is applicable to $\wlmclass{\gclass{\cfges} \cap \gclass{acyc}, \wclass{sr}}$.
In summary, if one of the value computation algorithms (a)--(c) is applicable, then Algorithm~\ref{alg:mmonoid} is applicable too.

\begin{sloppypar}
We conclude the investigation of the value computation algorithm by considering three classes of wRTG-LMs whose weight algebra is a particular M-monoid:
$\wlmclass{\gclass{all}, \walg{BD}}$,
$\wlmclass{\gclass{all}, \nbest}$, and
$\wlmclass{\gclass{all}, \wclass{int}}$.
It turns out that not every wRTG-LM in $\wlmclass{\gclass{all}, \walg{BD}}$ is closed (for an example, cf.\ Appendix~\ref{app:bd-restriction}).
Hence we first impose an additional restriction on this particular class.
\end{sloppypar}

We let \gls{wlmclass:bd} be the class of all wRTG-LMs $\overline G = \big((G, \alg L), \walg{BD}, \wt\big)$ in $\wlmclass{\gclass{all}, \walg{BD}}$ with $G = (N, \Sigma, A_0, R)$ such that for every $r \in R$, $\wt(r) = \tc{p}{r}$ with $p < 1$.
We remark that the condition $p < 1$ is sufficient to ensure that each wRTG-LM in $\wlmclass[<1]{\gclass{all}, \walg{BD}}$ is closed, but not necessary.
There may be weaker sufficient conditions which are more difficult to express, though.

\begin{boxtheorem}[restate={[name={}]thmapplicationsp}]\label{thm:applications2}
    Each wRTG-LM in each of the following three classes is closed:
    $\wlmclass[<1]{\gclass{all}, \walg{BD}}$,
    $\wlmclass{\gclass{all}, \nbest}$, and
    $\wlmclass{\gclass{all}, \wclass{int}}$.
\end{boxtheorem}

\begin{proof}
    For the proof, we refer to Appendix~\ref{app:applications}.
\end{proof}

By Theorem~\ref{thm:vca-terminating}, Corollary~\ref{cor:vca-correct} and Theorem~\ref{thm:applications2}, the value computation algorithm (Algorithm~\ref{alg:mmonoid}) is applicable to each class of wRTG-LMs that is mentioned in Theorem~\ref{thm:applications2} if we restrict their language algebra to $\cfges$.
We recall our comparison of algorithms in Table~\ref{tab:comparison2}.
Neither of algorithms (a)--(c) is in general applicable to any of the wRTG-LMs of Theorem~\ref{thm:applications2}, but Algorithm~\ref{alg:mmonoid} is applicable to each of them.

\subsection{M-monoid parsing algorithm}

By Theorem~\ref{thm:terminating-correct}, the M-monoid parsing algorithm is applicable to each class of wRTG-LMs that is mentioned in Theorem~\ref{thm:applications} or Theorem~\ref{thm:applications2} if we restrict them to finitely decomposable language algebras.

We continue to discuss two classes of nonlooping wRTG-LMs, each of which represents a particular parsing problem.
First we consider the class $\wlmclass{\gclass{nl} \cap \gclass{\findc}, \wclass{sr}}$.
It contains exactly those wRTG-LMs for which Goodman's algorithm can solve the semiring parsing problem.
By Theorem~\ref{thm:terminating-correct}, the M-monoid parsing algorithm is applicable to each wRTG-LM in the class $\wlmclass{\gclass{nl} \cap \gclass{\findc}, \wclass{sr} \cap \wclass{\dcomp}}$.
By the same argument as in the previous subsection we may extend this result to the class $\wlmclass{\gclass{nl} \cap \gclass{\findc}, \wclass{sr}}$.

Second we consider the class $\wlmclass{\gclass{YIELD} \cap \gclass{nl}, \wclass{ADP}}$ of wRTG-LMs.
It contains all those wRTG-LMs which are specifications of ADP problems.
Clearly $\gclass{YIELD} \subseteq \gclass{\findc}$ and by Lemma~\ref{lem:adp-mmonoid-complete-distributive}, $\wclass{ADP} \subseteq \wclass{dist} \cap \wclass{\dcomp}$.
Thus, by Theorem~\ref{thm:terminating-correct}, the M-monoid parsing algorithm is applicable to each wRTG-LM in $\wlmclass{\gclass{YIELD} \cap \gclass{nl}, \wclass{ADP}}$.

In the end, we come to a more general view on nonlooping wRTG-LMs.
By Theorem~\ref{thm:terminating-correct}, the M-monoid parsing algorithm is terminating and correct for every wRTG-LM whose language model is in $\gclass{nl} \cap \gclass{\findc}$ if its weight algebra is in $\wclass{dist} \cap \wclass{\dcomp}$.
Thus the M-monoid parsing algorithm is applicable to the rather general class $\wlmclass{\gclass{nl} \cap \gclass{\findc}, \wclass{dist} \cap \wclass{\dcomp}}$ of wRTG-LMs.

The following statement summarizes the findings of this section.

\begin{boxcorollary}\label{cor:applicability}
    The M-monoid parsing algorithm is applicable to the following classes of wRTG-LMs.
    \begin{enumerate}[label=(\arabic*)]
        \item $\wlmclass{\gclass{nl} \cap \gclass{\findc}, \wclass{sr}}$ -- this includes every input for which Goodman's semiring parsing algorithm terminates and is correct.
        \item $\wlmclass{\gclass{\findc}, \wclass{sup}}$ -- this includes every input of Nederhof's weighted deductive parsing algorithm.
        \item $\wlmclass[<1]{\gclass{\findc}, \walg{BD}}$.
        \item $\wlmclass{\gclass{\findc}, \nbest}$.
        \item $\wlmclass{\gclass{\findc}, \wclass{int}}$ -- thus the M-monoid parsing algorithm can compute the intersection of a grammar and a syntactic object.
        \item $\wlmclass{\gclass{YIELD} \cap \gclass{nl}, \wclass{ADP}}$ -- thus the M-monoid parsing algorithm can solve every ADP problem.
    \end{enumerate}
\end{boxcorollary}

Like the M-monoid parsing algorithm, the semiring parsing algorithm~\cite{Goodman1999} and the weighted deductive parsing algorithm~\cite{ned03} are only applicable if the language algebra of their input is finitely decomposable.
This is because they use a weighted deduction system in the first phase of their pipeline, too.
By (1) and (2) of Corollary~\ref{cor:applicability}, our approach subsumes semiring parsing and weighted deductive parsing.
The classes of (3) and (4) are essentially instances of the semiring parsing problem to which the M-monoid parsing algorithm is applicable even if the RTG-LM is \emph{looping} (i.e., not nonlooping).
Moreover, their weight algebras are not superior (in which case the weighted deductive parsing algorithm would be applicable).
Likewise (5) and (6) are in general outside the scope of both semiring parsing and weighted deductive parsing.

\subsection{Complexity}

We only discuss the complexity of the second phase of the M-monoid parsing algorithm, i.e., the value computation algorithm (Algorithm~\ref{alg:mmonoid}) because the first phase (canonical weighted deduction system) is executed on demand.
Thus the value computation algorithm is the main determinant of complexity and the canonical weighted deduction system only adds a factor which depends on the language algebra of the input.
Since the weighted parsing algorithms of~\cites{Goodman1999}{ned03} are two-phase pipelines that use a weighted deduction system in their first phase as well, we believe that abstracting from the first phase yields the most significant statement on complexity.

Now we compare the complexity of the value computation algorithm to the complexity of the algorithms of \textcite{Mohri2002}, \textcite{Knuth1977} and the second phase of \textcite{Goodman1999}.
We do this by restricting the inputs of the value computation algorithm to the input scenarios of the other algorithms.
Since there is no complexity bound on the operations in the weight algebra of a wRTG-LM (they can even be undecidable), it is not possible to give a general statement about the complexity of any of the considered algorithms.
Hence we abstract from the costs of these operations.

Mohri's algorithm is applicable to every wRTG-LM in $\wlmclass[\clsd]{\gclass{\cfges} \cap \gclass{mon}, \wclass{com.\;sr}}$.
Its complexity is polynomial in the maximal number  $n_{\max}$ of times the value of a nonterminal changes.
Our value computation algorithm has the same complexity if we restrict its inputs to $\wlmclass[\clsd]{\gclass{\cfges} \cap \gclass{mon}, \wclass{com.\;sr}}$.
We remark that $n_{\max}$ is in general not polynomial in the size of the input wRTG-LM\@.
Mohri circumvents this problem by specifying the order in which nonterminals are processed for well-known classes of inputs, e.g., acyclic graphs or superior weight algebras.
We can adapt this idea by imposing such an ordering on the iteration over the nonterminals in line~\ref{l:for-loop}.

Knuth's algorithm is applicable to every wRTG-LM in $\wlmclass{\gclass{\cfges}, \wclass{sup}}$.
Its complexity is in $O\big(|N'| \cdot (|N'| + |R'|)\big)$.
Our value computation algorithm has the same complexity if we restrict its inputs to $\wlmclass{\gclass{\cfges}, \wclass{sup}}$ (assuming that $|N'| \le |R'|$, which is usually the case).
This is because in every iteration of the repeat-until loop (lines~\ref{l:loop}--\ref{l:loop-condition}), at least one nonterminal is assigned a weight which stays the same across all future iterations (since the M-monoid is superior).

The second phase of Goodman's semiring parsing algorithm is applicable to every wRTG-LM in $\wlmclass{\gclass{acyc} \cap \gclass{\cfges}, \wclass{sr}}$.
It processes a topological ordering of its input and thus achieves a complexity in $O(|R'|)$.
If we restrict the inputs of our value computation algorithm to $\wlmclass{\gclass{acyc} \cap \gclass{\cfges}, \wclass{sr}}$, then its complexity is worse.
We can, however, use the topological ordering of the input in line~\ref{l:for-loop} of the value computation algorithm; then we achieve the same complexity as Goodman.
(Since Goodman requires this ordering to be precomputed, we take the liberty of doing so as well.)
Indeed, Mohri suggests to process acyclic graphs in topological order, too.

Finally, we note that, although our value computation algorithm (Algorithm~\ref{alg:mmonoid}) -- when restricted to the respective inputs -- has the same complexity as the other algorithms, in average performs more computations than those.
This is because in each iteration of lines~\ref{l:for-loop}--\ref{l:update-v}, the values of all nonterminals are recomputed.
In particular, in the final iteration of the repeat-until loop (lines~\ref{l:loop}--\ref{l:loop-condition}), the value of every nonterminal is unchanged.
We could avoid superfluous computations by using a direct generalization of Mohri's algorithm to the branching case rather than Algorithm~\ref{alg:mmonoid}.
However, the intricacies of such a generalization would exceed the scope of this paper.

\clearpage

\printglossaries
\printindex
\begingroup\setlength\emergencystretch{20pt}
\printbibliography
\endgroup

\appendix

\newpage
\section{Additional proofs}

In this appendix we have placed full proofs of some of the lemmas and claims of the previous sections.

\subsection{Proofs of statements from the preliminaries}\label{sec:proofs-preliminaries}

We recall the definition of the natural order on pairs of natural numbers.
For every $(a_1,b_1), (a_2,b_2) \in \mathbb N^2$ we have that $(a_1,b_1) \leq (a_2,b_2)$ if one of the following holds:
\begin{enumerate}
    \item $a_1 < a_2$, or
    \item $a_1 = a_2$ and $b_1 \leq b_2$.
\end{enumerate}

\begin{lemma}
    $(\mathbb N^2, \leq)$ is a well-order.
\end{lemma}

\begin{proof}
    We start with proving that $(\mathbb N^2, \leq)$ is a total order.
    Reflexivity follows directly from~(ii).
    Now let $(a_1,b_1), (a_2,b_2), (a_3,b_3) \in \mathbb N^2$.
    For transitivity let $(a_1,b_1) \leq (a_2,b_2)$ and $(a_2,b_2) \leq (a_3,b_3)$; there are four cases in which this may hold:
    \begin{enumerate}
        \item if $a_1 < a_2$ and $a_2 < a_3$, then by transitivity of $(\mathbb N,<)$ also $a_1 < a_3$, hence $(a_1,b_1) \leq (a_3,b_3)$,
        \item if $a_1 < a_2$ and $a_2 = a_3$, then clearly $a_1 < a_3$ and hence $(a_1,b_1) \leq (a_3,b_3)$,
        \item if $a_1 = a_2$ and $a_2 < a_3$, then clearly $a_1 < a_3$ and hence $(a_1,b_1) \leq (a_3,b_3)$, and
        \item if $a_1 = a_2$ and $a_2 = a_3$, then $a_1 = a_3$ and furthermore, we have that $b_1 \leq b_2$ and $b_2 \leq b_3$ because otherwise $(a_1,b_1) \leq (a_2,b_2)$ and $(a_2,b_2) \leq (a_3,b_3)$ would not hold.
            Now we obtain from the transitivity of $(\mathbb N,\leq)$ that $b_1 \leq b_3$ and thus $(a_1,b_1) \leq (a_3,b_3)$.
    \end{enumerate}
    For antisymmetry, let $(a_1,b_1) \leq (a_2,b_2)$ and $(a_2,b_2) \leq (a_1,b_1)$.
    Then we have to distinguish two cases:
    \begin{enumerate}
        \item if $a_1 = a_2$, then $b_1 \leq b_2$ and $b_2 \leq b_1$, thus by antisymmetry of $(\mathbb N,\leq)$ we obtain $b_1 = b_2$, and
        \item otherwise, we have that $a_1 < a_2$ and $a_2 < a_1$ which is equivalent to $a_1 \leq a_2 \land a_2 \leq a_1 \land a_1 \not= a_2$.
            By antisymmetry of $(\mathbb N,\leq)$, $a_1 \leq a_2 \land a_2 \leq a_1$ implies $a_1 = a_2$, but this is a contradiction.
    \end{enumerate}
    For totality there are two cases as well:
    \begin{enumerate}
        \item if $a_1 = a_2$, then by totality of $(\mathbb N,\leq)$ either $b_1 \leq b_2$ or $b_2 \leq b_1$, thus either $(a_1,b_1) \leq (a_2,b_2)$ or $(a_2,b_2) \leq (a_1,b_1)$, respectively, and
        \item otherwise, by totality of $(\mathbb N,<)$ either $a_1 < a_2$ or $a_2 < a_1$, thus either $(a_1,b_1) \leq (a_2,b_2)$ or $(a_2,b_2) \leq (a_1,b_1)$, respectively.
    \end{enumerate}
    Since $(\mathbb N^2,\leq)$ is reflexive, transitive, antisymmetric, and total, we conclude that it is a total order.

    It remains to show that $(\mathbb N^2,<)$, where~$<$ is the strict total ordering induced by~$\leq$, is well-founded.
    For this, let $I \subseteq \mathbb N^2$ such that $I \not= \emptyset$.
    We define $I_1 = \{ a \mid (a,b) \in I \}$.
    Clearly $I_1 \subseteq \mathbb N$ and $I_1 \not= \emptyset$, thus by well-foundedness of $(\mathbb N,<)$ there is an $a' \in I_1$ such that $a \not< a'$ for every $a \in I_1$.
    Now we define $I_{1,2} = \{ b \mid (a,b) \in I_1 \}$ and by the same argumentation obtain that there is a $b' \in I_{1,2}$ such that $b \not< b'$ for every $b \in I_{1,2}$.
    By definition of~$I_1$ and~$I_{1,2}$, $(a',b') \in I$ and by definition of~$<$, $(a,b) \not< (a',b')$ for every $(a,b) \in I$.
    Thus $(\mathbb N^2,<)$ is well-founded.
\end{proof}

\lempochains*

\begin{proof}
    Let $(A, \preceq)$ be a partial order, $n \in \mathbb N$, and $a_1, \dots, a_n \in A$ such that $a_1 \preceq \dots \preceq a_n$ and $a_1 = a_n$
    We show that $a_1 = \dots = a_n$ by contradiction.
    For this, assume that there are $i, j \in [n]$ such that $a_i \not= a_j$.
    Without loss of generality, we assume that $i < j$.
    Then, by transitivity of $\preceq$, $a_i \preceq a_j$.
    Since $a_i \not= a_j$, we have that $a_i \prec a_j$.
    Then, by transitivity of $\prec$, $a_1 \prec a_j$ and thus $a_1 \prec a_n$.
    This contradicts the fact that $a_1 = a_n$.
\end{proof}

\lemheightfinite*

\begin{proof}
    The proof is done by induction on~$h$.
    For the induction base let $h = 0$.
    Then
    \[ |\{ t \in \T_\Sigma \mid \height(t) \leq 0 \}| = |\{ \sigma \in \Sigma \mid \rk(\sigma) = 0 \}| \leq |\Sigma| \enspace. \]
    For the induction step let $h \in \mathbb N$.
    We assume (IH) that for every $h' \leq h$ it holds that $|\{ t \in \T_\Sigma \mid \height(t) \leq h' \}| \leq |\Sigma|^{(\sum_{i=0}^{h'} k^i)}$.
    Then
    \begin{align*}
        &|\{ t \in \T_\Sigma \mid \height(t) \leq h + 1 \}| \\
        &\leq |\{ \sigma(t_1,\dots,t_k) \mid \sigma \in \Sigma \ \text{and for every $i \in [k]$: $t_i \in \T_\Sigma$ and $\height(t_i) \leq h$} \}| \\
        &= |\Sigma| \cdot |\{ (t_1,\dots,t_k) \mid \text{for every $i \in [k]$: $t_i \in \T_\Sigma$ and $\height(t_i) \leq h$} \}| \\
        &= |\Sigma| \cdot |\{t \in \T_\Sigma \mid \height(t) \leq h\}|^k \\
        &\leq |\Sigma| \cdot (|\Sigma|^{(\sum_{i=0}^h k^i)})^k \tag{IH} \\
        &= |\Sigma| \cdot |\Sigma|^{(\sum_{i=1}^{h+1} k^i)} \\
        &= |\Sigma|^{(\sum_{i=0}^{h+1} k^i)} \enspace. \qedhere
    \end{align*}
\end{proof}

\lemnatordrt*

\begin{proof}
    Let $(\walg K, \oplus, \welem 0)$ be a monoid.
    For reflexivity of~$\preceq$, let $\welem k \in \walg K$.
    Since~$\welem 0$ is the identity element, we have that $\welem k \oplus \welem 0 = \welem k$.
    Thus $\welem k \preceq \welem k$.
    For transitivity of~$\preceq$, let $\welem k_1, \welem k_2, \welem k_3 \in \walg K$ such that $\welem k_1 \preceq \welem k_2$ and $\welem k_2 \preceq \welem k_3$.
    Then there are $\welem k, \welem k' \in \walg K$ such that $\welem k_1 \oplus \welem k = \welem k_2$ and $\welem k_2 \oplus \welem k' = \welem k_3$.
    Thus $\welem k_1 \oplus \welem k \oplus \welem k' = \welem k_2 \oplus \welem k' = \welem k_3$ and hence $\welem k_1 \preceq \welem k_3$.
\end{proof}

\lemnatord*

\begin{proof}
    Let $(\walg K, \oplus, \mathbb 0)$ be a monoid.
    First, we assume that~$\walg K$ is naturally ordered.
    Let $\welem k_1, \welem k_2, \welem k_3 \in \walg K$ with $\welem k_1 = \welem k_1 \oplus \welem k_2 \oplus \welem k_3$.
    Then by definition of~$\preceq$, $\welem k_1 \oplus \welem k_2 \preceq \welem k_1 \oplus \welem k_2 \oplus \welem k_3$ and $\welem k_1 = \welem k_1 \oplus \welem k_2 \oplus \welem k_3 \preceq \welem k_1 \oplus \welem k_2$.
    Thus, as~$\preceq$ is antisymmetric, $\welem k_1 = \welem k_1 \oplus \welem k_2$.

    Second, we assume that for every $\welem k_1, \welem k_2, \welem k_3 \in \walg K$ with $\welem k_1 = \welem k_1 \oplus \welem k_2 \oplus \welem k_3$ it holds that $\welem k_1 = \welem k_1 \oplus \welem k_2$.
    By Lemma~\ref{lem:natord-refl-trans}, $\preceq$ is reflexive and transitive.
    For antisymmetry, let $\welem k_1, \welem k_2 \in \walg K$ such that $\welem k_1 \preceq \welem k_2$ and $\welem k_2 \preceq \welem k_1$.
    Then there are $\welem k, \welem k'$ such that $\welem k_1 \oplus \welem k = \welem k_2$ and $\welem k_2 \oplus \welem k' = \welem k_1$.
    Hence
    \begin{align*}
        \welem k_1 &= \welem k_2 \oplus \welem k' = \welem k_1 \oplus \welem k \oplus \welem k' \\
        &= \welem k_1 \oplus \welem k \tag{by assumption} \\
        &= \welem k_2 \enspace.
    \end{align*}
    Thus~$\preceq$ is antisymmetric and~$\walg K$ is naturally ordered.
\end{proof}

\lemdcompnatord*

\begin{proof}
    Let $(\walg K, \oplus, \mathbb 0, \infsum)$ be a d-complete monoid.
    Furthermore, we let $\welem k, \welem l_1, \welem l_2 \in \walg K$ such that $\welem k \oplus \welem l_1 \oplus \welem l_2 = \welem k$.
    We define the family $(\welem k_i \mid i \in \mathbb N)$ of elements of~$\welem K$ such that for every $i \in \mathbb N$, $\welem k_i = \welem l_1 \oplus \welem l_2$.
    Then, for every $i \in \mathbb N$, $\welem k \oplus \welem k_i = \welem k$.
    Thus, by Lemma~\ref{lem:d-complete}~(ii),
    \begin{equation}\label{eq:lem-dcn}
        \welem k \oplus \infsum_{i \in \mathbb N} \welem k_i = \welem k \enspace.
    \end{equation}
    Hence
    \begin{align*}
        \welem k &= \welem k \oplus \infsum_{i \in \mathbb N} \welem k_i \tag{Eq.\,\ref{eq:lem-dcn}} \\
        =& \welem k \oplus \infsum_{i \in \mathbb N} (\welem l_1 \oplus \welem l_2)
        = \welem k \oplus \infsum_{i \in \mathbb N} \welem l_1 \oplus \infsum_{i \in \mathbb N} \welem l_2
        = \welem k \oplus \welem l_1 \oplus \infsum_{i \in \mathbb N \setminus \{ 0 \}} \welem l_1 \oplus \infsum_{i \in \mathbb N} \welem l_2 \\
        =& \welem k \oplus \welem l_1 \oplus \infsum_{i \in \mathbb N} \welem l_1 \oplus \infsum_{i \in \mathbb N} \welem l_2
        = \welem k \oplus \welem l_1 \oplus \infsum_{i \in \mathbb N} (\welem l_1 \oplus \welem l_2)
        = \welem k \oplus \welem l_1 \oplus \infsum_{i \in \mathbb N} \welem k_i
        = \welem k \oplus \infsum_{i \in \mathbb N} \welem k_i \oplus \welem l_1 \\
        =& \welem k \oplus \welem l_1 \enspace. \tag{Eq.\,\ref{eq:lem-dcn}}
    \end{align*}
    Thus, by Lemma~\ref{lem:natord-subsume}, $\walg K$ is naturally ordered.
\end{proof}

\lemiidc*

\begin{proof}
    Let $(\walg K, \oplus, \mathbb 0, \infsum)$ be a completely idempotent monoid.
    Furthermore, we let $\welem k \in \walg K$ and $(\welem k_i \mid i \in \mathbb N)$ be a family of elements of~$\walg K$ such that $\welem k \oplus \welem k_i = \welem k$ for every $i \in \mathbb N$.
    Then
    \[ \welem k \oplus \infsum_{i \in \mathbb N} \welem k_i = \infsum_{i \in \mathbb N} \welem k \oplus \infsum_{i \in \mathbb N} \welem k_i = \infsum_{i \in \mathbb N} (\welem k \oplus \welem k_i) = \infsum_{i \in \mathbb N} \welem k = \welem k \enspace. \]
    Thus, by Lemma~\ref{lem:d-complete}~(ii), $\walg K$ is d-complete.
\end{proof}

\subsection{Superior M-monoids}\label{sec:proof-superior-mmonoids}

\begin{lemma*}[cf.\ Section~\ref{sec:superior-mmonoids}]
    The tropical M-monoid $(\mathbb{R}_0^\infty,\min,\infty,\Omega_+,\inf)$ is superior.
\end{lemma*}

\begin{proof}
    We show that for every $c \in \mathbb{R}_0^1$ and $k \in \mathbb N$, $\mul^{(k)}_c \in \Omega_+$ is $\le$-superior.
    Let $k,i \in \mathbb N$ and $a,a_1,\dots,a_k,c \in \mathbb R_0^\infty$.
    \begin{enumerate}
        \item Assume that $a \le a_i$.
            Then
            \begin{align*}
                \mul_c^{(k)}(a_1,\dots,a_{i-1},a,a_{i+1},\dots,a_k) &= c + a_1 + \ldots + a_{i-1} + a + a_{i+1} + \ldots + a_k \\
                &\le c + a_1 + \ldots + a_{i-1} + a_i + a_{i+1} + \ldots + a_k \tag{*} \\
                &= \mul_c^{(k)}(a_1,\dots,a_{i-1},a_i,a_{i+1},\dots,a_k) \enspace,
            \end{align*}
            where (*) holds because
            \[ c + a_1 + \ldots + a_{i-1} + a + a_{i+1} + \ldots + a_k \le c + a_1 + \ldots + a_{i-1} + a_i + a_{i+1} + \ldots + a_k \]
            due to the monotonicity of~$+$.
        \item Observe that $a' \le a' + b'$ for every $a',b' \in \mathbb R_0^\infty$.
            Hence $\max \{ c,a_1,\dots,a_k \} \le c + a_1 + \ldots + a_k$.
            Then
            \[
                \max \{ a_1,\dots,a_k \} \le \max \{ c,a_1,\dots,a_k \}
                \le c + a_1 + \ldots + a_k
                = \mul_c^{(k)}(a_1,\dots,a_k) \enspace. \qedhere
            \]
    \end{enumerate}
\end{proof}

\begin{lemma*}[cf.\ Section~\ref{sec:superior-mmonoids}]
    The Viterbi M-monoid $(\mathbb{R}_0^1,\max,0,\Omega_\cdot,\sup)$ is superior.
\end{lemma*}

\begin{proof}
    Since superior M-monoids are defined to be of the form $(\walg K, \min, \mathbb 0, \Omega)$ for some total order $\walg K$, we have to refer to the inverse total order on $\rzo$, $(\rzo, \ge)$.
    We show that for every $c \in \mathbb{R}_0^1$ and $k \in \mathbb N$, $\mul^{(k)}_c \in \Omega_\cdot$ is $\ge$-superior.
    Let $b,a,a_1,\ldots,a_k \in \mathbb{R}_0^1$.
    \begin{enumerate}
        \item Assume that $a \ge a_i$.
            Then, by monotonicity of $\cdot$ in $\rzo$,
            \[
                b \cdot a_1 \cdot \ldots \cdot a_{i-1} \cdot a \cdot a_{i+1} \cdot \ldots \cdot a_k \ \ge \
                b \cdot a_1 \cdot \ldots \cdot a_{i-1} \cdot a_i \cdot a_{i+1} \cdot \ldots \cdot a_k \enspace,
            \]
            which proves the first condition.
        \item Since $0 \le a_i \le 1$ for each $a_i$ and also $0 \le b \le 1$, we have for each $i \in [k]$ that $b \cdot a_1 \cdot \ldots \cdot a_i \cdot \ldots \cdot a_k \le a_i$.
            Thus $\min \{a_1,\ldots,a_k\} \ \ge \ b \cdot a_1 \cdot \ldots \cdot a_i \cdot \ldots \cdot a_k$ which proves the second condition. \qedhere
    \end{enumerate}
\end{proof}

\subsection{Best derivation M-monoid is distributive and d-complete}\label{sec:proof-best-derivation-mmonoid}

\begin{lemma*}[cf.\ Example~\ref{ex:best-derivation-mmonoid}]\label{lem:best-derivation-mmonoid}
    The best derivation M-monoid is d-complete and distributive.
    Furthermore, $(0, \emptyset)$ is absorptive.
\end{lemma*}

\begin{proof}
    \begin{sloppypar}
    Clearly the best derivation M-monoid is a complete M-monoid.
    In order to show distributivity of~$\Omegav$ over $\maxv$, we let $p \in \rzo$, $r \in R$ with $\rk(r) = k$, $i \in [k]$, $\tc{p}{r} \in \Omegav$, and $(p', D'), (p_1, D_1), \dots, (p_k, D_k) \in \mathbb R_0^1 \times \mathcal P(\T_R)$.
    Then there are three cases:
    \end{sloppypar}
    \begin{enumerate}
        \item if $p_i < p'$, then
            \begin{align*}
                &\tc{p}{r} \big( (p_1, D_1), \dots, (p_{i-1}, D_{i-1}), \maxv \big( (p_i, D_i), (p', D') \big), (p_{i+1}, D_{i+1}), \dots, (p_k, D_k) \big) \\
                &= \tc{p}{r} \big( (p_1, D_1), \dots, (p_{i-1}, D_{i-1}), (p', D'), (p_{i+1}, D_{i+1}), \dots, (p_k, D_k) \big)
                \intertext{and since by monotonicity of~$\cdot$, $p \cdot p_1 \cdot \ldots \cdot p_k < p \cdot p_1 \cdot \ldots \cdot p_{i-1} \cdot p' \cdot p_{i+1} \cdot \ldots \cdot p_k$}
                &= \begin{aligned}[t]
                    \maxv \big(&\tc{p}{r}((p_1, D_1), \dots, (p_k, D_k)), \\
                    &\tc{p}{r}((p_1, D_1), \dots, (p_{i-1}, D_{i-1}), (p', D'), (p_{i+1}, D_{i+1}), \dots, (p_k, D_k)) \big) \enspace.
                \end{aligned}
            \end{align*}
        \item if $p' < p_i$, then we obtain
            \begin{align*}
                &\tc{p}{r} \big( (p_1, D_1), \dots, (p_{i-1}, D_{i-1}), \maxv \big( (p_i, D_i), (p', D') \big), (p_{i+1}, D_{i+1}), \dots, (p_k, D_k) \big) \\
                &= \begin{aligned}[t]
                    \maxv \big(&\tc{p}{r}((p_1, D_1), \dots, (p_k, D_k)), \\
                    &\tc{p}{r}((p_1, D_1), \dots, (p_{i-1}, D_{i-1}), (p', D'), (p_{i+1}, D_{i+1}), \dots, (p_k, D_k)) \big)
                \end{aligned}
            \end{align*}
            analogously to the first case, and
        \item if $p' = p_i$, then
            \begin{align*}
                &\tc{p}{r} \big( (p_1, D_1), \dots, (p_{i-1}, D_{i-1}), \maxv \big( (p_i, D_i), (p', D') \big), (p_{i+1}, D_{i+1}), \dots, (p_k, D_k) \big) \\
                &= \tc{p}{r} \big( (p_1, D_1), \dots, (p_{i-1}, D_{i-1}), (p_i, D_i \cup D'), (p_{i+1}, D_{i+1}), \dots, (p_k, D_k) \big) \enspace,
                \intertext{now $p \cdot p_1 \cdot \ldots \cdot p_k = p \cdot p_1 \cdot \ldots \cdot p_{i-1} \cdot p' \cdot p_{i+1} \cdot \ldots p_k$, hence}
                &= \begin{aligned}[t]
                    \maxv(&\tc{p}{r}((p_1, D_1), \dots, (p_k, D_k)), \\
                    &\tc{p}{r}((p_1, D_1), \dots, (p_{i-1}, D_{i-1}), (p', D'), (p_{i+1}, D_{i+1}), \dots, (p_k, D_k))) \enspace.
                \end{aligned}
            \end{align*}
            For the last step we remark that
            \[ \begin{aligned}[t]
                &\{ r(d_1, \dots d_k) \mid d_1 \in D_1, \dots d_{i-1} \in D_{i-1}, d_i \in D_i \cup D', d_{i+1} \in D_{i+1}, \dots, d_k \in D_k \} \\
                &= \begin{aligned}[t]
                    &\{ r(d_1, \dots d_k) \mid d_1 \in D_1, \dots d_k \in D_k \} \cup {} \\
                    &\{ r(d_1, \dots d_k) \mid d_1 \in D_1, \dots d_{i-1} \in D_{i-1}, d_i \in D', d_{i+1} \in D_{i+1}, \dots, d_k \in D_k \} \enspace.
                \end{aligned}
            \end{aligned} \]
    \end{enumerate}

    In order to show that~$\mathbb{BD}$ is d-complete, we show that it is completely idempotent.
    For this, let~$I$ be a countable set and $(p, D) \in \rzo$.
    Then
    \begin{align*}
        \infsum_{i \in I} (p, D) &= \Big( \sup \{ p \mid i \in I \}, \bigcup_{\substack{i \in I: \\p = \sup \{ p \mid i \in I \}}} D \Big) \\
        &= \Big( p, \bigcup_{i \in I} D \Big) \\
        &= (p, D) \enspace.
    \end{align*}
    Thus, by Lemma~\ref{lem:inf-idp-d-complete}, it is d-complete.

    In order to show absorptivity of $(0, \emptyset)$, we let $p \in \mathbb R_0^1$, $r \in R$ with $\rk(r) = k$, $\tc{p}{r} \in \Omegav$, and $(p_1, D_1), \dots, (p_k, D_k) \in \mathbb R_0^1 \times \mathcal P(\T_R)$.
    Now, if there is an $i \in [k]$ such that $(p_i, D_i) = (0, \emptyset)$, then $p \cdot p_1 \cdot \ldots \cdot p_{i-1} \cdot 0 \cdot p_{i+1} \cdot \ldots \cdot p_k = 0$ by absorptivity of~$0$ and
    \[ \{ r(d_1, \dots, d_k) \mid d \in D_1, \dots, d_{i-1} \in D_{i-1}, d_i \in \emptyset, d_{i+1} \in D_{i+1}, \dots, d_k \in D_k \} = \emptyset \enspace, \]
    hence
    \[ \tc{p}{r} \big( (p_1, D_1), \dots, (p_{i-1}, D_{i-1}), (0, \emptyset), (p_{i+1}, D_{i+1}), \dots, (p_k, D_k) \big) = (0, \emptyset) \enspace. \qedhere \]
\end{proof}

\subsection{N-best M-monoid is distributive and d-complete}\label{sec:proof-nbest-mmonoid}

\begin{lemma*}[cf.\ Example~\ref{ex:nbest-mmonoid}]\label{lem:nbest-mmonoid}
    The n-best M-monoid is distributive and d-complete and $\nbzeroes$ is absorptive.
\end{lemma*}

\begin{proof}
    Clearly the n-best M-monoid is a complete M-monoid.
    For distributivity of~$\Omega_n$ over $\maxn$ we first show that~$\cdot_n$ is commutative and distributive over $\maxn$.
    Commutativity of~$\cdot_n$ follows from the commutativity of~$\cdot$ in~$\rzo$ and for distributivity of~$\cdot_n$ over $\maxn$ we let $(a_1, \dots, a_n), (b_1, \dots, b_n), (c_1, \dots, c_n) \in \nbest$ and $f = \takenbest \big( (a_1, \dots, a_n, b_1, \dots, b_n) \big)$.
    Then
    \begin{align*}
        &\maxn \big( (a_1, \dots, a_n), (b_1, \dots, b_n) \big) \cdot_n (c_1, \dots, c_n) \\
        &= (f_1, \dots, f_n) \cdot_n (c_1, \dots, c_n) \\
        &= \takenbest \big( (f_1 \cdot c_1, \dots, f_1 \cdot c_n, \dots f_n \cdot c_1, \dots, f_n \cdot c_n) \big) \\
        \intertext{and by monotonicity of $\cdot$ in $\rzo$}
        &= \takenbest \big( (a_1 \cdot c_1, \dots, a_1 \cdot c_n, \dots a_n \cdot c_1, \dots, a_n \cdot c_n, b_1 \cdot c_1, \dots, b_1 \cdot c_n, \dots b_n \cdot c_1, \dots, b_n \cdot c_n) \big) \\
        &= \maxn \big( (a_1 \cdot c_1, \dots, a_1 \cdot c_n, \dots a_n \cdot c_1, \dots, a_n \cdot c_n), (b_1 \cdot c_1, \dots, b_1 \cdot c_n, \dots b_n \cdot c_1, \dots, b_n \cdot c_n) \big) \\
        &= \begin{aligned}[t]
            \maxn \big(&\takenbest((a_1 \cdot c_1, \dots, a_1 \cdot c_n, \dots a_n \cdot c_1, \dots, a_n \cdot c_n)), \\
            &\takenbest((b_1 \cdot c_1, \dots, b_1 \cdot c_n, \dots b_n \cdot c_1, \dots, b_n \cdot c_n)) \big)
        \end{aligned} \\
        &= \maxn \big( (a_1, \dots, a_n) \cdot_n (c_1, \dots, c_n), (b_1, \dots, b_n) \cdot_n (c_1, \dots, c_n) \big) \enspace.
    \end{align*}
    \begin{sloppypar}
    We thus obtain for every $k \in \mathbb N$, $\welem k \in \nbest$, $\mulnkk \in \Omega_n$, $i \in [k]$, and $(a_1, \dots, a_n), (a_{1,1}, \dots, a_{1,n}), \dots, (a_{k,1}, \dots, a_{k,n}) \in \nbest$
    \end{sloppypar}
    \begin{align*}
        &\begin{aligned}[t]
            \mulnkk \big(&(a_{1,1}, \dots, a_{1,n}), \dots, (a_{i-1,1} \dots, a_{i-1,n}), \\
            &\maxn((a_{i,1}, \dots, a_{i,n}), (a_1, \dots, a_n)), \\
            &(a_{i+1,1}, \dots, a_{i+1,n}), \dots, (a_{k,1}, \dots, a_{k,n}) \big)
        \end{aligned} \\
        &= \begin{aligned}[t]
            \nbweight &\cdot_n (a_{1,1}, \dots, a_{1,n}) \cdot_n \ldots \cdot_n (a_{i-1,1}, \dots, a_{i-1,n}) \\
            &\cdot_n \maxn \big( (a_{i,1}, \dots, a_{i,n}), (a_1, \dots, a_n) \big) \\
            &\cdot_n (a_{i+1,1}, \dots, a_{i+1,n}) \cdot_n \ldots \cdot_n (a_{k,1}, \dots, a_{k,n})
        \end{aligned}
        \intertext{and by commutativity of $\cdot_n$}
        &= \begin{aligned}[t]
            \maxn \big( (a_{i,1}, \dots, a_{i,n}), (a_1, \dots, a_n) \big) &\cdot_n \nbweight \\
            &\cdot_n (a_{1,1}, \dots, a_{1,n}) \cdot_n \ldots \cdot_n (a_{i-1,1}, \dots, a_{i-1,n}) \\
            &\cdot_n (a_{i+1,1}, \dots, a_{i+1,n}) \cdot_n \ldots \cdot_n (a_{k,1}, \dots, a_{k,n})
        \end{aligned}
        \intertext{and by distributivity of $\cdot_n$ over $\maxn$}
        &= \begin{aligned}[t]
            \maxn \big( &(a_{i,1}, \dots, a_{i,n}) \cdot_n \nbweight \\
            &\cdot_n (a_{1,1}, \dots, a_{1,n}) \cdot_n \ldots \cdot_n (a_{i-1,1}, \dots, a_{i-1,n}) \\
            &\cdot_n (a_{i+1,1}, \dots, a_{i+1,n}) \cdot_n \ldots \cdot_n (a_{k,1}, \dots, a_{k,n}), \\
            &(a_1, \dots, a_n)  \cdot_n \nbweight \\
            &\cdot_n (a_{1,1}, \dots, a_{1,n}) \cdot_n \ldots \cdot_n (a_{i-1,1}, \dots, a_{i-1,n}) \\
            &\cdot_n (a_{i+1,1}, \dots, a_{i+1,n}) \cdot_n \ldots \cdot_n (a_{k,1}, \dots, a_{k,n}) \big)
        \end{aligned}
        \intertext{and by commutativity of $\cdot_n$}
        &= \begin{aligned}[t]
            \maxn \big(&\nbweight \cdot_n (a_{1,1}, \dots, a_{1,n}) \cdot_n \ldots \cdot_n (a_{k,1}, \dots, a_{k,n}), \\
            &\nbweight \cdot_n (a_{1,1}, \dots, a_{1,n}) \cdot_n \ldots \cdot_n (a_{i-1,1}, \dots, a_{i-1,n}) \\
            &\cdot_n (a_1, \dots, a_n) \\
            &\cdot_n (a_{i+1,1}, \dots, a_{i+1,n}) \cdot_n \ldots \cdot_n (a_{k,1}, \dots, a_{k,n}) \big)
        \end{aligned} \\
        &= \begin{aligned}[t]
            \maxn \big(&\mulnkk((a_{1,1}, \dots, a_{1,n}), \dots, (a_{k,1}, \dots, a_{k,n})), \\
            &\mulnkk(\begin{aligned}[t]
                &(a_{1,1}, \dots, a_{1,n}), \dots, (a_{i-1,1}, \dots, a_{i-1,n}), \\
                &(a_1, \dots, a_n), \\
                &(a_{i+1,1}, \dots, a_{i+1,n}), \dots, (a_{k,1}, \dots, a_{k,n})) \big) \enspace.
            \end{aligned}
        \end{aligned}
    \end{align*}

    Now we show that the n-best M-monoid is d-complete.
    For this, let $(a_1, \dots, a_n) \in \nbest$ and $\big((a_{i,1}, \dots, a_{i,n}) \mid i \in I\big)$ be an $I$-indexed family over $\nbest$ such that for every $i \in I$, $(a_1, \dots, a_n) \oplus (a_{i,1}, \dots, a_{i,n}) = (a_1, \dots, a_n)$.
    Then for every $i \in \mathbb N$ we have that $a_n \ge a_{i,1}$.
    Thus $a_n \ge \sup \{ a_{i,j} \mid i \in I, j \in [n] \}$.
    Let~$\psi: J \to \mathbb N$ be a bijective mapping.
    We define the family $(f_i \mid i \in \mathbb N)$ such that for each $i \in [n]$, $f_i = a_i$ and for each $i \in \mathbb N \setminus [n]$, $f_i = a_{\lfloor i / n \rfloor, i \bmod n + 1}$.
    Then $\takenbest((f_i \mid i \in \mathbb N)) = (a_1, \dots, a_n)$.
    Thus
    \[ (a_1, \dots a_n) \oplus \infsum[\maxn]_{i \in I} (a_{i,1}, \dots, a_{i,n}) = (a_1, \dots, a_n) \enspace. \]
    Then, by Lemma~\ref{lem:d-complete}~(ii), the n-best M-monoid is d-complete.

    In order to show that $\nbzeroes$ is absorptive for~$\Omega_n$, we first show that it is absorptive for~$\cdot_n$.
    For this, we let $(a_1, \dots, a_n) \in \nbest$.
    Then
    \begin{align*}
        (a_1, \dots, a_n) \cdot_n \nbzeroes &= \takenbest(\underbrace{a_1 \cdot 0, \dots, a_1 \cdot 0}_{\text{$n$ times}}, \dots, \underbrace{a_n \cdot 0, \dots, a_n \cdot 0}_{\text{$n$ times}}) \\
        &= \takenbest(\underbrace{0, \dots, 0}_{\text{$n^2$ times}})
        \tag{$0$ is absorptive for $\cdot$} \\
        &= \nbzeroes
    \end{align*}
    Now absorptivity of $\nbzeroes$ for $\Omega_n$ is easy to see.
\end{proof}

\subsection{Definition of closed weighted RTG-based language models}\label{app:closed-definition}

\lemtranswf*

\begin{proof}
    Let $D \subseteq \T_R$ with $|D| \not= \emptyset$.
    We define the set $I = \{ \height(d) \mid d \in \T_R \}$.
    Clearly $I \subseteq \mathbb N$ and $I \not= \emptyset$.
    Thus, as $(\mathbb N, <)$ is well-founded, there is an $i \in I$ such that for every $i' \in I$, $i' \not< i$.
    We choose an arbitrary $d \in D$ such that $\height(d) = i$.
    We show that for every $d' \in D$ it does not hold that $d' (\vdash^+)^{-1} d$ by contradiction.
    For this, assume that there is a $d' \in D$ such that $d' (\vdash^+)^{-1} d$.
    Then $d \vdash^+ d'$ and thus by definition of~$\vdash$, $\height(d') < \height(d) = i$.
    This contradicts the fact that for every $i' \in I$, $i' \not< i$.
\end{proof}

\lemcutoutsubset*

\begin{proof}
    Let $d, d' \in \T_R$ such than $d \vdash^+ d'$.
    Then
    \begin{align*}
        \cotrees(d') &= \{ d'' \in \T_R \mid d' \vdash^+ d \} \\
        &\subset \{ d'' \in \T_R \mid d' \vdash^+ d'' \} \cup \{ d \} \\
        \intertext{and as $d \vdash^+ d'$, by transitivity of~$\vdash^+$}
        &\subseteq \{ d'' \in \T_R \mid d \vdash^+ d'' \} \\
        &= \cotrees(d) \enspace. \qedhere
    \end{align*}
\end{proof}

\subsection{Properties of closed weighted RTG-based language models}\label{app:closed-properties}

This subappendix contains the full proofs of Lemma~\ref{lem:closed-bigger-trees'}, Theorem~\ref{thm:outside-trees-subsumed} and Lemma~\ref{lem:outside-trees-spawned}.
We start with several auxiliary statements.

\begin{lemma}\label{lem:trees-bounded-height'}
    For every $c \in \mathbb N$ there is an $n \in \mathbb N$ such that for each $d \in \TRc$, $\height(d) < n$.
\end{lemma}

\begin{proof}
    We start with an auxiliary statement:
    for every $\rho \in R^*$ with $|\rho| = |R| + 1$ it holds that~$\rho$ is cyclic.
    For this, let $\rho \in R^*$ with $|\rho| = |R| + 1$.
    Since~$R$ is finite, there are $i,j \in [|R| + 1]$ with $i \not= j$ such that $\rho_i = \rho_j$, hence~$\rho$ is cyclic.

    Next we show that for every $c \in \mathbb N$, there is an $n \in \mathbb N$ such that for each $c' \in \mathbb N$ with $c' \leq c$ and $\rho \in R^*$ which is $c'$-cyclic it holds that $|\rho| < n$.
    Note that the number of strings $\rho \in R^*$ with $|\rho| = |R| + 1$ is $|R|^{|R|+1}$; we denote this number by~$m$.
    Now let $n = (c + 1) \cdot m$.
    We show that for every $\rho \in R^*$ with $|\rho| = n$ it holds that~$\rho$ is $c'$-cyclic for some $c' > c$.
    Clearly such~$\rho$ is cyclic, hence we let $c' \in \mathbb N$ with $c' \leq c$ and assume that there is a $\rho \in R^*$ such that $|\rho| = n$ and~$\rho$ is $c'$-cyclic.
    We let
    \[ \rho = \underbracket{\rho_1 \dots \rho_{m\vphantom{()}}} \dots \underbracket{\rho_{c \cdot m + 1} \dots \rho_{(c+1) \cdot m}} \enspace, \]
    where $\rho_i \in R^*$ and $|\rho_i| = |R| + 1$ for every $i \in [n]$.
    Since $n / m = c + 1$, there is an $i \in [n]$ such that~$\rho_i$ occurs at least $c + 1$ times in~$\rho$.
    Furthermore, since $|\rho_i| = |R| + 1$,~$\rho_i$ is cyclic.
    Thus there is an elementary cycle which occurs at least $c + 1$ times in~$\rho$, which contradicts our assumption that~$\rho$ is $c'$-cyclic for some $c' \leq c$.
    As the same argument can be made for any $n' \geq n$, we conclude that for every $\rho \in R^*$ such that~$\rho$ is $c'$-cyclic it holds that $|\rho| < n$.

    Then for every $d \in \TRc$ and $p \in \pos(d)$ it holds that $|p| < n - 1$ and thus $\height(d) < n - 1$, which proves the lemma.
\end{proof}

\begin{lemma}\label{lem:trc-finite}
    For every $c \in \mathbb N$ the set~$\TRc$ is finite.
\end{lemma}

\begin{proof}
    This follows directly from Lemmas~\ref{lem:trees-bounded-height'} and~\ref{lem:fixed-height-finite-trees'}.
\end{proof}

\lemgenclosed*

\begin{proof}
    The proof is done by induction on~$c'$.
    For the induction base, let $c' = c + 1$, then the statement of the lemma holds by Equation~\eqref{eq:c-closed}.
    For the induction step, let $c' \geq c + 1$.
    We assume that for each $d' \in (\T_R)$ and elementary cycle $w' \in R^*$ such that there is a leaf $p \in \pos(d')$ which is $(c',w')$-cyclic the following holds (IH):
    \[ \wthom{d'} \oplus \bigoplus_{d'' \in \cotrees(d', w')} \wthom{d''} = \bigoplus_{d'' \in \cotrees(d', w')} \wthom{d''} \enspace. \]
    Now we let $d \in (\T_R)$ and $w \in R^*$ be an elementary cycle such that there is a leaf $p \in \pos(d)$ which is $(c'+1,w)$-cyclic.
    We let $v_0, \dots, v_{c'+1} \in R^*$ such that $\seq(d, p) = v_0 w v_1 \dots w v_{c'+1}$, $w = r_1 \dots r_m$ with $r_i \in R$ for every $i \in [m]$, $r_m = \big(A \to \sigma(A_1, \dots, A_k)\big)$, $n = |v_0| + |w| + 1$, $s \in [k]$ such that $p_n = s$, and
    \[ D = \{ d' \in \cotrees(d, w) \mid d'[x_{s,A_s}]_{p_{1 \isep n}} = d[x_{s,A_s}]_{p_{1 \isep n}} \} \enspace. \]
    As $d \subseteq \cotrees(d, w)$
    \begin{align*}
        &\wthom{d} \oplus \bigoplus_{d' \in \cotrees(d, w)} \wthom{d'} \\
        &= \wthom{d} \oplus \bigoplus_{d \in D} \wthom{d'} \oplus \bigoplus_{d' \in \cotrees(d, w) \setminus D} \wthom{d'} \\
        \intertext{and since $\Omega$ distributes over $\oplus$}
        &= \wt(d)[x_{s,A_s}]_{p_{1 \isep n}} \left(\wthom{d|_{p_n}} \oplus \bigoplus_{d' \in D} \wthom{d'|_{p_n}}\right)_{\walg K} \oplus \bigoplus_{d' \in \cotrees(d, w) \setminus D} \wthom{d'}
        \intertext{and since $\{ d'|_{p_n} \mid d' \in D \} = \cotrees(d|_{p_n}, w)$, by IH}
        &= \wt(d)[x_{s,A_s}]_{p_{1 \isep n}}\left(\bigoplus_{d' \in D} \wthom{d'|_{p_n}}\right)_{\walg K} \oplus \bigoplus_{d' \in \cotrees(d, w) \setminus D} \wthom{d'} \\
        &= \bigoplus_{d \in D} \wthom{d'} \oplus \bigoplus_{d' \in \cotrees(d, w) \setminus D} \wthom{d'} \\
        &= \bigoplus_{d' \in \cotrees(d, w)} \wthom{d'} \enspace. \qedhere
    \end{align*}
\end{proof}

\begin{lemma}\label{lem:closed-bigger-trees-any}
    For every $d \in (\T_R)$ and $c' \in \mathbb N$ with $c' \geq c + 1$ such that~$d$ is $c'$-cyclic the following holds:
    \[ \wthom{d} \oplus \bigoplus_{d' \in \cotrees(d, w)} \wthom{d'} = \bigoplus_{d' \in \cotrees(d, w)} \wthom{d'} \enspace. \]
\end{lemma}

\begin{proof}
    This is a consequence of Lemma~\ref{lem:closed-bigger-trees'}.
\end{proof}

\begin{lemma}\label{lem:outside-cutout-trees-subsumed}
    For every $m \in \mathbb N$, $d \in \T_R \setminus \TRc$, and $B \subseteq \cotrees(d) \setminus \TRc$ with $|B| = m$ the following holds:
    \[ \bigoplus_{d' \in \cotrees(d)} \wthom{d'} = \bigoplus_{d' \in \cotrees(d) \setminus B} \wthom{d'} \enspace.\]
\end{lemma}

\begin{proof}
    Let $d \in \T_R$ and $m \in \mathbb N$.
    The proof is done by induction on~$m$.
    For the induction base, let $m = 0$.
    Then $B = \emptyset$ and for every $d \in \TRc$
    \[ \bigoplus_{d' \in \cotrees(d)} \wthom{d'} = \bigoplus_{d' \in \cotrees(d) \setminus \emptyset} \wthom{d'} \enspace.\]

    For the induction step, let $m \in \mathbb N$.
    We assume (IH) that for every $d \in \T_R \setminus \TRc$ and $B \subseteq \cotrees(d) \setminus \TRc$ with $|B| = m$ it holds that
    \[ \bigoplus_{d' \in \cotrees(d)} \wthom{d'} = \bigoplus_{d' \in \cotrees(d) \setminus B} \wthom{d'} \enspace.\]
    Now let $B \subseteq \cotrees(d) \setminus \TRc$ such that $|B| = m + 1$.
    Then, by Lemma~\ref{lem:transition-well-founded}, there is a $d' \in B$ such that for every $d'' \in B$ it does not hold that $d'' {(\vdash^+)}^{-1} d'$ and thus $d' \not{\vdash^+} d''$.
    Then
    \begin{align*}
        &\bigoplus_{d'' \in \cotrees(d) \setminus B} \wthom{d} \\
        &= \bigoplus_{d'' \in \cotrees(d') \setminus B} \wthom{d''} \oplus \bigoplus_{d'' \in (\cotrees(d) \setminus \cotrees(d')) \setminus B} \wthom{d''}
        \tag{Lemma~\ref{lem:subtree-cotrees-subset}} \\
        \intertext{and for every $d'' \in B$, as $d' \not{\vdash^+} d''$, we have that $d'' \not\in \cotrees(d')$; thus}
        &= \bigoplus_{d'' \in \cotrees(d')} \wthom{d''} \oplus \bigoplus_{d'' \in (\cotrees(d) \setminus \cotrees(d')) \setminus B} \wthom{d''} \\
        &= \bigoplus_{d'' \in \cotrees(d')} \wthom{d''} \oplus \wthom{d'} \oplus \bigoplus_{d'' \in (\cotrees(d) \setminus \cotrees(d')) \setminus B} \wthom{d''}
        \tag{Lemma~\ref{lem:closed-bigger-trees-any}} \\
        &= \bigoplus_{d'' \in \cotrees(d') \setminus (B \setminus \{ d' \})} \wthom{d''} \oplus \bigoplus_{d'' \in (\cotrees(d) \setminus \cotrees(d')) \setminus (B \setminus \{ d' \})} \wthom{d''} \\
        &= \bigoplus_{d'' \in \cotrees(d) \setminus (B \setminus \{ d' \})} \wthom{d''} \\
        &= \bigoplus_{d'' \in \cotrees(d)} \wthom{d''} \enspace. \tag*{(IH) \qedhere}
    \end{align*}
\end{proof}

\begin{lemma}\label{lem:cutout-trees-subsume'}
    For every $d \in \T_R \setminus \TRc$ the following holds:
    \[ \wthom{d} \oplus \bigoplus_{d' \in \cotrees(d) \cap \TRc} \wthom{d'} = \bigoplus_{d' \in \cotrees(d) \cap \TRc} \wthom{d'} \enspace. \]
\end{lemma}

\begin{proof}
    Let $d \in \T_R \setminus \TRc$.
    Then
    \begin{align*}
        \wthom{d} \oplus \bigoplus_{d' \in \cotrees(d) \cap \TRc} \wthom{d'} &=\wthom{d} \oplus \bigoplus_{d' \in \cotrees(d)} \wthom{d'}
        \tag{Lemma~\ref{lem:outside-cutout-trees-subsumed}} \\
        &=\bigoplus_{d' \in \cotrees(d)} \wthom{d'}
        \tag{Lemma~\ref{lem:closed-bigger-trees-any}} \\
        &=\bigoplus_{d' \in \cotrees(d) \cap \TRc} \wthom{d'}
        \tag*{(Lemma~\ref{lem:outside-cutout-trees-subsumed}) \qedhere}
    \end{align*}
\end{proof}

\thmclosed*

\begin{proof}
    Let $l \in \mathbb N$, $D \subseteq \TRc$, and $D' \subseteq \T_R \setminus \TRc$ such that $\bigcup_{d \in D'} (\cotrees(d) \cap \TRc) \subseteq D$.
    The proof is done by induction on~$l$.
    For the induction base, let $l = 0$.
    Then $B = \emptyset$ and the statement of the lemma holds.

    For the induction step, let $l \in \mathbb N$.
    We assume that for every $B \subseteq D'$ with $|B| = l$,
    \begin{equation}
        \bigoplus_{d \in D} \wthom{d} \oplus \infsum_{d \in D'} \wthom{d} = \bigoplus_{d \in D} \wthom{d} \oplus \infsum_{d \in D' \setminus B} \enspace. \tag{IH}
    \end{equation}
    Now we let $B \subseteq D'$ such that $|B| = l + 1$ and $d \in B$.
    Then, as $\cotrees(d) \cap \TRc \subseteq D$
    \begin{align*}
        &\bigoplus_{d' \in D} \wthom{d'} \oplus \infsum_{d' \in D' \setminus B} \wthom{d'} \\
        &=\bigoplus_{d' \in D \setminus (\cotrees(d) \cap \TRc)} \wthom{d'} \oplus \bigoplus_{d' \in \cotrees(d) \cap \TRc} \wthom{d'} \oplus \infsum_{d' \in D' \setminus B} \wthom{d'} \\
        &=\bigoplus_{d' \in D \setminus (\cotrees(d) \cap \TRc)} \wthom{d'} \oplus \bigoplus_{d' \in \cotrees(d) \cap \TRc} \wthom{d'} \oplus \wthom{d} \oplus \infsum_{d' \in D' \setminus B} \wthom{d'}
        \tag{Lemma~\ref{lem:cutout-trees-subsume'}, because $d \in \T_R \setminus \TRc$} \\
        &=\bigoplus_{d' \in D \setminus (\cotrees(d) \cap \TRc)} \wthom{d'} \oplus \bigoplus_{d' \in \cotrees(d) \cap \TRc} \wthom{d'} \oplus \infsum_{d' \in D' \setminus (B \setminus \{ d \})} \wthom{d'} \\
        &=\bigoplus_{d' \in D \setminus (\cotrees(d) \cap \TRc)} \wthom{d'} \oplus \bigoplus_{d' \in \cotrees(d) \cap \TRc} \wthom{d'} \oplus \infsum_{d' \in D'} \wthom{d'}
        \tag{IH} \\
        &= \bigoplus_{d' \in D} \wthom{d'} \oplus \infsum_{d' \in D'} \wthom{d'} \qedhere
    \end{align*}
\end{proof}

\lemclosed*

\begin{proof}
    The proof is done by induction on~$l$.
    For the induction base, let $l = 0$.
    Then $B = \emptyset$ and the statement of the lemma holds for every $A \in N$.

    For the induction step, let $l \in \mathbb N$.
    We assume that for every $A \in N$ and $B \subseteq (\T_R)_A \setminus \TRc$ with $|B| = l$,
    \begin{equation}
        \bigoplus_{d \in (\TRc)_A} \wthom{d} = \bigoplus_{d \in (\TRc)_A \cup B} \wthom{d} \enspace. \tag{IH}
    \end{equation}
    Now we let $A \in N$, $B \subseteq (\T_R)_A \setminus \TRc$ such that $|B| = l + 1$, and $d' \in B$.
    Then
    \begin{align*}
        \bigoplus_{d \in (\TRc)_A \cup B} \wthom{d} &= \bigoplus_{d \in (\TRc)_A} \oplus \bigoplus_{d \in B} \wthom{d} \tag{$B \cap \TRc = \emptyset$} \\
        &= \bigoplus_{d \in (\TRc)_A \setminus (\cotrees(d') \cap \TRc)} \wthom{d} \oplus \bigoplus_{d \in \cotrees(d') \cap \TRc} \wthom{d} \oplus \bigoplus_{d \in B} \wthom{d} \\
        &= \bigoplus_{d \in (\TRc)_A \setminus (\cotrees(d') \cap \TRc)} \wthom{d} \oplus \bigoplus_{d \in \cotrees(d') \cap \TRc} \wthom{d} \oplus \wthom{d'} \oplus \bigoplus_{\mathclap{d \in B \setminus \{ d' \}}} \wthom{d} \\
        &= \bigoplus_{d \in (\TRc)_A \setminus (\cotrees(d') \cap \TRc)} \wthom{d} \oplus \bigoplus_{d \in \cotrees(d') \cap \TRc} \wthom{d} \oplus \bigoplus_{\mathclap{d \in B \setminus \{ d' \}}} \wthom{d} \tag{Lemma~\ref{lem:cutout-trees-subsume'}, because $d \in \T_R \setminus \TRc$} \\
        &= \bigoplus_{d \in (\TRc)_A} \wthom{d} \oplus \bigoplus_{d \in B \setminus \{ d' \}} \wthom{d} \\
        &= \bigoplus_{d \in (\TRc)_A} \wthom{d} \enspace. \tag*{(IH) \qedhere}
    \end{align*}
\end{proof}

\subsection{Intersection is an instance of the M-monoid parsing problem}\label{app:intersection}

In this subappendix we give a full proof of Theorem~\ref{thm:intersection}.
For this, we let $(G,(\alg L,\phi))$ be an RTG-LM such that $G=(N,\Sigma,A_0,R)$ and $(\alg L,\phi)$ is a finitely decomposable language algebra.
Moreover, we let $a \in \alg L_{\sort(A_0)}$.
We consider the M-monoid parsing problem with the following input:
\begin{itemize}
    \item the wRTG-LM $((G,(\alg L,\phi)), \walg{K}((G,(\alg L,\phi)),a), \wt)$ where $\wt(r) = \omega_r$ for each $r \in R$ and
    \item $a$.
\end{itemize}
We show that $(G',(\alg L,\phi))$ is the $\psi$-intersection of $(G,(\alg L,\phi))$ and $a$, where
\begin{itemize}
    \item $G' = (N',\Sigma,[A_0,a],R')$ with $N'=\mathrm{lhs}(\fparse(a)) \cup \{[A_0,a]\}$ (we note that $\fparse(a)$ is a finite set), $R' = \fparse(a)$, and
    \item $\psi\colon N' \rightarrow N$ is defined by $\psi([A,b]) = A$ for each $[A,b] \in N'$.
\end{itemize}
We extend the mapping $\widehat{\psi}: \AST(G') \to \AST(G, a)$ such that $\widehat{\psi}: \T_{R'} \to \T_R$.
This is required for our proofs by structural induction.
Clearly, the extended mapping $\widehat{\psi}$ is not bijective in general and we will only show bijectivity of $\widehat{\psi}: \AST(G') \to \AST(G, a)$.

\begin{lemma}\label{lem:intersection-semantics}
    For every $d \in \T_{R'}$ it holds that $\sem[\alg L]{\pi_\Sigma(d)} = b$, where $\lhs(d(\varepsilon)) = [A, b]$ for some $A \in N'$.
\end{lemma}

\begin{proof}
    The proof is done by structural induction on $d$.
    We assume (IH) that for every $k \in \mathbb N$, $i \in [k]$, and $d_i \in \T_{R'}$ it holds that $\sem[\alg L]{\pi_\Sigma(d_i)} = a_i$, where $\lhs(d_i(\varepsilon)) = [A_i, a_i]$ for some $A_i \in N'$.
    Then for every $r \in R'$ with $r = \big( [A, b] \to \sigma([A_1, a_1], \dots, [A_k, a_k]) \big)$
    \begin{align*}
        \sem[\alg L]{\pi_\Sigma \big( r(d_1, \dots, d_k) \big)} &= \phi(\sigma) \big( \sem[\alg L]{\pi_\Sigma(d_1)}, \dots, \sem[\alg L]{\pi_\Sigma(d_k)} \big) \\
        &= \phi(\sigma)(a_1, \dots, a_k) \tag{IH} \\
        &= b \enspace. \tag*{(Definition of $P_{R,a}$) \qedhere}
    \end{align*}
\end{proof}

\begin{lemma}\label{lem:psi-and-pisigma}
    For every $d \in \T_{R'}$ it holds that $\sem[\alg L]{\pi_\Sigma(d)} = \sem[\alg L]{\pi_\Sigma(\widehat{\psi}(d))}$.
\end{lemma}

\begin{proof}
    The proof is done by structural induction on $d$.
    We assume (IH) that for every $k \in \mathbb N$, $i \in [k]$, and $d_i \in \T_{R'}$ it holds that $\sem[\alg L]{\pi_\Sigma(d_i)} = \sem[\alg L]{\pi_\Sigma(\widehat{\psi}(d_i))}$.
    Then for every $r \in R'$ with $r = \big( [A,b] \to \sigma([A_1,a_1] \dots, [A_k,a_k]) \big)$ we have
    \begin{align*}
        \sem[\alg L]{\pi_\Sigma(r(d_1, \dots, d_k))} &= \phi(\sigma) \big( \sem[\alg L]{\pi_\Sigma(d_1)}, \dots, \sem[\alg L]{\pi_\Sigma(d_k)} \big) \\
        &= \phi(\sigma) \big( \sem[\alg L]{\pi_\Sigma(\widehat{\psi}(d_1))}, \dots, \sem[\alg L]{\pi_\Sigma(\widehat{\psi}(d_k))} \big) \tag{IH} \\
        &= \sem[\alg L]{\pi_\Sigma \bigl( \psi(r)(\widehat{\psi}(d_1), \dots, \widehat{\psi}(d_k)) \bigr)} \\
        &= \sem[\alg L]{\pi_\Sigma \big( \widehat{\psi}(r(d_1, \dots, d_k)) \bigr)} \enspace. \tag*{\qedhere}
    \end{align*}
\end{proof}

\begin{lemma}\label{lem:intersection-injective}
    For every $d \in \T_R$ it holds that $|\{ d' \in \T_{R'} \mid \widehat{\psi}(d') = d \ \text{and} \ \sem[\alg L]{\pi_\Sigma(d')} = \sem[\alg L]{\pi_\Sigma(d)} \}| \le 1$.
\end{lemma}

\begin{proof}
    The proof is done by structural induction on $d$.
    We assume (IH) that for every $k \in \mathbb N$, $i \in [k]$, and $d_i \in \T_R$ it holds that $|\{ d_i' \in \T_{R'} \mid \widehat{\psi}(d_i') = d_i \ \text{and} \ \sem[\alg L]{\pi_\Sigma(d_i')} = \sem[\alg L]{\pi_\Sigma(d_i)} \}| \le 1$.
    Let $r \in R$ with $\reqrule$.
    Then
    \begin{align*}
        &\{ d' \in \T_{R'} \mid \widehat{\psi}(d') = r(d_1, \dots, d_k) \ \text{and} \ \sem[\alg L]{\pi_\Sigma(d')} = \sem[\alg L]{\pi_\Sigma(r(d_1, \dots, d_k))} \} \\
        &= \{ d' \in \T_{R'} \mid \begin{aligned}[t]
            &\widehat{\psi}(d') = r(d_1, \dots, d_k), \lhs(d'(\varepsilon)) = [A,\sem[\alg L]{\pi_\Sigma(r(d_1, \dots, d_k))}], \ \text{and} \\
            &\sem[\alg L]{\pi_\Sigma(d')} = \sem[\alg L]{\pi_\Sigma(r(d_1, \dots, d_k))} \}
        \end{aligned} \tag{Lemma~\ref{lem:intersection-semantics}} \\
        &= \{ \begin{aligned}[t]
            &\big([A,\sem[\alg L]{\pi_\Sigma(r(d_1, \dots, d_k))}] \to \sigma([A_1,a_1], \dots, [A_k,a_k])\big)(d_1', \dots, d_k') \mid \\
            &(a_1, \dots, a_k) \in \phi(\sigma)^{-1}(\sem[\alg L]{\pi_\Sigma(r(d_1, \dots, d_k))}), d_1' \in (\T_{R'})_{[A_1,a_1]}, \dots, d_k' \in (\T_{R'})_{[A_k,a_k]}, \\
            &\widehat{\psi}(d_1') = d_1, \dots, \widehat{\psi}(d_k') = d_k, \ \text{and} \ \sem[\alg L]{\pi_\Sigma(r'(d_1', \dots, d_k'))} = \sem[\alg L]{\pi_\Sigma(r(d_1, \dots, d_k))} \}
        \end{aligned} \tag{Definition of $P_{R,a}$} \\
        \intertext{(where $r' = \big([A,\sem[\alg L]{\pi_\Sigma(r(d_1, \dots, d_k))}] \to \sigma([A_1,a_1], \dots, [A_k,a_k])\big)$)}
        &= \{ \begin{aligned}[t]
            &\big([A,\sem[\alg L]{\pi_\Sigma(r(d_1, \dots, d_k))}] \to \sigma([A_1,\sem[\alg L]{\pi_\Sigma(d_1)}], \dots, [A_k,\sem[\alg L]{\pi_\Sigma(d_k)}])\big)(d_1', \dots, d_k') \mid \\
            &d_1' \in (\T_{R'})_{[A_1,\sem[\alg L]{\pi_\Sigma(d_1)}]}, \dots, d_k' \in (\T_{R'})_{[A_k,\sem[\alg L]{\pi_\Sigma(d_k)}]}, \\
            &\widehat{\psi}(d_1') = d_1, \dots, \widehat{\psi}(d_k') = d_k, \ \text{and} \ \sem[\alg L]{\pi_\Sigma(d_1')} = \sem[\alg L]{\pi_\Sigma(d_1)}, \dots, \sem[\alg L]{\pi_\Sigma(d_k')} = \sem[\alg L]{\pi_\Sigma(d_k)} \}
        \end{aligned} \tag{Lemma~\ref{lem:psi-and-pisigma}}
    \end{align*}
    This set has at most one element as $|\{ d_i' \in (\T_{R'})_{[A_i,\sem[\alg L]{\pi_\Sigma(d_i)}]} \mid \widehat{\psi}{(d_i')} = d_i \ \text{and} \ \sem[\alg L]{\pi_\Sigma(d_i')} = \sem[\alg L]{\pi_\Sigma(d_i)} \}| \le 1$ for every $i \in [k]$ by (IH).
\end{proof}

For every $d \in \T_R$ we let $R'(d) = \wthom{d}$.

\begin{lemma}\label{lem:intersection-surjective}
    For every $d \in \AST(G, a)$ there is a $d' \in \T_{R'(d)}$ such that $\widehat{\psi}(d') = d$.
\end{lemma}

\begin{proof}
    The proof is done by induction on $d$.
    We assume (IH) that for every $k \in \mathbb N$, $i \in [k]$, and $d_i \in \T_{R'}$ there is a $d_i' \in T_{R'(d_i)}$ such that $\widehat{\psi}(d_i') = d_i$.
    We let $r \in R$ with $\reqrule$ and $d = r(d_1, \dots, d_k)$.
    Then by (IH) for every $i \in [k]$ there is a $d_i' \in T_{R'(d_i)}$ with $\widehat{\psi}(d_i') = d_i$.
    Then by definition of $\omega_r$, $r' \in R'(d)$ with $r' = \big( [A,b] \to \sigma([A_1,a_1], \dots [A_k,a_k]) \big)$, $[A_i,a_i] = \lhs(d_i'(\varepsilon))$ for every $i \in [k]$ and $b = \phi(\sigma)(a_1, \dots, a_k)$.
    Moreover, $d_1', \dots, d_k' \in \T_{R'(d)}$.
    Thus $r'(d_1', \dots, d_k') \in \T_{R'(d)}$.
    Clearly $\widehat{\psi}(r'(d_1', \dots, d_k')) = d$.
\end{proof}

\thmintersection*

\begin{proof}
    First we show that $\widehat{\psi}: \AST(G') \to \AST(G, a)$ is bijective by showing that it is injective and surjective.
    For injectivity let $d_1, d_2 \in \AST(G')$ such that $\widehat{\psi}(d_1) = \widehat{\psi}(d_2)$.
    By Lemma~\ref{lem:intersection-semantics}, $\sem[\alg L]{\pi_\Sigma(d_1)} = a = \sem[\alg L]{\pi_\Sigma(d_2)}$.
    Then by Lemma~\ref{lem:psi-and-pisigma}, there is a $d \in \AST(G, a)$ such that $\widehat{\psi}(d_1) = d = \widehat{\psi}(d_2)$.
    Then by Lemma~\ref{lem:intersection-injective}, $d_1 = d_2$; hence $\widehat{\psi}: \AST(G') \to \AST(G, a)$ is injective.

    For surjectivity, let $d \in \AST(G, a)$.
    Then by Lemma~\ref{lem:intersection-surjective}, there is a $d' \in \T_{R'(d)}$ such that $\widehat{\psi}(d') = d$.
    Since $R' = \bigcup_{d \in \AST(G, a)} R'(d)$, it holds that $d' \in \T_{R'}$.
    Then by definition of $\widehat{\psi}$ and Lemma~\ref{lem:intersection-semantics}, $\lhs(d'(\varepsilon)) = [A_0, a]$ and thus $d' \in \AST(G)$; hence $\widehat{\psi}: \AST(G') \to \AST(G, a)$ is surjective.

    Now we show that $L(G')_{\alg L} = L(G)_{\alg L} \cap \{ a \}$.
    For this, we distinguish two cases.
    \begin{enumerate}
        \item If $a \in L(G)_{\alg L}$,
            \begin{align*}
                L(G')_{\alg L} &= \{ \sem[\alg L]{\pi_\Sigma(d)} \mid d \in \AST(G) \} \\
                &= \{ \sem[\alg L]{\pi_\Sigma(d)} \mid d \in (\T_{R'})_{[A_0,a]}) \} \\
                &= \{ a \mid d \in (\T_{R'})_{[A_0,a]}) \} \tag{Lemma~\ref{lem:intersection-semantics}} \\
                &= \{ a \} \\
                &= L(G)_{\alg L} \cap \{ a \} \enspace. \tag{$a \in L(G)_{\alg L}$}
            \end{align*}
        \item Otherwise $a \not\in L(G)_{\alg L}$, then $\AST(G, a) = \emptyset$.
            Thus, since $\widehat{\psi}: \AST(G') \to \AST(G, a)$ is bijective, $\AST(G') = \emptyset$ as well.
            Consequently
            \[
                L(G')_{\alg L} = \emptyset = L(G)_{\alg L} \cap \{ a \} \enspace. \qedhere
            \]
    \end{enumerate}
\end{proof}

\subsection{ADP algebra is a d-complete and distributive M-monoid}\label{sec:proof-adp-mmonoid}

\lemadpmmonoid*

\begin{proof}
    \begin{sloppypar}
    Let $(\walg{K}', \oplus, \emptyset, \Sigma',\psi',\infsum)$ be the algebra associated with~$\walg{K}$ and~$h$.
    We show that $(\walg{K}', \oplus, \emptyset, \Sigma',\psi',\infsum)$ is a d-complete and distributive M-monoid in three steps.
    We begin by proving that $(\walg{K}', \oplus, \emptyset, \Sigma',\psi')$ is an M-monoid.
    First, $(\walg{K}',\oplus,\emptyset)$ is a commutative monoid, as
    \end{sloppypar}
    \begin{itemize}
        \item $\walg{K}'$ is a set, and $\oplus: \walg{K}' \otimes \walg{K}' \rightarrow \walg{K}'$,
            i.e., for every $F_1,F_2 \in \walg{K}'$ it holds that $F_1 \oplus F_2 \in \walg{K}'$.
            For the proof of this claim, let $F_1,F_2 \in \walg{K}'$.
            Now we distinguish two cases:
            \begin{enumerate}
                \item If there is an $s \in S$ such that $F_1,F_2 \subseteq \walg{K}_s$, then $F_1 \oplus F_2 = h_s(F_1 \cup F_2)$.
                    Obviously $F_1 \cup F_2 \subseteq \walg{K}_s$ and then by the definition of~$\walg{K}'$, $h_s(F_1 \cup F_2) \in \walg{K}'$.
                \item Otherwise $F_1 \oplus F_2 = \bot$ and $\bot \in \walg{K}'$ by definition.
            \end{enumerate}
        \item Commutativity of~$\oplus$ easily follows from the commutativity of~$\cup$.
       \item We show that $\oplus$ is associative, i.e., for every $F_1,F_2,F_3 \in \walg{K}'$ it holds that $(F_1 \oplus F_2) \oplus F_3 = F_1 \oplus (F_2 \oplus F_3)$, by the following case analysis.
            Let $F_1,F_2,F_3 \in \walg{K}'$.
            Now either
            \begin{enumerate}
                \item there is an $s \in S$ such that $F_1,F_2 \subseteq \walg{K}_s$, then,
                    \begin{enumerate}
                        \item if also $F_3 \subseteq \walg{K}_s$, then
                            \begin{align*}
                                (F_1 \oplus F_2) \oplus F_3 &= h_s(h_s(F_1 \cup F_2) \cup F_3) \\
                                &= h_s(h_s(F_1 \cup F_2) \cup h_s(F_3)) \tag{$h$ is idempotent} \\
                                &= h_s((F_1 \cup F_2) \cup F_3) \tag{Equation~\ref{eq:obj-function'}} \\
                                &= h_s(F_1 \cup (F_2 \cup F_3)) \tag{$\cup$ is associative} \\
                                &= h_s(F_1 \cup h_s(F_2 \cup F_3)) \tag{$h$ is idempotent} \\
                                &= F_1 \oplus (F_2 \oplus F_3) \enspace,
                            \end{align*}
                            or,
                        \item if $F_3 \not\subseteq \walg{K}_s$, then $(F_1 \oplus F_2) \oplus F_3 = \bot = F_2 \oplus F_3 = F_1 \oplus (F_2 \oplus F_3)$,
                    \end{enumerate}
                    or
                \item there is no such $s \in S$ and hence $F_1 \oplus F_2 = \bot = (F_1 \oplus F_2) \oplus F_3$.
                    Now it may be that there is an $s' \in S$ such that $F_2,F_3 \subseteq \walg{K}_{s'}$,
                    then $F_2 \oplus F_3 \subseteq \walg{K}_{s'}$, but $F_1 \not \subseteq \walg{K}_{s'}$ and hence $F_1 \oplus (F_2 \oplus F_3) = \bot$.
                    Otherwise $F_2 \oplus F_3 = \bot$ and hence $F_1 \oplus (F_2 \oplus F_3) = \bot$.
                    (We note that, if we had not added~$\bot$ to~$\walg{K}'$ and still chosen $\emptyset$ as the identity element, then in this case~$\oplus$ would not be associative.)
            \end{enumerate}
        \item As $\emptyset \subseteq \walg{K}_s$ for any $s \in S$, we have $\emptyset \in \walg{K}'$.
            We show that~$\emptyset$ is the identity element by showing that $\emptyset \oplus F = F$ for every $F \in \walg{K}'$.
            Then the other condition, $F \oplus \emptyset = F$, will follow from the commutativity of~$\oplus$.
            Let $F \in \walg{K}'$.
            Again, we distinguish two cases:
            \begin{enumerate}
                \item If $F = \bot$, then $\emptyset \oplus \bot = \bot$ by definition.
                \item Otherwise there is an $s \in S$ such that $F \subseteq \walg{K}_s$.
                    Since $\emptyset \subseteq \walg{K}_s$, we have that $\emptyset \oplus F = h_s(\emptyset \cup F) = h_s(F)$ and as~$h$ is idempotent, $h_s(F) = F$.
            \end{enumerate}
    \end{itemize}
    Second, $(\walg{K}',\psi')$ is a $\Sigma'$-algebra as~$\Sigma'$ is a ranked set and $\psi'(\sigma)(F_1,\dots,F_k) \in \walg{K}'$ for every $k \in \mathbb N$, $\sigma \in \Sigma'_k$, and $F_1,\dots,F_k \in \walg{K}'$ which we show by the following case analysis.
    Let $k \in \mathbb N$, $\sigma \in \Sigma'_k$, and $F_1,\dots,F_k \in \walg{K}'$.
    Then:
    \begin{enumerate}
        \item If $\sigma = t$ with $t \in (\T_\Sigma(X_{s_1 \dots s_k}))_s$, then there are two possibilities:
            \begin{enumerate}
                \item If $F_i \subseteq \walg{K}_{s_i}$ for every $i \in [k]$, then
                    \begin{align*}
                        \psi'(\sigma)(F_1,\dots,F_k) &= h_s(t_{\walg{K}}(F_1,\dots,F_k)) \\
                        &= h_s\big(\{ t_{\walg{K}}(a_1,\dots,a_k) \mid a_1 \in F_1,\dots,a_k \in F_k \}\big)
                    \end{align*}
                    \begin{sloppypar}
                    and by definition of~$t_{\walg{K}}$, $t_{\walg{K}}(a_1,\dots,a_k) \in \walg{K}_s$ for every $a_1 \in F_1,\dots,a_k \in F_k$.
                    Hence $t_{\walg{K}}(F_1,\dots,F_k) \subseteq \walg{K}_s$ and by definition of~$\walg{K}'$, $h_s\big(t_{\walg{K}}(F_1,\dots,F_k)\big) \in \walg{K}'$.
                    \end{sloppypar}
                \item Otherwise $\psi'(\sigma)(F_1,\dots,F_k) = \bot \in \walg{K}'$.
            \end{enumerate}
        \item If $\sigma = \welem{0}^k$, then $\psi'(\sigma)(F_1,\dots,F_k) = \emptyset \in \walg{K}'$.
    \end{enumerate}
    Finally, $\welem{0}^k \in \Sigma'$ and $\psi'(\welem{0}^k)(F_1,\dots,F_k) = \emptyset$ for every $k \in \mathbb N$ and $F_1,\dots,F_k \in \walg{K}'$ by definition.

    The operation~$\infsum$ fulfils the axioms of an infinitary sum operation on~$\walg{K}'$, as the following case analysis shows.
    Let  $(F_i \mid i \in I)$ be an $I$-indexed family over $\walg{K}'$. Then:
    \begin{itemize}
        \item if $I=\emptyset$, then $\infsum_{i \in \emptyset} F_i = \emptyset$,
        \item if $I = \{n\}$ and $F_n \in \walg{K}'$, then it holds that either
            \begin{enumerate}
                \item $F_n = \bot$, then $\infsum_{i \in \{ n \}} F_i = \bot$, or
                \item $F_n \subseteq \walg{K}_s$ for some $s \in S$, then $\infsum_{i \in \{ n \}} F_i = h_s(F_n) = F_n$,
            \end{enumerate}
        \item if $I = \{m,n\}$ with $m \not= n$ and $F_m,F_n \in \walg{K}'$, then it holds that either
            \begin{enumerate}
                \item there is an $s \in S$ such that $F_m,F_n \subseteq \walg{K}_s$, then
                    \[ \infsum_{i \in \{ m,n \}} F_i = h_s \Big( \bigcup_{i \in \{ m,n \}} F_i \Big) = h_s(F_m \cup F_n) = F_m \oplus F_n \enspace, \text{ or} \]
                \item otherwise $\infsum_{i \in \{ m,n \}} F_i = \bot = F_m \oplus F_n$,
            \end{enumerate}
        \item for every $J$-partition of~$I$ it holds that either
            \begin{enumerate}
                \item there is an $s \in S$ such that $F_i \subseteq \walg{K}_s$ for every $i \in I$, then
                    \begin{align*}
                        \infsum_{i \in I} F_i = h_s\big(\bigcup_{i \in I} F_i\big) &= h_s\Big(\bigcup_{j \in J} \big(\bigcup_{i \in I_j} F_i\big)\Big) \\
                        &= h_s\Big(\bigcup_{j \in J} h_s\big(\bigcup_{i \in I_j} F_i\big)\Big) \tag{Equation~\ref{eq:obj-function'}} \\
                        &= h_s\Big(\bigcup_{j \in J} \big(\infsum_{i \in I_j} F_i\big) \Big) = \infsum_{j \in J} \big( \infsum_{i \in I_j} F_i \big) \enspace, \text{or}
                    \end{align*}
                \item there is an $i' \in I$ such that $F_{i'} = \bot$; then there is a $j' \in J$ such that $i' \in I_{j'}$, hence $\infsum_{i \in I_{j'}} F_i = \bot$ and thus
                    \[ \infsum_{j \in J} \big( \infsum_{i \in I_j} F_i \big) = \bot = \infsum_{i \in I} F_i \enspace, \text{or} \]
                \item there are $i_1,i_2 \in I$ such that $F_{i_1} \subseteq \walg{K}_{s_1}$ and $F_{i_2} \subseteq \walg{K}_{s_2}$ with $s_1,s_2 \in S$ and $s_1 \not= s_2$.
                    Now we have to distinguish two cases:
                    \begin{enumerate}
                        \item there is a $j' \in J$ such that $i_1,i_2 \in I_{j'}$, then $\infsum_{i \in I_{j'}} F_i = \bot$ and hence
                            \[ \infsum_{j \in J} \big( \infsum_{i \in I_j} F_i \big) = \bot = \infsum_{i \in I} F_i \enspace, \]
                        \item $i_1 \in I_{j_1}$ and $i_2 \in I_{j_2}$ with $j_1,j_2 \in J$ and $j_1 \not= j_2$.
                            Then there are $F_1 \subseteq \walg{K}_{s_1}$ and $F_2 \subseteq \walg{K}_{s_2}$ such that  $\infsum_{i \in I_{j_1}} F_i = F_1$ and $\infsum_{i \in I_{j_2}} F_i = F_2$, but since $s_1 \not= s_2$
                            \[ \infsum_{j \in J} \big( \infsum_{i \in I_j} F_i \big) = \bot = \infsum_{i \in I} F_i \enspace. \]
                    \end{enumerate}
            \end{enumerate}
    \end{itemize}

    In order to show that~$\walg K'$ is even d-complete, we let $F \in \walg k'$ and $(F_i \mid i \in I)$ be an $I$-indexed family over~$\walg K'$ such that for every $i \in I$, $F \oplus F_i = F$.
    Then, by definition of~$\oplus$, we have to distinguish two cases.
    \begin{enumerate}
        \item If there is an $s \in S$ such that $F \subseteq \walg K_s$ and for every $i \in I$, $F_i \subseteq \walg K_s$, then for every $i \in I$, $h_s(F \cup F_i) = F$.
            Thus
            \begin{align*}
                F \oplus \infsum_{i \in I} F_i &= h_s \left( F \cup \infsum_{i \in I} F_i \right) \\
                &= h_s \left( F \cup h_s \left( \bigcup_{i \in I} F_i \right) \right) \\
                &= h_s \left( h_s(F) \cup h_s \left( \bigcup_{i \in I} F_i \right) \right) \tag{$F \in \walg K'$} \\
                &= h_s \left( F \cup \bigcup_{i \in I} F_i \right) \tag{Equation~\ref{eq:obj-function'}} \\
                &= h_s \left( \bigcup_{i \in I} \, (F \cup F_i) \right) \tag{$\bigcup$ is idempotent} \\
                &= h_s \left( \bigcup_{i \in I} h_s(F \cup F_i) \right) \tag{Equation~\ref{eq:obj-function'}} \\
                &= h_s \left( \bigcup_{i \in I} F \right) \\
                &= h_s(F) \tag{$\bigcup$ is idempotent} \\
                &= F \enspace. \tag{$F \in \walg K'$}
            \end{align*}
        \item Otherwise, there is an $i \in I$ such that $F \oplus F_i = \bot$.
            But then also $F = \bot$.
            Thus, by definition of~$\oplus$, $F \oplus \infsum_{i \in I} F_i = \bot$.
    \end{enumerate}

    Distributivity of~$\walg{K}'$ is implied by the fact that $h$ satisfies Bellman's principle of optimality, which we will show next.
    Let $k \in \mathbb N$, $s,s_1,\dots,s_k \in S$, $\sigma \in \Sigma'_k$, $F_1,\dots,F_k,F' \in \walg{K}'$, and $i \in [k]$.
    We consider two cases; first, assume that $\sigma = t$ with $t \in (\T_\Sigma(X_{s_1 \dots s_k}))_s$.
    Now there are two possibilities:

    \begin{enumerate}
        \item If $F_{i'} \subseteq \walg{K}_{s_{i'}}$ for every $i' \in [k]$ and $F' \in \walg K_{s_i}$, then
            \begin{align*}
                &\psi'(\sigma)\Big(F_1,\dots,F_{i-1},F_i \oplus F',F_{i+1},\dots,F_k\Big) \\
                &= \psi'(\sigma)\Big(F_1,\dots,F_{i-1},h_{s_i}\big(F_i \cup F'\big),F_{i+1},\dots,F_k\Big) \\
                &= h_s\Big({t_{\walg{K}}}\big(F_1,\dots,F_{i-1},h_{s_i}\big(F_i \cup F'\big),F_{i+1},\dots,F_k\big)\Big) \\
                &= h_s\Big({t_{\walg{K}}}\big(h_{s_1}(F_1),\dots,h_{s_{i-1}}(F_{i-1}),h_{s_i}\big(F_i \cup F'\big),h_{s_{i+1}}(F_{i+1}),\dots,h_{s_k}(F_k)\big)\Big) \tag{$h$ is idempotent} \\
                &= h_s\Big({t_{\walg{K}}}\big(F_1,\dots,F_{i-1},F_i \cup F',F_{i+1},\dots,F_k\big)\Big) \tag{Equation~\ref{eq:bellman'}} \\
                &= h_s\Big({t_{\walg{K}}}(F_1,\dots,F_{i-1},F_i,F_{i+1},\dots,F_k) \cup {t_{\walg{K}}}(F_1,\dots,F_{i-1},F',F_{i+1},\dots,F_k)\Big) \\
                &= h_s\Big(h_s\big({t_{\walg{K}}}(F_1,\dots,F_{i-1},F_i,F_{i+1},\dots,F_k)\big) \cup h_s\big({t_{\walg{K}}}(F_1,\dots,F_{i-1},F',F_{i+1},\dots,F_k)\big)\Big) \tag{Equation~\ref{eq:obj-function'}} \\
                &= h_s\Big(\psi'(\sigma)(F_1,\dots,F_{i-1},F_i,F_{i+1},\dots,F_k) \cup \psi'(\sigma)(F_1,\dots,F_{i-1},F',F_{i+1},\dots,F_k)\Big) \\
                &= \psi'(\sigma)(F_1,\dots,F_{i-1},F_i,F_{i+1},\dots,F_k) \oplus \psi'(\sigma)(F_1,\dots,F_{i-1},F',F_{i+1},\dots,F_k) \enspace.
            \end{align*}
        \item If there is an $i' \in [k]$ such that $F_{i'} \not\subseteq \walg{K}_{s_{i'}}$ or $F' \not\subseteq \walg{K}_{s_i}$, then
            \begin{align*}
                \psi'(\sigma)\Big(F_1,\dots,F_{i-1},F_i \oplus F',F_{i+1},\dots,F_k\Big) &= \bot \\
                \intertext{and furthermore}
                \psi'(\sigma)(F_1,\dots,F_{i-1},F_i,F_{i+1},\dots,F_k) &= \bot
                \intertext{or}
                \psi'(\sigma)(F_1,\dots,F_{i-1},F',F_{i+1},\dots,F_k) &= \bot \enspace.
            \end{align*}
            Hence
            \begin{align*} &\psi'(\sigma)(F_1,\dots,F_{i-1},F_i,F_{i+1},\dots,F_k) \oplus \psi'(\sigma)(F_1,\dots,F_{i-1},F',F_{i+1},\dots,F_k) = \bot \enspace.
              \end{align*}
    \end{enumerate}

    Second, assume that $\sigma = \welem{0}^k$.
    Then
    \begin{align*}
        &\psi'(\sigma)\Big(F_1,\dots,F_{i-1},F_i \oplus F',F_{i+1},\dots,F_k\Big)
        = \emptyset
        = \emptyset \cup \emptyset
        = h_s\big(\emptyset \cup \emptyset\big) \\
        &= h_s\big(\psi'(\sigma)(F_1,\dots,F_{i-1},F_i,F_{i+1},\dots,F_k) \cup \psi'(\sigma)(F_1,\dots,F_{i-1},F',F_{i+1},\dots,F_k)\big) \\
        &= \psi'(\sigma)(F_1,\dots,F_{i-1},F_i,F_{i+1},\dots,F_k) \oplus \psi'(\sigma)(F_1,\dots,F_{i-1},F',F_{i+1},\dots,F_k) \enspace. \qedhere
    \end{align*}
\end{proof}

\subsection{ADP is an instance of the M-monoid parsing problem}\label{app:adp-mmonoid-parsing}

This subappendix contains the full proof of Theorem~\ref{thm:ADP-M-monoid}.
We start with an auxiliary statement.

\begin{lemma}\label{lem:ast-mmonoid-is-tree-adp}
    For every $d \in \T_R$ it holds that $\sem[\walg K']{\wt(d)} = \{ \pi_\Sigma(d)_{\walg{K}} \}$.
\end{lemma}

\begin{proof}
    We prove the statement of the lemma by structural induction over~$d$.

    For the induction base, let $d = (A \to t)$ in~$R$.
    Then $d \in R_{(\varepsilon,A)}$ and hence $t \in \T_\Sigma$.
    (We recall that $\T_\Sigma = \T_\Sigma(X_\varepsilon)$.)
    Now for both cases, $t \in (\T_\Sigma)_\sans$ or $t \in (\T_\Sigma)_\sinp$, the proof of $\sem[\walg K']{\wt(d)} = \{ \pi_\Sigma(d)_{\walg{K}} \}$ is the same.
    Thus, for every $s \in \{\sinp,\sans\}$, we have
    \[ \sem[\walg K']{\wt(d)} = \psi'(t) = h_s\big(\{t_{\walg K}\}\big) = \{t_{\walg K}\} = \{ \pi_\Sigma(d)_{\walg K} \} \enspace. \]

    For the induction step, we let $d \in \T_R$ be of the form $r(d_1,\dots,d_k)$ for some $k \in \mathbb N$ with $r = (A \to t)$ in~$R$.
    Then there are $A_1,\dots,A_k \in N$ such that $r \in R_{(A_1 \dots A_k,A)}$ and $d_i \in (\T_R)_{A_i}$ for each $i \in [k]$.
    We assume (IH) that for every $i \in [k]$, $\sem[\walg K']{\wt(d_i)} = \{ \pi_\Sigma(d_i)_{\walg{K}} \}$.
    Furthermore, let~$t'$ be obtained from~$t$ by replacing the $i$th occurrence of a nonterminal in~$t$ by~$x_i$ for every $i \in [k]$.
    Again, for both cases, $t \in (\T_\Sigma(N))_\sans$ or $t \in (\T_\Sigma(N))_\sinp$, the proof of $\sem[\walg K']{\wt(d)} = \{ \pi_\Sigma(d)_{\walg{K}} \}$ is the same.
    Thus, for every $s \in \{ \sinp,\sans \}$, we have
    \begin{align*}
        \sem[\walg K']{\wt(r(d_1,\dots,d_k))} &= \psi'(t')(\sem[\walg K']{\wt(d_1)},\dots,\sem[\walg K']{\wt(d_k)}) \enspace, \\
        \intertext{now for every $i \in [k]$ and $r' = (A_i \to t'')$ in~$R$ we have that $\sort(A_i) = \sort(t'')$; thus $\sem[\walg K']{\wt(d_i)} \subseteq \walg K_{\sort(A_i)}$ and we can continue with:}
        &= h_s\big(t'_{\walg{K}}(\sem[\walg K']{\wt(d_1)},\dots,\sem[\walg K']{\wt(d_k)})\big) \\
        &= h_s\big(t'_{\walg{K}}(\{\pi_\Sigma(d_1)_{\walg{K}}\},\dots,\{\pi_\Sigma(d_k)_{\walg{K}}\})\big) \tag{IH} \\
        &= \{t'_{\walg{K}}(\pi_\Sigma(d_1)_{\walg{K}},\dots,\pi_\Sigma(d_k)_{\walg{K}})\} \\
        &= \{\big(t'_{\T_\Sigma}(\pi_\Sigma(d_1),\dots,\pi_\Sigma(d_k))\big)_{\walg{K}} \} \tag{Observation~\ref{obs:tree-derived-operations}} \\
        &= \{\pi_\Sigma(r(d_1,\dots,d_k))_{\walg{K}}\} \enspace. \tag*{\qedhere}
    \end{align*}
\end{proof}

Now we are able to prove Theorem~\ref{thm:ADP-M-monoid}.

\thmadpmmonoid*

\begin{proof}
    Let $(G,(\lalg{YIELD}^\Sigma,\phi))$ with $G = (N,\Sigma,A_0,R)$ be an $S$-sorted yield grammar over $\Sigma$,
    $(\walg{K},\psi)$ be an $S$-sorted $\Sigma$-algebra,
    and~$h$ be an objective function for~$\walg{K}$ that satisfies Bellman's principle of optimality.
    Moreover, let
    \[ ((G,(\lalg{YIELD}^\Sigma,\phi)), (\walg{K}', \oplus, \emptyset, \Sigma',\infsum), \wt) \]
    be the wRTG-LM constructed as in Theorem~\ref{thm:ADP-M-monoid} and $w \in (\Sigma_{(\varepsilon,\sinp)})^*$.
    In this proof, we write $\yield$ rather than $\yield_{\Sigma_{(\varepsilon,\sinp)}}$ for the sake of readability.
    \begin{align*}
        \fparse(w) &= \infsum_{d \in (\T_R)_{A_0}: \, \sem[\lalg{YIELD}^\Sigma]{\pi_\Sigma(d)} \, = \, \langle w, \sans \rangle} \sem[\walg K']{\wt(d)} \\
        &= \infsum_{d \in (\T_R)_{A_0}: \yield(\pi_\Sigma(d)) = w} \sem[\walg K']{\wt(d)} \\
        &= \infsum_{d \in \pi_\Sigma^{-1}(L(G) \cap \yield^{-1}(w))} \sem[\walg K']{\wt(d)} \\
        &= \infsum_{t \in L(G) \cap \yield^{-1}(w)} \sem[\walg K']{\wt(\pi_\Sigma^{-1}(t))} \tag{$G$ is unambiguous} \\
        &= \infsum_{t \in L(G) \cap \yield^{-1}(w)} \{ \pi_\Sigma(\pi_\Sigma^{-1}(t))_{\walg{K}} \} \tag{Lemma~\ref{lem:ast-mmonoid-is-tree-adp}} \\
        &= \infsum_{t \in L(G) \cap \yield^{-1}(w)} \{ t_{\walg{K}} \} \\
        &= h_\sans \left( \bigcup_{t \in L(G) \cap \yield^{-1}(w)} \{ \sem[\walg K]{t} \} \right) \tag{$\sort(t) = \sans$ for every $t \in L(G)$} \\
        &= \adp(w) \enspace. \tag*{\qedhere}
    \end{align*}
\end{proof}

\subsection{Each weight-preserving weighted deduction system is sound and complete}\label{app:weight-preserving-wds}

\lemwdspreserving*

\begin{proof}
    \begin{sloppypar}
    Let $\overline G = ((G,\alg L),\walg K,\wt)$ in $\wlmclass{\gclass{\alg L},\walg{K}}$ with $G = (N,\Sigma,A_0,R)$, $a \in \alg L_{\sort(A_0)}$, $\wds_{\walg K,\walg K}: \wlmclass{\gclass{\alg L},\walg{K}} \times \alg L \to \wlmclass{\gclass{\cfges},\walg K}$ be a weight-preserving weighted deduction system, and $\wds_{\walg K,\walg K}(\overline G,a) = ((G',\lalg{CFG}^\emptyset),\walg K,\wt')$.
    \end{sloppypar}

    If $\varepsilon \in \sem[\lalg{CFG}^\emptyset]{L(G')}$, then there is a $d \in (\T_{R'})_{A_0'}$ such that $\sem[\lalg{CFG}^\emptyset]{\pi_\Sigma(d)} = \varepsilon$.
    Then $\psi^{-1}(d) \in (\T_R)_{A_0}$ and $\sem{\pi_\Sigma(\psi^{-1}(d))} = a$.
    Thus $a \in \sem{L(G)}$ and $\wds_{\walg K,\walg K}$ is sound.
    If $a \in \sem{L(G)}$, then there is a $d \in (\T_R)_{A_0}$ such that $\sem{\pi_\Sigma(d)} = a$.
    Then $\psi(d) \in (\T_{R'})_{A_0'}$.
    Since $\sem[\lalg{CFG}^\emptyset]{L(G')} \subseteq \{ \varepsilon \}$ we have that $\sem[\lalg{CFG}^\emptyset]{\pi_\Sigma(d)} = \varepsilon$;
    thus $\varepsilon \in \sem[\lalg{CFG}^\emptyset]{L(G')}$ and $\wds_{\walg K,\walg K}$ is complete.
\end{proof}

\subsection{The canonical weighted deduction system is weight-preserving}\label{app:cnc-weight-preserving}

\lemcncwp*

\begin{proof}
    Let $\overline G = ((G, \alg L), \walg K, \wt)$ in $\wlmclass{\gclass{\alg L}, \walg K}$ with $G = (N,\Sigma,A_0,R)$, $a \in \alg L_{\sort(A_0)}$, and $\cnc(\overline G, a) = ((G', \cfges), \walg K, \wt')$ with $G' = (N',\Sigma',A_0',R')$.
    Next we will define the mapping $\psi: \AST(G, a) \to \AST(G')$ according to the definition of weight-preserving mappings.
    For this, we first define the auxiliary mapping
    \[
        \psi': \{ d \in \T_R \mid \sem[\alg L]{\pi_\Sigma(d)} \in \factors(a) \} \to \T_{R'}
    \]
    by induction (which is not possible for $\psi$). Let $d \in \T_R$ with $\sem[\alg L]{\pi_\Sigma(d)} \in \factors(a)$.
    If
    \begin{itemize}
        \item $d$ has the form $r(d_1,\dots,d_k)$ with $r = (A \to t)$ with $\yield_N(t) = A_1 \dots A_k$, $k \in \mathbb N$ and $A_1,\dots,A_k \in N$,
        \item for every $i \in [k]$, we have $a_i = \sem[\alg L]{\pi_\Sigma(d_i)}$, and
        \item for every $i \in [k]$ there is a $t_i \in \T_\Sigma(N)$ such that $t_i$ is the right-hand side of the rule $d_i(\varepsilon)$,
    \end{itemize}
    then we let
    \begin{align*}
        \psi'(d) &= r'(\psi'(d_1),\dots,\psi'(d_k)) \enspace, \text{ where }\\
        r' &= \left([A,t,t'_{\alg L}(a_1,\dots,a_k)] \to \langle x_1 \dots x_k \rangle ([A_1,t_1,a_1],\dots,[A_k,t_k,a_k])\right)
    \end{align*}
    and~$t'$ is obtained from~$t$ by replacing the $i$th occurrence of a nonterminal by~$x_i$ for every $i \in [k]$.

    It can be seen that, for every $d \in \T_R$ with $\sem[\alg L]{\pi_\Sigma(d)} \in \factors(a)$, the sets $\pos(d)$ and $\pos(\psi'(d))$ are equal, and that the mapping~$\psi'$ is bijective.

    Next we define the mapping
    \[
        \psi: \{ d \in (\T_R)_{A_0} \mid \sem[\alg L]{\pi_\Sigma(d)} = a \} \to (\T_{R'})_{A_0'}
    \]
    for each $d \in (\T_R)_{A_0}$ of the form $r(d_1,\dots,d_k)$ with $r = (A_0 \to t)$ and $\sem[\alg L]{\pi_\Sigma(d)} = a$ by
    \[ \psi(d) = \big([A_0,a] \to \langle x_1 \rangle ([A_0,t,a])\big)\Big(\psi'(d)\Big) \enspace. \]
    Then~$\psi$ is bijective, too, and we have that $\pos(\psi(d)) = \{ \varepsilon \} \cup \{ 1 \} \circ \pos(d)$ for every $d \in (\T_R)_{A_0}$ with $\sem[\alg L]{\pi_\Sigma(d)} = a$.

    By the definition of $\cnc$, for every $d \in (\T_R)_{A_0}$ with $\sem[\alg L]{\pi_\Sigma(d')} = a$ and for every  $p \in \pos(d)$ it holds that $\wt(d(p)) = \wt(\psi(d)(1p))$ and $\wt(\psi(d)(\varepsilon)) = \id(\walg K)$.
    Since $\pos(\psi(d)) = \{ \varepsilon \} \cup \{ 1 \} \circ \pos(d)$, we have that $\wthom{d} = \wthom{\psi(d)}$.
    Thus $\cnc$ is weight-preserving. By Lemma~\ref{lem:weight-eq-implies-sound-and-complete} it is also sound and complete.
\end{proof}

\subsection{Applying the canonical weighted deduction system to nonlooping wRTG-LMs yields acyclic wRTG-LMs}\label{app:nl-cnc-acyc}

This subappendix contains the full proof of Lemma~\ref{lem:no-loops-cnc-acyclic}.
We start with an auxiliary statement.

\begin{lemma}\label{lem:cnc-semantics}
    \sloppy
    For every wRTG-LM $\overline G = \big((G, \alg L), \walg K, \wt\big)$ with $G = (N, \Sigma, A_0, R)$, $a_0 \in \alg L$, $\big((G', \cfges), \walg K, \wt'\big) = \cnc(\overline G, a_0)$ with $G' = (N', \Sigma', A_0', R')$, and $d \in \T_{R'}$ of the form $r(d_1, \dots, d_k)$ with $r = \big([A,t,b] \to \langle x_1 \dots x_k \rangle ([A_1,t_1,a_1], \dots, [A_k,t_k,a_k])\big)$ the following holds:
    $\psi^{-1}(d)_{\alg L} = b$, where~$\psi$ is defined as in the proof of Lemma~\ref{lem:cnc-weight-preserving}.
\end{lemma}

\begin{proof}
    Let $\overline G = \big((G, \alg L), \walg K, \wt\big)$ with $G = (N, \Sigma, A_0, R)$ be a wRTG-LM, $a_0 \in \alg L$, $\big((G', \cfges), \walg K, \wt'\big) = \cnc(\overline G, a_0)$ with $G' = (N', \Sigma', A_0', R')$, and $d \in \T_{R'}$ of the form $r(d_1, \dots, d_k)$ with $r = ([A,t,b] \to \langle x_1 \dots x_k \rangle ([A_1,t_1,a_1], \dots, [A_k,t_k,a_k]))$.
    We show the statement of the lemma by structural induction on~$d$.
    For the induction step, assume that for every $i \in [k]$ and $d_i \in \T_{R'}$ of the form $r_i(d_{i,1}, \dots, d_{i,k_i})$ with $r_i = ([A_i,t_i,a_i] \to \langle x_1 \dots x_{k_i} \rangle ([A_{i,1},t_{i,1},a_{i,1}], \dots, [A_{i,k_i},t_{i,k_i},a_{i,k_i}]))$ the following holds:
    $\psi^{-1}(d_i)_{\alg L} = a_i$.
    Then
    \begin{align*}
        \psi^{-1}(d)_{\alg L} &= t'_{\alg L}(\psi^{-1}(d_1)_{\alg L}, \dots, \psi^{-1}(d_k)_{\alg L}) \\
        &= t'_{\alg L}(a_1, \dots, a_k) \tag{IH} \\
        &= b \enspace, \tag{definition of $\cnc$}
    \end{align*}
    where~$t'$ is obtained from~$t$ by replacing the $i$th occurrence of a nonterminal by~$x_i$ for every $i \in [k]$.
\end{proof}

\lemnlcncacyc*

\begin{proof}
    Let $\overline G = \big((G, \alg L), \walg K, \wt\big)$ in $\wlmclass{\gclass{nl} \cap \gclass{\findc}, \wclass{all}}$ with $G = (N, \Sigma, A_0, R)$.
    Then for every $d \in \T_R$ and $p, p' \in \pos(d)$ it holds that $d(p) = d(p')$ and $(d|_p)_{\alg L} = (d|_{p'})_{\alg L}$ imply $p = p'$.
    We give an indirect proof for the lemma.
    For this, let $a_0 \in \alg L$ and assume that $\cnc(\overline G, a_0) \not\in \wlmclass{\gclass{acyc}, \wclass{all}}$.
    Let $\cnc(\overline G, a_0) = \big((G', \cfges), \walg K, \wt'\big)$ with $G' = (N', \Sigma', A_0', R')$.
    Then there is a $d \in \T_{R'}$ which is not acyclic, i.e., there is a leaf $p \in \pos(d)$ such that~$p$ is cyclic.
    Thus there is are $i, j \in [|p|]$ with $i < j$ such that $d(p_i) = d(p_j)$.
    By definition of $\cnc$, $[A_0,a_0]$ does not occur in the right-hand side of any rule in~$R'$, hence $p_i \not= \varepsilon$ and $\lhs(d(p_i)) = [A,t,a]$ for some $A \in N'$, $t \in \T_{\Sigma'}(N')$, and $a \in \alg L$.
    Then, by Lemma~\ref{lem:cnc-semantics}, $\psi^{-1}(d|_{p_i})_{\alg L} = \psi^{-1}(d|_{p_j})_{\alg L} = a$, where~$\psi$ is defined as in the proof of Lemma~\ref{lem:cnc-weight-preserving}.
    Let $d' = \psi^{-1}(d)$; we remark that $d' \in \T_R$.
    Now there are $p, p' \in \pos(d')$ such that $p \not= p'$, $d'(p) = d'(p')$, and $(d'|_p)_{\alg L} = (d'|_{p'})_{\alg L}$, which contradicts the definition of~$\overline G$.
\end{proof}

\subsection{General statements about Algorithm~\ref{alg:mmonoid}}\label{app:vca-general}

\lemvisbigsum*

\begin{proof}
    The proof is done by induction on~$n$.
    For the induction base let $n = 0$.
    Then for every $A \in N'$ we have that
    \[ V_0(A) \overset{\text{Line~\ref{l:init-v}}}= \welem 0 = \bigoplus_{d \in \emptyset} \wtphom{d} \overset{\text{Line~\ref{l:init-v}}}= \bigoplus_{d \in \mathcal V_0(A)} \wtphom{d} \enspace. \]

    For the induction step, let $n \in \mathbb N$.
    We assume (IH) that for every $A \in N'$ it holds that $V_n(A) = \bigoplus_{d \in \mathcal V_n(A)} \wtphom{d}$.
    Then for $\select_n = A$,
    \begin{align*}
        V_{n+1}(A) &= \bigoplus_{\substack{\ruleindex}} \wt'(r)\big(V_n(A_1), \dots, V_n(A_k)\big)
        \tag{Observation~\ref{obs:v-nplus1}} \\
        &= \bigoplus_{\substack{\ruleindex}} \wt'(r)\left(\bigoplus_{d_1 \in \mathcal V_n(A_1)} \wtphom{d_1}, \dots, \bigoplus_{d_k \in \mathcal V_n(A_k)} \wtphom{d_k}\right)
        \tag{IH} \\
        &= \bigoplus_{\substack{\ruleindex}} \quad \bigoplus_{d_1 \in \mathcal V_n(A_1), \dots, d_k \in \mathcal V_n(A_k)} \wt'(r)\left(\wtphom{d_1}, \dots, \wtphom{d_k}\right)
        \tag{$\wt'(r)$ distributes over $\oplus$} \\
        &= \bigoplus_{\substack{\ruleindex \\ d_1 \in \mathcal V_n(A_1), \dots, d_k \in \mathcal V_n(A_k), \\ d = r(d_1, \dots, d_k)}} \wtphom{d}
        \\
        &= \bigoplus_{d \in \mathcal V_{n+1}(A)} \wtphom{d} \enspace,
        \tag{Observation~\ref{obs:v-nplus1}}
    \end{align*}
    and for every $A' \in N' \setminus \{ A \}$,
    \begin{align*}
        V_{n+1}(A') &= V_n(A')
        \tag{Observation~\ref{obs:v-nplus1}} \\
        &= \bigoplus_{d \in \mathcal V_n(A')} \wtphom{d}
        \tag{IH} \\
        &= \bigoplus_{d \in \mathcal V_{n+1}(A')} \wtphom{d} \enspace.
        \tag*{(Observation~\ref{obs:v-nplus1}) \qedhere}
    \end{align*}
\end{proof}

\lemmcvmonotone*

\begin{proof}
    Let $n \in \mathbb N$, $A \in N'$, and $d \in \mathcal V_n(A)$.
    We show that $d \in \mathcal V_{n'}(A)$ for each $n \in \mathbb N$ with $n' > n$ by structural induction on~$d$.
    For the induction base let $d \in R'$, then by lines~\ref{l:init-vnew}--\ref{l:update-vnew} and~\ref{l:update-v} we have that $d \in \mathcal V_n(\lhs(d))$ for every $n \in \mathbb N_+$.
    Furthermore, by line~\ref{l:init-v}, $\mathcal V_0(A) = \emptyset$ for every $A \in N'$.
    Therefore the implication holds.

    For the induction step, let $d = r(d_1, \dots, d_k)$ and $\reqrule$ with $k > 0$.
    We assume (IH) that for every $i \in [k]$, $n, n' \in \mathbb N$, and $d_i \in (\T_{R'})_{A_i}$ with $n' > n$ the following holds: if $d_i \in \mathcal V_n(A_i)$, then $d_i \in \mathcal V_{n'}(A_i)$.
    Now, if $d \in \mathcal V_n(A)$, then there is an $n_0 < n$ such that~$d$ is first added to $\mathcal V(A)$ in the $n_0$th iteration of the inner for loop.
    Then by Observation~\ref{obs:v-nplus1}, $d_i \in \mathcal V_{n_0 + 1}(A_i)$ for every $i \in [k]$.
    Then by~(IH), for every $i \in [k]$ and $n_0' \in \mathbb N$ with $n_0' > n_0$ it holds that $d_i \in \mathcal V_{n_0'}(A_i)$.
    Thus for every $n' \geq n$, by Observation~\ref{obs:v-nplus1}, $d \in \mathcal V_{n'+1}(A)$.
\end{proof}

\subsection{Termination of Algorithm~\ref{alg:mmonoid}}\label{app:vca-termination}

\begin{lemma}\label{lem:subtrees-in-mcv}
    For every $d \in \T_{R'}$, $n \in \mathbb N$, and $A \in N'$ the following holds: if $d \in \mathcal V_n(A)$, then for every $p \in \pos(d)$: $d|_p \in \mathcal V_n(\lhs(d(p)))$.
\end{lemma}

\begin{proof}
    Let $n \in \mathbb N$, $A \in N'$, and $d \in \mathcal V_n(A)$.
    We show that $d|_p \in \mathcal V_n(A)$ for each $p \in \pos(d)$ by structural induction on~$d$.
    For the induction base, let $d \in R'$; then $\pos(d) = \{ \varepsilon \}$ and $d|_\varepsilon = d \in \mathcal V_n(A)$.
    For the induction step, let $d = r(d_1, \dots, d_k)$ and $\reqrule$ with $k > 0$.
    We assume (IH) that for every $i \in [k]$, $n \in \mathbb N$, and $d_i \in (\T_{R'})_{A_i}$ the following holds: if $d_i \in \mathcal V_n(A_i)$, then for every $p \in \pos(d_i)$: $d_i|_p \in \mathcal V_n(\lhs(d_i(p)))$.
    Now if $d \in \mathcal V_n(A)$, then there is an $n_0 < n$ such that~$d$ is first added to $\mathcal V(A)$ in the $n_0$th iteration of the inner for loop.
    Then by Observation~\ref{obs:v-nplus1}, $d_i \in \mathcal V_{n_0 + 1}(A_i)$ for every $i \in [k]$.
    Furthermore, by Lemma~\ref{lem:mcv-monotone}, $d_i \in \mathcal V_n(A_i)$ for every $i \in [k]$.
    Now for every $p \in \pos(d)$ with $p = ip'$ for some $i \in [k]$, the statement of the lemma follows from~(IH), and for $p = \varepsilon$ it trivially holds.
\end{proof}

\begin{lemma}\label{lem:cutout-tree-growth}
    For every $d \in \T_{R'}$, $n, l \in \mathbb N$, and elementary cycle $w \in (R')^*$ the following holds:
    if there are $p, p' \in \pos(d)$ such that $p \prefof p'$, $\seq(d, p, p') = w$, and $d|_{p_{1 \isep |p| - l}} \in \mathcal V_n(\lhs(d|_{p_{1 \isep |p| - l}}(\varepsilon)))$, then $(d[d|_{p'}]_{p})|_{p_{1 \isep |p| - l}} \in \mathcal V_n(\lhs(d|_{p_{1 \isep |p| - l}}(\varepsilon)))$.
\end{lemma}

\begin{proof}
    Let $d \in \T_{R'}$, $n, l \in \mathbb N$, $w \in (R')^*$ be an elementary cycle, and $p, p' \in \pos(d)$ such that $p \prefof p'$, $\seq(d, p, p') = w$, and $d|_{p_{1 \isep |p| - l}} \in \mathcal V_n(\lhs(d|_{p_{1 \isep |p| - l}}(\varepsilon)))$
    We show that $(d[d|_{p'}]_{p})|_{p_{1 \isep |p| - l}} \in \mathcal V_n(\lhs(d|_{p_{1 \isep |p| - l}}(\varepsilon)))$ by induction on~$l$.
    For the induction base, let $l = 0$.
    We remark that $d(p) = d(p')$.
    Then, since $d|_{p_{1 \isep |p| - l}} = d|_p \in \mathcal V_n(\lhs(d(p)))$, we have that $(d[d|_{p'}]_{p})|_{p_{1 \isep |p| - l}} = d|_{p'} \in \mathcal V_n(\lhs(d(p)))$ by Lemma~\ref{lem:subtrees-in-mcv}.

    For the induction step, let $l \in \mathbb N$.
    We assume (IH) that for every $n \in \mathbb N$ the following holds:
    if $d|_{p_{1 \isep |p| - l}} \in \mathcal V_n(\lhs(d|_{p_{1 \isep |p| - l}}(\varepsilon)))$, then $(d[d|_{p'}]_{p})|_{p_{1 \isep |p| - l}} \in \mathcal V_n(\lhs(d|_{p_{1 \isep |p| - l}}(\varepsilon)))$.
    Now we distinguish two cases.
    \begin{enumerate}
        \item If $l \ge |p|$, then $d|_{p_{1 \isep |p| - l}} = d = d|_{p_{1 \isep |p| - (l+1)}}$.
            Thus $d|_{p_{1 \isep |p| - (l+1)}} \in \mathcal V_n(\lhs(d|_{p_{1 \isep |p| - (l+1)}}(\varepsilon)))$ implies $(d[d|_{p'}]_{p})|_{p_{1 \isep |p| - (l+1)}} \in \mathcal V_n(\lhs(d|_{p_{1 \isep |p| - (l+1)}}(\varepsilon)))$ by (IH).
        \item Otherwise, we let $n_0 \in \mathbb N$ such that $d|_{p_{1 \isep |p| - (l+1)}}$ is first added to $\mathcal V(\lhs(d|_{p_{1 \isep |p| - (l+1)}}(\varepsilon)))$ in the $n_0$th iteration of the inner for loop.
            If $d|_{p_{1 \isep |p| - (l+1)}} \in \mathcal V_n(\lhs(d|_{p_{1 \isep |p| - (l+1)}}(\varepsilon)))$, then $n_0 < n$.
            Then by Lemma~\ref{lem:subtrees-in-mcv}, $d|_{p_{1 \isep |p| - l}} \in \mathcal V_{n_0 + 1}(\lhs(d|_{p_{1 \isep |p| - l}}(\varepsilon)))$.
            Then by (IH), $(d[d|_{p'}]_{p})|_{p_{1 \isep |p| - l}} \in \mathcal V_{n_0 + 1}(\lhs(d|_{p_{1 \isep |p| - l}}(\varepsilon)))$ and by Observation~\ref{obs:v-nplus1}, $(d[d|_{p'}]_{p})|_{p_{1 \isep |p| - (l+1)}} \hspace{-1mm}\in \mathcal V_{n_0 + 1}(\lhs(d|_{p_{1 \isep |p| - (l+1)}}(\varepsilon)))$.
            Finally, by Lemma~\ref{lem:mcv-monotone}, $(d[d|_{p'}]_{p})|_{p_{1 \isep |p| - (l+1)}} \in \mathcal V_n(\lhs(d|_{p_{1 \isep |p| - (l+1)}}(\varepsilon)))$. \qedhere
    \end{enumerate}
\end{proof}

\begin{lemma}\label{lem:cut-cycle'}
    For every $d \in \T_{R'}, n \in \mathbb N$, $A \in N'$, and elementary cycle $w \in (R')^*$ the following holds:
    if $d \in \mathcal V_n(A)$ and there are $p, p' \in \pos(d)$ such that $p \prefof p'$ and $\seq(d, p, p') = w$, then $\cotrees(d, w) \subseteq \mathcal V_n(A)$.
\end{lemma}

\begin{proof}
    Let $d \in \T_{R'}$, $n \in \mathbb N$, $A \in N'$, and $w \in (R')^*$ such that there are $p, p' \in \pos(d)$ with $p \prefof p'$, $\seq(p, p') = w$, $w$ is an elementary cycle, and $d \in \mathcal V_n(A)$.
    We show that for every $l \in \mathbb N$ and $d' \in \T_{R'}$ with $d \nwtrans{l} d'$ it holds that $d' \in \mathcal V_n(A)$ by induction on~$l$.
    For the induction base, let $l = 0$.
    Then for every $d' \in \T_{R'}$ with $d \nwtrans{0} d'$ it holds that $d' = d$, and $d \in \mathcal V_n(A)$ by definition.

    For the induction step, let $l \in \mathbb N$.
    We assume (IH) that for every $d' \in \T_{R'}$, $n \in \mathbb N$, and $A \in N'$ with $d \nwtrans{l} d'$ the following holds:
    if $d \in \mathcal V_n(A)$, then $d' \in \mathcal V_n(A)$.
    Now let $d'' \in \T_{R'}$ such that $d \nwtrans{l+1} d''$.
    Then there is a $d' \in \T_{R'}$ such that $d \nwtrans{l} d'$ and $d' \wtrans d''$.
    Then, if $d \in \mathcal V_n(A)$, we have that $d' \in \mathcal V_n(A)$ by (IH).
    Moreover, there are $p, p' \in \pos(d)$ such that $\seq(d, p, p') = w$ and $d'[d'_{p'}]_p = d''$.
    Thus, by Lemma~\ref{lem:cutout-tree-growth}, $d'' \in \mathcal V_n(A)$.

    Now, as for every $l \in \mathbb N$ and $d' \in \T_{R'}$ with $d \nwtrans{l} d'$ it holds that $d \in \mathcal V_n(A)$ implies $d' \in \mathcal V_n(A)$, we obtain that for every $d' \in \T_{R'}$ with $d \wtrans^+ d'$ the same implication holds.
    Thus $d' \in \mathcal V_n(A)$ for every $d' \in \cotrees(d, w)$.
\end{proof}

\lemcutcycles*

\begin{proof}
    This is a consequence of Lemma~\ref{lem:cut-cycle'}.
\end{proof}

\begin{lemma}\label{lem:outside-tree}
    For every $n \in \mathbb N$ the following holds:
    if $\Delta_n(A) \cap \TRpc = \emptyset$, then $V_{n+1}(A) = V_n(A)$, where $\select_n = A$.
\end{lemma}

\begin{proof}
    Let $n \in \mathbb N$ and $\select_n = A$.
    Then $\Delta_n(A) = \emptyset$ or $\Delta_n \subseteq \T_{R'} \setminus \TRpc$.
    For every $d \in \Delta_n(A)$ we have that $d \in \mathcal V_{n+1}(A)$ and thus by, Lemma~\ref{lem:cut-cycles},
    \begin{align*}
        \cotrees(d) \cap \TRpc &\subseteq \mathcal V_{n+1}(A) \cap \TRpc \\
        &= \big( \mathcal V_n(A) \disunion \Delta_n(A) \big) \cap \TRpc
        \tag{Lemma~\ref{lem:mcv-monotone}} \\
        &= \big( \mathcal V_n(A) \cap \TRpc \big) \disunion \big( \underbrace{\Delta_n(A) \cap \TRpc}_{= \, \emptyset} \big)
        \tag{distributivity of~$\cap$ over~$\disunion$} \\
        &= \mathcal V_n(A) \cap \TRpc \\
        &\subseteq \mathcal V_n(A) \enspace.
    \end{align*}
    Furthermore, by Lemma~\ref{lem:mcv-monotone}, $\mathcal V_n(A) \subseteq \mathcal V_{n+1}(A)$, and since $\Delta_n(A) = \emptyset$ or $\Delta_n \subseteq \T_{R'} \setminus \TRpc$
    \begin{align*}
        V_{n+1}(A) = \bigoplus_{d \in \mathcal V_{n+1}(A)} \wtphom{d} &= \bigoplus_{d \in \mathcal V_n(A)} \wtphom{d} \oplus \bigoplus_{d \in \Delta_n(A)} \wtphom{d}
        \tag{Lemma~\ref{lem:v-is-bigsum}} \\
        &= \bigoplus_{d \in \mathcal V_n(A)} \wtphom{d}
        \tag{Theorem~\ref{thm:outside-trees-subsumed}} \\
        &= V_n(A) \enspace. \tag*{(Lemma~\ref{lem:v-is-bigsum}) \qedhere}
    \end{align*}
\end{proof}

\lemmcvgrowsonchange*

\begin{proof}
    Let $n \in \mathbb N$ and $A \in N'$.
    If $V_{n+1}(A) \not= V_n(A)$, then we obtain $\select_n = A$ from Observation~\ref{obs:v-nplus1}.
    Then by Lemma~\ref{lem:outside-tree}, $\Delta_n(A) \cap \TRpc \not= \emptyset$.
    Furthermore, by Lemma~\ref{lem:mcv-monotone}, $\mathcal V_n(A) \subseteq \mathcal V_{n+1}(A)$.
    Thus
    \begin{align*}
        \mathcal V_{n+1}(A) \cap \TRpc &= \big(\mathcal V_n(A) \disunion \Delta_n(A)\big) \cap \TRpc
        \\
        &= (\mathcal V_n(A) \cap \TRpc) \disunion \underbrace{\big(\Delta_n(A) \cap \TRpc\big)}_{\not= \, \emptyset}
        \\
        &\supset \mathcal V_n(A) \cap \TRpc \enspace. \qedhere
    \end{align*}
\end{proof}

\subsection{Correctness of Algorithm~\ref{alg:mmonoid}}\label{app:vca-correctness}

\lempassthrough*

\begin{proof}
    Let $n \in \mathbb N$, $d \in \TRpc$ of the form $r(d_1, \dots, d_k)$ with $r = \big(A \to \sigma(A_1, \dots, A_k)\big)$, $\welem k_1, \dots, k_k \in \walg K$, and $I \in [k]$ such that $\select_n = A$, for every $i \in [k] \setminus I$, $d_i \in \mathcal V_n(A_i)$, and for every $i \in I$, $V_n(A_i) = V_n(A_i) \oplus \welem k_i$.
    Then by Observation~\ref{obs:v-nplus1}
    \begin{align}\label{eq:through-outer}
        V_{n+1}(A) &= \bigoplus_{\substack{r' \in R': \\ r' = (B \to \sigma(B_1, \dots, B_{k'}))}} \wt'(r')\big(V_n(B_1), \dots, V_n(B_{k'})\big)
        \nonumber \\
        &= \bigoplus_{\substack{r' \in R' \setminus \{ r \}: \\ r' = (B \to \sigma(B_1, \dots, B_{k'}))}} \wt'(r')\big(V_n(B_1), \dots, V_n(B_{k'})\big) \oplus \wt'(r)\big(V_n(A_1), \dots, V_n(A_k)\big) \enspace.
    \end{align}
    Now for each $i \in [k]$, we let $\welem l_i \in \walg K$ be as in the statement of the lemma and define the set
    \[ S_i = \begin{cases}
        \mathcal V_n(A_i)
        &\text{if $i \in I$} \\
        \{ d_i \}
        &\text{otherwise.}
    \end{cases} \]
    Then by Lemma~\ref{lem:v-is-bigsum} and distributivity of $\wt'(r)$ over $\oplus$
    \begin{align}\label{eq:through-inner}
        &\wt'(r)\big(V_n(A_1), \dots, V_n(A_k)\big)
        \nonumber \\
        &= \bigoplus_{(d_1', \dots, d_k') \in \mathcal V_n(A_1) \times \dots \times \mathcal V_n(A_k)} \wt'(r)\big(\wtphom{d_1}, \dots, \wtphom{d_k}\big)
        \nonumber \\
        &= \begin{aligned}[t]
            &\bigoplus_{(d_1', \dots, d_k') \in \mathcal V_n(A_1) \times \dots \times \mathcal V_n(A_k) \setminus S_1 \times \dots \times S_k} \wt'(r)\big(\wtphom{d_1}, \dots, \wtphom{d_k}\big) \\
            &\oplus \bigoplus_{(d_1', \dots, d_k') \in S_1 \times \dots \times S_k} \wt'(r)\big(\wtphom{d_1}, \dots, \wtphom{d_k}\big) \enspace.
        \end{aligned}
    \end{align}
    Now
    \begin{align*}
        \bigoplus_{(d_1', \dots, d_k') \in S_1 \times \dots \times S_k} \wt'(r)\big(\wtphom{d_1}, \dots, \wtphom{d_k}\big) &= \wt'(r)(U_1, \dots, U_k)
        \intertext{where for every $i \in [k]$, if $i \in I$, then $U_i = V_n(A_i)$, else $U_i = \wtphom{d_i}$,}
        &= \wt'(r)(U_1', \dots, U_k')
        \intertext{where for every $i \in [k]$, if $i \in I$, then $U_i' = V_n(A_i) \oplus \welem k_i$, else $U_i' = \wtphom{d_i}$,}
        &= \wt'(r)(U_1, \dots, U_k) \oplus \wt'(r)(\welem l_1, \dots, \welem l_k) \enspace.
    \end{align*}

    Thus we obtain by Equation~\ref{eq:through-inner}
    \begin{align*}
        &\wt'(r)\big(V_n(A_1), \dots, V_n(A_k)\big)
        \\
        &= \begin{aligned}[t]
            &\bigoplus_{(d_1', \dots, d_k') \in \mathcal V_n(A_1) \times \dots \times \mathcal V_n(A_k) \setminus S_1 \times \dots \times S_k} \wt'(r)\big(\wtphom{d_1}, \dots, \wtphom{d_k}\big) \\
            &\oplus \wt'(r)(U_1, \dots, U_k) \oplus \wt'(r)(\welem l_1, \dots, \welem l_k)
        \end{aligned}
        \\
        &= \begin{aligned}[t]
            &\bigoplus_{(d_1', \dots, d_k') \in \mathcal V_n(A_1) \times \dots \times \mathcal V_n(A_k) \setminus S_1 \times \dots \times S_k} \wt'(r)\big(\wtphom{d_1}, \dots, \wtphom{d_k}\big) \\
            &\oplus \bigoplus_{(d_1', \dots, d_k') \in S_1 \times \dots \times S_k} \wt'(r)\big(\wtphom{d_1}, \dots, \wtphom{d_k}\big) \oplus \wt'(r)(\welem l_1, \dots, \welem l_k)
        \end{aligned}
        \\
        &= \wt'(r)\big(V_n(A_1), \dots, V_n(A_k)\big) \oplus \wt'(r)(\welem l_1, \dots, \welem l_k) \enspace.
    \end{align*}

    Finally, by Equation~\ref{eq:through-outer}, we obtain
    \begin{align*}
        V_{n+1}(A) &= \begin{aligned}[t]
            &\bigoplus_{\substack{r' \in R' \setminus \{ r \}: \\ r' = (B \to \sigma(B_1, \dots, B_{k'}))}} \wt'(r')\big(V_n(B_1), \dots, V_n(B_{k'})\big) \\
            &\oplus \wt'(r)\big(V_n(A_1), \dots, V_n(A_k)\big) \oplus \wt'(r)(\welem l_1, \dots, \welem l_k)
        \end{aligned}
        \\
        &= \bigoplus_{\substack{r' \in R': \\ r' = (B \to \sigma(B_1, \dots, B_{k'}))}} \wt'(r')\big(V_n(B_1), \dots, V_n(B_{k'})\big) \oplus \wt'(r)(\welem l_1, \dots, \welem l_k)
        \\
        &= V_{n+1}(A) \oplus \wt'(r)(\welem l_1, \dots, \welem l_k) \enspace. \tag*{(Observation~\ref{obs:v-nplus1}) \qedhere}
    \end{align*}
\end{proof}

\subsection{Termination and correctness of the M-monoid parsing algorithm}\label{app:mpa-properties}

In this subappendix we give a full proof of Lemma~\ref{lem:1and2-closed}.
For this, we need a few auxiliary definitions and lemmas.

Let $a$ be a syntactic object and $\overline G = \big((G, \alg L), \walg K, \wt\big)$ be a wRTG-LM with $G = (N, \Sigma, A_0, R)$ such that $\alg L$ is finitely decomposable and $\cnc(\overline G, a) = \big((G', \cfges), \walg K, \wt'\big)$ with $G' = (N', \Sigma', A_0', R')$.
We define a partial mapping $\psi: R' \pto\! R$ such that for each $r = \big([A,t,a_0] \to \langle x_1 \dots x_k \rangle([A_1,t_1,a_1], \dots, [A_k,t_k,a_k])\big)$ in $R'$, $\psi(r) = A \to t$.
Furthermore, we define the mapping $\psi': \T_{R'} \to \T_R$ such that for each $d \in \T_{R'}$, if $\lhs(d(\varepsilon)) = \big([A_0,a] \to \langle x_1 \rangle([A_0,t,a])\big)$ for some $t \in \T_{\Sigma'}(N)$, then $\psi'(d) = \overline{\psi}(d|_1)$, and otherwise $\psi'(d) = \overline{\psi}(d)$, where $\overline{\psi}$ is the $N$-sorted tree relabeling induced by $\psi$.
We also let $\psi(w) = \psi(w_1) \dots \psi(w_{|w|})$ for every $w \in (R')^*$.

\begin{lemma}\label{lem:psi-weight}
    For every $d \in \T_{R'}$ it holds that $\wt'(d)_{\walg K} = \wthom{\psi'(d)}$.
\end{lemma}

\begin{proof}
    Let $d \in \T_{R'}$.
    We let $d' = d|_1$ if there is a $t \in \T_{\Sigma'}(N)$ such that $\lhs(d(\varepsilon)) = \big([A_0,a] \to \langle x_1 \rangle([A_0,t,a])\big)$ and otherwise, $d' = d$.
    We note that if $d' = d|_1$, then $\wt(d(\varepsilon)) = \id$ by definition of $\cnc$ and thus $\wthom{d} = \wthom{d'}$ in both cases.
    Now, by definition of $\cnc$, for every $p \in \pos(d')$ it holds that $\wt'(d'(p)) = \wt(\overline{\psi}(d')(p))$ and hence $\wt'(d')_{\walg K} = \wthom{\overline{\psi}(d')} = \wthom{\psi'(d)}$.
\end{proof}

\begin{lemma}\label{lem:psi-cutout}
    For every $c \in \mathbb N$, $d \in \T_{R'}$ and elementary cycle $w \in (R')^*$ such that there is a leaf $p \in \pos(d)$ which is $(c+1,w)$-cyclic the following holds:
    $\psi'(\cotrees(d, w)) = \cotrees(\psi'(d), \psi(w))$.
\end{lemma}

\begin{proof}
    Let $c \in \mathbb N$, $d \in \T_{R'}$ and $w \in (R')^*$ be an elementary cycle such that there is a leaf $p \in \pos(d)$ which is $(c+1,w)$-cyclic.
    Then
    \begin{align*}
        d' \in \psi'(\cotrees(d, w)) &\iff \exists d'' \in \T_{R'}: d \wtrans^+ d'' \land \psi'(d'') = d' \\
        &\iff \exists d'' \in \T_{R'}, p, p' \in \pos(d): \seq(d, p, p') = w \land d[d|_{p'}]_p = d'' \land \psi'(d'') = d' \\
        &\iff \exists p, p' \in \pos(\psi'(d)): \seq(\psi'(d), p, p') = \psi(w) \land \psi'(d)[\psi'(d)|_{p'}]_p = d' \\
        &\iff \psi'(d) \wtrans^+ d' \\
        &\iff d' \in \cotrees(\psi'(d), \psi(w)) \enspace. \qedhere
    \end{align*}
\end{proof}

\begin{lemma}\label{lem:psi-preserves-cycles}
    For every $d \in \T_{R'}$ and $w \in (R')^*$ the following holds:
    if there is a leaf $p \in \pos(d)$ which is $(c+1,w)$-cyclic, then there are a $c' \in \mathbb N$ and a leaf $p' \in \pos(\psi'(d))$ which is $(c',\psi(w))$-cyclic and $c' > c + 1$.
\end{lemma}

\begin{proof}
    Let $d \in \T_{R'}$, $w \in (R')^*$ such that there is a leaf $p \in \pos(d)$ which is $(c+1,w)$-cyclic.
    Then there are $v_0, \dots, v_{c+1} \in (R')^*$ such that $\seq(d, p) = v_0 w v_1 \dots w v_{c+1}$ and for every $i \in [0, c+1]$, $w$ is not a substring of $v_i$.
    If $\lhs(d(\varepsilon)) = \big([A_0,a] \to \langle x_1 \rangle([A_0,t,a])\big)$ for some $t \in \T_{\Sigma'}(N)$, then we let $p' = 1p$ and $p_0 = 1$ and otherwise, we let $p' = p$ and $p_0 = \varepsilon$.
    Then $\seq(\psi'(d), p_0, p') = \psi(v_0) \psi(w) \psi(v_1) \dots \psi(w) \psi(v_{c+1})$.
    Furthermore, for every $i \in [0, c+1]$, $\psi(v_i) = v_{i,0} \psi(w) v_{i,1} \dots \psi(w) v_{i,c_i}$ with $c_i \in \mathbb N$ and for every $j \in [0, c_i]$, $v_{i,j} \in R^*$ and $\psi(w)$ is not a substring of $v_{i,j}$.
    Thus $\psi'(d)$ is $(c', \psi(w))$-cyclic for $c' = c + 1 + \sum_{i \in [0, c]} c_i$.
\end{proof}

\lemclosedpreserved*

\begin{proof}
    Let $a$ be a syntactic object and $\overline G = \big((G, \alg L), \walg K, \wt\big)$ be a wRTG-LM with $G = (N, \Sigma, A_0, R)$ such that $\alg L$ is finitely decomposable and $\cnc(\overline G, a) = \big((G', \cfges), \walg K, \wt'\big)$ with $G' = (N', \Sigma', A_0', R')$.
    If there is a $c \in \mathbb N$ such that $\overline G$ is $c$-closed, then for every $d \in \T_{R'}$ and elementary cycle $w \in (R')^*$ such that there is a leaf $p \in \pos(d)$ which is $(c+1,w)$-cyclic
    \begin{align*}
        \wt'(d)_{\walg K} \oplus \bigoplus_{d' \in \cotrees(d, w)} \wt'(d')_{\walg K} &= \wthom{\psi'(d)} \oplus \bigoplus_{d' \in \cotrees(d, w)} \wthom{\psi'(d')} \tag{Lemma~\ref{lem:psi-weight}} \\
        &= \wthom{\psi'(d)} \oplus \bigoplus_{d' \in \cotrees(\psi'(d), \psi(w))} \wthom{d'} \tag{Lemma~\ref{lem:psi-cutout}} \\
        &= \bigoplus_{d' \in \cotrees(\psi'(d), \psi(w))} \wthom{d'} \tag{Lemma~\ref{lem:closed-bigger-trees'}} \\
        &= \bigoplus_{d' \in \cotrees(d, w)} \wthom{\psi'(d')} \tag{Lemma~\ref{lem:psi-cutout}} \\
        &= \bigoplus_{d' \in \cotrees(d, w)} \wt'(d)_{\walg K} \enspace. \tag{Lemma~\ref{lem:psi-weight}}
    \end{align*}
    We note that Lemma~\ref{lem:closed-bigger-trees'} can be applied due to Lemma~\ref{lem:psi-preserves-cycles}.

    If $\overline G$ is nonlooping and the weight algebra of $\overline G$ is in $\wclass{\dcomp} \cap \wclass{dist}$, then, by Lemma~\ref{lem:no-loops-cnc-acyclic}, $\cnc(\overline G, a)$ is in $\wlmclass{\gclass{acyc}, \wclass{\dcomp} \cap \wclass{dist}}$ for every syntactic object $a$.
    Then, by Lemma~\ref{lem:acyclic-closed}, it is also closed.
\end{proof}

\subsection{Application scenarios employ closed wRTG-LMs}\label{app:applications}

This subappendix contains the full proofs of Theorem~\ref{thm:applications} and Theorem~\ref{thm:applications2}.

\begin{lemma}[restate={[name={}]lemfiniteclosed}]\label{lem:finite-closed}
    Every wRTG-LM in $\wlmclass{\gclass{all}, \wclass{\finidpo}}$ is closed.
\end{lemma}

\begin{proof}
    Let $\overline G = \big((G, \alg L), \walg K, \wt\big)$ in $\wlmclass{\gclass{all}, \wclass{\finidpo}}$ with $G = (N, \Sigma, A_0, R)$.
    Then $\walg K$ is finite and idempotent and there is a partial order $(\walg K, \preceq)$ such that for every $k \in \mathbb N$, $\omega \in \Omega_k$, and $\welem k_1, \dots, \welem k_k \in \walg K$: $\maxord \{ \welem k_1, \dots, \welem k_k \} \preceq \omega(\welem k_1, \dots, \welem k_k)$.
    We show that $\overline G$ is $|\walg K|$-closed.
    For this, let $n = |\walg K| + 1$, $G = (N, \Sigma, A_0, R)$, $d \in \T_R$, $p \in \pos(d)$ be a leaf, and $w \in R^*$ be an elementary cycle such that~$p$ is $(n, w)$-cyclic.
    Then there are $v_0, \dots, v_n \in R^*$ such that for every $i \in [0,n]$, $w$ is not a substring of $v_i$ and $\seq(d,p) = v_0 w v_1 \dots w v_n$.
    We let $w = r_1 \dots r_m$ with $r_i \in R$ for every $i \in [m]$ and $r_m = (A \to t)$ with $\yield_N(t) = A_1 \dots A_k$ and $A_i \in N$ for every $i \in [k]$.

    We consider the set $D = \{ d|_{p_{1 \isep j}} \mid i \in [n], j = \sum_{l=0}^{i-1} |v_l| + i \cdot |w| \}$.
    Since $|D| = n$, there are $d_1, d_2 \in D$ such that $d_1 \not= d_2$ and $\wthom{d_1} = \wthom{d_2}$.
    Let $i, j \in [n]$ and $i', j' \in \mathbb N$ such that $i' = \sum_{l=0}^{i-1} |v_l| + i \cdot |w|$, $j' = \sum_{l=0}^{j-1} |v_l| + j \cdot |w|$, $d_1 = d|_{p_{1 \isep i'}}$, and $d_2 = d|_{p_{1 \isep j'}}$.
    Without loss of generality, we assume that $i < j$.
    Let $q = j - |w| + 1$.
    By Lemma~\ref{lem:po-chains}, $d|_{p_{1 \isep q}} = d|_{p_{1 \isep j'}}$.
    We let $s \in [k]$ such that $p_q = s$.
    Then
    \begin{align*}
        \wthom{d} &= \wt(d)[x_{s, A_s}]_{p_{1 \isep q}} \left( \wthom{d|_{p_{1 \isep q}}} \right)_{\walg K} \\
        &= \wt(d)[x_{s, A_s}]_{p_{1 \isep q}} \left( \wthom{d_2} \right)_{\walg K} \\
        &= \wthom{d[d_2]_{p_{1 \isep q}}} \enspace. \tag{$d|_{p_{1 \isep q'}}(\varepsilon) = d_2(\varepsilon)$}
    \end{align*}
    Then, as $\walg K$ is idempotent,
    \[ \wthom{d} \oplus \wthom{d[d_2]_{p_{1 \isep q}}} = \wthom{d[d_2]_{p_{1 \isep q}}} \enspace. \]
    Finally, as $d[d_2]_{p_{1 \isep q}} \in \cotrees(d)$, we have that
    \[ \wthom{d} \oplus \bigoplus_{d' \in \cotrees(d, w)} \wthom{d'} = \bigoplus_{d' \in \cotrees(d, w)} \wthom{d'} \]
    and thus $\overline G$ is $|\walg K|$-closed.
\end{proof}

\begin{lemma}[restate={[name={}]lemsupclosed}]\label{lem:sup-closed}
    Every wRTG-LM in $\wlmclass{\gclass{all}, \wclass{sup}}$ is closed.
\end{lemma}

\begin{proof}
    Let $\overline G = \big((G, \alg L), \walg K, \wt\big)$ in $\wlmclass{\gclass{all}, \wclass{sup}}$.
    We show that $\overline G$ is $0$-closed.
    For this, let $G = (N, \Sigma, A_0, R)$, $d \in \T_R$, $p \in \pos(d)$ be a leaf, and $w \in R^*$ be an elementary cycle such that~$p$ is $(1, w)$-cyclic.
    Then there are $v_0, v_1 \in R^*$ such that $w$ is not a substring of $v_0$ or $v_1$ and $\seq(d, p) = v_0 w v_1$.
    We let $v_0 = r_1 \dots r_m$, $r_m = (A \to t)$ with $\yield_N(t) = A_1 \dots A_k$ and $A_i \in N$ for every $i \in [k]$, and $d' = d[d|_{p_{1 \isep m + |w|}}]_{p_{1 \isep m + 1}}$.
    We will show that $\wthom{d} \oplus \wthom{d'} = \wthom{d'}$.

    First, we let $\wthom{d|_{p_{1 \isep m + |w|}}} = \welem k$ and $\wthom{d|_{p_{1 \isep m + 1}}} = \welem k'$ (we are not interested in the particular value).
    As $\walg K$ is superior, we have that $\welem k \preceq_\oplus \welem k'$.
    Thus $\wthom{d|_{p_{1 \isep m + 1}}} \oplus \wthom{d|_{p_{1 \isep m + |w|}}} = \wthom{d|_{p_{1 \isep m + |w|}}}$.
    Therefore
    \begin{align*}
        &\wthom{d} \oplus \wthom{d'} \\
        &= (d[x_{s, A_s}]_{p_{1 \isep m + 1}})\left(\wthom{d|_{p_{1 \isep m + 1}}} \oplus \wthom{d|_{p_{1 \isep m + |w|}}}\right)_{\walg K}
        \tag{distributivity of $\Omega$ over $\oplus$} \\
        &= (d[x_{s, A_s}]_{p_{1 \isep m + 1}})\left(\wthom{d|_{p_{1 \isep m + |w|}}}\right)_{\walg K} \\
        &= \wthom{d'} \enspace.
    \end{align*}

    Now, as $d' \in \cotrees(d, w)$, this entails that
    \[ \wthom{d} \oplus \bigoplus_{d'' \in \cotrees(d, w)} \wthom{d''} = \bigoplus_{d'' \in \cotrees(d, w)} \wthom{d''} \]
    and thus~$\overline G$ is 0-closed.
\end{proof}

\begin{lemma}[restate={[name={}]lemacyclicclosed}]\label{lem:acyclic-closed}
    Every wRTG-LM in $\wlmclass{\gclass{acyc}, \wclass{\dcomp} \cap \wclass{dist}}$ is closed.
\end{lemma}

\begin{proof}
    Let $\big((G, \alg L), \walg K, \wt\big) \in \wlmclass{\gclass{acyc}, \wclass{\dcomp} \cap \wclass{dist}}$ with $G = (N, \Sigma, A_0, R)$.
    Since every $d \in \T_R$ is acyclic we have that $\T_R = \T_R^{(0)}$.
    Thus, by Definition of closed, $\overline G$ is $0$-closed.
\end{proof}

\thmapplications*

\begin{proof}
    This is a consequence of Lemmas~\ref{lem:finite-closed},~\ref{lem:sup-closed}, and~\ref{lem:acyclic-closed}.
\end{proof}

\begin{lemma}[restate={[name={}]lembdclosed}]\label{lem:bd-closed}
    Every wRTG-LM in $\wlmclass[<1]{\gclass{all}, \walg{BD}}$ is closed.
\end{lemma}

\begin{proof}
    Let $\overline G = \big((G, \alg L), \bd, \wt\big)$ in $\wlmclass[<1]{\gclass{all}, \walg{BD}}$.
    We show that $\overline G$ is $0$-closed.
    For this, let $G = (N, \Sigma, A_0, R)$, $d \in \T_R$, $p \in \pos(d)$ be a leaf, and $w \in R^*$ be an elementary cycle such that~$p$ is $(1, w)$-cyclic.
    Then there are $v_0, v_1 \in R^*$ such that $w$ is not a substring of $v_0$ or $v_1$ and $\seq(d, p) = v_0 w v_1$.
    We let $v_0 = r_1 \dots r_m$, $r_m = (A \to t)$ with $\yield_N(t) = A_1 \dots A_k$ and $A_i \in N$ for every $i \in [k]$, and $d' = d[d|_{p_{1 \isep m + |w|}}]_{p_{1 \isep m + 1}}$.
    We will show that $\maxv\bigl(\wthom[\bd]{d}, \wthom[\bd]{d'}\bigr) = \wthom[\bd]{d'}$.

    First, we let $\wthom[\bd]{d|_{p_{1 \isep m + |w|}}} = (q, D)$.
    Then $\wthom[\bd]{d|_{p_{1 \isep m + 1}}} = (q', D')$, where~$q'$ is a product of~$q$ and other elements from~$\mathbb R_0^1$ and $D' \in \mathcal P(\T_R)$ (we are not interested in the particular values).
    Since by definition of $\wlmclass[<1]{\gclass{all}, \walg{BD}}$ each factor of that product is less than~$1$, we have that $q' < q$ by monotonicity of~$\cdot$ in $\rzo$.
    Thus $\maxv\Bigl(\wthom[\bd]{d|_{p_{1 \isep m + 1}}}, \wthom[\bd]{d|_{p_{1 \isep m + |w|}}}\Bigr) = \wthom[\bd]{d|_{p_{1 \isep m + |w|}}}$.
    Therefore
    \begin{align*}
        &\maxv\bigl(\wthom[\bd]{d}, \wthom[\bd]{d'}\bigr) \\
        &= (d[x_{s, A_s}]_{p_{1 \isep m + 1}})\left(\maxv\Bigl(\wthom[\bd]{d|_{p_{1 \isep m + 1}}}, \wthom[\bd]{d|_{p_{1 \isep m + |w|}}}\Bigr)\right)_{\bd}
        \tag{distributivity of $\Omegav$ over $\oplus$} \\
        &= (d[x_{s, A_s}]_{p_{1 \isep m + 1}})\left(\wthom[\bd]{d|_{p_{1 \isep m + |w|}}}\right)_{\bd} \\
        &= \wthom[\bd]{d'} \enspace.
    \end{align*}

    Now, as $d' \in \cotrees(d, w)$, this entails that
    \[ \maxv\bigl(\wthom[\bd]{d}, \maxv_{d'' \in \cotrees(d, w)} \wthom[\bd]{d''}\bigr) = \maxv_{d'' \in \cotrees(d, w)} \wthom[\bd]{d''} \]
    and thus~$\overline G$ is 0-closed.
\end{proof}

\begin{lemma}[restate={[name={}]lemnbestclosed}]\label{lem:nbest-closed}
    Every wRTG-LM in $\wlmclass{\gclass{all}, \nbest}$ is closed.
\end{lemma}

\begin{proof}
    Let $n \in \mathbb N$ and $\overline G = \big((G, \alg L), \nbest, \wt\big)$ in $\wlmclass{\gclass{all}, \nbest}$.
    We show that $\overline G$ is $(n-1)$-closed.
    For this, let $G = (N, \Sigma, A_0, R)$, $d \in \T_R$, $p \in \pos(d)$ be a leaf, and $w \in R^*$ be an elementary cycle such that~$p$ is $(n, w)$-cyclic.
    Then there are $v_0, \dots, v_n \in R^*$ such that for every $i \in [0,n]$, $w$ is not a substring of $v_i$ and $\seq(d, p) = v_0 w v_1 \dots w v_n$.
    We let $v_0 = r_1 \dots r_m$, $r_m = (A \to t)$ with $\yield_N(t) = A_1, \dots, A_k$ and $A_i \in N$ for every $i \in [k]$, $ s \in [k]$ such that $p_{m+1} = s$, $m'_n = 1$, $d'_n = d$, and for every $i \in [n]$
    \begin{align*}
        m'_{i-1} &= m'_i + |v_{n-i-1}| + |w|
        \\
        s'_i &\in \mathbb N \ \text{such that} \ p_{m'_{i-1} + 1} = s'_i \\
        d'_{i-1} &= d'_i[d|_{p_{1 \isep m'_{i-1} + |w|}}]_{p_{m'_n \isep m'_n + |v_0|} \dots p_{m'_{i-1} \isep m'_{i-1} + |v_{n-i}|}} \enspace.
    \end{align*}
    We will show that
    \[ \wthom[\nbest]{d} \oplus \bigoplus_{i = 0}^{n-1} \wthom[\nbest]{d'_i} = \bigoplus_{i = 0}^{n-1} \wthom[\nbest]{d'_i} \enspace. \]

    First, we define $d_i'' = d_i'|_{p_{1 \isep m+1}}$ for every $i \in [0,n-1]$.
    Then we let $w = r_1 \dots r_l$ and for every $i \in [l]$, we let $\wt(r_i) = \mulnkki{i}$ with $k_i = \rk(r_i)$ and $\welem k_i \in \nbest$.
    Then there are $d_1, \dots, d_{k_1} \in \T_R$ such that $d|_{p_{m+1}} = r_1(d_1, \dots, d_{k_1})$.
    Thus
    \[
        \wthom[\nbest]{d|_{p_{1 \isep m+1}}} = \mulnkki{1}\left(\wthom[\nbest]{d_1}, \dots, \wthom[\nbest]{d_k}\right)
    \]
    and by recursively applying $\wt$ to $d(p_{2 \isep i})$ for each $i \in [m + 2, |p| - |v_n|]$
    \begin{align*}
        &= \mulnkki{1}\bigg(\begin{aligned}[t]
            &\wthom[\nbest]{d_1}, \dots, \wthom[\nbest]{d_{s-1}}, \\
            &\mulnkki{2}\bigg(\begin{aligned}[t]
                &\wthom[\nbest]{(d_s)|_1}, \dots, \wthom[\nbest]{d_{s'-1}}, \\
                &\dots \\
                &\mulnkki{l}\bigg(\begin{aligned}[t]
                    &\wthom[\nbest]{d|_{p_{1 \isep m + l} 1}}, \dots, \wthom[\nbest]{d_{p_{1 \isep m + l} s''-1}}, \\
                    &\dots \\
                    &\mulnkki{l}\Big(\wthom[\nbest]{d|_{p_{1 \isep |p| - |v_n|} 1}}, \dots, \wthom[\nbest]{d_{p_{1 \isep |p| - |v_n|} k}}\Big), \\
                    &\dots, \\
                    &\wthom[\nbest]{d_{p_{1 \isep m + l} s''+1}}, \dots, \wthom[\nbest]{d|_{p_{1 \isep m + l} k}} \Big),
                \end{aligned} \\
                &\dots, \\
                &\wthom[\nbest]{d_{s'+1}}, \dots, \wthom[\nbest]{(d_s)|_k}\bigg)
            \end{aligned} \\
            &\wthom[\nbest]{d_{s+1}}, \dots, \wthom[\nbest]{d_k}\bigg) \enspace,
        \end{aligned}
        \intertext{where $p_{m+2} = s'$ and $p_{m + l + 1} = s''$,}
        &= \takenbest\Big(\begin{aligned}[t]
            &\nbweighti{1} \cdot_n \wthom[\nbest]{d_1} \cdot_n \ldots \cdot_n \wthom[\nbest]{d_{s-1}} \\
            &\cdot_n \begin{aligned}[t]
                &\nbweighti{2} \cdot_n \wthom[\nbest]{(d_s)|_1} \cdot_n \ldots \cdot_n \wthom[\nbest]{d_{s'-1}} \\
                &\cdot_n \ldots \\
                &\cdot_n \begin{aligned}[t]
                    &\nbweighti{l} \cdot_n \wthom[\nbest]{d|_{p_{1 \isep m + l} 1}} \cdot_n \ldots \cdot_n \wthom[\nbest]{d_{p_{1 \isep m + l} s''-1}} \\
                    &\cdot_n \ldots \\
                    &\nbweighti{l} \cdot_n \wthom[\nbest]{d|_{p_{1 \isep |p| - |v_n|} 1}} \cdot_n \ldots \cdot_n \wthom[\nbest]{d_{p_{1 \isep |p| - |v_n|} k}} \\
                    &\cdot_n \ldots \\
                    &\cdot_n \wthom[\nbest]{d_{p_{1 \isep m + l} s''+1}} \cdot_n \ldots \cdot_n \wthom[\nbest]{d|_{p_{1 \isep m + l} k}}
                \end{aligned} \\
                &\cdot_n \ldots, \\
                &\cdot_n \wthom[\nbest]{d_{s'+1}} \cdot_n \ldots \cdot_n \wthom[\nbest]{(d_s)|_k}
            \end{aligned} \\
            &\cdot_n \wthom[\nbest]{d_{s+1}} \cdot_n \ldots \cdot_n \wthom[\nbest]{d_k}\bigg) \enspace,
        \end{aligned}
        \intertext{where we have skipped the inner applications of $\takenbest$ for readability. Thus, by commutativity of $\cdot_n$, there is a $\welem k \in \walg K$ such that we continue}
        &= \takenbest\Big((\nbweighti{1} \cdot_n \ldots \cdot_n \nbweighti{l})^n \cdot_n \welem k\Big) \enspace.
    \end{align*}
    We let $\wthom[\nbest]{d|_{p_{1 \isep m+1}}} = (a_1, \dots, a_n)$.
    Then for every $i \in [n]$, there is an $i' \in [n]$ such that
    \[ a_i = \takenbest\Big((\nbweighti{1} \cdot_n \ldots \cdot_n \nbweighti{l})^n \cdot_n \welem k\Big)_{i'} \enspace. \]
    Now as for every $i \in [n-1]$, $(\welem k_1 \cdot \ldots \cdot \welem k_l)^i \ge (\welem k_1 \cdot \ldots \cdot \welem k_l)^n$ and $(\welem k_1 \cdot \ldots \cdot \welem k_l)^0 = 1 \ge (\welem k_1 \cdot \ldots \cdot \welem k_l)^n$, for every $i \in [n]$ and $i' \in [0, n-1]$ by monotonicity of~$\cdot$ we have that $a_i \le (\wthom[\nbest]{d''_{i'}})_1$.
    Thus
    \[ \wthom[\nbest]{d|_{p_{1 \isep m+1}}} \oplus \bigoplus_{i=0}^{n-1} \wthom[\nbest]{d''_i} = \bigoplus_{i=0}^{n-1} \wthom[\nbest]{d''_i} \enspace. \]
    Therefore
    \begin{align*}
        &\wthom[\nbest]{d} \oplus \bigoplus_{i=0}^{n-1} \wthom[\nbest]{d'_i} \\
        &= d[x_{s,A_s}]_{p_{1 \isep m+1}}\left(\wthom[\nbest]{d|_{p_{1 \isep m+1}}} \oplus \bigoplus_{i=0}^{n-1} \wthom[\nbest]{d''_i}\right)_{\walg K}
        \tag{distributivity of $\Omegav$ over $\oplus$} \\
        &= d[x_{s,A_s}]_{p_{1 \isep m+1}}\left(\bigoplus_{i=0}^{n-1} \wthom[\nbest]{d''_i}\right)_{\walg K} \\
        &= \bigoplus_{i=0}^{n-1} \wthom[\nbest]{d'_i} \enspace.
    \end{align*}

    Now, as $d'_i \in \cotrees(d, w)$ for every $i \in [0, n-1]$, this entails that
    \[ \wthom[\nbest]{d} \oplus \bigoplus_{d' \in \cotrees(d, w)} \wthom[\nbest]{d'} = \bigoplus_{d' \in \cotrees(d, w)} \wthom[\nbest]{d'} \]
    and thus~$\overline G$ is $n-1$-closed.
\end{proof}

\begin{lemma}[restate={[name={}]lemintersectionclosed}]\label{lem:intersection-closed}
    Every wRTG-LM in $\wlmclass{\gclass{all}, \wclass{int}}$ is closed.
\end{lemma}

\begin{proof}
    Let $\overline G = \big((G, \alg L), (\walg K, \cup, \emptyset, \omega), \wt\big)$ in $\wlmclass{\gclass{all}, \wclass{int}}$.
    Clearly, $(\walg K, \subseteq)$ is a partial order.
    By definition of~$\Omega$, for every $k \in \mathbb N$, $\omega \in \Omega_k$, and $\welem k_1, \dots, \welem k_k \in \walg K$ it holds that $\welem k_i \subseteq \omega(\welem k_1, \dots, \welem k_k)$ for each $i \in [k]$.
    Thus and since $\walg K$ is finite, idempotent and distributive, by Lemma~\ref{lem:finite-closed}, $\overline G$ is closed.
\end{proof}

\begingroup\setlength\emergencystretch{10pt}
\thmapplicationsp*
\endgroup

\begin{proof}
    This is a consequence of Lemmas~\ref{lem:bd-closed},~\ref{lem:nbest-closed}, and~\ref{lem:intersection-closed}.
\end{proof}

\subsection{Restriction of best derivation M-monoid is necessary}\label{app:bd-restriction}

\begin{lemma}\label{lem:bd-restriction}
    There is a wRTG-LM in $\wlmclass{\gclass{all}, \walg{BD}} \setminus \wlmclass[<1]{\gclass{all}, \walg{BD}}$ which is not closed.
\end{lemma}

\begin{example}
    Let $\overline G = \big((G, \cfges), \walg K, \wt\big) \in \wlmclass{\gclass{all}, \mathbb{BD}}$ with $G = (N, \Sigma, A_0, R)$ such that there are $r, r' \in R$ with $r = \big(A \to \langle x_1 \rangle(A)\big)$, $r' = (A \to \langle \varepsilon \rangle)$, $\wt(r) = \tc{1}{r}$, and $\wt(r') = \tc{p}{r'}$ for some $p \in \rzo$.
    We let
    \[ d^{(c)} = \underbrace{r( \dots r(}_{\text{$c$ times}} r' \underbrace{) \dots )}_{\text{$c$ times}} \enspace. \]
    We show shat for every $c \in \mathbb N$,
    \[ \wthom{d^{(c)}} \oplus \bigoplus_{d' \in \cotrees(d^{(c)})} \wthom{d'} \not= \bigoplus_{d' \in \cotrees(d^{(c)})} \wthom{d'} \]
    by an indirect proof.
    Assume that there is a $c \in \mathbb N$ such that
    \[ \wthom{d^{(c)}} \oplus \bigoplus_{d' \in \cotrees(d^{(c)})} \wthom{d'} = \bigoplus_{d' \in \cotrees(d^{(c)})} \wthom{d'} \enspace. \]
    Then
    \begin{align*}
        \wthom{d^{(c)}} \oplus \bigoplus_{d' \in \cotrees(d^{(c)})} \wthom{d'} &= (p, \{ d^{(c)} \}) \oplus \bigoplus_{\substack{c' \in \mathbb N: \\ c' < c}} (p, \{ d^{(c')} \})
        \tag{as $\oplus$ is idempotent} \\
        &= \bigoplus_{\substack{c' \in \mathbb N: \\ c' \le c}} (p, \{ d^{(c')} \}) \\
        &\not= \bigoplus_{\substack{c' \in \mathbb N: \\ c' < c}} (p, \{ d^{(c')} \}) = \bigoplus_{d' \in \cotrees(d^{(c)})} \wthom{d'} \enspace,
    \end{align*}
    which contradicts our assumption.
    Now, since for every $c \in \mathbb N$ we have that the leaf $1^c \in \pos(d)$ is $(\lfloor c / 2 \rfloor, rr)$-cyclic, it follows that~$\overline G$ is not $(\lfloor c / 2 \rfloor)$-closed.
    Thus~$\overline G$ is not closed.
\end{example}
\end{document}